\documentclass[12pt]{article}
\usepackage[utf8]{inputenc}
\usepackage{amsfonts,amsmath,amsthm,amssymb,mathrsfs,graphicx}
\usepackage{tabularx,lipsum,ltablex,array,makecell,booktabs,caption,multirow,float,wasysym}
\usepackage{fnpct} 
\usepackage[hang,flushmargin]{footmisc} 
\usepackage[letterpaper, margin=1.0 in]{geometry}
\usepackage{verbatim}
\usepackage{setspace}

\usepackage[american]{babel}
\usepackage[longnamesfirst,authoryear]{natbib}
\cleardoublepage
\normalbaselines 

\usepackage{hyperref,url}
\hypersetup{colorlinks,allcolors=blue}
\usepackage[capitalise]{cleveref}

\graphicspath{{./figures/}}

\newtheorem{theorem}{Theorem}[section]

\newtheorem{lemma}{Lemma}[section]
\theoremstyle{definition}

\theoremstyle{definition}
\newtheorem{remark}{Remark}[section]

\newtheorem*{notation}{Notation}

\numberwithin{equation}{section}
\theoremstyle{definition}

\newcommand{\Keywords}[1]{\par\noindent {\small{\em Keywords\/}: #1}} 
\newcommand{\JELclass}[1]{\par\noindent {\small{\em JEL classification\/}: #1}} 

\begin{document}
\title{Semiparametric Discrete Choice Models for Bundles}

\author{{Fu Ouyang}\footnote{Corresponding author. School of Economics, The University of Queensland, St Lucia, QLD 4072, Australia. Tel: +61 451 478 622. Email: \href{f.ouyang@uq.edu.au}{f.ouyang@uq.edu.au}.} \\ {\normalsize University of Queensland} \and {Thomas Tao Yang}\footnote{Research School of Economics, The Australian National University, Canberra, ACT 0200, Australia. Email: \href{tao.yang@anu.edu.au}{tao.yang@anu.edu.au}.} \\ {\normalsize Australian National University}}
\maketitle

\begin{abstract}
\noindent We propose two approaches to estimate semiparametric discrete choice models for bundles. Our first approach is a kernel-weighted rank estimator based on a matching-based identification strategy. We establish its complete asymptotic properties and prove the validity of the nonparametric bootstrap for inference. We then introduce a new multi-index least absolute deviations (LAD) estimator as an alternative, of which the main advantage is its capacity to estimate preference parameters on both alternative- and agent-specific regressors. Both methods can account for arbitrary correlation in disturbances across choices, with the former also allowing for interpersonal heteroskedasticity. We also demonstrate that the identification strategy underlying these procedures can be extended naturally to panel data settings, producing an analogous localized maximum score estimator and a LAD estimator for estimating bundle choice models with fixed effects. We derive the limiting distribution of the former and verify the validity of the numerical bootstrap as an inference tool. All our proposed methods can be applied to general multi-index models. Monte Carlo experiments show that they perform well in finite samples.
\par \bigskip
\JELclass{C13, C14, C35.}
\par \medskip
\Keywords{Multi-index models; Rank correlation; Maximum score; Least absolute deviations; Panel data.}
\par \vspace{2cm} 
\end{abstract}

\onehalfspacing
\section{Introduction}

\label{Introduction} 
In many circumstances, consumers purchase bundles of
goods (e.g., both chips and salsa) or services (e.g., a combination of
mobile phone, home internet, and cable TV plans), instead of a single good or
service. The literature on bundle choice is less well developed than that on
ordinary multinomial choice models. An important but less extensively studied
empirical question in industrial organization and marketing research relates
to complementary or substitutive effects that may explain the bundle choice
behavior of consumers. We refer readers to \shortcites{berry2014structural}\cite{berry2014structural} for a review
of this literature. In an empirical study, \cite{Gentzkow2007} estimated a
parametric bundle choice model analyzing demand for print and online
newspapers. Similarly, using aggregate data, \cite{Fan2013} examined
ownership consolidation in the newspaper market, where households on the
demand side may purchase two newspapers as a bundle. A substantial
literature focuses on bundle choice behavior, analysis of product demand,
pricing strategy, customer subscriptions, and brand collaboration, among
others. Examples include \cite{manski1980empirical}, \cite{train1987demand},
\cite{hendel1999estimating}, \cite{chung2003general}, \cite{dube2004multiple}%
, \cite{nevo2005academic}, \cite{augereau2006coordination}, \cite%
{song2006measuring}, \cite{foubert2007shopper}, \cite%
{liu2010complementarities}, \cite{gentzkow2014competition}, \cite{lewbel2019sparse}, \cite
{KimEtal2020}, \cite{ershov2021estimating}, and \cite{iaria2021empirical}. The models in most of these
applications are parametric.\footnote{%
See \cite{Gentzkow2007} pp. 722--723 for a brief review of pervasive
parametric methods in the literature.}

There has been a growing interest in recent years in developing nonparametric and semiparametric methods for discrete choice models with bundles. \cite{SherKim2014} explored the identification of bundle choice
models without stochastic components in the utility from the perspective of
microeconomic theory, assuming a finite population of consumers. \cite%
{dunker2022nonparametric} studied nonparametric identification of
market-level demand models in a general framework in which the demand for
bundles is nested. \cite{FoxLazzati2017} provided nonparametric
identification results for both single-agent discrete choice bundle models
and binary games of complete information, and established the mathematical
equivalence between these models. \cite{AllenRehbeck2019} studied
nonparametric identification of a class of latent utility models with
additively separable unobserved heterogeneity by exploiting asymmetries in
cross-partial derivatives of conditional average demands. \cite{OYZ_infinity}
used the ``identification at infinity'' for
bundle choice models in cross-sectional settings. \cite{allen2022latent}
established partial identification results for complementarity
(substitutability) in a latent utility model with binary quantities of two
goods. These methods assume cross-sectional data. \cite{wang2023testing} proposed a semiparametric framework for identifying and testing substitution and complementarity patterns in multinomial choice models with bundles, which can be applied to both cross-sectional and panel data settings.

Our paper differs from these studies in many respects. Using individual-level data, we study semiparametric point identification and estimation of preference coefficients for bundle choice models in both cross-sectional and panel data settings. Our methods do not impose distributional constraints on unobserved error terms and allow for arbitrary correlation among them. The robustness afforded by the distribution-free specification makes our approaches competitive alternatives to existing parametric methods. Our paper is most closely related to \cite{wang2023testing} in terms of model specification and the assumptions regarding the dependence of unobservables on observed covariates. However, our proposed estimation and inference methods are built on distinct identification inequalities and criterion functions.

For the cross-sectional model, our identification results rely on exogeneity conditions similar to \cite{AllenRehbeck2019, allen2022latent}) and \cite{wang2023testing}, instead of either the exclusion
restrictions (``special regressors'') in \cite{FoxLazzati2017} or aggregate data and valid instrumental variables required by \cite{dunker2022nonparametric}. We first consider a maximum rank correlation (MRC) estimator in the spirit of \cite{Han1987}, using a matching insight. A recent study by \cite{KOT2021} used a similar method to study multinomial choice models. We validate nonparametric bootstrapping for inference and introduce a criterion-based test for interaction effects. Additionally, we propose a novel least absolute deviations (LAD) procedure as an alternative to the matching-based MRC method. The LAD estimator offers the advantage of estimating coefficients for both choice- and individual-specific covariates. It is worth noting that while we introduce these methods in the context of discrete choice models with bundles, their fundamental principles can be readily adapted to estimate a wide range of ``multi-index'' models with minor modifications.

We then extend the insight gained from the identification of cross-sectional models to panel data models with fixed effects. In this context, our identification relies on a mild conditional homogeneity assumption, which is also employed by \cite{PakesPorter2016}, \cite{ShiEtal2018}, and  \cite{gao2020robust}, among others, for studying ordinary panel data multinomial choice models. Based on a series of moment inequalities, we propose a maximum-score-type estimation procedure in the spirit of \cite{Manski1987} and derive its complete asymptotic properties. We also justify the validity of applying the numerical bootstrap for inference and propose a test for the presence of interaction effects. As in the cross-sectional model, we also introduce a panel data LAD estimation method, which offers the same advantages of accommodating both choice- and individual-specific covariates. Our methods for constructing these estimations are general and can be applied to models with similar multi-index structures.

This paper contributes to the fast-growing literature on a broad class of multi-index models. Leading examples are the multinomial choice model and its variants. In addition to the previously mentioned studies, this literature includes works by \cite{Lee1995}, \cite{Lewbel2000}, \cite{AltonjiMatzkin2005}, \cite{Fox2007}, \cite{berry2009nonparametric}, \cite{AhnEtal2018}, \cite{YanYoo2019}, and \cite{chernozhukov2019nonseparable}, to name a few. These approaches, however, do not directly extend to bundle choice models because of the non-mutually exclusive choice set and distinctive random utilities associated with bundles. Another example of multi-index models involves the study of dyadic networks, which has gained increasing attention among researchers. Some of the notable works in this area include \cite{graham2017econometric}, \cite{toth2017estimation}, \cite{candelaria2020semiparametric}, and \cite{gao2023logical}.

The remainder of this paper is organized as follows. Section \ref{SEC2} presents our MRC and LAD methods for the cross-sectional bundle choice model. Section \ref{SEC3} studies a panel data bundle choice model with fixed effects, introducing our localized MS and panel data LAD procedures. In Section \ref{monte}, we assess the finite sample performance of these proposed procedures through Monte Carlo experiments. Section \ref{conclude} concludes the paper. All tables, proofs, extensions, and further discussions are collected in the supplementary appendix (Appendixes \ref{appendixC}–\ref{appendixE}).

For ease of reference, the notations maintained throughout this paper are
listed here.

\begin{notation}
All vectors are column vectors. $\mathbb{R}^{p}$ is a $p$-dimensional
Euclidean space equipped with the Euclidean norm $\Vert \cdot \Vert $, and $%
\mathbb{R}_{+}\equiv \{x\in \mathbb{R}|x\geq 0\}$. We reserve letters $i$ and $m$ for indexing agents, $j$ and $l$
for indexing alternatives, and $s$ and $t$ for indexing time periods. The
first element of a vector $v$ is denoted by $v^{(1)}$ and the sub-vector
comprising its remaining elements is denoted by $\tilde{v}$. $P(\cdot
)$ and $\mathbb{E}[\cdot ]$ denote probability and expectation,
respectively. $1[\cdot ]$ is an indicator function that equals 1 when the
event in the brackets occurs, and 0 otherwise. For two random vectors $U$
and $V$, the notation $U\overset{d}{=}V|\cdot $ means that $U$ and $V$ have
identical distribution conditional on $\cdot $, and $U\perp V|\cdot $ means
that $U$ and $V$ are independent conditional on $\cdot $. Symbols $\setminus$, $^{\prime }$, $\Leftrightarrow $, $\Rightarrow$, $\propto$, $\overset{d}{\rightarrow }$, and $\overset{p}{\rightarrow }$ represent set
difference, matrix transposition, if and only if, implication, proportionality,
convergence in distribution, and convergence in probability, respectively. For any (random) positive
sequences $\{a_{N}\}$ and $\{b_{N}\}$, $a_{N}=O(b_{N})$ ($O_{p}(b_{N})$) means
that $a_{N}/b_{N}$ is bounded (bounded in probability) and $a_{N}=o(b_{N})$
($o_{p}(b_{N})$) means that $a_{N}/b_{N}\rightarrow0$ ($a_{N}/b_{N}%
\overset{p}{\rightarrow}0$).  
\end{notation}

\section{Cross-Sectional Model}\label{SEC2}
Throughout this paper, we focus on a choice model in
which the choice set $\mathcal{J}$ consists of three mutually exclusive
alternatives (numbered $0$--$2$) and a bundle of alternatives $1$ and $2$;
that is, $\mathcal{J}=\{0,1,2,(1,2)\}$. This simple model is sufficient to
illustrate the main intuition that runs through both cross-sectional and
panel data models.\footnote{We discuss models with more alternatives in Appendix \ref{appendixAdd_3}.}

For ease of exposition, we re-number alternatives in $\mathcal{J}$ with $2$%
-dimensional vectors of binary indicators $d=(d_{1},d_{2})\in \{0,1\}\times
\{0,1\}$, where $d_{1}$ and $d_{2}$ indicate if alternative 1 and 2 are
chosen, respectively. In this way the choice set $\mathcal{J}$ can be
one-to-one mapped to the set $\mathcal{D}=\{(0,0),(1,0),(0,1),(1,1)\}$. An
agent chooses the alternative in $\mathcal{D}$ to maximize the latent
utility
\begin{equation}
U_{d}=\sum_{j=1}^{2}F_{j}(X_{j}^{\prime }\beta ,\epsilon _{j})\cdot
d_{j}+\eta \cdot F_{b}(W^{\prime }\gamma)\cdot d_{1}\cdot d_{2},
\label{crossutility}
\end{equation}%
where $F_{j}(\cdot ,\cdot )$'s and $F_{b}(\cdot )$ are unknown (to the
econometrician) $\mathbb{R}^{2}\mapsto \mathbb{R}$ and $\mathbb{R}\mapsto
\mathbb{R}$, respectively, functions strictly monotonic in each of their
arguments, $X_{j}\in \mathbb{R}^{k_{1}}$ collects alternative-specific covariates affecting the
utility associated with stand-alone alternative $j$, $W\in \mathbb{R}%
^{k_{2}} $ is a vector of explanatory variables characterizing the
interaction effects of the bundle (e.g., bundle discount), $(\epsilon
_{1},\epsilon _{2},\eta )\in \mathbb{R}^{2}\times \mathbb{R}_{+}$ captures
unobserved (to the econometrician) heterogeneous effects, and $(\beta,\gamma)\in \mathbb{R}^{k_{1}+k_{2}}$ are
unknown preference parameters to estimate. 

We assume that $F_{b}(0)=0$ so that the sign of $W^{\prime
}\gamma$ has its economic meaning (see below for details). For instance, $%
F_{b}\left( x\right) $ can be $x$ or $x^{3}$. The specification of model (%
\ref{crossutility}) takes the linear forms in \cite{FoxLazzati2017} as a
special case. We note that \cite{FoxLazzati2017} does rely on the additive
separability while ours does not. The utility of choosing $(0,0)$ is
normalized to 0. The utilities of choosing $(1,0)$ and $(0,1)$ are $F_{1}\left(
X_{1}^{\prime }\beta ,\epsilon _{1}\right) $ and $F_{2}$($X_{2}^{\prime
}\beta ,\epsilon _{2}$), respectively. The utility of choosing the bundle is
the sum of the two stand-alone utilities plus the interaction term $%
\eta \cdot F_{b}(W^{\prime }\gamma )$ which captures either complementary (%
$W^{\prime }\gamma >0$) or substitution ($W^{\prime }\gamma <0$) effects. $(\epsilon _{1},\epsilon _{2})$ are
idiosyncratic shocks associated with each stand-alone alternative as in
common multinomial choice models, and $\eta $ reflects unobserved
heterogeneity for the bundle. In what follows, we assume $\eta \geq 0$ so that $%
W^{\prime }\gamma >0$ indicates interaction effect to be complementary, and
otherwise substitutive.  

Given the latent utility model (\ref{crossutility}), the observed dependent
variable $Y_{d}$ takes the form
\begin{equation}
Y_{d}=1[U_{d}>U_{d^{\prime }},\forall d^{\prime }\in \mathcal{D}\setminus d].\footnote{Throughout this paper, we do not  consider the case in which $U_{d}=U_{d^{\prime }}$. As will be clear below, such utility ties occur with probability 0 under our identification consitions.}
\label{choicemodel}
\end{equation}%
Let $Z\equiv (X_{1},X_{2},W)$. The probabilities of choosing alternatives $d\in \mathcal{D}$ can be
expressed as
\begin{align}
 & (P(Y_{(0,0)}=1|Z),P(Y_{(1,0)}=1|Z),P(Y_{(0,1)}=1|Z),P(Y_{(1,1)}=1|Z))' \label{crossprob}\\
= & \left[\begin{array}{c}
P(\max \{F_{1}\left( X_{1}^{\prime }\beta ,\epsilon
_{1}\right) ,F_{2}\left( X_{2}^{\prime }\beta ,\epsilon _{2}\right)
,F_{1}\left( X_{1}^{\prime }\beta ,\epsilon _{1}\right) +F_{2}\left(
X_{2}^{\prime }\beta ,\epsilon _{2}\right) +\eta \cdot F_{b}(W^{\prime
}\gamma )\}<0|Z)\\
P(\max \{0,F_{2}\left( X_{2}^{\prime }\beta ,\epsilon
_{2}\right) ,F_{1}\left( X_{1}^{\prime }\beta ,\epsilon _{1}\right)
+F_{2}\left( X_{2}^{\prime }\beta ,\epsilon _{2}\right) +\eta \cdot F_{b}(W^{\prime
}\gamma )\}<F_{1}\left( X_{1}^{\prime }\beta ,\epsilon
_{1}\right) |Z)\\
P(\max \{0,F_{1}\left( X_{1}^{\prime }\beta ,\epsilon
_{1}\right) ,F_{1}\left( X_{1}^{\prime }\beta ,\epsilon _{1}\right)
+F_{2}\left( X_{2}^{\prime }\beta ,\epsilon _{2}\right) +\eta \cdot F_{b}(W^{\prime
}\gamma )\}<F_{2}\left( X_{2}^{\prime }\beta ,\epsilon
_{2}\right) |Z)\\
P(\max \{0,F_{1}\left( X_{1}^{\prime }\beta ,\epsilon
_{1}\right) ,F_{2}\left( X_{2}^{\prime }\beta ,\epsilon _{2}\right)
\}<F_{1}\left( X_{1}^{\prime }\beta ,\epsilon _{1}\right) +F_{2}\left(
X_{2}^{\prime }\beta ,\epsilon _{2}\right) +\eta \cdot F_{b}(W^{\prime
}\gamma )|Z)
\end{array}\right]. \nonumber
\end{align}

\subsection{Rank Correlation Method}\label{sec:rcm}
Section \ref{sec:cross_identify} shows the identification result. Section \ref{estimator} presents the MRC estimator and its asymptotic properties. Bootstrap inference procedure and a criterion-based test for interaction effects are discussed in Sections \ref{SEC:infer_cross} and \ref{SEC:test}, respectively.
 
\subsubsection{Identification}\label{sec:cross_identify}
When $(\epsilon _{1},\epsilon _{2},\eta )\perp Z$, expression (\ref%
{crossprob}) implies that $P(Y_{(1,0)}=1|X_{1}=x_{1},X_{2}=x_{2},W=w)$ and $%
P(Y_{(1,1)}=1|X_{1}=x_{1},X_{2}=x_{2},W=w)$ are both increasing in $%
x_{1}^{\prime }\beta $ for some constant vectors $x_{2}$ and $w$. That is, for $d_2\in\{0,1\}$,
\begin{equation}
x_{1}^{\prime }\beta \geq \bar{x}_{1}^{\prime }\beta  
\Leftrightarrow P(Y_{\left( 1,d_{2}\right)
}=1|Z=(x_{1},x_{2},w))\geq P(Y_{\left( 1,d_{2}\right)
}=1|Z=(\bar{x}_{1},x_{2},w)).  \label{crossmi1}
\end{equation}%
Similarly, for cases with $d_{1}=0$ and any $d_{2},$ $x_{1}^{\prime }\beta
\geq \bar{x}_{1}^{\prime }\beta $ is the ``if-and-only-if'' condition to the
second inequality in (\ref{crossmi1}) but with \textquotedblleft $\leq $%
\textquotedblright\ instead.

Once $\beta $ is identified, we can move on to identify $\gamma $ using the
following moment inequalities for $w^{\prime }\gamma $ constructed by fixing
$V(\beta)\equiv(X_{1}^{\prime }\beta ,X_{2}^{\prime }\beta )$ at some constant vector $%
(v_{1},v_{2})$. For $d=(1,1)$,
\begin{align}
& w^{\prime }\gamma \geq \bar{w}^{\prime }\gamma  \notag \\
\Leftrightarrow & P(Y_{(1,1)}=1|V(\beta) =(v_{1},v_{2}),W=w)\geq P(Y_{(1,1)}=1|V(\beta) =(v_{1},v_{2}),W=\bar{w}).  \label{crossmi3}
\end{align}
For all $d\in \mathcal{D}\setminus (1,1)$, $w^{\prime }\gamma \geq \bar{w}%
^{\prime }\gamma $ is the ``if-and-only-if'' condition for the second inequality
in (\ref{crossmi3}) but with \textquotedblleft $\leq $\textquotedblright\
instead.

To implement the above idea, we assume a random sample in Assumption C1 below. Thus, we can take two independent copies of $\left( Y,Z\right) $, i.e., $\left( Y_{i},Z_{i}\right) $ and $\left( Y_{m},Z_{m}\right) $ from the
sample. For (\ref{crossmi1}), we can match $X_{2}$ and $W$ for this two observations, so then we can rank their choice probabilities of $Y_{\left( 1,d_{2}\right) }$ using the index $X_{1}^{\prime }\beta$. This idea forms the basis of our MRC estimator, which we introduce in the next section. Note that $\gamma$ can be alternatively identified by matching $(X_{1},X_{2})$ across agents. However, some exploratory simulation studies have shown that matching the estimated index $V(\hat{\beta})$ (as suggested in (\ref{crossmi3})) can yield better finite sample performance than element-by-element matching.  

To establish the identification of $\beta $ and $\gamma $ based on (\ref%
{crossmi1})--(\ref{crossmi3}), the following conditions are sufficient:
\begin{itemize}
\item[\textbf{C1}] (i) $\{(Y_{i},Z_{i})\}_{i=1}^{N}$ are i.i.d. across $i$, (ii) $(\epsilon _{1},\epsilon
_{2},\eta )\perp Z$, and (iii) the joint distribution of $(\epsilon
_{1},\epsilon _{2},\eta )$ is absolutely continuous on $\mathbb{R}^{2}\times
\mathbb{R}_{+}$.

\item[\textbf{C2}] For any pair of $(i,m)$ and $j=1,2$, denote $%
X_{imj}=X_{ij}-X_{mj}$, $W_{im}=W_{i}-W_{m}$, and $V_{im}(\beta)= V_{i}(\beta)-V_{m}(\beta)$. Then, (i) $X_{im1}^{(1)}$ ($%
X_{im2}^{(1)}$) has almost everywhere (a.e.) positive Lebesgue density on $%
\mathbb{R}$ conditional on $\tilde{X}_{im1}$ ($\tilde{X}_{im2}$) and
conditional on $(X_{im2},W_{im})$ ($(X_{im1},W_{im})$) in a neighborhood of $%
(X_{im2},W_{im})$ ($(X_{im1},W_{im})$) near zero, (ii) Elements in $X_{im1}$ ($%
X_{im2}$), conditional on $(X_{im2},W_{im})$ ($(X_{im1},W_{im})$) in a
neighborhood of $(X_{im2},W_{im})$ ($(X_{im1},W_{im})$) near zero, are
linearly independent, (iii) $W_{im}^{(1)}$ has a.e. positive Lebesgue
density on $\mathbb{R}$ conditional on $\tilde{W}_{im}$ and conditional on $V_{im}(\beta)$ in a neighborhood of $V_{im}(\beta)$ near zero, and (iv)
Elements in $W_{im}$, conditional on $V_{im}(\beta)$ in a neighborhood of $V_{im}(\beta)$ near zero, are linearly independent.

\item[\textbf{C3}] $(\beta,\gamma)\in
\mathcal{B}\times\mathcal{R}$, where $\mathcal{B}=\{b\in\mathbb{R}^{k_1}|b^{(1)}=1,\Vert
b\Vert\leq C_b\}$ and $\mathcal{R}=\{r\in\mathbb{R}^{k_2}|r^{(1)}=1,\Vert
r\Vert\leq C_r\}$ for some constants $C_b,C_r>0$.

\item[\textbf{C4}] $F_{1}(\cdot ,\cdot )$, $F_{2}(\cdot ,\cdot )$, and $%
F_{b}(\cdot )$ are respectively $\mathbb{R}^{2}\mapsto
\mathbb{R}$, $\mathbb{R}^{2}\mapsto \mathbb{R}$, and $\mathbb{R}\mapsto
\mathbb{R}$ functions strictly increasing in each of their arguments. $F_{b}(0)=0$.
\end{itemize}

From the previous discussion, it can be seen that Assumptions C1 and C4 are the keys to establishing the moment inequalities (\ref{crossmi1})--(\ref{crossmi3}). It is important to note that Assumption C1 allows for arbitrary correlation among $(\epsilon _{1},\epsilon _{2},\eta )$. Assumptions C2(i) and C2(iii) are standard requirements for MRC and MS types of estimators. Assumptions C2(ii) and C2(iv) are regular rank conditions. Furthermore, Assumption C3 defines a compact parameter space and normalizes the scale of coefficients, which is standard practice for discrete choice models. Essentially, we assume that $\beta^{(1)}$ and $\gamma^{(1)}$ are not equal to zero. 

The following identification result for model (\ref{crossutility})--(\ref{choicemodel}%
) is proved in Appendix \ref{appendixA}.
\begin{theorem}
\label{T:crossidentify} Suppose Assumptions C1--C4 hold. Then $\beta $ and $\gamma $ are identified.
\end{theorem}

\begin{remark}
The matching-based method discussed here is distribution-free, accommodating arbitrary correlation among $(\epsilon_{1},\epsilon_{2},\eta)$ and interpersonal heteroskedasticity. Furthermore, as demonstrated in Sections \ref{estimator}--\ref{SEC:infer_cross}, the resulting estimator is computationally tractable and has an asymptotic normal distribution, rendering inference relatively straightforward, particularly when employing bootstrap methods. However, this method has two main limitations. Firstly, it becomes difficult to apply when the model has many covariates, as this decreases the number of potential matches. Secondly, in cases where ``common regressors'' are present, it cannot estimate coefficients related to these variables unless their coefficients are assumed to be not ``alternative-specific''.
\end{remark}


\subsubsection{Localized MRC Estimator}\label{estimator} 
The local monotonic relations established in (\ref%
{crossmi1})--(\ref{crossmi3}) naturally motivate a two-step localized MRC
estimation procedure. Note that the probability of obtaining perfectly
matched observations is zero for continuous regressors. Following the
literature, we propose to use kernel weights as an approximation to the
matching. These steps are described in turn below. 

In the first step, we consider the localized MRC estimator $\hat{\beta}$ of $%
\beta $, analogous to the MRC estimator proposed in \cite{Han1987} and \cite{KOT2021}.
Specifically, $\hat{\beta}=\arg \max_{\beta \in \mathcal{B}}\mathcal{L}%
_{N,\beta }^{K}(b)$. $\mathcal{L}_{N,\beta }^{K}(b)$ is defined as
\begin{align}
\mathcal{L}_{N,\beta }^{K}(b)=\sum_{i=1}^{N-1}\sum_{m>i}\sum_{d\in \mathcal{D%
}}& \{\mathcal{K}_{h_{N}}(X_{im2},W_{im})(Y_{md}-Y_{id})\text{sgn}%
(X_{im1}^{\prime }b)\cdot \left( -1\right) ^{d_{1}}  \notag \\
& +\mathcal{K}_{h_{N}}(X_{im1},W_{im})(Y_{md}-Y_{id})\text{sgn}%
(X_{im2}^{\prime }b)\cdot \left( -1\right) ^{d_{2}}\}  \label{crossobjbK}
\end{align}%
with $\mathcal{K}_{h_{N}}(X_{imj},W_{im})\equiv
h_{N}^{-(k_{1}+k_{2})}\prod_{\iota =1}^{k_{1}}K(X_{imj,\iota
}/h_{N})\prod_{\iota =1}^{k_{2}}K(W_{im,\iota }/h_{N})$, where $X_{imj,\iota
}$ ($W_{im,\iota }$) is the $\iota $-th element of vector $X_{imj}$ ($W_{im}$%
), $K\left( \cdot \right) $ is a standard kernel density function, and $%
h_{N} $ is a bandwidth sequence that converges to 0 as $N\rightarrow \infty $%
.\footnote{%
In practice, kernel functions and bandwidths can be different for each
univariate matching variable. Here we assume they are the same for all
covariates just for notational convenience.} Obviously, $\mathcal{K}%
_{h_{N}}(X_{imj},W_{im})\approx\left[
X_{imj}=0,W_{im}=0\right] $ for $j=1,2,$ as $h_{N}\rightarrow 0.$

In the second step, we obtain $\hat{\gamma}=\arg \max_{\gamma \in \mathcal{R}%
}\mathcal{L}_{N,\gamma }^{K}(r;\hat{\beta})$ with
\begin{equation}
\mathcal{L}_{N,\gamma }^{K}(r;\hat{\beta})=\sum_{i=1}^{N-1}\sum_{m>i}%
\mathcal{K}_{\sigma _{N}}(V_{im}(\hat{\beta}))(Y_{i(1,1)}-Y_{m(1,1)})\text{%
sgn}(W_{im}^{\prime }r),  \label{crossobjrK}
\end{equation}%
where $\mathcal{K}_{\sigma _{N}}(V_{im}(\hat{%
\beta}))=\sigma _{N}^{-2}K(X_{im1}^{\prime }\hat{\beta}/\sigma
_{N})K(X_{im2}^{\prime }\hat{\beta}/\sigma _{N})$ and $\sigma _{N}$ is a
bandwidth sequence that converges to 0 as $N\rightarrow \infty $. Again, $%
\mathcal{K}_{\sigma _{N}}(V_{im}(\hat{\beta}))\approx 1%
[ V_{im}(\hat{\beta})=0] $ as $%
\sigma _{N}\rightarrow 0$.

The rest of this section establishes the asymptotic properties of the
estimators computed based on criterion functions (\ref{crossobjbK}) and (\ref%
{crossobjrK}). To do so, we need to place Assumptions C5--C7 in addition to Assumptions C1--C4 and introduce some new notations. Assumption C5 is a purely technical assumption and standard in the literature. To save space, we present it in Appendix \ref{appendixA} before the proof of Theorem \ref{T:crossAsymp}. Here we only present Assumptions C6--C7 since they are of practical importance.

\begin{itemize}
\item[\textbf{C6}] $K(\cdot )$ is continuously differentiable and assumed to
satisfy: (i) $\sup_{\upsilon }|K(\upsilon )|<\infty $, (ii) $\int K(\upsilon
)\text{d}\upsilon =1$, (iii) for any positive integer $\iota \leq \max
 \{ \kappa _{\beta },\kappa _{\gamma } \}$ where $\kappa_\beta$ is an even integer greater than $k_1+k_2$ and $\kappa_\gamma$ is an even integer greater than 2, $\int \upsilon ^{\iota }K(\upsilon )\text{d}\upsilon =0$ if $\iota < \max\{\kappa_{\beta},\kappa_{\gamma}\}$, $\int \upsilon ^{\iota }K(\upsilon )\text{d}\upsilon\neq 0$ if $\iota=\max\{\kappa_{\beta},\kappa_{\gamma}\}$, and $\int
\vert \upsilon ^{\max\{\kappa_{\beta},\kappa_{\gamma
}\} }\vert K(\upsilon )\text{d}\upsilon <\infty $.\footnote{Here we use the same kernel function $K(\cdot)$ for estimating $\beta$ and $\gamma$ to avoid additional notations. However, in practice, one can use kernel functions with different orders by replacing $\max\{\kappa_{\beta},\kappa_{\gamma}\}$ with $\kappa_{\beta}$ and $\kappa_{\gamma}$,respectively. }
\item[\textbf{C7}] $h_{N}$ and $\sigma _{N}$ are sequences of positive
numbers such that as $N\rightarrow \infty $: (i) $h_{N}\rightarrow 0$, $%
\sigma _{N}\rightarrow 0$, (ii) $\sqrt{N}h_{N}^{k_{1}+k_{2}}\rightarrow
\infty $, $\sqrt{N}\sigma _{N}^{2}\rightarrow \infty $, and (iii) $\sqrt{N}%
h_{N}^{\kappa _{\beta }}\rightarrow 0$, $\sqrt{N}\sigma _{N}^{\kappa
_{\gamma }}\rightarrow 0$.
\end{itemize}

Assumptions C6 and C7 impose mild and standard restrictions on kernel functions and tuning parameters. These assumptions are essential to establish the consistency of the estimator and to ensure that the bias term becomes asymptotically negligible when deriving the limiting distribution. Assumption C7 requires the order of the kernel function to be at least $\max \left\{ \kappa _{\beta },\kappa _{\gamma}\right\}$. Using similar arguments employed in \cite{Han1987}, \cite{Sherman1993, Sherman1994AoS, Sherman1994ET}, and \cite{AbrevayaEtal2010}, we show the $\sqrt{N}$-consistency and asymptotic normality of the MRC estimators. The proof can be found in Appendix \ref{appendixA}.

\begin{theorem}
\label{T:crossAsymp} If Assumptions C1--C7 hold, then we have
\begin{itemize}
\item[(i)] $\sqrt{N}(\hat{\beta}-\beta )\overset{d}{\rightarrow }%
N(0,4\left\{ \mathbb{E}[\nabla ^{2}\varrho _{i}(\beta )]\right\} ^{-1}%
\mathbb{E}[\nabla \varrho _{i}(\beta )\nabla \varrho _{i}(\beta )^{\prime
}]\left\{ \mathbb{E}[\nabla ^{2}\varrho _{i}(\beta )]\right\} ^{-1})$ or,
alternatively, $\hat{\beta}-\beta $ has the linear representation:
\begin{equation*}
\hat{\beta}-\beta =-\frac{2}{N}\left\{ \mathbb{E}[\nabla ^{2}\varrho
_{i}(\beta )]\right\} ^{-1}\sum_{i=1}^{N}\varrho _{i}(\beta
)+o_{p}(N^{-1/2}),
\end{equation*}%
with $\varrho _{i}\left( \cdot \right) $ defined in (\ref{EQ:rhoi}).

\item[(ii)] $\sqrt{N}(\hat{\gamma}-\gamma )\overset{d}{\rightarrow }N\left(
0,\mathbb{E}[\nabla ^{2}\tau _{i}(\gamma )]^{-1}\mathbb{E}\left[ \varDelta%
_{i}\varDelta_{i}^{\prime }\right] \mathbb{E}[\nabla ^{2}\tau _{i}(\gamma
)]^{-1})\right) $, where
\begin{equation*}
\varDelta_{i}\equiv 2\nabla \tau _{i}(\gamma )+2\mathbb{E}\left[ \nabla
_{13}^{2}\mu \left( V_{i}(\beta ),V_{i}(\beta ),\gamma \right) \right]
\left\{ \mathbb{E}[\nabla ^{2}\varrho _{i}(\beta )]\right\} ^{-1}\varrho
_{i}(\beta ),
\end{equation*}%
with $\tau _{i}\left( \cdot \right) $ and $\mu \left( \cdot ,\cdot ,\cdot
\right) $ defined in (\ref{EQ:Tau}) and (\ref{EQ:Mu}), respectively, and $\nabla _{13}^{2}$ denotes the
second order derivative w.r.t. the first and third arguments.
\end{itemize}
\end{theorem}

We provide some intuition on the $\sqrt{N}$ convergence rate in Appendix \ref{appendixAdd}.

\subsubsection{Inference\label{SEC:infer_cross}}

Theorem \ref{T:crossAsymp} indicates that our localized MRC estimators are asymptotically normal and have asymptotic variances with the usual sandwich structure. The expressions for the asymptotic variances consist of first and second derivatives of the limit of the expectation of the maximands in (\ref{crossobjbK}) and (\ref{crossobjrK}). To use these results for statistical inference, \cite{Sherman1993} proposed applying the numerical derivative method of \cite{PakesPollard1989}. \cite{HongEtal2015} investigated the application of the numerical derivative method in extremum estimators, including second-order U-statistics. As an alternative, \cite{CavanaghSherman1998} suggested nonparametrically estimating these quantities. These procedures, however, require selecting additional tuning parameters.


To avoid this complexity, we propose conducting inference using the classic nonparametric bootstrap for ease of implementation. \cite{Subbotin2007} proved the consistency of the nonparametric bootstrap for \citeauthor{Han1987}'s (\citeyear{Han1987}) MRC estimator.\footnote{\cite{JinEtal2001} proposed an alternative resampling method for U-statistics by perturbing the criterion function repeatedly.} 
The structure of the criterion functions of our estimators is similar to that of the standard MRC estimator, but they do differ in two ways. First, our estimators require matching and thus contain kernel functions. Second, the criterion function for estimating $\hat{\gamma}$ contains a first-step ($\sqrt{N}$-consistent) $\hat{\beta}$ to approximate the true value of $\beta$. These two differences usually do not make the bootstrap inconsistent. For example, \cite{Horowitz2001} demonstrated the consistency of the bootstrap for a range of estimators involved with kernel functions; \cite{ChenEtal2003} demonstrated the consistency of the bootstrap for estimators with a first-step estimation component and a well-behaved criterion function.

The bootstrap algorithm is standard. Draw $\{ (Y_{i} ^{\ast},Z_{i}^{\ast})\} _{i=1}^{N}$ independently from the sample $\{ (Y_{i},Z_{i})\} _{i=1}^{N}$ with replacement. Obtain the bootstrap estimator $\hat {\beta}^{\ast}$ is obtained from
\begin{align}
\hat{\beta}^{\ast} =\arg\max_{b\in\mathcal{B}}\mathcal{L}_{N,\beta}%
^{K\ast}\left( b\right) \equiv &\arg\max_{b\in\mathcal{B}}\sum_{i=1}^{N-1}%
\sum_{m>i}\sum_{d\in\mathcal{D}} \{ \mathcal{K}_{h_{N}}\left(
X_{im2}^{\ast},W_{im}^{\ast}\right) Y_{mid}^{\ast}\text{sgn}\left( X_{im1}%
^{\ast\prime}b\right) \cdot\left( -1\right) ^{d_{1}}
\nonumber\\
& + \mathcal{K}_{h_{N}}\left( X_{im1}^{\ast},W_{im}^{\ast}\right)
Y_{mid}^{\ast}\text{sgn}\left( X_{im2}^{\ast\prime}b\right) \cdot\left( -1\right)
^{d_{2}}\} .\label{EQ:betabootstrap}
\end{align}
Then, obtain the bootstrap estimator $\hat{\gamma}^{\ast}$ through
\begin{equation}
\hat{\gamma}^{\ast}=\arg\max_{r\in\mathcal{R}}\mathcal{L}_{N,\gamma}^{K\ast
}( r,\hat{\beta}^{\ast}) \equiv\arg\max_{r\in\mathcal{R}}%
\sum_{i=1}^{N-1}\sum_{m>i}\mathcal{K}_{\sigma_{N}} ( V_{im}^{\ast} (
\hat{\beta}^{\ast} ) ) Y_{im\left( 1,1\right) }^{\ast
}\text{sgn}\left( W_{im}^{\ast\prime}r\right) . \label{EQ:gammabootstrap}%
\end{equation}
Repeat the above steps for $B$ times to yield a sequence of bootstrap estimates $\{(\hat{\beta}_{(b)}^{\ast},\hat{\gamma}_{(b)}^{\ast})\}_{b=1}^{B}$.

\begin{theorem} \label{T:cross_boot}
Suppose Assumptions C1--C7 hold. Then, we have
\[
\sqrt{N} ( \hat{\beta}^{\ast}-\hat{\beta} ) \overset{d}{\rightarrow
}N ( 0,4\left\{ \mathbb{E}[\nabla^{2}\varrho_{i}(\beta)]\right\}
^{-1}\mathbb{E}[\nabla\varrho_{i}(\beta)\nabla\varrho_{i}(\beta)^{\prime
}]\left\{ \mathbb{E}[\nabla^{2}\varrho_{i}(\beta)]\right\} ^{-1} ) \textrm{ and}  
\]
\[
\sqrt{N}\left( \hat{\gamma}^{\ast}-\hat{\gamma}\right)
\overset{d}{\rightarrow}N\left( 0,\mathbb{E}[\nabla^{2}\tau_{i}(\gamma
)]^{-1}\mathbb{E}\left[ \varDelta_{i}\varDelta_{i}^{\prime}\right]
\mathbb{E}[\nabla^{2}\tau_{i}(\gamma)]^{-1}\right) 
\]
conditional on the sample. 
\end{theorem}
This result is proved in Appendix \ref{appendixD}.  As a result, $\{(\hat{\beta}_{(b)} ^{\ast},\hat{\gamma}_{(b)}^{\ast})\}_{b=1}^{B}$ can be used to compute the standard errors or confidence intervals of the estimators $\hat{\beta}$ and $\hat{\gamma}$.  

\subsubsection{Testing $\eta$}\label{SEC:test}
Our identification requires $\gamma^{(1)}\neq 0$ and a scale normalization $\gamma^{(1)}=1$. As a consequence, it is impossible to test for the presence of the interaction effect by simply looking at the significance of $\left\Vert \gamma \right\Vert$. This observation highlights the role of $\eta$--it determines the magnitude of the interaction effect between the two stand-alone alternatives. In this section, we introduce a criterion-based procedure to test whether $\eta$ degenerates to zero. We formulate the null hypothesis as
\begin{equation*}
\mathbb{H}_{0}:\eta >0\text{ almost surely and }E\left( \eta \right) >0,
\end{equation*}%
and the alternative hypothesis as%
\begin{equation*}
\mathbb{H}_{1}:\eta =0\text{ almost surely.}
\end{equation*}%

Before presenting the theoretical result, we explain the idea underlying the test. We define $\mathcal{\hat{L}}_{N}^{K}\left( \hat{\gamma}\right) $ as
\begin{equation*}
\mathcal{\hat{L}}_{N}^{K} ( \hat{\gamma} ) =\frac{1}{\sigma
_{N}^{2}N\left( N-1\right) }\sum_{i\neq m}K_{\sigma _{N},\gamma } (
V_{im} ( \hat{\beta} )  ) Y_{im\left( 1,1\right) }\text{sgn}%
\left( W_{im}^{\prime }\hat{\gamma}\right) .
\end{equation*}%
Suppose we know the value of $\beta $. Define
\begin{equation*}
\mathcal{L}_{N}^{K}\left( r\right) =\frac{1}{\sigma _{N}^{2}N\left(
N-1\right) }\sum_{i\neq m}K_{\sigma _{N},\gamma }\left( V_{im}\left( \beta
\right) \right) Y_{im\left( 1,1\right) }\text{sgn}\left( W_{im}^{\prime
}r\right) .
\end{equation*}%
Clearly, $\mathcal{L}_{N}^{K}\left( r\right) $ is a second-order
U-statistic, and%
\begin{equation*}
\mathcal{L}_{N}^{K}\left( r\right) \overset{p}{\rightarrow }\mathcal{\bar{L}}%
\left( r\right) \equiv f_{V_{im}\left( \beta \right) }\left( 0\right)
\mathbb{E}\left[ Y_{im\left( 1,1\right) }\text{sgn}\left( W_{im}^{\prime
}r\right) |V_{im}\left( \beta \right) =0\right] .
\end{equation*}%
If $\eta =0$ almost surely,
\begin{equation*}
\mathcal{\bar{L}}\left( r\right) =0\text{ for any }r
\end{equation*}%
since $Y_{im\left( 1,1\right) }$ is independent of $W_{im}$ conditional on
$V_{im}\left( \beta \right) =0$ and $\mathbb{E}[ Y_{im\left(
1,1\right) }|V_{im}\left( \beta \right) =0] =0$. On the contrary, if $%
\eta >0$ almost surely and $E\left( \eta \right) >0$, then the signs of $W_{im}^{\prime }\gamma $ and $P\left( Y_{i\left( 1,1\right)
}|V_{i}\left( \beta \right) =v,W_{i}\right) -P\left( Y_{m\left( 1,1\right)
}|V_{m}\left( \beta \right) =v,W_{m}\right) $ should be the same. 
As a result,
\begin{equation*}
\mathcal{\bar{L}}\left( \gamma \right) >0.
\end{equation*}%
In light of this observation, testing $\mathbb{H}_0$ is equivalent to testing $\mathcal{\bar{L}}(\gamma) > 0$.

By some standard analysis for U-statistics, the limiting distribution of $%
\mathcal{L}_{N}^{K}\left( \gamma \right) $ is
\begin{equation*}
\sqrt{N}\left( \mathcal{L}_{N}^{K}\left( \gamma \right) -\mathcal{\bar{L}}%
\left( \gamma \right) \right) \overset{d}{\rightarrow }N\left( 0,\Delta_\gamma
\right) ,
\end{equation*}%
with
\begin{equation*}
\Delta_\gamma =4\text{Var}\left\{ f_{V_{im}(\beta)}(0)\mathbb{E}\left[ Y_{im\left( 1,1\right) }\text{sgn}%
\left( W_{im}^{\prime }\gamma \right) |Z_{i},V_{m}\left( \beta \right)
=V_{i}\left( \beta \right) \right] \right\} .
\end{equation*}%
We can then test whether $\mathcal{\bar{L}}\left( \gamma \right) >0$ using
the above limiting distribution.

The analysis for our case is more complicated by the fact that we need to first
estimate $\beta $ and $\gamma$ before running the test. In Theorem \ref{TH:testing} below, we show that the
plugged-in $\hat{\beta}$ and $\hat{\gamma}$ affect $\mathcal{\hat{L}}%
_{N}^{K}\left( \hat{\gamma}\right) $ at the rates of $N^{-1}\sigma _{N}^{-2}$
and $N^{-1}$, respectively, around the true values ($\beta ,\gamma $). These
rates are much faster than $N^{-1/2}$. Consequently, the asymptotic behavior of $\mathcal{\hat{L}}_N^K(\hat{\gamma})$ is the same as that of $\mathcal{L}_N^K(\gamma)$. The limiting distribution of $\mathcal{\hat{L}}_N^K(\hat{\gamma})$ under the null hypothesis is provided in the following theorem. The proof is included in Appendix \ref{appendixA}.

\begin{theorem}
\label{TH:testing}Suppose Assumptions C1-C7 hold. Then, under $\mathbb{H}%
_{0},$%
\begin{equation*}
\sqrt{N} ( \mathcal{\hat{L}}_{N}^{K}\left( \hat{\gamma}\right) -\mathcal{%
\bar{L}}\left( \gamma \right) ) \overset{d}{\rightarrow }N\left(
0,\Delta \right) .
\end{equation*}
\end{theorem}

To perform the test, we propose using a standard bootstrap procedure. Specifically, we independently draw $\{ (Y_{i} ^{\ast},Z_{i}^{\ast})\} _{i=1}^{N}$ with replacement from the sample $\{ (Y_{i},Z_{i})\} _{i=1}^{N}$ for a total of $B$ times. For each draw, we calculate
\begin{equation*}
\mathcal{\hat{L}}_{N}^{K\ast }\left( \hat{\gamma}\right) =\frac{1}{\sigma
_{N}^{2}N\left( N-1\right) }\sum_{i\neq m}K_{\sigma _{N},\gamma } (
V_{im}^{\ast } ( \hat{\beta} )  ) Y_{im\left( 1,1\right)
}^{\ast }\text{sgn}\left( W_{im}^{\ast \prime }\hat{\gamma}\right) .
\end{equation*}%
Note that there is no need to re-compute $\hat{\beta}$ and $\hat{\gamma}$ using the bootstrap samples. We construct a 95\% confidence interval for $\mathcal{\bar{L}}(\gamma)$ as a one-sided interval $[ Q_{0.05}( \mathcal{\hat{L}}_{N}^{K\ast }\left( \hat{\gamma}\right) ) ,+\infty )$, where $Q_{0.05}(\mathcal{\cdot})$ denotes the 5\% quantile of $\mathcal{\hat{L}}_N^{K^*}(\hat{\gamma})$. By the same reason as we provided above, the plugged-in $\hat{\beta}$ and $\hat{\gamma}$ have no impact on the asymptotics of $\mathcal{\hat{L}}_{N}^{K\ast }\left( \hat{\gamma}\right) .$ The validity of the bootstrap for U-statistics has been established in the literature (see, e.g., \cite{ArconesGine}). Since this analysis is standard, we omit it here for brevity.

\subsection{Least Absolute Deviations Method}\label{SEC:LAD}
The main limitation of the MRC estimation proposed in Section \ref{sec:rcm} is its inability to estimate coefficients for regressors that are common across different alternatives, often referred to as common regressors. For example, individual-specific characteristics like gender may be present in multiple stand-alone utilities as well as the bundle utility. In such cases, the coefficients for these regressors are often alternative-specific. In this section, we propose an innovative least absolute deviations (LAD) estimator, offering the advantage of estimating coefficients for both common and alternative-specific regressors. The LAD estimator is also robust in that it does not impose any distributional assumptions on the error terms and allows for arbitrary correlations among them, aligning with the flexibility of the MRC estimator in these respects. 

Consider the following model extended from the model specified in (\ref{crossutility}) and (\ref{choicemodel}):
\begin{align}
U_{d}  &  =\sum_{j=1}^{2}F_{j}(X_{j}^{\prime}\beta+S^{\prime}\rho_{j}%
,\epsilon_{j})\cdot d_{j}+\eta\cdot F_{b}(W^{\prime}\gamma+S^{\prime}\rho
_{b})\cdot d_{1}\cdot d_{2},\nonumber\\
Y_{d}  &  =1[U_{d}>U_{d^{\prime}},\forall d^{\prime}\in\mathcal{D}\setminus
d], \label{eq:lad_model}%
\end{align}
where $S\in\mathbb{R}^{k_{3}}$ is a vector of common regressors affecting the
utilities of all stand-alone alternatives and the bundle. With a slight abuse
of notation, we let $Z\equiv(X_{1},X_{2},W,S)$ in this section. Then, under
Assumption C4 and $(\epsilon_{1},\epsilon_{2},\eta)\perp Z$, for two
independent copies of $(Y,Z)$, $(Y_{i},Z_{i})$ and $(Y_{m},Z_{m})$, we can
write the following inequalities from model (\ref{eq:lad_model}) and equation
(\ref{crossprob}) (with modified latent utilities specified in (\ref{eq:lad_model})):
\begin{align}
\{X_{im1}^{\prime}\beta+S^{\prime}_{im}\rho_{1}\geq0,X_{im2}^{\prime}%
\beta+S_{im}^{\prime}\rho_{2}\leq0,W_{im}^{\prime}\gamma+S^{\prime}_{im}\rho
_{b}\leq0\}  &  \Rightarrow\Delta p_{(1,0)}(Z_{i},Z_{m})\geq
0,\label{eq:lad_ineq_1}\\
\{X_{im1}^{\prime}\beta+S^{\prime}_{im}\rho_{1}\leq0,X_{im2}^{\prime}%
\beta+S_{im}^{\prime}\rho_{2}\geq0,W_{im}^{\prime}\gamma+S^{\prime}_{im}\rho
_{b}\geq0\}  &  \Rightarrow\Delta p_{(1,0)}(Z_{i},Z_{m})\leq0,
\label{eq:lad_ineq_2}%
\end{align}
where $S_{im}\equiv S_{i}-S_{m}$ and $\Delta p_{(1,0)}(Z_{i},Z_{m})\equiv
P(Y_{i(1,0)}=1|Z_{i})-P(Y_{m(1,0)}=1|Z_{m})$.

We propose the following LAD criterion function to identify the unknown
parameters $\theta\equiv(\beta,\gamma,\rho_{1},\rho_{2},\rho_{b})$ in model
(\ref{eq:lad_model}):
\begin{equation}
Q_{(1,0)}(\vartheta)=\mathbb{E}\left[  q_{(1,0)}(Z_{i},Z_{m},\vartheta
)\right]  , \label{eq:lad_obj}%
\end{equation}
where $\vartheta=(b,r,\varrho_{1},\varrho_{2},\varrho_{b})$ is an arbitrary
vector in the parameter space $\Theta$ containing $\theta$ and
\begin{align*}
q_{(1,0)}(Z_{i},Z_{m},\vartheta)=  &  [|I_{im(1,0)}^{+}(\vartheta)-1[\Delta
p_{(1,0)}(Z_{i},Z_{m})\geq0]|+|I_{im(1,0)}^{-}(\vartheta)-1[\Delta
p_{(1,0)}(Z_{i},Z_{m})\leq0]|]\\
&  \times\lbrack I_{im(1,0)}^{+}(\vartheta)+I_{im(1,0)}^{-}(\vartheta)]
\end{align*}
with
\begin{align}
I_{im(1,0)}^{+}(\vartheta)  &  =1[X_{im1}^{\prime}b+S_{im}^{\prime
}\varrho_{1}\geq0,X_{im2}^{\prime}b+S_{im}^{\prime}\varrho_{2}\leq
0,W_{im}^{\prime}r+S^{\prime}_{im}\varrho_{b}\leq0]\text{ and}%
\label{eq:lad_ineq_1_2}\\
I_{im(1,0)}^{-}(\vartheta)  &  =1[X_{im1}^{\prime}b+S_{im}^{\prime
}\varrho_{1}\leq0,X_{im2}^{\prime}b+S^{\prime}_{im}\varrho_{2}\geq
0,W_{im}^{\prime}r+S^{\prime}_{im}\varrho_{b}\geq0] \label{eq:lad_ineq_2_2}%
\end{align}
being indicator functions respectively for the events defined on the left hand
sides of (\ref{eq:lad_ineq_1}) and (\ref{eq:lad_ineq_2}), but evaluated at $\vartheta$.

The construction of the function $q_{(1,0)}(Z_{i},Z_{m},\vartheta)$ might seem
complex at first glance, but it has an intuitive interpretation. Inequalities
(\ref{eq:lad_ineq_1}) and (\ref{eq:lad_ineq_2}) imply that using the true
parameter $\theta$, we can flawlessly predict the occurrence of the event
$\{\Delta p_{(1,0)}(Z_{i},Z_{m})\geq0\}$ ($\{\Delta p_{(1,0)}(Z_{i},Z_{m}%
)\leq0\}$) when $I^{+}_{im(1,0)}(\theta)=1$ ($I^{-}_{im(1,0)}(\theta)=1$). However,
for any $\vartheta\neq\theta$, such predictions may lead to errors, where
$I^{+}_{im(1,0)}(\vartheta)=1$ ($I^{-}_{im(1,0)}(\vartheta)=1$) does not align
with the observed data ${\Delta p_{(1,0)}(Z_{i},Z_{m})\geq0}$ (${\Delta
p_{(1,0)}(Z_{i},Z_{m})\leq0}$). Besides, when $I^{+}_{im(1,0)}(\vartheta
)=I^{-}_{im(1,0)}(\vartheta)=0$, no prediction can be made for the sign of
$\Delta p_{(1,0)}(Z_{i},Z_{m})$. Then, treating $q_{(1,0)}(Z_{i}%
,Z_{m},\vartheta)$ as a loss function, we can interpret it as the loss
incurred from making incorrect predictions for the sign of $\Delta
p_{(1,0)}(Z_{i},Z_{m})$ using $\vartheta$:
\begin{itemize}
\item[-] A wrong prediction incurs a loss of 2, i.e., $q_{(1,0)}(Z_{i}%
,Z_{m},\vartheta)=2$ when $\{I^{+}_{im(1,0)}(\vartheta)=1,\Delta
p_{(1,0)}(Z_{i},Z_{m})\leq0\}$ or $\{I^{-}_{im(1,0)}(\vartheta)=1,\Delta
p_{(1,0)}(Z_{i},Z_{m})\geq0\}$ occurs.
\item[-] A correct prediction or no prediction incurs 0 loss, i.e.,
$q_{(1,0)}(Z_{i},Z_{m},\vartheta)=0$ when $\{I^{+}_{im(1,0)}(\vartheta
)=1,\Delta p_{(1,0)}(Z_{i},Z_{m})\geq0\}$ or $\{I^{-}_{im(1,0)}(\vartheta
)=1,\Delta p_{(1,0)}(Z_{i},Z_{m})\leq0\}$ or $\{I^{+}_{im(1,0)}(\vartheta
)=I^{-}_{im(1,0)}(\vartheta)=0\}$ occurs.
\end{itemize}

Therefore, $Q_{(1,0)}(\vartheta)$ represents the expected loss of making wrong
predictions for the sign of $\Delta p_{(1,0)}(Z_{i},Z_{m})$ based on
inequalities (\ref{eq:lad_ineq_1}) and (\ref{eq:lad_ineq_2}). It is evident
that the true parameter $\theta$ minimizes $Q_{(1,0)}(\vartheta)$ since it
never gives wrong predictions.

The following conditions are sufficient for establishing the identification of
$\theta$ via (\ref{eq:lad_obj}), which are analogous to Assumptions C1--C4:

\begin{itemize}
\item[\textbf{CL1}] (i) $\{(Y_{i},Z_{i})\}_{i=1}^{N}$ are i.i.d. across $i$,
(ii) $(\epsilon_{1},\epsilon_{2},\eta)\perp Z$, and (iii) the joint
distribution of $(\epsilon_{1},\epsilon_{2},\eta)$ is absolutely continuous on
$\mathbb{R}^{2}\times\mathbb{R}_{+}$.

\item[\textbf{CL2}] For any pair of $(i,m)$ and $j=1,2$, denote $X_{imj}%
=X_{ij}-X_{mj}$, $W_{im}=W_{i}-W_{m}$, and $S_{im}= S_{i}-S_{m}$. Then,
(i) $X_{im1}^{(1)}$ ($X_{im2}^{(1)}$) has a.e. positive Lebesgue density on
$\mathbb{R}$ conditional on $(\tilde{X}_{im1},X_{im2},W_{im},S_{im})$
($(\tilde{X}_{im2},X_{im1},W_{im},S_{im})$), (ii) Elements in $(X_{im1},S_{im})$
($(X_{im2},S_{im})$), conditional on $(X_{im2},W_{im})$ ($(X_{im1}
,W_{im})$), are linearly independent, (iii) $W_{im}^{(1)}$ has a.e.
positive Lebesgue density on $\mathbb{R}$ conditional on $(\tilde{W}%
_{im},X_{im1},X_{im2},S_{im})$, and (iv) Elements in $(W_{im},S_{im})$, conditional on
$(X_{im1},X_{im2})$, are linearly independent.

\item[\textbf{CL3}] $\theta=(\beta,\gamma,\rho_{1},\rho_{2},\rho_{b})\in
\Theta\subset\mathbb{R}^{k_{1}+k_{2}+3k_{3}}$, where $\Theta=\{\vartheta
=(b,r,\varrho_{1},\varrho_{2},\varrho_{b})|b^{(1)}=r^{(1)}=1,\Vert
\vartheta\Vert\leq C\}$ for some constant $C>0$.

\item[\textbf{CL4}] $F_{1}(\cdot,\cdot)$, $F_{2}(\cdot,\cdot)$, and
$F_{b}(\cdot)$ are respectively $\mathbb{R}^{2}\mapsto\mathbb{R}$,
$\mathbb{R}^{2}\mapsto\mathbb{R}$, and $\mathbb{R}\mapsto\mathbb{R}$ functions
strictly increasing in each of their arguments. $F_{b}(0)=0$.
\end{itemize}

The following theorem, proved in Appendix \ref{appendixA}, summarizes our identification result.

\begin{theorem}
\label{thm:cross_theta_identification} Under Assumptions CL1--CL4,
$Q_{(1,0)}(\theta)<Q_{(1,0)}(\vartheta)$ for all $\vartheta\in\Theta
\setminus\{\theta\}$.
\end{theorem}

The true parameter $\theta$ can also be identified using an alternative
population criterion function, replacing the $1[\Delta p_{(1,0)}(Z_{i}%
,Z_{m})\geq0]$ and $1[\Delta p_{(1,0)}(Z_{i},Z_{m})\leq0]$ in
(\ref{eq:lad_obj}) with $\Delta p_{(1,0)}(Z_{i},Z_{m})$ and $-\Delta
p_{(1,0)}(Z_{i},Z_{m})$, respectively:
\begin{equation}
Q_{(1,0)}^{D}(\vartheta)=\mathbb{E}\left[  q_{(1,0)}^{D}(Z_{i},Z_{m}%
,\vartheta)\right]  , \label{eq:lad_pop_obj}%
\end{equation}
where
\begin{align*}
q_{(1,0)}^{D}(Z_{i},Z_{m},\vartheta)=  &  [|I_{im(1,0)}^{+}(\vartheta)-\Delta
p_{(1,0)}(Z_{i},Z_{m})|+|I_{im(1,0)}^{-}(\vartheta)+\Delta p_{(1,0)}%
(Z_{i},Z_{m})|]\\
&  \times\lbrack I_{im(1,0)}^{+}(\vartheta)+I_{im(1,0)}^{-}(\vartheta
)]+[1-(I_{im(1,0)}^{+}(\vartheta)+I_{im(1,0)}^{-}(\vartheta))].
\end{align*}
In fact, minimizing criterion functions (\ref{eq:lad_pop_obj}) and
(\ref{eq:lad_obj}) is equivalent since $|\Delta p_{(1,0)}(Z_{i},Z_{m})|<1$
and thus
\begin{align*}
\left\{q_{(1,0)}(Z_{i},Z_{m},\vartheta)=0\right\}   &  \Leftrightarrow
\left\{q_{(1,0)}^{D}(Z_{i},Z_{m},\vartheta)=1\right\}  \text{ and}\\
\left\{q_{(1,0)}(Z_{i},Z_{m},\vartheta)=2\right\}   &  \Leftrightarrow
\left\{q_{(1,0)}^{D}(Z_{i},Z_{m},\vartheta)=1+2|\Delta p_{(1,0)}(Z_{i}%
,Z_{m})|\right\}  .
\end{align*}

Though criterion function (\ref{eq:lad_obj}) is more intuitive, in practice we
recommend to construct estimation procedure based on (\ref{eq:lad_pop_obj})
rather than (\ref{eq:lad_obj}). This is because $\Delta p_{(1,0)}(Z_{i}%
,Z_{m})$ is unobservable, and hence a feasible estimation procedure will need
to plug in its (uniformly) consistent estimate $\Delta\hat{p}_{(1,0)}%
(Z_{i},Z_{m})$. However, in finite samples $\Delta p_{(1,0)}(Z_{i},Z_{m})$ and
$\Delta\hat{p}_{(1,0)}(Z_{i},Z_{m})$ can have different signs with positive
probability, leading to finite-sample bias in the estimation. Using
(\ref{eq:lad_pop_obj}) can effectively reduce the impact of such problem in
comparison to (\ref{eq:lad_obj}). To illustrate this, consider the situation
where $\Delta p_{(1,0)}(Z_{i},Z_{m})>0$ but $\Delta\hat{p}_{(1,0)}(Z_{i}%
,Z_{m})<0$. A simple calculation shows that this will generate an error of
size 2 to $q_{(1,0)}(Z_{i},Z_{m},\vartheta)$ in (\ref{eq:lad_obj}) and an
error of size $2|\Delta\hat{p}_{(1,0)}(Z_{i},Z_{m})|$ to $q_{(1,0)}^{D}%
(Z_{i},Z_{m},\vartheta)$ in (\ref{eq:lad_pop_obj}). Since $\Delta\hat
{p}_{(1,0)}(Z_{i},Z_{m})\overset{p}{\rightarrow}\Delta p_{(1,0)}(Z_{i},Z_{m}%
)$, as the sample size increases, we expect not only $P(\Delta p_{(1,0)}%
(Z_{i},Z_{m})>0,\Delta\hat{p}_{(1,0)}(Z_{i},Z_{m})<0)\rightarrow0$ but also
$\Delta\hat{p}_{(1,0)}(Z_{i},Z_{m})\approx0$ when its sign is different from
that of $\Delta p_{(1,0)}(Z_{i},Z_{m})$. Consequently, the bias caused by the
inconsistency in signs between $\Delta p_{(1,0)}(Z_{i},Z_{m})$ and $\Delta
\hat{p}_{(1,0)}(Z{i},Z_{m})$ decreases much faster when using
(\ref{eq:lad_pop_obj}) compared to (\ref{eq:lad_obj}). From this perspective,
criterion function (\ref{eq:lad_pop_obj}) can be considered a debiased version
of (\ref{eq:lad_obj}).

Suppose there is a uniformly consistent estimator $\Delta\hat{p}_{(1,0)}%
(Z_{i},Z_{m})$ for $\Delta p_{(1,0)}(Z_{i},Z_{m})$. We propose the following
LAD estimator $\hat{\theta}$ for $\theta$:
\begin{equation}
\hat{\theta}=\arg\min_{\vartheta\in\Theta}\sum_{i=1}^{N-1}\sum_{m>i}\hat
{q}^{D}_{im(1,0)}(\vartheta), \label{eq:lad_estimator}%
\end{equation}
where
\begin{align*}
\hat{q}_{im(1,0)}^{D}(\vartheta)=  &  [\vert I_{im(1,0)}^{+}(\vartheta)-\Delta
\hat{p}_{(1,0)}(Z_{i},Z_{m})\vert+\vert I_{im(1,0)}^{-}(\vartheta)+\Delta\hat
{p}_{(1,0)}(Z_{i},Z_{m})\vert]\\
&  \times[I_{im(1,0)}^{+}(\vartheta)+I_{im(1,0)}^{-}(\vartheta)]+[1-(I_{im(1,0)}
^{+}(\vartheta)+I_{im(1,0)}^{-}(\vartheta))].
\end{align*}

We need the following assumption on top of Assumptions CL1--CL4 to establish the
consistency of the estimator $\hat{\theta}$ proposed in
(\ref{eq:lad_estimator}):
\begin{itemize}
\item[\textbf{CL5}] $\sup_{(z_{i},z_{m})\in\mathcal{Z}^{2}}\vert\Delta\hat{p}_{(1,0)}(z_{i},z_{m})-\Delta p_{(1,0)}(z_{i},z_{m})\vert\overset{p}{\rightarrow}0$
as $n\rightarrow\infty$, where $\mathcal{Z}$ is the support of $Z$.
\end{itemize}

Assumption CL5, when combined with Assumption CL2, essentially requires a
uniformly consistent estimator for $\Delta p_{(1,0)}(Z_{i},Z_{m})$ on a
non-compact support.\footnote{In empirical industrial organization and marketing research, choice
models are often studied using aggregate data. Researchers can observe the
aggregated choice probabilities (market shares) along with covariates of
market-level. In this scenario, Assumption CL5 can be satisfied under mild
conditions. See Section 6 of \cite{ShiEtal2018} for a more detailed
discussion.} To our knowledge, the nearest neighbor estimator
(\cite{devroye1978uniform}) and the kernel regression estimator
(\cite{hansen2008uniform}) can be shown to satisfy such uniform consistency
condition under certain restrictions on the tail probability bounds for $Z$.
For example, \cite{hansen2008uniform} proved that when $\Delta\hat{p}%
_{(1,0)}(Z_{i},Z_{m})$ is obtained using kernel regression, Assumption CL5 is
satisfied if elements in $Z$ have distribution tails thinner than polynomial.
Some recent advances in machine learning can also be used to estimate $\Delta
p_{(1,0)}(Z_{i},Z_{m})$, especially for high-dimensional cases. However,
similar results for uniform convergence with non-compact support are lacking
in the literature. \cite{kohler2021rate} demonstrated that, with certain
restrictions on the network architecture and smoothness assumptions on the
regression function, the multi-layer neural network regression can circumvent
the curse of dimensionality. They derive a uniform convergence rate of the
multi-layer neural network estimator, assuming a compact support for
covariates. In our simulation studies presented in Section \ref{monte}, we use the
neural network method to estimate the choice probabilities in a panel data
bundle choice model (see Section \ref{SEC:LAD_panel}). Our panel data LAD estimator using such
plug-in estimates performs satisfactorily despite the lack of theoretical guarantees.

The following theorem, proved in Appendix \ref{appendixA}, shows the consistency of
$\hat{\theta}$.

\begin{theorem}
\label{thm:cross_theta_consistency} Suppose Assumptions CL1--CL5 hold. We have
$\hat{\theta}\overset{p}{\rightarrow}\theta$.
\end{theorem}

Using the plug-in nonparametric estimates $\Delta\hat{p}_{(1,0)}(Z_{i},Z_{m})$
complicates the derivation of the convergence rate and asymptotic distribution
of the estimator, especially for its panel data counterpart which will be
discussed in Section \ref{SEC:LAD_panel}. We typically need $\Delta\hat{p}_{(1,0)}(Z_{i},Z_{m})$
converges to $\Delta p_{(1,0)}(Z_{i},Z_{m})$ sufficiently fast to prevent its
prediction error from dominating the rate of the estimator. More
importantly, we conjecture that applying certain transformations or smoothing
techniques (e.g., \cite{Horowitz1992}) to criterion function (\ref{eq:lad_obj}%
) (or (\ref{eq:lad_pop_obj})) may yield an estimator with a better
convergence rate and facilitate the derivation of its asymptotic distribution.
However, exploring these issues is a non-trivial task, and due to space
constraints, we must defer a comprehensive investigation to future research.
We conclude this section with a final remark.

\begin{remark}\label{remark:LAD1}
In addition to (\ref{eq:lad_ineq_1}) and (\ref{eq:lad_ineq_2}), we can also
construct population and sample criterion functions based on the following
inequalities:
\begin{align}
\{X_{im1}^{\prime}\beta+S^{\prime}_{im}\rho_{1}\leq0,X_{im2}^{\prime}%
\beta+S_{im}^{\prime}\rho_{2}\geq0,W_{im}^{\prime}\gamma+S^{\prime}_{im}\rho
_{b}\leq0\}  &  \Rightarrow\Delta p_{(0,1)}(Z_{i},Z_{m})\geq
0,\label{eq:lad_ineq_3}\\
\{X_{im1}^{\prime}\beta+S^{\prime}_{im}\rho_{1}\geq0,X_{im2}^{\prime}%
\beta+S_{im}^{\prime}\rho_{2}\leq0,W_{im}^{\prime}\gamma+S^{\prime}_{im}\rho
_{b}\geq0\}  &  \Rightarrow\Delta p_{(0,1)}(Z_{i},Z_{m})\leq
0,\label{eq:lad_ineq_4}\\
\{X_{im1}^{\prime}\beta+S^{\prime}_{im}\rho_{1}\geq0,X_{im2}^{\prime}%
\beta+S_{im}^{\prime}\rho_{2}\geq0,W_{im}^{\prime}\gamma+S^{\prime}_{im}\rho
_{b}\geq0\}  &  \Rightarrow\Delta p_{(1,1)}(Z_{i},Z_{m})\geq
0,\label{eq:lad_ineq_5}\\
\{X_{im1}^{\prime}\beta+S^{\prime}_{im}\rho_{1}\leq0,X_{im2}^{\prime}%
\beta+S_{im}^{\prime}\rho_{2}\leq0,W_{im}^{\prime}\gamma+S^{\prime}_{im}\rho
_{b}\leq0\}  &  \Rightarrow\Delta p_{(1,1)}(Z_{i},Z_{m})\leq0,
\label{eq:lad_ineq_6}\\
\{X_{im1}^{\prime}\beta+S^{\prime}_{im}\rho_{1}\leq0,X_{im2}^{\prime}%
\beta+S_{im}^{\prime}\rho_{2}\leq0,W_{im}^{\prime}\gamma+S^{\prime}_{im}\rho
_{b}\leq0\}  &  \Rightarrow\Delta p_{(0,0)}(Z_{i},Z_{m})\geq
0,\label{eq:lad_ineq_7}\\
\{X_{im1}^{\prime}\beta+S^{\prime}_{im}\rho_{1}\geq0,X_{im2}^{\prime}%
\beta+S_{im}^{\prime}\rho_{2}\geq0,W_{im}^{\prime}\gamma+S^{\prime}_{im}\rho
_{b}\geq0\}  &  \Rightarrow\Delta p_{(0,0)}(Z_{i},Z_{m})\leq
0,
\label{eq:lad_ineq_8}%
\end{align}
where $\Delta p_{d}(Z_{i},Z_{m})$ for $d\in\{(0,1),(1,1),(0,0)\}$ are similarly
defined as $\Delta p_{(1,0)}(Z_{i},Z_{m})$ in (\ref{eq:lad_ineq_1})--(\ref{eq:lad_ineq_2}). Note that the following two events are not informative
in the sense that they are not sufficient to predict the sign of any $\Delta
p_{d}(Z_{i},Z_{m})$:
\begin{align*}
&  \{X_{im1}^{\prime}\beta+S^{\prime}_{im}\rho_{1}\le0,X_{im2}^{\prime}%
\beta+S_{im}^{\prime}\rho_{2}\leq0,W_{im}^{\prime}\gamma+S^{\prime}_{im}\rho
_{b}\geq0\},\\
&  \{X_{im1}^{\prime}\beta+S^{\prime}_{im}\rho_{1}\geq0,X_{im2}^{\prime}%
\beta+S_{im}^{\prime}\rho_{2}\geq0,W_{im}^{\prime}\gamma+S^{\prime}_{im}\rho
_{b}\leq0\}.
\end{align*}

For any $d\in\{(0,1),(1,1),(0,0)\}$ corresponding to (\ref{eq:lad_ineq_3}%
)--(\ref{eq:lad_ineq_4}), (\ref{eq:lad_ineq_5})--(\ref{eq:lad_ineq_6}), and
(\ref{eq:lad_ineq_7})--(\ref{eq:lad_ineq_8}), respectively, $(I_{imd}%
^{+}(\vartheta),I_{imd}^{-}(\vartheta))$ can be defined analogously to
$(I_{im(1,0)}^{+}(\vartheta),I_{im(1,0)}^{-}(\vartheta))$ in
(\ref{eq:lad_ineq_1_2})--(\ref{eq:lad_ineq_2_2}). Then, any single criterion
$Q_{d}(\vartheta)=\mathbb{E}[q_{d}(Z_{i},Z_{m},\vartheta)]$ ($Q_{d}%
^{D}(\vartheta)=\mathbb{E}[q_{d}^{D}(Z_{i},Z_{m},\vartheta)]$) or their convex
combination can be used to identify $\theta$. For estimation purposes, we
recommend to use a sample criterion constructed by combining the sample
analogues of all available population criteria (e.g., summing them up), to
achieve higher finite sample efficiency. For example, in the simulation
studies presented in Section \ref{monte}, we compute
\[
\hat{\theta}=\arg\min_{\vartheta\in\Theta}\sum_{i=1}^{N-1}\sum_{m>i}\sum
_{d\in\mathcal{D}}\hat{q}_{imd}^{D}(\vartheta),
\]
where $\mathcal{D}=\{(0,0),(1,0),(0,1),(1,1)\}$ and $\hat{q}_{imd}^{D}(\vartheta)$ is similarly defined as in (\ref{eq:lad_estimator}).
\end{remark}

\section{Panel Data Model}\label{SEC3}
The structure of this section is similar to that of Section \ref{SEC2}. Section \ref{SEC3.1} introduces a localized MS procedure, in which we establish the identification and the limiting distribution of the estimator (Section \ref{SEC3.1.1}), show the validity of the numerical bootstrap inference procedure (Section \ref{SEC:infer_panel}), and propose a test for the interaction effects (Section \ref{SEC:paneltest}). Section \ref{SEC:LAD_panel} develops a LAD estimation analogous to Section \ref{SEC:LAD}. 

\subsection{Maximum Score Estimation}\label{SEC3.1}
\subsubsection{Localized MS Estimator}\label{SEC3.1.1}
The increasing availability of panel data sets provides new opportunities
for the econometrician to control for unobserved heterogeneity across
agents. This helps with the relaxing of the strict exogeneity restriction
placed in the cross-sectional model. In this section we apply the
matching-based identification strategy presented in Section \ref{sec:rcm}
to a panel data bundle choice model. The latent utilities and observed
choices are expressed as follows:
\begin{equation}
U_{dt}=\sum_{j=1}^{2}F_{j}(X_{jt}^{\prime }\beta ,\alpha _{j},\epsilon
_{jt})\cdot d_{j}+ \eta _{t}\cdot F_{b}\left(W_{t}^{\prime }\gamma
+\alpha _{b} \right) \cdot d_{1}\cdot d_{2},  \label{panelutility}
\end{equation}%
and
\begin{equation}
Y_{dt}=1[U_{dt}>U_{d^{\prime }t},\forall d^{\prime }\in \mathcal{D}\setminus
d],  \label{panelchoice}
\end{equation}%
where we use $t=1,...,T$ to denote time periods and suppress the agent
subscript $i$ to simplify notations. Again, we assume that functions $F_j(\cdot,\cdot,\cdot)$'s and $F_b(\cdot,\cdot)$ are strictly increasing in their arguments. The specification
considered here is a natural extension of the cross-sectional model
(\ref{crossutility})--(\ref{choicemodel}). Similar extensions are common in various
discrete choice models. See, e.g., \cite{PakesPorter2016}, \cite{ShiEtal2018}, and \cite{gao2020robust}. 
Note that the random utility specified in (\ref{panelutility})
includes a set of unobserved (to the econometrician) agent-specific fixed
effects $\alpha \equiv (\alpha _{1},\alpha _{2},\alpha _{b})$
associated with the two stand-alone alternatives and the bundle,
respectively. Denote $Z_{t}=(X_{1t},X_{2t},W_{t})$ and $\xi _{t}=(\epsilon _{1t},\epsilon _{2t},\eta
_{t})$. In line with the literature on fixed effects methods, we
place no restrictions on the distribution of $\alpha $ conditional on $%
Z^{T}\equiv \left( Z_{1},...,Z_{T}\right)$
and $\xi ^{T}\equiv \left( \xi _{1},...,\xi _{T}\right)
$. 

Here we consider the identification and estimation of model (\ref%
{panelutility})--(\ref{panelchoice}) with $T<\infty $ and $N\rightarrow
\infty $ (i.e., short panel). For any $T\geq 2$, our identification strategy
relies on a similar conditional homogeneity assumption as those adopted by \cite%
{Manski1987}, \cite{PakesPorter2016}, \cite{ShiEtal2018}, and \cite{gao2020robust} for binary and
multinomial choice models. Specifically, we assume $\xi _{s}\overset{d}{=}%
\xi _{t}|(\alpha ,Z_{s},Z_{t})$ for any two time periods $s$ and $t$.

This restriction is much weaker than the strong exogeneity condition needed for the cross-sectional model. However, thanks to the panel data structure, it suffices to establish the following moment inequalities in the presence of the fixed effects $\alpha $: for all $d$ with $d_{1}=1$ and any fixed $\left(
x_{2},w,c\right)$,
\begin{equation}
x_{1}^{\prime }\beta \geq \tilde{x}_{1}^{\prime }\beta\Leftrightarrow\begin{array}{c}
P(Y_{\left( 1,d_{2}\right)t
}=1|X_{1t}=x_{1},X_{2t}=x_{2},W_t=w,\alpha =c)\\
\geq\\
P(Y_{\left( 1,d_{2}\right)s
}=1|X_{1s}=\tilde{x}_{1},X_{2s}=x_{2},W_s=w,\alpha =c).
\end{array}\label{panelmib1}
\end{equation}
Similarly, for all $d$ with $d_{1}=0$ and any fixed $\left(
x_{2},w,c\right)$, $x_{1}^{\prime }\beta \geq
\tilde{x}_{1}^{\prime }\beta $ is the if-and-only-if condition to the second
inequality in (\ref{panelmib1}) but with a ``$\leq $'' instead.

For $d=(1,1)$ and any fixed $\left( x_{1},x_{2},c\right)
$, we have
\begin{equation}
w^{\prime }\gamma \geq \tilde{w}^{\prime }\gamma\Leftrightarrow\begin{array}{c}
P(Y_{\left( 1,1\right)t
}=1|X_{1t}=x_{1},X_{2t}=x_{2},W_t=w,\alpha =c)\\
\geq\\
P(Y_{(1,1)s}=1|X_{1s}=x_{1},X_{2s}=x_{2},W_s=\tilde{w},\alpha =c).
\end{array}\label{panelmir3}
\end{equation}
Similarly, for all $d\neq (1,1)$ and any fixed $\left(
x_{1},x_{2},c\right)$, $w^{\prime }\gamma
\geq \tilde{w}^{\prime }\gamma $ is the if-and-only-if condition to the
second inequality in (\ref{panelmir3}) but with a ``$\leq $'' instead.

The moment inequalities (\ref{panelmib1})--(\ref{panelmir3})\footnote{%
Once $\beta $ is identified through (\ref{panelmib1}), $\gamma $ can be
alternatively identified via matching indexes $(X_{1t}^{\prime }\beta
,X_{2t}^{\prime }\beta )$; that is, for $d=(1,1)$ and any fixed $\left(
v_{1},v_{2}\right)$, $w^{\prime }\gamma \geq \tilde{w}^{\prime
}\gamma \Leftrightarrow P(Y_{(1,1)t}=1|X^{\prime }_{1t}\beta =v_{1},X_{2t}^{\prime
}\beta =v_{2},W_t=w,\alpha =c)\geq P(Y_{(1,1)s}=1|X^{\prime }_{1s}\beta
=v_{1},X_{2s}^{\prime }\beta =v_{2},W_s=\tilde{w},\alpha =c))$. However, due to
some technical issues (see Appendix \ref{appendixAdd}), we do not adopt
these moment inequalities to construct the MS estimator of $\gamma $.} are
derived from the monotonicity conditions of agents' choices, analogous to
(\ref{crossmi1})--(\ref{crossmi3}) for the cross-sectional model. The main
difference is that, in the presence of the fixed effects, we match and make
our comparisons within agents over time, as opposed to pairs of agents.
These moment inequalities are the foundation of our identification result
for model (\ref{panelutility})--(\ref{panelchoice}).

We impose the following conditions for identification. For notational
convenience, we let $X_{jts} \equiv X_{jt}-X_{js}$ and $W_{ts} \equiv
W_{t}-W_{s}$ for all $j=1,2$ and $(s,t)$.

\begin{enumerate}
\item[\textbf{P1}] (i) $\{(Y_{i}^{T},Z_{i}^{T})\}_{i=1}^{N}$ are i.i.d,
where $Y_{i}^{T}\equiv
\{(Y_{i(0,0)t},Y_{i(1,0)t},Y_{i(0,1)t},Y_{i(1,1)t})\}_{t=1}^T$, (ii) for almost all $%
(\alpha ,Z^{T})$, $\xi _{t}\overset{d}{=}\xi _{s}|\left( \alpha
,Z_{t},Z_{s}\right) $, and (iii) the distribution of $\xi _{t}$ conditional
on $\alpha $ is absolutely continuous w.r.t. the Lebesgue measure on $%
\mathbb{R}^{2}\times \mathbb{R}_{+}$ for all $t$.

\item[\textbf{P2}] $(\beta,\gamma)\in
\mathcal{B}\times \mathcal{R}$, where $\mathcal{B}=\{b\in \mathbb{R}%
^{k_{1}}|\left\Vert b\right\Vert =1,b^{(1)}\neq 0\}$ and $\mathcal{R}=\{r\in
\mathbb{R}^{k_{2}}|\left\Vert r\right\Vert =1,r^{(1)}\neq 0\}$.\footnote{We adopt the same scale normalization as in \cite{KimPollard1990} and \cite{SeoOtsu2018}. This simplifies the asymptotic analysis of our MS estimators, enabling us to use the results established in these works directly. It is easy to verify that the normalization used here is equivalent to that in Assumption C3.}

\item[\textbf{P3}] (i) $X_{1ts}^{(1)}$ ($X_{2ts}^{(1)}$) has a.e. positive
Lebesgue density on $\mathbb{R}$ conditional on $\tilde{X}_{1ts}$ ($\tilde{X}%
_{2ts}$) and conditional on $\left( X_{2ts},W_{ts}\right)$ ($\left( X_{1ts},W_{ts}\right)
$) in a neighborhood of $\left( X_{2ts},W_{ts}\right)$ ($\left( X_{1ts},W_{ts}\right)$) near zero, (ii) the support of $X_{1ts}$ ($X_{2ts}$),
conditional on $\left( X_{2ts},W_{ts}\right)$
($\left( X_{1ts},W_{ts}\right)$) in a
neighborhood of $\left( X_{2ts},W_{ts}\right)$
($\left( X_{1ts},W_{ts}\right)$) near zero,
is not contained in any proper linear subspace of $\mathbb{R}^{k_{1}}$,
(iii) $W_{ts}^{(1)}$ has a.e. positive Lebesgue density on $\mathbb{R}$
conditional on $\tilde{W}_{ts}$ and conditional on $\left(
X_{1ts},X_{2ts}\right) $ in a neighborhood of $\left( X_{1ts},X_{2ts}\right)
$ near zero, and (iv) the support of $W_{ts}$, conditional on $\left(
X_{1ts},X_{2ts}\right) $ in a neighborhood of $\left( X_{1ts},X_{2ts}\right)
$ near zero, is not contained in any proper linear subspace of $\mathbb{R}%
^{k_{2}}$.

\item[\textbf{P4}] $F_{j}\left( \cdot ,\cdot ,\cdot \right) ,$ $j=1,2,$ and $%
F_{b}\left( \cdot \right) $ are strictly increasing in their arguments. $F_{b}(0)=0$.
\end{enumerate}

Assumptions P1--P3 are analogous to the identification conditions used by
\cite{Manski1987}. Note that Assumption P1 allows for arbitrary correlation
between the fixed effects and the observed covariates, provided that the
correlation is time stationary. This assumption substantially relaxes the
strict exogeneity condition in the cross-sectional case, but the price to
pay is a slower convergence rate, as shown in Theorem \ref{T:paneldist}
below. A technical explanation for the slower convergence rate can be found
in Appendix \ref{appendixAdd}.

We summarize our identification results in the following theorem. The proof is omitted for brevity, as it closely resembles that of the cross-sectional model.

\begin{theorem}
\label{Thm:panelid} Suppose Assumptions P1--P4 hold. Then $\beta$ and $%
\gamma $ are identified.
\end{theorem}

The monotonic relations (\ref{panelmib1})--(\ref{panelmir3}) motivate the
following two-step localized MS estimation procedure. Given a random sample
of $N$ agents $i=1,...,N$, we obtain the estimator $\hat{\beta}$ of $\beta$
through maximizing the criterion function
\begin{align}
\mathcal{L}_{N,\beta}^{P,K}(b)=\sum_{i=1}^{N}\sum_{t>s}\sum_{d\in\mathcal{D}%
} & \{\mathcal{K}_{h_{N}}(X_{i2ts},W_{its})(Y_{ids}-Y_{idt})\text{sgn}%
(X_{i1ts}^{\prime}b)\cdot\left( -1\right) ^{d_{1}}  \notag \\
& +\mathcal{K}_{h_{N}}(X_{i1ts},W_{its})(Y_{ids}-Y_{idt})\text{sgn}%
(X_{i2ts}^{\prime}b)\cdot\left( -1\right) ^{d_{2}}\},  \label{panelobjbK}
\end{align}
where $\mathcal{K}_{h_{N}}(X_{ijts},W_{its})\equiv%
\prod_{l=1}^{k_{1}}h_{N}^{-1}K\left( \left. X_{ijts,l}\right/ h_{N}\right)
\prod_{l=1}^{k_{2}}h_{N}^{-1}K\left( \left. W_{its,l}\right/ h_{N}\right) $
for $j=1,2$, $K\left( \cdot\right) $ is a kernel density function, and $%
h_{N} $ is a bandwidth sequence that converges to 0 as $N\rightarrow\infty$.
Similarly, we compute the estimator $\hat{\gamma}$ of $\gamma$ by maximizing the
criterion function 
\begin{equation}
\mathcal{L}_{N,\gamma}^{P,K}(r)=\sum_{i=1}^{N}\sum_{t>s}\mathcal{K}%
_{\sigma_{N}}(X_{i1ts},X_{i2ts})(Y_{i(1,1)t}-Y_{i(1,1)s})\text{sgn}%
(W_{its}^{\prime}r),  \label{panelobjrK}
\end{equation}
where $\mathcal{K}_{\sigma_{N}}(X_{i1ts},X_{i2ts})\equiv\prod_{l=1}^{k_{1}}%
\sigma_{N}^{-1}K\left( \left. X_{i1ts,l}\right/ \sigma_{N}\right)
\prod_{l=1}^{k_{1}}\sigma_{N}^{-1}K\left( \left. X_{i2ts,l}\right/
\sigma_{N}\right) $ for $j=1,2$, and $\sigma_{N}$ is a bandwidth sequence
that converges to 0 as $N\rightarrow\infty$. Unlike the cross-sectional case, here we choose not to estimate $\gamma$ by
matching $X_{jt}^{\prime}\hat{\beta}$ and $X_{js}^{\prime}\hat{\beta}$, $%
j=1,2$. The reason is explained in Appendix \ref{appendixAdd}. Note that criterion functions (%
\ref{panelobjbK}) and (\ref{panelobjrK}) take value 0 for observations whose
choice is time invariant; that is, our approach uses only data on
\textquotedblleft switchers\textquotedblright\ in a way similar to the
estimator in \cite{Manski1987}.

To establish the asymptotic properties of $\hat{\beta}$ and $\hat{\gamma}$,
we need the following conditions in addition to Assumptions P1--P4:
\begin{enumerate}
\item[\textbf{P5}] Let $f_{X_{jt},W_{t}}(\cdot)$ denote the density of the
random vector $(X_{jt},W_{t})$ for all $j=1,2$,
and $f_{X_{1ts},X_{2ts}}(\cdot)$ denote the density of the random vector $%
\left( X_{1ts},X_{2ts}\right)$. $f_{X_{jt},W_{t}}(\cdot)$'s and $%
f_{X_{1ts},X_{2ts}}(\cdot)$ are absolutely continuous, bounded from above on
their supports, strictly positive in a neighborhood of zero, and twice
continuous differentiable a.e. with bounded derivatives.

\item[\textbf{P6}] For all $d\in\mathcal{D}$, $\mathbb{E}[Y_{dts}\text{sgn}%
\left( X_{1ts}^{\prime}b\right) |X_{2ts},W_{ts}]$ and $\mathbb{E}[Y_{dts}%
\text{sgn}\left( X_{2ts}^{\prime}b\right) |X_{1ts},W_{ts}]$ are twice
continuously differentiable w.r.t. $b$ a.e. with bounded derivatives, and $%
\mathbb{E}[Y_{dts}\text{sgn}\left( W_{ts}^{\prime}r\right) |X_{1ts},X_{2ts}]$
is twice continuously differentiable w.r.t. $r$ a.e. with bounded
derivatives.

\item[\textbf{P7}] $K(\cdot)$ is a function of bounded variation and has a
compact support. It satisfies: (i) $K(v)\geq0$ and $\sup_{v}|K(v)|<\infty$,
(ii) $\int K(v)$d$v=1$, (iii) $\int vK(v)$d$v=0$, (iv) $\int v^{2}K(v)$d$%
v<\infty$ and (v) $K(\cdot)$ is twice continuously differentiable with
bounded first and second derivatives.

\item[\textbf{P8}] $(\hat{\beta},\hat{\gamma})$ satisfies
\begin{equation*}
N^{-1}\mathcal{L}_{N,\beta}^{P,K}(\hat{\beta})\geq\max_{b\in\mathcal{B}%
}N^{-1}\mathcal{L}_{N,\beta}^{P,K}(b)-o_{p}((Nh_{N}^{k_{1}+k_{2}})^{-2/3})
\end{equation*}
and
\begin{equation*}
N^{-1}\mathcal{L}_{N,\gamma}^{P,K}(\hat{\gamma})\geq\max_{r\in\mathcal{R}%
}N^{-1}\mathcal{L}_{N,\gamma}^{P,K}(r)-o_{p}((N\sigma_{N}^{2k_{1}})^{-2/3}).
\end{equation*}

\item[\textbf{P9}] $h_{N}$ and $\sigma_{N}$ are sequences of positive
numbers such that as $N\rightarrow\infty$: (i) $h_{N}\rightarrow0$ and $%
\sigma _{N}\rightarrow0$, (ii) $Nh_{N}^{k_{1}+k_{2}}\rightarrow\infty$ and $%
N\sigma_{N}^{2k_{1}}\rightarrow\infty,$ and (iii) $%
(Nh_{N}^{k_{1}+k_{2}})^{2/3}h_{N}^{2}\rightarrow0$ and $(N%
\sigma_{N}^{2k_{1}})^{2/3}\sigma_{N}^{2}\rightarrow0$.
\end{enumerate}

The boundedness and smoothness restrictions placed on Assumptions P5--P7 are
regularity conditions needed for proving the uniform convergence of the
criterion functions to their population analogues. Assumption P8 is a
standard technical condition. Assumption P9 is also standard, where P9(iii) is included to ensure that the bias term from the kernel estimation is asymptotically negligible. Note that Assumption P6 implicitly assumes that the second
moments of $X_{jts}$'s and $W_{ts}$ exist.

For the asymptotic distribution, we focus on the case where all regressors
are continuous, and introduce the following notations to ease the
exposition. Let
\begin{align*}
\phi_{Ni}\left( b\right) & \equiv\sum_{t>s}\sum_{d\in\mathcal{D}} \{
\mathcal{K}_{h_{N}}(X_{i2ts},W_{its})Y_{idst}\left( -1\right) ^{d_{1}}\left(
1 [ X_{i1ts}^{\prime}b>0 ] -1 [ X_{i1ts}^{\prime}\beta>0 ] \right) \\
& +\mathcal{K}_{h_{N}}(X_{i1ts},W_{its})Y_{idst}\left( -1\right)
^{d_{2}}\left( 1 [ X_{i2st}^{\prime}b>0 ] -1 [ X_{i2st}^{\prime}\beta>0 ]
\right) \}
\end{align*}
and
\begin{equation*}
\varphi_{Ni}\left( r\right) \equiv\sum_{t>s}\mathcal{K}_{%
\sigma_{N}}(X_{i1ts},X_{i2ts})Y_{i(1,1)ts} ( 1 [ W_{its}^{\prime}r>0 ] -1 [
W_{its}^{\prime}\gamma>0 ] ) .
\end{equation*}
Note that $\hat{\beta}$ and $\hat{\gamma}$ can be equivalently obtained from
\begin{equation*}
\hat{\beta}=\arg\max_{b\in\mathcal{B}}N^{-1}\sum_{i=1}^{N}\phi_{Ni}\left(
b\right) \text{ and }\hat{\gamma}=\arg\max_{r\in\mathcal{R}}N^{-1}\sum
_{i=1}^{N}\varphi_{Ni}\left( r\right) ,
\end{equation*}
because $\text{sgn}\left( \cdot\right) =2\times1\left[ \cdot\right] -1$ and
adding terms not related to $b$ or $r$ has no effects on the optimization.

\begin{theorem}
\label{T:paneldist} Suppose Assumptions P1--P9 hold. Then,
\begin{equation*}
(Nh_{N}^{k_{1}+k_{2}})^{1/3}(\hat{\beta}-\beta)\overset{d}{\rightarrow}%
\arg\max_{\rho\in\mathbb{R}^{k_{1}}}\mathcal{Z}_{1}\left( \rho\right) ,
\end{equation*}
where $\mathcal{Z}_{1}$ is a Gaussian process taking values in $\ell^{\infty
}\left( \mathbb{R}^{k_{1}}\right) ,$ with $\mathbb{E}\left( \mathcal{Z}%
_{1}\left( \rho\right) \right) =\frac{1}{2}\rho^{\prime}\mathbb{V}\rho,$
covariance kernel $\mathbb{H}_{1}\left( \rho_{1},\rho_{2}\right) ,$ and $%
\rho,\rho_{1},\rho_{2}\in\mathbb{R}^{k_{1}},$ and%
\begin{equation*}
(N\sigma_{N}^{2k_{1}})^{1/3}\left( \hat{\gamma}-\gamma\right) \overset{d}{%
\rightarrow}\arg\max_{\delta\in\mathbb{R}^{k_{2}}}\mathcal{Z}_{2}\left(
\delta\right) ,
\end{equation*}
where $\mathcal{Z}_{2}$ is a Gaussian process taking values in $\ell^{\infty
}\left( \mathbb{R}^{k_{2}}\right) ,$ with $\mathbb{E}\left( \mathcal{Z}%
_{2}\left( \delta\right) \right) =\frac{1}{2}\delta^{\prime}\mathbb{W}\delta$%
, covariance kernel $\mathbb{H}_{2}\left( \delta_{1},\delta _{2}\right) ,$
and $\delta,\delta_{1},\delta_{2}\in\mathbb{R}^{k_{2}}.$ $\mathbb{V}$, $%
\mathbb{W}$, $\mathbb{H}_{1}$, and $\mathbb{H}_{2}$ are defined,
respectively, by equations (\ref{EQ:V}), (\ref{EQ:W}), (\ref{EQ:H1}), and (%
\ref{EQ:H2}) in Appendix \ref{appendixB}.  
\end{theorem}

The convergence rates of $\hat{\beta}$ and $\hat{\gamma}$ decline as the number of alternatives, $J$, increases due to the necessity of matching more covariates. We illustrate this in the model with $J=3$ in Section \ref{SEC:idenJ3P}. In line with existing results on the MS estimators (e.g., \cite{KimPollard1990} and \cite{SeoOtsu2018}), the limiting distributions of $\hat{\beta}$ and $\hat{\gamma}$ are non-Gaussian, and their rates of convergence are slower than $N^{-1/3}$. We refer interested readers to Appendix \ref{appendixAdd} for further insights into the convergence rates of both cross-sectional and panel data estimators. Inference using the limiting distribution derived in Theorem \ref{T:paneldist} directly is rather difficult. Therefore, we recommend a bootstrap-based procedure for the inference in the next section. One alternative is to adopt a smoothed MS approach (e.g., \cite{Horowitz1992}), which may yield a faster and asymptotically normal estimator. We leave this topic for future research.

\subsubsection{Inference\label{SEC:infer_panel}}
\cite{AbrevayaHuang2005} proved the inconsistency of the classic bootstrap
for the ordinary MS estimator. Our panel data estimators are indeed of the
MS type and we thus expect that the classic bootstrap does not work for them
either. 
Valid inference can be made by the $m$-out-of-$n$ bootstrap, according to
\cite{LeePun2006}. Recently, \cite{HongLi2020} proposed the numerical
bootstrap, and showed the superior performance of this procedure over the $m$%
-out-of-$n$ bootstrap. For our panel data estimators, another advantage of
the numerical bootstrap is that we do not need to choose another set of
bandwidths ($h_{N}$ and $\sigma_{N}$) for estimation using the bootstrap
series as with the $m$-out-of-$n$ bootstrap. Based on these considerations,
we propose to conduct the inference using the numerical bootstrap.

The numerical bootstrap estimators $\hat{\beta}^{\ast}$ and $\hat{\gamma }%
^{\ast}$ are obtained as follows. First, draw $\{
(Y_{i}^{T\ast},Z_{i}^{T\ast})\} _{i=1}^{N}$ independently
from the collection of the sample values $\{(
Y_{i}^{T},Z_{i}^{T})\} _{i=1}^{N}$ with replacement. Then,
obtain $\hat{\beta}^{\ast}$ from
\begin{equation}
\hat{\beta}^{\ast}=\arg\max_{b\in\mathcal{B}}N^{-1}\sum_{i=1}^{N}\phi
_{Ni}\left( b\right) +\left( N\varepsilon_{N1}\right) ^{1/2}\cdot\lbrack
N^{-1}\sum_{i=1}^{N}\phi_{Ni}^{\ast}\left( b\right)
-N^{-1}\sum_{i=1}^{N}\phi_{Ni}\left( b\right) ],  \label{EQ:betabootstrapP}
\end{equation}
where $\varepsilon_{N1}$ is a tuning parameter to be discussed later, and $%
\phi_{Ni}^{\ast}\left( b\right) $ is the same as $\phi_{Ni}\left( b\right) $
except it uses the sampling series $\{
(Y_{i}^{T\ast},Z_{i}^{T\ast})\} _{i=1}^{N}$ as inputs. Similarly, compute $\hat{\gamma }%
^{\ast}$ from
\begin{equation}
\hat{\gamma}^{\ast}=\arg\max_{r\in\mathcal{R}}N^{-1}\sum_{i=1}^{N}\varphi
_{Ni}\left( r\right) +\left( N\varepsilon_{N2}\right) ^{1/2}\cdot\lbrack
N^{-1}\sum_{i=1}^{N}\varphi_{Ni}^{\ast}\left( r\right)
-N^{-1}\sum_{i=1}^{N}\varphi_{Ni}\left( r\right) ],
\label{EQ:gammabootstrapP}
\end{equation}
where $\varepsilon_{N2}$ is a tuning parameter, and $\varphi_{Ni}^{\ast
}\left( r\right) $ is similarly defined using bootstrap series.

When $\varepsilon _{N1}^{-1}=\varepsilon _{N2}^{-1}=N,$ the numerical
bootstrap reduces to the classic nonparametric bootstrap. The numerical
bootstrap excludes this case and requires $N\varepsilon _{N1}\rightarrow
\infty $ and $N\varepsilon _{N2}\rightarrow \infty $ as $N\rightarrow \infty
$. We note that $\varepsilon _{N1}^{-1}$ and $\varepsilon _{N2}^{-1}$ play a
similar role to the $m$ in the $m$-out-of-$n$ bootstrap (use only $m$
observations for the estimation).\footnote{%
Note that our $\varepsilon _{N}$ was written as $\varepsilon _{N}^{1/2}$ in
\cite{HongLi2020}.} The following conditions for the $\varepsilon _{N}$'s
are required for the validity of the numerical bootstrap:

\begin{enumerate}
\item[\textbf{P10}] $\varepsilon _{N1}$ and $\varepsilon _{N2}$ are
sequences of positive numbers such that as $N\rightarrow \infty $: (i) $%
\varepsilon _{N1}\rightarrow 0$ and $\varepsilon _{N2}\rightarrow 0$, (ii) $%
N\varepsilon _{N1}\rightarrow \infty $ and $N\varepsilon _{N2}\rightarrow
\infty $, (iii) $\varepsilon _{N1}^{-1}h_{N}^{k_{1}+k_{2}}\rightarrow \infty
$ and $\varepsilon _{N2}^{-1}\sigma _{N}^{2k_{1}}\rightarrow \infty $, and
(iv) $(\varepsilon _{N1}^{-1}h_{N}^{k_{1}+k_{2}})^{2/3}h_{N}^{2}\rightarrow
0 $ and $(\varepsilon _{N2}^{-1}\sigma _{N}^{2k_{1}})^{2/3}\sigma
_{N}^{2}\rightarrow 0$.
\end{enumerate}

The following theorem shows the validity of the numerical bootstrap, whose proof is deferred to Appendix \ref%
{appendixB}. We note that our estimators do not directly satisfy all conditions required by
\cite{HongLi2020}, for example, condition (vi) in their Theorem 4.1 is not the case here. 

\begin{theorem}
\label{Thm:panelinfer} Suppose Assumptions P1--P10 hold. Then
\begin{equation*}
(\varepsilon _{N1}^{-1}h_{N}^{k_{1}+k_{2}})^{1/3}(\hat{\beta}^{\ast }-\hat{%
\beta})\overset{d}{\rightarrow }\arg \max_{\rho \in \mathbb{R}^{k_{1}}}%
\mathcal{Z}_{1}^{\ast }\left( \rho \right) \text{\textit{\ }\emph{%
conditional on the sample},}
\end{equation*}%
and
\begin{equation*}
(\varepsilon _{N2}^{-1}\sigma _{N}^{2k_{1}})^{1/3}\left( \hat{\gamma}^{\ast
}-\hat{\gamma}\right) \overset{d}{\rightarrow }\arg \max_{\delta \in \mathbb{%
R}^{k_{2}}}\mathcal{Z}_{2}^{\ast }\left( \delta \right) \text{ \emph{%
conditional on the sample},}
\end{equation*}%
where $\mathcal{Z}_{1}^{\ast }\left( \rho \right) $ and $\mathcal{Z}%
_{2}^{\ast }\left( \delta \right) $ are independent copies of $\mathcal{Z}%
_{1}\left( \rho \right) $ and $\mathcal{Z}_{2}\left( \delta \right) $
defined in Theorem \ref{T:paneldist}, respectively.
\end{theorem}

We conclude this section by discussing the choice of the tuning parameters. Recall that $\hat{\beta}$
is of rate $(Nh_{N}^{k_{1}+k_{2}})^{-1/3}$, and additionally we need $%
(Nh_{N}^{k_{1}+k_{2}})^{2/3}h_{N}^{2}\rightarrow 0$ (to handle the bias). To
attain a fast rate of convergence, we tend to set $h_{N}$ as large as
possible. For example, we may set $(Nh_{N}^{k_{1}+k_{2}})^{2/3}h_{N}^{2}%
\propto \log \left( N\right) ^{-1}$. For the same reason, we may set $%
(N\sigma _{N}^{2k_{1}})^{2/3}\sigma _{N}^{2}\propto \log \left( N\right)
^{-1}$. These lead to $h_{N}\propto \log \left( N\right) ^{-\frac{3}{%
2k_{1}+2k_{2}+6}}N^{-\frac{1}{k_{1}+k_{2}+3}}$ and $\sigma _{N}\propto \log
\left( N\right) ^{^{-\frac{3}{4k_{1}+6}}}N^{-\frac{1}{2k_{1}+3}}.$ As
recommended by \cite{HongLi2020}, we can choose $\varepsilon _{N1}$ and $%
\varepsilon _{N2}$ such that $\varepsilon
_{N1}^{-1}h_{N}^{k_{1}+k_{2}}\propto (Nh_{N}^{k_{1}+k_{2}})^{2/3}$ and $%
\varepsilon _{N2}^{-1}\sigma _{N}^{2k_{1}}\propto (N\sigma
_{N}^{2k_{1}})^{2/3},$ which then implies that%
\begin{equation}
\varepsilon _{N1}\propto N^{-\frac{k_{1}+k_{2}+2}{k_{1}+k_{2}+3}}\log \left(
N\right) ^{-\frac{k_{1}+k_{2}}{2k_{1}+2k_{2}+6}}\text{ and }\varepsilon
_{N2}\propto N^{-\frac{2k_{1}+2}{2k_{1}+3}}\log \left( N\right) ^{-\frac{%
k_{1}}{2k_{1}+3}}.  \label{varepsilon}
\end{equation}

\subsubsection{Testing $\protect\eta$}\label{SEC:paneltest}
In this section, we use a similar approach as in Section \ref{SEC:test} to propose a test for whether $\eta_t$ tends to degenerate to 0. The hypotheses are formulated as
\begin{align*}
&\mathbb{H}_{0}:\eta_t >0\text{ almost surely and }\mathbb{E}\left( \eta_t \right) >0, \\
& \mathbb{H}_{1}:\eta_t =0\text{ almost surely.}
\end{align*}
We will just briefly outline how to construct this test since the approach is similar to that of Section \ref{SEC:test}.

The test is based on the criterion 
\begin{equation*}
N^{-1}\mathcal{L}_{N,\gamma }^{P,K}(r)=N^{-1}\sum_{i=1}^{N}\sum_{t>s}%
\mathcal{K}_{\sigma _{N}}(X_{i1ts},X_{i2ts})(Y_{i(1,1)t}-Y_{i(1,1)s})\text{%
sgn}(W_{its}^{\prime }r).
\end{equation*}%
We also define
\begin{equation*}
\mathcal{\bar{L}}^{P}\left( r\right) =\sum_{t>s}f_{X_{1ts},X_{2ts}}\left(
0,0\right) \mathbb{E}\left[ \left. (Y_{i(1,1)t}-Y_{i(1,1)s})\text{sgn}%
(W_{its}^{\prime }r)\right\vert X_{1ts}=0,X_{2ts}=0\right] .
\end{equation*}%
Then, we have
\begin{equation*}
N^{-1}\mathcal{L}_{N,\gamma }^{P,K}(r)\overset{p}{\rightarrow }\mathcal{\bar{%
L}}^{P}\left( r\right)
\end{equation*}%
uniformly in a small neighborhood of $\gamma .$ Using the same reasoning as before, we can conclude that under $\mathbb{H}_{0},$ $\mathcal{\bar{L}}^{P}\left(
\gamma \right) >0,$ and under $\mathbb{H}_{1},$ $\mathcal{\bar{L}}^{P}\left(
\gamma \right) =0.$ We show the limiting distribution of $N^{-1}\mathcal{L}%
_{N,\gamma }^{P,K}(\hat{\gamma})$ in the following theorem. The proof is
deferred to Appendix \ref{appendixB}.

\begin{theorem}
\label{T:paneltest} Suppose Assumptions P1--P9 hold. Then, under $\mathbb{H}%
_{0},$
\begin{equation*}
\sqrt{N\sigma _{N}^{2k_{1}}}( N^{-1}\mathcal{L}_{N,\gamma }^{P,K}(\hat{%
\gamma})-\mathcal{\bar{L}}^{P}\left( \gamma \right) ) \overset{d}{%
\rightarrow }N( 0,\Delta ^{P}_{\gamma}) ,
\end{equation*}%
where
\begin{equation*}
\Delta ^{P}_{\gamma}=\lim_{N\rightarrow \infty }\text{\emph{Var}}\left( \sigma
_{N}^{k_{1}}\sum_{t>s}\mathcal{K}_{\sigma
_{N}}(X_{i1ts},X_{i2ts})(Y_{i(1,1)t}-Y_{i(1,1)s})\text{\emph{sgn}}%
(W_{its}^{\prime }\gamma )\right) .
\end{equation*}
\end{theorem}

Some standard calculations can verify that $\Delta ^{P}_{\gamma}$ is well defined. Importantly, the
plug-in estimator $\hat{\gamma}$ does not affect the asymptotics of $N^{-1}\mathcal{L}%
_{N,\gamma }^{P,K}(\hat{\gamma})$. To understand this, we notice from the proof that
\begin{equation*}
N^{-1}\mathcal{L}_{N,\gamma }^{P,K}(\hat{\gamma})-N^{-1}\mathcal{L}%
_{N,\gamma }^{P,K}(\gamma )=O_{p} ( \left\Vert \hat{\gamma}-\gamma
\right\Vert ^{2} ) +O_{p} ( \sigma _{N}^{2} ) +O_{p} (
\left( N\sigma _{N}^{2k_{1}}\right) ^{-2/3} ) .
\end{equation*}%
Since $\hat{\gamma}-\gamma =O_{p} (  ( N\sigma
_{N}^{2k_{1}} ) ^{-1/3} ) $ by Theorem \ref{T:paneldist}, we can show, by Assumption P9, that 
\begin{equation*}
\sqrt{N\sigma_{N}^{2k_{1}}} ( N^{-1}\mathcal{L}_{N,\gamma }^{P,K}(\hat{\gamma})-N^{-1}%
\mathcal{L}_{N,\gamma }^{P,K}(\gamma ) ) =o_{p}\left( 1\right).
\end{equation*}
For the implementation, one can employ the standard bootstrap with no need to
re-estimate $\hat{\gamma}$. The validity of this approach can be justified in a similar way to what was discussed in Section  \ref{SEC:test}.

\subsection{LAD Estimator for Panel Data Models}\label{SEC:LAD_panel}
The matching-based MS procedure proposed in Section \ref{SEC3.1} shares the same limitation as the MRC estimator in Section \ref{sec:rcm}, as it cannot identify the coefficients for (time-varying) common regressors (e.g., income). In this section, we introduce a LAD estimation method for panel data bundle choice models, similar to the one discussed in Section \ref{SEC:LAD} for cross-sectional models. This LAD estimator can estimate coefficients for both common and alternative-specific regressors. It offers the same level of robustness as the localized MS estimation (\ref{panelobjbK})--(\ref{panelobjrK}), meaning it does not rely on distributional assumptions and can handle general correlations between fixed effects and regressors. The price to pay is that its implementation requires a first-step nonparametric estimation of choice probabilities.

We extend the model (\ref{panelutility})--(\ref{panelchoice}) to accommodate common regressors, similar to what we did in Section \ref{SEC:LAD} for the cross-sectional model. We consider the following panel data model:
\begin{align}
U_{dt} & =\sum_{j=1}^{2}F_{j}(X_{jt}'\beta+S_{t}'\rho_{j},\alpha_{j},\epsilon_{jt})\cdot d_{j}+\eta_{t}\cdot F_{b}(W_{t}'\gamma+S_{t}'\rho_{b}+\alpha_{b})\cdot d_{1}\cdot d_{2},\label{eq:panel_lad_1}\\
Y_{dt} & =1[U_{dt}>U_{d't},\forall d'\in\mathcal{D}\setminus d],\label{eq:panel_lad_2}
\end{align}
where $S_{t} \in \mathbb{R}^{k_{3}}$ represents a vector of common regressors that enter more than one stand-alone alternative or the bundle. All other aspects of the model remain the same as in (\ref{panelutility})--(\ref{panelchoice}). We denote $Z_{t}=(X_{1t},X_{2t},W_{t},S_{t})$ and use $Z^{T}=(Z_{1},...,Z_{T})$ to collect all observed covariates in (\ref{eq:panel_lad_1})--(\ref{eq:panel_lad_2}).

Suppose $\xi_{s} \overset{d}{=} \xi_{t} | (\alpha, Z_{t}, Z_{s})$ for any two time periods $t$ and $s$. We can write the following identification inequalities:
\begin{align}
\{X_{1ts}'\beta+S_{ts}'\rho_{1}\geq0,X_{2ts}'\beta+S_{ts}'\rho_{2}\leq0,W_{ts}'\gamma+S_{ts}'\rho_{b}\leq0\} & \Rightarrow\Delta p_{(1,0)}(Z_{t},Z_{s})\geq0,\label{eq:panel_lad_3}\\
\{X_{1ts}'\beta+S_{ts}'\rho_{1}\leq0,X_{2ts}'\beta+S_{ts}'\rho_{2}\geq0,W_{ts}'\gamma+S_{ts}'\rho_{b}\geq0\} & \Rightarrow\Delta p_{(1,0)}(Z_{t},Z_{s})\geq0,\label{eq:panel_lad_4}
\end{align}
where $S_{ts} \equiv S_{t} - S_{s}$ and $\Delta p_{(1,0)}(Z_{t}, Z_{s}) \equiv P(Y_{(1,0)t} = 1 | Z_{t}, Z_{s}) - P(Y_{(1,0)s} = 1 | Z_{t}, Z_{s})$. It is clear that (\ref{eq:panel_lad_3}) and (\ref{eq:panel_lad_4}) are analogous to (\ref{eq:lad_ineq_1}) and (\ref{eq:lad_ineq_2}). However, in the context of a panel data model, two differences from the cross-sectional model should be noted:
\begin{itemize}
\item[-] We are comparing the same agent in different time periods, rather than different agents.
\item[-] $\Delta p_{(1,0)}(Z_{t}, Z_{s})$ is defined as the difference between choice probabilities of $(1,0)$ in time periods $t$ and $s$, conditioning on all observable covariates from these two periods. This is because we do not impose any restriction on the way that $\alpha$ is dependent on $(Z_{t}, Z_{s})$, and thus we must keep both $Z_{t}$ and $Z_{s}$ in the conditioning sets when integrating out the $\alpha$ in (\ref{eq:panel_lad_3}) and (\ref{eq:panel_lad_4}).\footnote{Recall that
in cross-sectional model, $\Delta p_{(1,0)}(Z_{i},Z_{m})\equiv P(Y_{(1,0)i}=1|Z_{i})-P(Y_{(1,0)m}=1|Z_{m})$ since $(Y_{i},Z_{i})$ and $(Y_{m},Z_{m})$ are assumed to be independent.} 
\end{itemize}

Inequalities (\ref{eq:panel_lad_3}) and (\ref{eq:panel_lad_4}) lead to the following criterion to identify model parameters $\theta \equiv (\beta, \gamma, \rho_{1}, \rho_{2}, \rho_{b})$:
\begin{equation}
Q_{ts(1,0)}^P(\vartheta)=\mathbb{E}\left[  q_{(1,0)}^P(Z_{t},Z_{s},\vartheta
)\right]  , \label{eq:panel_lad_5}%
\end{equation}
where $\vartheta=(b,r,\varrho_{1},\varrho_{2},\varrho_{b})$ is an arbitrary
vector in the parameter space $\Theta$ containing $\theta$ and
\begin{align*}
q_{(1,0)}^P(Z_{t},Z_{s},\vartheta)=  &  [|I_{ts(1,0)}^{+}(\vartheta)-1[\Delta
p_{(1,0)}(Z_{t},Z_{s})\geq0]|+|I_{ts(1,0)}^{-}(\vartheta)-1[\Delta
p_{(1,0)}(Z_{t},Z_{s})\leq0]|]\\
&  \times\lbrack I_{ts(1,0)}^{+}(\vartheta)+I_{ts(1,0)}^{-}(\vartheta)]
\end{align*}
with
\begin{align}
I_{ts(1,0)}^{+}(\vartheta)  &  =1[X_{ts1}^{\prime}b+S_{ts}^{\prime
}\varrho_{1}\geq0,X_{ts2}^{\prime}b+S_{ts}^{\prime}\varrho_{2}\leq
0,W_{ts}^{\prime}r+S^{\prime}_{ts}\varrho_{b}\leq0]\text{ and}%
\label{eq:panel_lad_6}\\
I_{ts(1,0)}^{-}(\vartheta)  &  =1[X_{ts1}^{\prime}b+S_{ts}^{\prime
}\varrho_{1}\leq0,X_{ts2}^{\prime}b+S^{\prime}_{ts}\varrho_{2}\geq
0,W_{ts}^{\prime}r+S^{\prime}_{ts}\varrho_{b}\geq0]. \label{eq:panel_lad_7}%
\end{align}

Under the following assumptions, Theorem \ref{thm:panel_lad_identification} establishes the identification
of $\theta$ based on (\ref{eq:panel_lad_5}). The proof is omitted due to its similarity to that of Theorem \ref{thm:cross_theta_identification}. 
\begin{enumerate}
\item[\textbf{PL1}] (i) $\{(Y_{i}^{T},Z_{i}^{T})\}_{i=1}^{N}$ are i.i.d,
where $Y_{i}^{T}\equiv
\{(Y_{i(0,0)t},Y_{i(1,0)t},Y_{i(0,1)t},Y_{i(1,1)t})\}_{t=1}^T$, (ii) for almost all $%
(\alpha ,Z^{T})$, $\xi _{t}\overset{d}{=}\xi _{s}|\left( \alpha
,Z_{t},Z_{s}\right) $, and (iii) the distribution of $\xi _{t}$ conditional
on $\alpha $ is absolutely continuous w.r.t. the Lebesgue measure on $%
\mathbb{R}^{2}\times \mathbb{R}_{+}$ for all $t$.
\item[\textbf{PL2}]$\theta=(\beta,\gamma,\rho_{1},\rho_{2},\rho_{b})\in
\Theta\subset\mathbb{R}^{k_{1}+k_{2}+3k_{3}}$, where $\Theta=\{\vartheta
=(b,r,\varrho_{1},\varrho_{2},\varrho_{b})|b^{(1)}=r^{(1)}=1,\Vert
\vartheta\Vert\leq C\}$ for some constant $C>0$.
\item[\textbf{PL3}] (i) $X_{1ts}^{(1)}$ ($X_{2ts}^{(1)}$) has a.e. positive
Lebesgue density on $\mathbb{R}$ conditional on $(\tilde{X}_{1ts},X_{2ts},W_{ts},S_{ts})$ ($(\tilde{X}_{2ts},X_{1ts},W_{ts},S_{ts})$), (ii) the support of $(X_{1ts},S_{ts})$ ($(X_{2ts},S_{ts})$) is not contained in any proper linear subspace of $\mathbb{R}^{k_{1}+k_3}$ conditional on $(X_{2ts},W_{ts})$ ($(X_{1ts},W_{ts})$), (iii) $W_{ts}^{(1)}$ has a.e. positive Lebesgue density on $\mathbb{R}$
conditional on $(\tilde{W}_{ts},X_{1ts},X_{2ts},S_{ts})$, and (iv) the support of $(W_{ts},S_{ts})$ is not contained in any proper linear subspace of $\mathbb{R}^{k_{2}+k_3}$ conditional on $(X_{1ts},X_{2ts})$.
\item[\textbf{PL4}] $F_{j}\left( \cdot ,\cdot ,\cdot \right) ,$ $j=1,2,$ and $%
F_{b}\left( \cdot\right) $ are strictly increasing in their arguments. $F_{b}(0)=0$.
\end{enumerate}

\begin{theorem}\label{thm:panel_lad_identification}
Under Assumptions PL1--PL4,
$Q_{ts(1,0)}^P(\theta)<Q_{ts(1,0)}^P(\vartheta)$ for all $\vartheta\in\Theta
\setminus\{\theta\}$.
\end{theorem}

Following similar arguments to (\ref{eq:lad_pop_obj}), we can show that $\theta$ can be equivalently identified by the following criterion, which results in less bias in its empirical counterpart:
\begin{equation}
Q_{ts(1,0)}^{P,D}(\vartheta)=\mathbb{E}\left[  q_{(1,0)}^{P,D}(Z_{t},Z_{s}%
,\vartheta)\right]  , \label{eq:panel_lad_pop_obj}
\end{equation}
where
\begin{align*}
q_{(1,0)}^{P,D}(Z_{t},Z_{s},\vartheta)=  &  [|I_{ts(1,0)}^{+}(\vartheta)-\Delta
p_{(1,0)}(Z_{t},Z_{s})|+|I_{ts(1,0)}^{-}(\vartheta)+\Delta p_{(1,0)}%
(Z_{t},Z_{s})|]\\
&  \times\lbrack I_{ts(1,0)}^{+}(\vartheta)+I_{ts(1,0)}^{-}(\vartheta
)]+[1-(I_{ts(1,0)}^{+}(\vartheta)+I_{ts(1,0)}^{-}(\vartheta))].
\end{align*}

Then, assuming a uniformly consistent estimator $\Delta\hat{p}_{(1,0)}%
(Z_{t},Z_{s})$ for $\Delta p_{(1,0)}(Z_{t},Z_{s})$, we propose the following
LAD estimator $\hat{\theta}$ for $\theta$:
\begin{equation}
\hat{\theta}=\arg\min_{\vartheta\in\Theta}\sum_{i=1}^{N-1}\sum_{t>s}\hat
{q}^{P,D}_{its(1,0)}(\vartheta), \label{eq:panel_lad_estimator}%
\end{equation}
where
\begin{align*}
\hat{q}_{its(1,0)}^{P,D}(\vartheta)=  &  [\vert I_{its(1,0)}^{+}(\vartheta)-\Delta
\hat{p}_{(1,0)}(Z_{it},Z_{is})\vert+\vert I_{its(1,0)}^{-}(\vartheta)+\Delta\hat
{p}_{(1,0)}(Z_{it},Z_{is})\vert]\\
&  \times[I_{its(1,0)}^{+}(\vartheta)+I_{its(1,0)}^{-}(\vartheta)]+[1-(I_{its(1,0)}%
^{+}(\vartheta)+I_{its(1,0)}^{-}(\vartheta))].
\end{align*}
In practice, one can employ identification inequalities analogous to those presented in Remark \ref{remark:LAD1} (i.e., (\ref{eq:lad_ineq_3})--(\ref{eq:lad_ineq_8})) and follow a procedure similar to (\ref{eq:panel_lad_3})--(\ref{eq:panel_lad_estimator}) to compute $\hat{q}_{its(0,1)}^{P,D}(\vartheta)$, $\hat{q}_{its(0,0)}^{P,D}(\vartheta)$, and $\hat{q}_{its(1,1)}^{P,D}(\vartheta)$. Collectively, a more efficient $\hat{\theta}$ can be obtained by minimizing the sample criterion $\sum_{i=1}^{N}\sum_{t>s}\sum_{d\in\mathcal{D}}\hat{q}_{itsd}^{P,D}(\vartheta)$, where $\mathcal{D}=\{(0,0),(1,0),(0,1),(1,1)\}$.

The following assumption, together with Assumptions PL1--PL4, is sufficient for the
consistency of the estimator $\hat{\theta}$ obtained from (\ref{eq:panel_lad_estimator}):
\begin{itemize}
\item[\textbf{PL5}] $\sup_{(z_{t},z_{s})\in\mathcal{Z}_{t}\times\mathcal{Z}_{s}}\vert\Delta\hat{p}_{(1,0)}(z_{t},z_{s})-\Delta p_{(1,0)}(z_{t},z_{s})\vert\overset{p}{\rightarrow}0$
as $n\rightarrow\infty$ for all $t\neq s$, where $\mathcal{Z}_{t}$
is the support of $Z_{t}$.
\end{itemize}

\begin{theorem}\label{thm:panel_lad_consistency}
Suppose Assumptions PL1--PL5 hold. We have
$\hat{\theta}\overset{p}{\rightarrow}\theta$.
\end{theorem}
The proof of Theorem \ref{thm:panel_lad_consistency} follows similar arguments to Theorem \ref{thm:cross_theta_consistency} and is omitted for brevity.

\section{Monte Carlo Experiments}\label{monte}
In this section, we investigate the finite sample performance of our proposed approaches, in both cross-sectional and panel data models, by means of Monte Carlo experiments. We set sample sizes $N=250,500,1000$ for cross-sectional designs, and $N=1000,2500,5000$ for panel data designs, considering the slower convergence rates. All simulation results are obtained from 1000 independent replications. We report these results in tables collected in Appendix \ref{appendixC}.

We start from a benchmark cross-sectional design (Design 1) specified as follows:
\begin{align*}
U_{id} & =\sum_{j=1}^{2}( X_{ij,1}+\beta X_{ij,2}+\rho_j S_i+\epsilon
_{ij})\cdot d_{j}+\eta_{i}\cdot( W_{i,1}+\gamma W_{i,2})\cdot
d_{1}\cdot d_{2},\\
Y_{id} & =1[U_{id}>U_{id^{\prime}},\forall d^{\prime}\neq d],
\end{align*}
where $\beta =\gamma=\rho_1=\rho_2=1$ and $d=(d_{1} ,d_{2})\in\{(0,0),(1,0),(0,1),(1,1)\}$. We impose the scale normalization (Assumption C4) by setting the coefficients on $ X_{ij,1}$ and $W_{i,1}$ to be 1, and thus $(\beta ,\gamma,\rho_1,\rho_2)$ are free parameters to estimate. In this design, we let $\left\{ \{X_{ij,\iota}\}_{j=1,2;\iota=1,2},\{W_{i,\iota}\}_{\iota=1,2},S_i,\{\epsilon _{ij}\}_{j=1,2},\eta_{i}\right\} _{i=1,...,N}$ be independent across $i,j$, and $\iota$ and from each other. $X_{ij,1}$'s and $W_{i,1}$ are $\textrm{Logistic}(0,1)$, $X_{ij,2}$'s are $\textrm{Bernoulli}(1/3)$, $W_{i,2}$, $S_i$, and $\epsilon_{ij}$'s are $N(0,1)$, and $\eta_{i}$ is $\text{Beta}(2,2)$. 

To implement our MRC estimation, note that larger bandwidth means smaller variation. In light of this observation and requirements under Assumptions C6 and C7, we employ the sixth-order Gaussian kernel ($\kappa_{\beta}=6$) and tuning parameter (bandwidth) $h_{N}=c_{1}\hat{\sigma}\cdot N^{-1/8}\log(N)^{1/6}$ for ``matching'' the four continuous regressors ($k_{1}=k_{2}=2$) in the first-step estimation, where $\hat{\sigma}$ represents the standard deviation of the corresponding matching variable. Under this choice, $\sqrt{N}h_{N}^{4}\propto\log(N) ^{2/3}\rightarrow\infty$ and $\sqrt{N}h_{N}^{6}\rightarrow0$. For similar reasons, we adopt the fourth-order Gaussian kernel ($\kappa_{\gamma}=4$) and bandwidth $\sigma_{N}=c_{2}\hat{\sigma}\cdot N^{-1/4} \log(N)^{1/4}$ for the second step, and thus $\sqrt{N}\sigma_{N}^{2} \propto\log(N) ^{1/2}\rightarrow\infty$ and $\sqrt{N}\sigma_{N}^{4} \rightarrow0$. The discrete variable $X_{ij,2}$ is exactly matched across observations in the first step by using an indicator function. To test the sensitivity of our methods to the choice of bandwidths, we experiment with several values of $c_{1}\in\{0.6, 1,1.4,1.8\}$ and several values of $c_{2}\in\{1.6,2,2.4,2.8\}$. It turns out that the results are not sensitive to the choice of $(c_{1},c_{2})$. Thus, to save space, we only report the results of $(c_{1},c_2)=(1,2)$, for which the root mean squared errors (RMSE) are the smallest among all choices. 

The implementation of our proposed LAD estimator requires a first-stage
nonparametric estimation of $\Delta p_{d}(Z_{i},Z_{m})$. For the cross-sectional designs of this section, we use the standard Nadaraya-Watson estimator for conditional probability density (see, e.g., \cite{LiRacine2007}). We employ the fourth-order Gaussian kernel\footnote{We also try the second-order Gaussian kernel, but the simulation results exhibit larger bias than those obtained with higher-order kernels. This is in line with our discussion in Section \ref{SEC:LAD}, i.e., the bias in $\Delta\hat{p}_{d}(Z_{i},Z_{m})$, when making $\Delta\hat{p}_{d}(Z_{i},Z_{m})$ and $\Delta p_{d}(Z_{i},Z_{m})$ be of different signs, leads to the bias in the LAD estimation. Therefore, it is crucial to prioritize estimators that incur smaller bias to estimate $\Delta p_{d}(Z_{i},Z_{m})$.} for continuous variables and the Aitchison-Aitken kernel for (unordered) discrete variables. To save computation time, we simply choose bandwidths for continuous variables using Silverman's ``rule-of-thumb'' and set $N^{-1}$ as the smoothing parameter for discrete variables. When there is no need to repeat the computation thousands of times, we recommend considering more robust, data-driven methods, such as the cross-validation (CV) method proposed by \cite{hall2004cross}, to select all the tuning parameters. The computation is conducted using the state-of-art R package \texttt{np} (\cite{hayfield2008nonparametric}).  

To compute the MRC and LAD estimators, we apply the differential evolution (DE) algorithm (\cite{StornPrice1997}) to globally search the optima for their corresponding sample criterion functions proposed in previous sections. The implementation uses the well-developed and tested R package \texttt{DEoptim} (\cite{mullen2011deoptim}). 

Results for Design 1 are presented in tables numbered ``1'' and so on for other designs. The performance of the MRC estimator, its nonparametric bootstrap inference, and the LAD estimator are reported in tables labeled ``A'', ``B'', and ``C'', respectively. For example, Table \ref{tab:1A} summarizes the performance of the MRC estimator for Design 1, in which we report the estimator's mean bias (MBIAS) and RMSE. Since these statistics are sensitive to outliers, we also present the median bias (MED) and the median absolute deviations (MAD). Table \ref{tab:1B} reports the empirical coverage frequencies (COVERAGE) and lengths (LENGTH) of the 95\% confidence intervals (CI) constructed using the standard bootstrap with $B=299$ independent draws and estimations. Recall that the MRC procedure can only estimate $(\beta,\gamma)$, while the LAD method can also estimate the coefficients $(\rho_1,\rho_2)$ on the common regressor $S_i$. Therefore, in Tables \ref{tab:1A} and \ref{tab:1B}, we only report results for $(\beta,\gamma)$, while in Table \ref{tab:1C} we report results for $(\beta,\gamma,\rho_1,\rho_2)$.

The results for Design 1 are in line with our asymptotic theory. As the sample size increases, the RMSE of the MRC estimator shrinks at the parametric rate, indicating that these procedures are $\sqrt{N}$-consistent. We observe a similar root-$n$ convergence rate for the LAD estimator, although we do not formally derive it in this paper. The MRC estimator has greater RMSE than the LAD estimator, as the latter makes full use of the data and can be more efficient. The nonparametric bootstrap for the MRC estimation also performs well and yields shrinking confidence intervals, with coverage rates approaching 95\% as the sample size increases.

The outcome of Monte Carlo experiments can be significantly affected by the design of the experiment itself. For instance, estimators tend to perform better when all i.i.d. regressors and errors are used compared to cases where either the regressors or errors are correlated. In order to evaluate the performance of our estimators with correlated regressors and errors, we modify Design 1 by letting $\textrm{Corr}(X_{i1,1},X_{i2,1})=\textrm{Corr}(X_{i1.2},X_{i2,2})=\textrm{Corr}(\epsilon_{i1},\epsilon_{i2})=0.5$. All other aspects of Design 1 remain the same. We refer to this modified version of the benchmark design as Design 2. Like Design 1, we use the same kernels and bandwidths in Design 2 and report the simulation results in Tables \ref{tab:2A}--\ref{tab:2C}. The results are very similar to those for Design 1, though as expected, the performance slightly deteriorates. This confirms our theoretical results that our approaches allow for arbitrary correlation in errors.

We proceed to analyze the finite sample performance of the MS estimation, its numerical bootstrap, and the LAD estimator for panel data bundle choice models. We consider two designs (Designs 3 and 4) that have the same choice set as the cross-sectional designs ($\{(0,0),(1,0),(0,1),(1,1)\}$) and a panel of two time periods. We examine sample sizes of $N=1000,2500,5000$ for each panel data design. We construct 95\% CI based on $B=299$ independent draws and estimations for inference. The first panel data design (Design 3) is specified as follows:
\begin{align*}
U_{idt} & =\sum_{j=1}^{2}(X_{ijt,1}+\beta X_{ijt,2}+\rho_j S_{it}+\alpha
_{ij}+\epsilon_{ijt})\cdot d_{j}+\eta_{it}\cdot(W_{it,1}+\gamma
 W_{it,2}+\alpha_{ib})\cdot d_{1}\cdot d_{2},\\
Y_{idt} & =1[U_{idt}>U_{id^{\prime}t},\forall d^{\prime}\neq d],
\end{align*}
where $\beta =\gamma=\rho_1=\rho_{2}=1$. $\left\{ \{X_{ijt,\iota}\}_{j=1,2;\iota=1,2},\{W_{it,\iota}\}_{\iota=1,2},S_{it},\{\epsilon_{ijt}\}_{j=1,2},\eta_{it}\right\} _{i=1,...,N;t=1,2}$ are mutually independent random variables, where $X_{ijt,1}$'s and $W_{it,1}$'s are $\textrm{Logistic}(0,1)$, $X_{ijt,2}$'s are $\textrm{Bernoulli}(1/3)$, $S_{it}$'s, $W_{it,2}$'s, and $\epsilon_{ijt}$'s are $N(0,1)$, and $\eta_{it}\sim\text{Beta}(2,2)$. We let $\alpha_{ij}=(X_{ij1,1} +X_{ij2,1})/4$ for $j=1,2$ and $\alpha_{ib}=(W_{i1,1}+W_{i2,1})/4$. For scale normalization (Assumption P3), we set the coefficients on $X_{ijt,1}$ and $W_{it,1}$ to be 1.\footnote{Note that we employ the scale normalization in Assumption C3 rather than Assumption P2. The two normalization conventions are equivalent, but the former is often more computationally convenient and easier to interpret.} In this design, the (unobserved) fixed effects are correlated with regressors, which invalidates the matching method in the cross-sectional case.

In the second panel data design (Design 4), the random utility model and coefficients are the same as in Design 3. However, similar to Design 2 for the cross-sectional setting, we set $\textrm{Corr}(X_{ijt,1},X_{ils,1})=\textrm{Corr}(X_{ijt,2},X_{ils,2})=\textrm{Corr}(\epsilon_{ijt},\epsilon_{ils})=0.25$ for all $j,l\in\{1,2\}$ and $t,s\in\{1,2\}$ with either $j\neq l$ or $t\neq s$, and $\{S_{it},\{W_{it,\iota}\}_{\iota=1,2},\eta_{it}\}_{i=1,...N;t=1,2}$ to have serial correlation 0.25. All other aspects remain the same.

To implement the MS estimation for these designs, we use the second-order Gaussian kernel and $h_{N}=\sigma_{N}=c_{3}\hat{\sigma}\cdot N^{-1/7}\log(N)^{-1/14}$. We have two more tuning parameters to implement the numerical bootstrap: $\varepsilon_{N1}$ and $\varepsilon_{N2}$. We set $\varepsilon_{N1}=\varepsilon_{N2}=c_{4}\cdot N^{-5/7}\log(N)^{-5/14}$.\footnote{Though \cite{HongLi2020} recommended $\varepsilon_{N1},\varepsilon_{N2}\propto N^{-6/7}\log(N)^{-2/7}$ (see (\ref{varepsilon}) in our Section \ref{SEC:infer_panel}), some tentative Monte Carlo simulations suggest that a slight modification like this leads to a better finite sample performance.} This choice of $(h_{N},\sigma_{N},\varepsilon_{N1}, \varepsilon_{N2})$ meets Assumptions P9 and P10. We experiment with $c_{3},c_4\in\{1,1.5,2, 2.5\}$. It turns out that the results are not sensitive to the choice of $(c_{3},c_{4})$. Therefore, to save space, we report only the results with $(c_{3},c_4)=(2,2)$. 

To implement our proposed LAD estimator, we need to estimate $\Delta p_{d}(Z_{it},Z_{is})$ nonparametrically in the first stage. It is important to note that in order to estimate the probability of $Y_{idt}=1$, we need to use observed (time-varying) regressors from both periods $t$ and $s$ due to the weak restrictions imposed on the unobservables of the model in Assumption P1. This makes the panel data LAD estimator more susceptible to the curse of dimensionality than the cross-sectional model. To address this issue, we employ nonparametric regression based on artificial neural networks with multiple hidden layers. Recent advances in statistics, such as \cite{kohler2021rate} and the references therein, establish the convergence rates of such estimators and demonstrate their ability to handle the curse of dimensionality under certain restrictions on the regression function and network (sparsity) architecture.\footnote{Strictly speaking, the rate results of \cite{kohler2021rate} only apply to cases where the covariates have compact supports, which is not the case for our simulation designs. However, despite the lack of theoretical guarantees, our simulation results do not flag any problems.} Our simulation studies simply use neural networks with two hidden layers and three neurons per hidden layer for computational efficiency. The R package $\texttt{neuralnet}$ (\cite{fritsch2019package}) is used for computation. In practice, one can consider using CV techniques to select the number of hidden layers and neurons.

Tables \ref{tab:3A}--\ref{tab:3C} and Tables \ref{tab:4A}--\ref{tab:4C} show the simulation results for Designs 3 and 4, respectively. According to the results, our MS and LAD estimators, calculated using the DE algorithm, perform reasonably well and yield similar results for these two designs. The MBIAS and RMSE of these estimators decrease as the sample size increases, indicating the consistency of our estimators. The MS estimator has a slower than cube root-$N$ rate of convergence, while the LAD estimator appears to have a rate of convergence between $\sqrt{N}$ and cube root-$N$. The inference results are also consistent with the theory. When the sample size is small, the CI obtained by the numerical bootstrap tends to be too wide, resulting in coverage rates slightly higher than 95\%. However, as the sample size increases, we observe that the CI shrink and their coverage rates approach 95\%.

\section{Conclusions}\label{conclude}
This paper develops estimation and inference procedures for semiparametric discrete choice models for bundles. We first propose a two-step maximum rank correlation estimator in the cross-sectional setting and establish its $\sqrt{N}$-consistency and asymptotic normality. Based on these results, we further show the validity of using the nonparametric bootstrap to carry out inference and propose a criterion-based test for the interaction effects. Additionally, we introduce a multi-index least absolute deviations (LAD) estimator to handle models with both alternative- and agent-specific covariates, which, to our knowledge, is new in the literature. We extend these approaches to the panel data setting, where the resulting maximum score (MS) and LAD estimators can consistently estimate bundle choice models with agent-, alternative-, and bundle-specific fixed effects. We also derive the limiting distribution of the MS estimator, justify the application of the numerical bootstrap for making inferences, and introduce an analogous test for interaction effects based on the MS criterion. It is important to note that although we propose these methods for estimating bundle choice models, they can be applied with minor modifications to many other multi-index models. A small-scale Monte Carlo study shows that our proposed methods perform adequately in finite samples. 

This paper focuses on scenarios where the researcher can access individual-level data. We note that our proposed methods can also be adapted for models that use aggregate data. In these cases, the researcher can observe the aggregated choice probabilities (e.g., market shares) in many markets and the market-level covariates for each alternative. These types of models are commonly used in empirical industrial organizations (see, e.g., \cite{Fan2013} and \cite{ShiEtal2018}). A detailed exploration of this adaptation would require different estimators and asymptotics. Due to space limitations, we do not delve into this topic in this paper and reserve it for future study.

The work here leaves many other open questions for future research. One promising subsequent research is to establish the asymptotic properties of the proposed LAD estimators for both cross-sectional and panel data models, which can further lead to the development of easy-to-implement inference methods. Another potential area of research is to explore the possibility of smoothing the criterion functions to achieve faster convergence rates and attain asymptotic normality for the MS and LAD estimators.

\begin{center} {\large\bf Supplementary Materials} \end{center}
Appendixes \ref{appendixA}–\ref{appendixE} are included in a supplementary appendix. Appendix \ref{appendixA} provides proofs for theorems presented in Section \ref{SEC2} for our cross-sectional estimators, except for Theorem \ref{T:cross_boot}, whose technical details are provided in Appendix \ref{appendixD} due to its length. Appendix \ref{appendixB} provides proofs for all the theorems for our panel data estimators introduced in Section \ref{SEC3}. Appendix \ref{appendixAdd} provides additional discussions on the convergence rates of proposed procedures, the construction of our MS estimator, and the identification of models with more alternatives and bundles. Finally, Appendix \ref{appendixE} collects the proofs of all the technical lemmas used in Appendices \ref{appendixA} and \ref{appendixB}.

\begin{center} {\large\bf Acknowledgements} \end{center}
This work is partially supported by the Faculty of Business, Economics, and Law (BEL), University of Queensland, via the 2019 BEL New Staff Research Start-Up Grant. We thank Shakeeb Khan, Dong-Hyuk Kim, Firmin Tchatoka, Hanghui Zhang, and Yu Zhou for their valuable feedback. We are also grateful for the helpful comments and discussions from all seminar and conference participants at the University of Adelaide and the 31st Australia New Zealand Econometric Study Group Meeting. All errors are our responsibility.

\part*{{\protect\Large Appendix}}\label{appendix}
\begin{appendix}
\section{Tables}\label{appendixC}

\begin{table}[H]
  \renewcommand\thetable{1A} \centering\centering
  \caption{Design 1, Performance of MRC Estimators}\label{tab:1A}
    \begin{tabular}{lccccccccc}
    \toprule
          & \multicolumn{4}{c}{$\beta$}      &       & \multicolumn{4}{c}{$\gamma$} \\
\cmidrule{2-5}\cmidrule{7-10}          & MBIAS  & RMSE  & MED   & MAD   &       & MBIAS  & RMSE  & MED   & MAD \\
\cmidrule{2-5}\cmidrule{7-10}    $N=250$ & 0.068 & 0.534 & 0.075 & 0.424 &       & 0.101 & 0.452 & 0.074 & 0.318 \\
    $N=500$ & 0.024 & 0.367 & 0.028 & 0.274 &       & 0.065 & 0.314 & 0.038 & 0.208 \\
    $N=1000$ & 0.024 & 0.237 & 0.016 & 0.154 &       & 0.040 & 0.215 & 0.034 & 0.131 \\
    \bottomrule
    \end{tabular}%
  
\end{table}%

\begin{table}[H]
  \renewcommand\thetable{1B} \centering\centering
  \caption{Design 1, Performance of Nonparametric Bootstrap}\label{tab:1B}
    \begin{tabular}{lccccc}
    \toprule
          & \multicolumn{2}{c}{$\beta$} &       & \multicolumn{2}{c}{$\gamma$} \\
\cmidrule{2-3}\cmidrule{5-6}          & COVERAGE & LENGTH &       & COVERAGE & LENGTH \\
\cmidrule{2-3}\cmidrule{5-6}    $N=250$ & 0.835 & 2.505 &       & 0.915 & 2.145 \\
    $N=500$ & 0.862 & 1.719 &       & 0.918 & 1.492 \\
    $N=1000$ & 0.926 & 1.164 &       & 0.938 & 1.002 \\
    \bottomrule
    \end{tabular}%
  
\end{table}%

\begin{table}[H]
  \renewcommand\thetable{1C} \centering\centering
  \caption{Design 1, Performance of LAD Estimators}\label{tab:1C}
    \begin{tabular}{lccccccccc}
    \toprule
          & \multicolumn{4}{c}{$\beta$}      &       & \multicolumn{4}{c}{$\gamma$} \\
\cmidrule{2-5}\cmidrule{7-10}          & MBIAS  & RMSE  & MED   & MAD   &       & MBIAS  & RMSE  & MED   & MAD \\
\cmidrule{2-5}\cmidrule{7-10}    $N=250$ & 0.102 & 0.274 & 0.090 & 0.180 &       & 0.040 & 0.418 & 0.017 & 0.247 \\
    $N=500$ & 0.112 & 0.213 & 0.114 & 0.144 &       & 0.032 & 0.274 & 0.004 & 0.173 \\
    $N=1000$ & 0.106 & 0.164 & 0.100 & 0.110 &       & 0.030 & 0.189 & 0.022 & 0.124 \\
    \midrule
          &       &       &       &       &       &       &       &       &  \\
    \midrule
          & \multicolumn{4}{c}{$\rho_1$}    &       & \multicolumn{4}{c}{$\rho_2$} \\
\cmidrule{2-5}\cmidrule{7-10}          & MBIAS  & RMSE  & MED   & MAD   &       & MBIAS  & RMSE  & MED   & MAD \\
\cmidrule{2-5}\cmidrule{7-10}    $N=250$ & 0.022 & 0.208 & 0.011 & 0.133 &       & 0.019 & 0.200 & 0.011 & 0.128 \\
    $N=500$ & 0.010 & 0.134 & 0.005 & 0.092 &       & 0.009 & 0.135 & 0.001 & 0.089 \\
    $N=1000$ & 0.010 & 0.095 & 0.007 & 0.062 &       & 0.013 & 0.092 & 0.011 & 0.061 \\
    \bottomrule
    \end{tabular}%
 
\end{table}%

\begin{table}[H]
  \renewcommand\thetable{2A} \centering\centering
  \caption{Design 2, Performance of MRC Estimators}\label{tab:2A}
    \begin{tabular}{lccccccccc}
    \toprule
          & \multicolumn{4}{c}{$\beta$}      &       & \multicolumn{4}{c}{$\gamma$} \\
\cmidrule{2-5}\cmidrule{7-10}          & MBIAS  & RMSE  & MED   & MAD   &       & MBIAS  & RMSE  & MED   & MAD \\
\cmidrule{2-5}\cmidrule{7-10}    $N=250$ & 0.067 & 0.585 & 0.051 & 0.452 &       & 0.083 & 0.471 & 0.040 & 0.325 \\
    $N=500$ & 0.061 & 0.433 & 0.038 & 0.336 &       & 0.072 & 0.328 & 0.058 & 0.216 \\
    $N=1000$ & 0.029 & 0.283 & 0.015 & 0.182 &       & 0.040 & 0.219 & 0.024 & 0.135 \\
    \bottomrule
    \end{tabular}%
  
\end{table}%

\begin{table}[H]
  \renewcommand\thetable{2B} \centering\centering
  \caption{Design 2, Performance of Nonparametric Bootstrap}\label{tab:2B}
    \begin{tabular}{lccccc}
    \toprule
          & \multicolumn{2}{c}{$\beta$} &       & \multicolumn{2}{c}{$\gamma$} \\
\cmidrule{2-3}\cmidrule{5-6}          & COVERAGE & LENGTH &       & COVERAGE & LENGTH \\
\cmidrule{2-3}\cmidrule{5-6}    $N=250$ & 0.870 & 2.847 &       & 0.911 & 2.292 \\
    $N=500$ & 0.894 & 1.992 &       & 0.927 & 1.566 \\
    $N=1000$ & 0.925 & 1.358 &       & 0.933 & 1.067 \\
    \bottomrule
    \end{tabular}%
  
\end{table}%

\begin{table}[H]
  \renewcommand\thetable{2C} \centering\centering
  \caption{Design 2, Performance of LAD Estimators}\label{tab:2C}
    \begin{tabular}{lccccccccc}
    \toprule
                    & \multicolumn{4}{c}{$\beta$}      &       & \multicolumn{4}{c}{$\gamma$} \\
\cmidrule{2-5}\cmidrule{7-10}          & MBIAS  & RMSE  & MED   & MAD   &       & MBIAS  & RMSE  & MED   & MAD \\
\cmidrule{2-5}\cmidrule{7-10}    $N=250$ & 0.109 & 0.312 & 0.102 & 0.205 &       & 0.059 & 0.491 & 0.018 & 0.297 \\
    $N=500$ & 0.107 & 0.222 & 0.103 & 0.147 &       & 0.048 & 0.308 & 0.023 & 0.189 \\
    $N=1000$ & 0.112 & 0.172 & 0.113 & 0.124 &       & 0.028 & 0.209 & 0.013 & 0.143 \\
    \midrule
          &       &       &       &       &       &       &       &       &  \\
    \midrule
          & \multicolumn{4}{c}{$\rho_1$}    &       & \multicolumn{4}{c}{$\rho_2$} \\
\cmidrule{2-5}\cmidrule{7-10}          & MBIAS  & RMSE  & MED   & MAD   &       & MBIAS  & RMSE  & MED   & MAD \\
\cmidrule{2-5}\cmidrule{7-10}    $N=250$ & -0.042 & 0.209 & -0.050 & 0.139 &       & -0.036 & 0.211 & -0.047 & 0.154 \\
    $N=500$ & -0.034 & 0.146 & -0.040 & 0.098 &       & -0.037 & 0.147 & -0.037 & 0.098 \\
    $N=1000$ & -0.041 & 0.105 & -0.043 & 0.074 &       & -0.037 & 0.100 & -0.039 & 0.073 \\
    \bottomrule
    \end{tabular}%
  
\end{table}%

\begin{table}[H]
  \renewcommand\thetable{3A} \centering\centering
  \caption{Design 3, Performance of MS Estimators}\label{tab:3A}
    \begin{tabular}{lccccccccc}
    \toprule
          & \multicolumn{4}{c}{$\beta$}      &       & \multicolumn{4}{c}{$\gamma$} \\
\cmidrule{2-5}\cmidrule{7-10}          & MBIAS  & RMSE  & MED   & MAD   &       & MBIAS  & RMSE  & MED   & MAD \\
\cmidrule{2-5}\cmidrule{7-10}    $N=1000$ & -0.008 & 0.371 & 0.003 & 0.244 &       & -0.014 & 0.411 & -0.057 & 0.276 \\
    $N=2500$ & 0.006 & 0.332 & 0.003 & 0.236 &       & -0.004 & 0.375 & -0.040 & 0.246 \\
    $N=5000$ & -0.002 & 0.291 & 0.001 & 0.235 &       & 0.009 & 0.367 & -0.016 & 0.252 \\
    \bottomrule
    \end{tabular}%
 
\end{table}%

\begin{table}[H]
  \renewcommand\thetable{3B} \centering\centering
  \caption{Design 3, Performance of Numerical Bootstrap}\label{tab:3B}
    \begin{tabular}{lccccc}
    \toprule
          & \multicolumn{2}{c}{$\beta$} &       & \multicolumn{2}{c}{$\gamma$} \\
\cmidrule{2-3}\cmidrule{5-6}          & COVERAGE & LENGTH &       & COVERAGE & LENGTH \\
\cmidrule{2-3}\cmidrule{5-6}    $N=1000$ & 0.970 & 2.275 &       & 0.967 & 2.837 \\
    $N=2500$ & 0.967 & 1.973 &       & 0.971 & 2.527 \\
    $N=5000$ & 0.954 & 1.743 &       & 0.967 & 2.256 \\
    \bottomrule
    \end{tabular}%
  
\end{table}%

\begin{table}[H]
  \renewcommand\thetable{3C} \centering\centering
  \caption{Design 3, Performance of LAD Estimators}\label{tab:3C}
    \begin{tabular}{lccccccccc}
    \toprule
                    & \multicolumn{4}{c}{$\beta$}      &       & \multicolumn{4}{c}{$\gamma$} \\
\cmidrule{2-5}\cmidrule{7-10}          & MBIAS  & RMSE  & MED   & MAD   &       & MBIAS  & RMSE  & MED   & MAD \\
\cmidrule{2-5}\cmidrule{7-10}              $N=1000$ & 0.046 & 0.375 & 0.036 & 0.250 &       & 0.131 & 0.740 & 0.021 & 0.466 \\
    $N=2500$ & 0.026 & 0.205 & 0.020 & 0.133 &       & 0.069 & 0.440 & 0.027 & 0.281 \\
    $N=5000$ & 0.002 & 0.155 & -0.003 & 0.107 &       & 0.067 & 0.327 & 0.039 & 0.214 \\
    \midrule
          &       &       &       &       &       &       &       &       &  \\
    \midrule
          & \multicolumn{4}{c}{$\rho_1$}    &       & \multicolumn{4}{c}{$\rho_2$} \\
\cmidrule{2-5}\cmidrule{7-10}          & MBIAS  & RMSE  & MED   & MAD   &       & MBIAS  & RMSE  & MED   & MAD \\
\cmidrule{2-5}\cmidrule{7-10}    $N=1000$ & 0.027 & 0.320 & -0.006 & 0.198 &       & 0.031 & 0.318 & 0.000 & 0.203 \\
    $N=2500$ & 0.012 & 0.172 & 0.011 & 0.112 &       & 0.009 & 0.167 & 0.012 & 0.114 \\
    $N=5000$ & 0.002 & 0.118 & -0.002 & 0.080 &       & 0.000 & 0.126 & -0.006 & 0.086 \\
    \bottomrule
    \end{tabular}%
  
\end{table}%

\begin{table}[H]
  \renewcommand\thetable{4A} \centering\centering
  \caption{Design 4, Performance of MS Estimators}\label{tab:4A}
    \begin{tabular}{lccccccccc}
    \toprule
          & \multicolumn{4}{c}{$\beta$}      &       & \multicolumn{4}{c}{$\gamma$} \\
\cmidrule{2-5}\cmidrule{7-10}          & MBIAS  & RMSE  & MED   & MAD   &       & MBIAS  & RMSE  & MED   & MAD \\
\cmidrule{2-5}\cmidrule{7-10}    $N=1000$ & -0.014 & 0.353 & -0.012 & 0.241 &       & -0.003 & 0.411 & -0.034 & 0.280 \\
    $N=2500$ & 0.000 & 0.318 & -0.007 & 0.236 &       & 0.014 & 0.385 & -0.006 & 0.260 \\
    $N=5000$ & 0.016 & 0.274 & 0.017 & 0.232 &       & 0.027 & 0.354 & 0.002 & 0.246 \\
    \bottomrule
    \end{tabular}%
  
\end{table}%

\begin{table}[H]
  \renewcommand\thetable{4B} \centering\centering
  \caption{Design 4, Performance of Numerical Bootstrap}\label{tab:4B}
    \begin{tabular}{lccccc}
    \toprule
          & \multicolumn{2}{c}{$\beta$} &       & \multicolumn{2}{c}{$\gamma$} \\
\cmidrule{2-3}\cmidrule{5-6}          & COVERAGE & LENGTH &       & COVERAGE & LENGTH \\
\cmidrule{2-3}\cmidrule{5-6}    $N=1000$ & 0.976 & 2.159 &       & 0.972 & 2.720 \\
    $N=2500$ & 0.959 & 1.856 &       & 0.969 & 2.383 \\
    $N=5000$ & 0.956 & 1.642 &       & 0.962 & 2.115 \\
    \bottomrule
    \end{tabular}%
  
\end{table}%

\begin{table}[H]
  \renewcommand\thetable{4C} \centering\centering
  \caption{Design 4, Performance of LAD Estimators}\label{tab:4C}
    \begin{tabular}{lccccccccc}
    \toprule
                    & \multicolumn{4}{c}{$\beta$}      &       & \multicolumn{4}{c}{$\gamma$} \\
\cmidrule{2-5}\cmidrule{7-10}          & MBIAS  & RMSE  & MED   & MAD   &       & MBIAS  & RMSE  & MED   & MAD \\
\cmidrule{2-5}\cmidrule{7-10}    $N=1000$ & 0.060 & 0.337 & 0.045 & 0.215 &       & 0.132 & 0.666 & 0.025 & 0.406 \\
    $N=2500$ & 0.011 & 0.172 & 0.000 & 0.113 &       & 0.068 & 0.404 & 0.026 & 0.246 \\
    $N=5000$ & 0.002 & 0.128 & -0.001 & 0.087 &       & 0.052 & 0.340 & 0.006 & 0.209 \\
    \midrule
          &       &       &       &       &       &       &       &       &  \\
    \midrule
          & \multicolumn{4}{c}{$\rho_1$}    &       & \multicolumn{4}{c}{$\rho_2$} \\
\cmidrule{2-5}\cmidrule{7-10}          & MBIAS  & RMSE  & MED   & MAD   &       & MBIAS  & RMSE  & MED   & MAD \\
\cmidrule{2-5}\cmidrule{7-10}    $N=1000$ & 0.069 & 0.328 & 0.029 & 0.197 &       & 0.044 & 0.314 & 0.011 & 0.204 \\
    $N=2500$ & -0.015 & 0.163 & -0.024 & 0.110 &       & -0.009 & 0.167 & -0.019 & 0.116 \\
    $N=5000$ & -0.028 & 0.124 & -0.035 & 0.087 &       & -0.024 & 0.128 & -0.024 & 0.087 \\
    \bottomrule
    \end{tabular}%
 
\end{table}

\section{Proofs for the Cross-Sectional Model}\label{appendixA}

\subsection{Proof of Theorem \ref{T:crossidentify}}

\begin{proof}
[Proof of Theorem \ref{T:crossidentify}]It suffices to show the identification of $\beta $ based on (\ref{crossmi1})
as the same arguments can be applied to the identification of $\beta $ and $%
\gamma $ based on similar moment inequalities. Denote $\Omega _{im}=\{X_{i2}=X_{m2},W_{i}=W_{m}\}$. By Assumption C1, the
monotonic relation (\ref{crossmi1}) implies that for all $d\in \mathcal{D}$
with $d_{1}=1$, $\beta $ maximizes
\begin{equation*}
Q_{1}(b)\equiv \mathbb{E}%
[(P(Y_{i(1,d_{2})}=1|Z_{i})-P(Y_{m(1,d_{2})}=1|Z_{m}))\text{sgn}%
(X_{im1}^{\prime }b)|\Omega _{im}]
\end{equation*}%
for each pair of $(i,m)$. To show that $\beta $ attains a unique maximum,
suppose that there is a $b\in \mathcal{B}$ such that $Q_{1}(b)=Q_{1}(\beta )$%
. In what follows, we show that $b=\beta $ must hold. To this end, we write $P[%
\text{sgn}(X_{im1}^{\prime }b)\neq \text{sgn}(X_{im1}^{\prime }\beta
)|\Omega _{im}]=P[(\tilde{X}_{im1}^{\prime }\beta<-X_{im1}^{(1)}<\tilde{X}_{im1}^{\prime }b)\cup (%
\tilde{X}_{im1}^{\prime }b<-X_{im1}^{(1)}<\tilde{X}%
_{im1}^{\prime }\beta)|\Omega _{im}]$. Note that when $b\neq\beta$,  $P(\tilde{X}%
_{im1}^{\prime }\beta\neq\tilde{X}_{im1}^{\prime }b|\Omega _{im})>0$ under Assumption C2(ii). Then by Assumption C2(i), we have $P[(\tilde{X}_{im1}^{\prime }\beta<-X_{im1}^{(1)}<\tilde{X}_{im1}^{\prime }b)\cup (%
\tilde{X}_{im1}^{\prime }b<-X_{im1}^{(1)}<\tilde{X}_{im1}^{\prime }\beta)|\Omega _{im}]>0$. Combining with (\ref{crossmi1}), this result implies that the event of $X_{im1}^{\prime }b$ and $P(Y_{i(1,d_{2})}=1|Z_{i})-P(Y_{m(1,d_{2})}=1|Z_{m})$ having different signs occurs with strictly positive probability, and hence $Q_{1}(b)<Q_{1}(\beta )$, a contradiction. Since $b$ is arbitrary, we conclude that $Q_{1}(b)\leq Q_{1}(\beta )$ and the equality holds only when $b=\beta$. This completes the proof.
\end{proof}

\subsection{Proof of Theorem \ref{T:crossAsymp}}\label{appendix_asym}
The proof of Theorem \ref{T:crossAsymp} involves utilizing several technical lemmas (i.e., Lemmas \ref{lemmaA1}--\ref{lemmaA7}). We will present these lemmas prior to the main proof. Before delving into the technical details, we will first introduce some new notations, state Assumption C5 that were omitted from the main text,  and provide an outline of the proof process. In what follows, we assume that all regressors
are continuous to streamline the exposition.\footnote{This is a purely technical assumption, which is made to facilitate the derivation of the asymptotic variance of the proposed estimators. With discrete covariates, the asymptotic results summarized in Theorem \ref{T:crossAsymp} still valid but with more complex expressions.}

The following notations will be used for convenience: For all $d\in\mathcal{D}$ and $j=1,2$,
\begin{itemize}
\item[-] $Y_{imd}\equiv Y_{id}-Y_{md}$.

\item[-] $f_{X_{imj},W_{im}}(\cdot )$ and $f_{V_{im}\left( \beta \right)
}(\cdot )$ denote the PDF of the random vectors $(X_{imj},W_{im})$ and $V_{im}(\beta )$, respectively. $%
f_{X_{2},W|X_{1}}(\cdot )$ ($f_{X_{1},W|X_{2}}(\cdot )$) denotes the PDF of $%
(X_{2},W)$ ($(X_{1},W)$) conditional on $X_{1}$ ($X_{2}$). $f_{V(\beta )}(\cdot )$ ($%
f_{V(\beta )|W}(\cdot )$) is the PDF of $V(\beta )$ (conditional on $W$%
). 

\item[-] For any function $g,\ \nabla g(v)$ and $\nabla ^{2}g(v)$ denote the
gradient and Hessian matrix for $g(\cdot )$ evaluated at $v$, respectively.

\item[-] $V_{i}\equiv V_{i}(\beta)$, $\hat{V}_{i}\equiv V_{i}(\hat{\beta})$,
$V_{im}\equiv V_{i}-V_{m}$, and $\hat{V}_{im}\equiv\hat{V}_{i}-\hat{V}_{m}$.

\item[-] For realized values, $v_{i}\equiv v_{i}(\beta)$, $\hat{v}_{i}\equiv
v_{i}(\hat{\beta})$, $v_{im}\equiv v_{i}-v_{m}$, and $\hat{v}_{im}\equiv
\hat{v}_{i}-\hat{v}_{m}$.

\item[-] $B(v_{i},v_{m},w_{i},w_{m})\equiv\mathbb{E}%
[Y_{im(1,1)}|V_{i}=v_{i},V_{m}=v_{m},W_{i}=w_{i},W_{m}=w_{m}]$.

\item[-] $S_{im}(r)\equiv \text{sgn}(w_{im}^{\prime }r)-\text{sgn}%
(w_{im}^{\prime }\gamma )$.

\item[-] Define
\begin{equation}
\varrho _{i}(b)\equiv -\sum_{d\in \mathcal{D}}\{\varrho _{i21d}(b)\cdot
(-1)^{d_{1}}+\varrho _{i12d}(b)\cdot (-1)^{d_{2}}\},  \label{EQ:rhoi}
\end{equation}
where
\begin{equation*}
\varrho _{ijld}(b)\equiv \mathbb{E}\left[ Y_{imd}\left[ \text{sgn}%
(X_{iml}^{\prime }b)-\text{sgn}(X_{iml}^{\prime }\beta )\right]
|Z_{i},X_{mj}=X_{ij},W_{m}=W_{i}\right]
\end{equation*}%
for $j,l=1,2$ with $j\neq l$.

\item[-] Define
\begin{equation}
\tau _{i}(r) \equiv \mathbb{E}[Y_{im(1,1)}\left[ \text{sgn}(W_{im}^{\prime
}r)-\text{sgn}(W_{im}^{\prime }\gamma )\right] |V_{m}\left( \beta \right)
=V_{i}\left( \beta \right) ,W_{i}]. \label{EQ:Tau}
\end{equation}
\item[-] Define 
\begin{align}
&\mu (v_{1},v_{2},r)  \label{EQ:Mu} \\
\equiv &\mathbb{E}\left[ \left. Y_{im(1,1)}\left[ \text{sgn}(W_{im}^{\prime
}r)-\text{sgn}(W_{im}^{\prime }\gamma )\right] \left(
\begin{array}{c}
X_{im1}^{\prime } \\
X_{im2}^{\prime }%
\end{array}%
\right) \right\vert V_{i}\left( \beta \right) =v_{1},V_{m}\left( \beta
\right) =v_{2}\right] f_{V\left( \beta \right) }\left( v_{1}\right). \nonumber
\end{align}
Let $\nabla _{1}\mu (v_{1},v_{2},r)$ denote the partial derivative
of $\mu (v_{1},v_{2},r)$ w.r.t. its first argument. Denote $\nabla
_{1k}^{2}\mu (v_{1},v_{2},r)$ as the partial derivative of $\nabla _{1}\mu
(v_{1},v_{2},r)$ w.r.t. its $k$-th argument and $\nabla _{33}^{2}\nabla
_{1}\mu (v_{1},v_{2},r)$ as the Hessian matrix of $\nabla _{1}\mu
(v_{1},v_{2},r)$ w.r.t. its third argument.

\item[-] $K_{h_{N},\beta}(\cdot)\equiv h_{N}^{k_{1}+k_{2}}\mathcal{K}_{h_{N}}(\cdot )$ and $K_{\sigma_{N},\gamma}(\cdot)\equiv \sigma _{N}^{2}\mathcal{K}_{\sigma _{N}}(\cdot )$. In proofs, we also write $K_{\beta}(\cdot/h_{N})$ and $K_{\gamma}(\cdot/\sigma_{N})$ for $K_{h_{N},\beta}(\cdot)$ and $K_{\sigma_{N},\gamma}(\cdot)$, respectively, when we need to apply changes of variables.
\end{itemize}

The following regularity conditions that were omitted in the main text will be used in proving Lemmas \ref{lemmaA1}--\ref{lemmaA7} and Theorem \ref{T:crossAsymp}: 
\begin{itemize}
\item[\textbf{C5}] (i) $f_{V_{im}\left( \beta \right) }(\cdot )$ and $%
f_{X_{imj},W_{im}}(\cdot )$ for $j=1,2$ are bounded from above on their
supports, strictly positive in a neighborhood of zero, and twice
continuously differentiable with bounded second derivatives, (ii) For all $%
d\in \mathcal{D}$, $b\in \mathcal{B}$, and $r\in \mathcal{R}$, $\mathbb{E}%
[Y_{imd}[\text{sgn}(X_{im1}^{\prime }b)-\text{sgn}(X_{im1}^{\prime }\beta
)]|X_{im2}=\cdot ,W_{im}=\cdot ]$, $\mathbb{E}[Y_{imd}[\text{sgn}%
(X_{im2}^{\prime }b)-\text{sgn}(X_{im2}^{\prime }\beta )]|X_{im1}=\cdot
,W_{im}=\cdot ]$, and $\mathbb{E}[Y_{im(1,1)}[\text{sgn}(W_{im}^{\prime }r)-%
\text{sgn}(W_{im}^{\prime }\gamma )]|V_{im}\left( \beta \right) =\cdot ]$
are continuously differentiable with bounded first derivatives, (iii) $%
\mathbb{E}[Y_{imd}|Z_{i}=\cdot ,Z_{m}=\cdot ]$ is $\kappa _{1\beta }^{\text{%
th}}$ continuously differentiable with bounded $\kappa _{1\beta }^{\text{th}%
} $ derivatives. $f_{X_{2},W|X_{1}}(\cdot )$ ($f_{X_{1},W|X_{2}}(\cdot )$)
is $\kappa _{2\beta }^{\text{th}}$ continuously differentiable with bounded $%
\kappa _{2\beta }^{\text{th}}$ derivatives. Denote $\kappa _{\beta }=\kappa
_{1\beta }+\kappa _{2\beta }$. $\kappa _{\beta }$ is an even integer greater
than $k_{1}+k_{2}$, (iv) $\mathbb{E}[Y_{im(1,1)}|V_{i}=\cdot ,V_{m}=\cdot
,W_{i}=\cdot ,W_{m}=\cdot ]$ is $\kappa _{1\gamma }^{\text{th}}$
continuously differentiable with bounded $\kappa _{1\gamma }^{\text{th}}$
derivatives, and $f_{V(\beta )|W}(\cdot )$ is $\kappa _{2\gamma }^{\text{th}%
} $ continuously differentiable with bounded $\kappa _{2\gamma }^{\text{th}}$
derivatives. Denote $\kappa _{\gamma }=\kappa _{1\gamma }+\kappa _{2\gamma }$%
. $\kappa _{\gamma }$ is an even integer greater than 2, and (v) $\mathbb{E}%
[\Vert X_{imj}\Vert ^{2}]<\infty $ for all $j=1,2,$ and $\mathbb{E}[\Vert
W_{im}\Vert ^{2}]<\infty $. All derivatives in this assumption are with
respect to $\cdot $.
\end{itemize}
The boundedness and smoothness restrictions under Assumption C5 are needed for proving the uniform convergence of the criterion functions to their population analogues. 

In what follows, we will only show the asymptotics of $\hat{\gamma}$. The
proof for $\hat {\beta}$ is omitted since one can derive the asymptotics of $%
\hat{\beta}$ by repeating the proof process for $\hat{\gamma}$ but skipping
the step handling the plug-in first step estimates (i.e., Lemma \ref{lemmaA7}%
). Throughout this section, we take it as a given that $\hat{\beta}$ is $%
\sqrt{N}$-consistent and asymptotically normal.

For ease of illustration, with a bit of abuse of notation, we will work with
criterion function
\begin{equation*}
\hat{\mathcal{L}}_{N}^{K}(r)\equiv\frac{1}{\sigma_{N}^{2}N(N-1)}\sum_{i\neq
m}K_{\sigma_{N},\gamma}(V_{im}(\hat{\beta}))h_{im}(r),
\end{equation*}
where $K_{\sigma_{N},\gamma}(\cdot)\equiv\sigma_{N}^{2}\mathcal{K}_{\sigma
_{N}}(\cdot)$ and $h_{im}(r)\equiv Y_{im(1,1)}[\text{sgn}(W_{im}^{\prime }r)-%
\text{sgn}(W_{im}^{\prime}\gamma)]$. Note that here we subtract the term $\text{sgn}(W_{im}^{\prime}\gamma)$ from
criterion function (\ref{crossobjrK}), analogous to \cite{Sherman1993}.
Doing this does not affect the value of the estimator, and will facilitate
the proofs that follow. Besides, we define
\begin{equation*}
\mathcal{L}_{N}^{K}(r)=\frac{1}{\sigma_{N}^{2}N(N-1)}\sum_{i\neq
m}K_{\sigma_{N},\gamma}(V_{im}(\beta))h_{im}(r)
\end{equation*}
and
\begin{equation*}
\mathcal{L}(r)\equiv f_{V_{im}(\beta)}(0)\mathbb{E}[h_{im}(r)|V_{im}(%
\beta)=0].
\end{equation*}

We establish the consistency of $\hat{\gamma }$ in Lemmas \ref{lemmaA1}--\ref%
{lemmaA3} by applying Theorem 2.1 in \cite{NeweyMcFadden1994}. The key step
is to show the uniform convergence of $\hat{\mathcal{L}}_{N}^{K}(r)$ for $%
r\in\mathcal{R}$. To this end, we bound the differences among $\hat {%
\mathcal{L}}_{N}^{K}(r)$, $\mathcal{L}_{N}^{K}(r)$, and $\mathcal{L}(r)$ in
Lemmas \ref{lemmaA1} and \ref{lemmaA2}.

The next step is to show the asymptotic normality of $\hat{\gamma}$ by
applying Theorem 2 of \cite{Sherman1994AoS}. Sufficient conditions for this
theorem are that $\hat{\gamma}-\gamma=O_{p}(N^{-1/2})$ and uniformly over a
neighborhood of $\gamma$ with a radius proportional to $N^{-1/2}$,
\begin{equation}
\hat{\mathcal{L}}_{N}^{K}(r)=\frac{1}{2}(r-\gamma)^{\prime}\mathbb{V}%
_{\gamma }(r-\gamma)+\frac{1}{\sqrt{N}}(r-\gamma)^{\prime}\varPsi%
_{N}+o_{p}(N^{-1}),  \label{eq:A2}
\end{equation}
where $\mathbb{V}_{\gamma}$ is a negative definite matrix and $\varPsi_{N}$
is asymptotically normal, with mean zero and variance $\mathbb{V}_{\varPsi}$%
. To verify equation (\ref{eq:A2}), we first show that uniformly over a
neighborhood of $\gamma$, $\mathcal{R}_{N}\equiv\{r\in\mathcal{R}|\Vert
r-\gamma\Vert\leq\delta_{N}\}$ with $\{\delta_{N}\}=O(N^{-\delta})$ for some
$0<\delta\leq1/2$,
\begin{equation}
\hat{\mathcal{L}}_{N}^{K}(r)=\frac{1}{2}(r-\gamma)^{\prime}\mathbb{V}%
_{\gamma }(r-\gamma)+\frac{1}{\sqrt{N}}(r-\gamma)^{\prime}\varPsi%
_{N}+o_{p}(\Vert r-\gamma\Vert^{2})+O_{p}(\varepsilon_{N}),  \label{eq:A3}
\end{equation}
which is the task of Lemmas \ref{lemmaA4}--\ref{lemmaA7}. The $%
O_{p}(\varepsilon_{N})$ term in equation (\ref{eq:A3}) can be shown to be of
order $o_{p}(N^{-1})$ uniformly over $\mathcal{R}_{N}$. Further, by Theorem
1 of \cite{Sherman1994ET}, equation (\ref{eq:A3}) implies $\hat{\gamma}%
-\gamma=O_{p}(\sqrt{\varepsilon_{N}}\vee N^{-1/2})=O_{p}(N^{-1/2})$. This
rate result, together with equation (\ref{eq:A3}), further verifies equation
(\ref{eq:A2}).

To obtain equation (\ref{eq:A3}), we will work with the following expansion
\begin{align}
\hat{\mathcal{L}}_{N}^{K}(r) & =\mathcal{L}_{N}^{K}(r)+\Delta\mathcal{L}%
_{N}^{K}(r)+R_{N}  \notag \\
& =\mathbb{E}[\mathcal{L}_{N}^{K}(r)]+\frac{2}{N}\sum_{i}\left\{ \mathbb{E}[%
\mathcal{L}_{N}^{K}(r)|Z_{i}]-\mathbb{E}[\mathcal{L}_{N}^{K}(r)]\right\}
+\rho_{N}(r)+\Delta\mathcal{L}_{N}^{K}(r)+R_{N},  \label{eq:A4}
\end{align}
where
\begin{equation}
\Delta\mathcal{L}_{N}^{K}(r)\equiv\frac{1}{\sigma_{N}^{3}N(N-1)}\sum_{i\neq
m}\nabla K_{\sigma_{N},\gamma}(V_{im}(\beta))^{\prime}(V_{im}(\hat{\beta }%
)-V_{im}\left( \beta\right) )h_{im}(r),  \label{eq:A5}
\end{equation}%
\begin{equation*}
\rho_{N}(r)=\mathcal{L}_{N}^{K}(r)-\frac{2}{N}\sum_{i}\mathbb{E}[\mathcal{L}%
_{N}^{K}(r)|Z_{i}]+\mathbb{E}[\mathcal{L}_{N}^{K}(r)],
\end{equation*}
and $R_{N}$ denotes the remainder term in the expansion of higher order (as $%
\sqrt{N}\sigma_{N}\rightarrow\infty$). The first three terms in (\ref{eq:A4}%
) are the U-statistic decomposition for $\mathcal{L}_{N}^{K}(r)$ (see
e.g., \cite{Sherman1993} and \cite{Serfling2009}). Lemmas \ref{lemmaA4}--\ref%
{lemmaA6} establish asymptotic properties of these three terms,
respectively. A linear representation for the fourth term in (\ref{eq:A4})
is derived in Lemma \ref{lemmaA7}.

Here we present Lemmas \ref{lemmaA1}--\ref{lemmaA7}, whose proofs are
relegated to Appendix \hyperref[appendixE]{F}. 

\begin{lemma}
\label{lemmaA1} Under Assumptions C1--C7, $\sup_{r\in \mathcal{R}}|\hat{%
\mathcal{L}}_{N}^{K}(r)-\mathcal{L}_{N}^{K}(r)|=o_{p}(1)$.
\end{lemma}

\begin{lemma}
\label{lemmaA2} Under Assumptions C1 and C5--C7, $\sup_{r\in \mathcal{R}%
}|\mathcal{L}_{N}^{K}(r)-\mathcal{L}(r)|=o_{p}(1)$.
\end{lemma}

\begin{lemma}
\label{lemmaA3} Under Assumptions C1--C7,, $\hat{\gamma}\overset{p}{%
\rightarrow }\gamma $.
\end{lemma}

\begin{lemma}
\label{lemmaA4} Under Assumptions C1--C7,, uniformly over $\mathcal{R}_{N}$,
we have
\begin{equation*}
\mathbb{E}[\mathcal{L}_{N}^{K}(r)]=\frac{1}{2}(r-\gamma )^{\prime }\mathbb{V}%
_{\gamma }(r-\gamma )+o(\Vert r-\gamma \Vert ^{2}),
\end{equation*}%
where $\mathbb{V}_{\gamma }\equiv \mathbb{E}[\nabla ^{2}\tau _{i}(\gamma )]$.
\end{lemma}

\begin{lemma}
\label{lemmaA5} Under Assumptions C1--C7,, uniformly over $\mathcal{R}_{N}$,
we have
\begin{equation*}
\frac{2}{N}\sum_{m}\mathbb{E}[\mathcal{L}_{N}^{K}(r)|Z_{m}]-2\mathbb{E}[%
\mathcal{L}_{N}^{K}(r)]=\frac{1}{\sqrt{N}}(r-\gamma )^{\prime }\varPsi%
_{N,1}+o_{p}(\Vert r-\gamma \Vert ^{2}),
\end{equation*}%
where $\varPsi_{N,1}\equiv N^{-1/2}\sum_{m}2\nabla \tau _{m}(\gamma )$.
\end{lemma}

\begin{lemma}
\label{lemmaA6} Under Assumptions C1 and C6--C7, uniformly over $\mathcal{R}%
_{N}$, $\rho _{N}(r)=O_{p}(N^{-1}\sigma _{N}^{-2})$.
\end{lemma}

\begin{lemma}
\label{lemmaA7} Suppose
\begin{equation}
\hat{\beta}-\beta=\frac{1}{N}\sum_{i}\psi_{i,\beta}+o_{p}(N^{-1/2}),
\label{eq:A10}
\end{equation}
where $\psi_{i,\beta}$ is the influence function (of $Z_{i}$). $%
\psi_{i,\beta }$ is i.i.d across $i$ with $\mathbb{E}\left(
\psi_{i,\beta}\right) =0$.\footnote{%
In fact, the linear representation of $\hat{\beta}-\beta$ in (\ref{eq:A10})
can be obtained by applying the same arguments in Theorem 2 of \cite%
{Sherman1993} to a representation for $\mathcal{L}_{N,\beta}^{K}(b)$
analogous to (\ref{eq:A2}).} Further, suppose Assumptions C1--C7 hold. Then
uniformly over $\mathcal{R}_{N}$, we have
\begin{equation*}
\Delta\mathcal{L}_{N}^{K}(r)=\frac{1}{\sqrt{N}}(r-\gamma)^{\prime }\varPsi%
_{N,2}+o_{p}(\Vert r-\gamma\Vert^{2})+O_{p}(N^{-1}\sigma_{N}^{-2}),
\end{equation*}
with $\varPsi_{N,2}=\frac{1}{\sqrt{N}}\sum_{i}\left(
-\int\nabla_{13}^{2}\mu(v_{m},v_{m},\gamma)f_{V}(v_{m})dv_{m}\right)
\psi_{i,\beta}$, where $\mu(v_{i},v_{m},r)\equiv
G(v_{i},v_{m},r)f_{V}(v_{i}) $,
\begin{equation*}
G(v_{i},v_{m},r)\equiv\mathbb{E}\left[ B(V_{i},V_{m},W_{i},W_{m})S_{im}(r)%
\left(
\begin{array}{c}
x_{im1}^{\prime} \\
x_{im2}^{\prime}%
\end{array}
\right) |V_{i}=v_{i},V_{m}=v_{m}\right] ,
\end{equation*}
$\nabla_{1}\mu(\cdot,\cdot,\cdot)$ denotes the partial derivative of $%
\mu(\cdot,\cdot,\cdot)$ w.r.t. its first argument, and $\nabla_{13}^{2}\mu(%
\cdot,\cdot,\cdot)$ denotes the partial derivative of $\nabla_{1}\mu
(\cdot,\cdot,\cdot)$ w.r.t. its third argument.
\end{lemma}

\begin{proof}[Proof of Theorem \protect\ref{T:crossAsymp}]
Putting results in Lemmas \ref{lemmaA4}--\ref{lemmaA7} together, we write
equation (\ref{eq:A4}) as
\begin{equation}
\hat{\mathcal{L}}_{N}^{K}(r)=\frac{1}{2}(r-\gamma)^{\prime}\mathbb{V}%
_{\gamma }(r-\gamma)+\frac{1}{\sqrt{N}}(r-\gamma)^{\prime}\varPsi%
_{N}+o_{p}(\Vert r-\gamma\Vert^{2})+O_{p}(\varepsilon_{N})   \label{eq:A13}
\end{equation}
where $\varPsi_{N}=\varPsi_{N,1}+\varPsi_{N,2}$ and $\varepsilon_{N}=N^{-1}%
\sigma_{N}^{-2}$. Theorem 1 of \cite{Sherman1994ET} then implies that $\hat{%
\gamma}-\gamma=O_{p}(\sqrt{\varepsilon_{N}})=O_{p}(N^{-1/2}\sigma _{N}^{-1})$%
.

Next, take $\delta_{N}=O(\sqrt{\varepsilon_{N}})$ and $\mathcal{R}_{N}=\{r\in%
\mathcal{R}|\Vert r-\gamma\Vert\leq\delta_{N}\}$. We repeat the proof for
Lemma \ref{lemmaA6} and deduce from a Taylor expansion around $\gamma$ that $%
\sup_{r\in\mathcal{R}_{N}}\mathbb{E}[\rho_{im}^{\ast}(r)^{2}]=O(%
\sigma_{N}^{2}\delta_{N}^{2})$. Apply Theorem 3 of \cite{Sherman1994ET} to
see that uniformly over $\mathcal{R}_{N}$, $\rho_{N}(r)=O_{p}(N^{-1}\sigma
_{N}^{\lambda-2}\delta_{N}^{\lambda})$ where $0<\lambda<1$. Then we have
\begin{equation*}
\rho_{N}(r)=O_{p}(N^{-1}\sigma_{N}^{\lambda-2}\delta_{N}^{%
\lambda})=O_{p}(N^{-1})O_{p}(N^{-\lambda/2}\sigma_{N}^{-2})=o_{p}(N^{-1})
\end{equation*}
by invoking Assumption C7 and choosing $\lambda$ sufficiently close to 1.
This result in turn implies that the $O_{p}(\varepsilon_{N})$ term in (\ref%
{eq:A13}) has order $o_{p}(N^{-1})$, and hence $\hat{\gamma}-\gamma
=O_{p}(N^{-1/2})$ by applying Theorem 1 of \cite{Sherman1994ET} once again.

Now (\ref{eq:A13}) can be expressed as
\begin{equation*}
\hat{\mathcal{L}}_{N}^{K}(r)=\frac{1}{2}(r-\gamma)^{\prime}\mathbb{V}%
_{\gamma }(r-\gamma)+\frac{1}{\sqrt{N}}(r-\gamma)^{\prime}\varPsi%
_{N}+o_{p}(N^{-1}).
\end{equation*}
Let $\varDelta_{i}\equiv2\nabla\tau_{i}(\gamma)-\left(
\int\nabla_{13}^{2}\mu(v_{m},v_{m},\gamma)f_{V}(v_{m})\text{d}v_{m}\right)
\psi_{i,\beta}$. Note that $\mathbb{E}[\varDelta_{i}]=0$ since $\mathbb{E}%
[\nabla\tau_{i}(\gamma)]=0$ and $\mathbb{E}[\psi_{i,\beta}]=0$. We deduce
from Assumption C7 and Lindeberg-L\'{e}vy CLT that $\varPsi_{N}\overset{d}{%
\rightarrow }N(0,\mathbb{V}_{\varPsi})$ where $\mathbb{V}_{\varPsi}=\mathbb{E%
}[\varDelta_{i}\varDelta_{i}^{\prime}]$. The asymptotic normality of $\hat{%
\gamma}$ then follows from Theorem 2 of \cite{Sherman1994AoS}, i.e., $\sqrt{N%
}(\hat{\gamma}-\gamma)\overset{d}{\rightarrow}N(0,\mathbb{V}_{\gamma }^{-1}%
\mathbb{V}_{\varPsi}\mathbb{V}_{\gamma}^{-1})$.
\end{proof}

\subsection{Proof of of Theorem \ref{TH:testing}}
Recall that $\mathcal{\hat{L}}_{N}^{K}\left( r\right) $ was redefined in
Appendix \ref{appendix_asym}. However, for the proof of Theorem \ref{TH:testing}, it is more convenient to still use the definition in the main body of the paper, specifically,
\begin{equation*}
\mathcal{\hat{L}}_{N}^{K}\left( \hat{\gamma}\right) =\frac{1}{\sigma
_{N}^{2}N\left( N-1\right) }\sum_{i\neq m}K_{\sigma _{N},\gamma } (
V_{im} ( \hat{\beta} )  ) Y_{im\left( 1,1\right) }\text{sgn}%
\left( W_{im}^{\prime }\hat{\gamma}\right) .
\end{equation*}

\begin{lemma}
\label{LE:L_N^K_beta}Suppose Assumptions C1--C7 hold. Then, we have%
\begin{equation*}
\frac{1}{\sigma _{N}^{2}N\left( N-1\right) }\sum_{i\neq m} [ K_{\sigma
_{N},\gamma } ( V_{im} ( \hat{\beta} )  ) -K_{\sigma
_{N},\gamma }\left( V_{im}\left( \beta \right) \right)  ] Y_{im\left(
1,1\right) }\text{\emph{sgn}}\left( W_{im}^{\prime }\gamma \right)
=O_{p} ( N^{-1}\sigma _{N}^{-2} ) .
\end{equation*}
\end{lemma}

\begin{lemma}
\label{LE:L_N^K}Suppose Assumptions C1--C7 hold. Then uniformly over $r\in
\mathcal{R}$,
\begin{eqnarray*}
\mathcal{\hat{L}}_{N}^{K}\left( r\right) &=&\frac{1}{2}\left( r-\gamma
\right) ^{\prime }\mathbb{V}_{\gamma }\left( r-\gamma \right) +o_{p}\left(
N^{-1/2}\left\Vert r-\gamma \right\Vert \right) +o_{p}\left( \left\Vert
r-\gamma \right\Vert ^{2}\right) \\
&&+\frac{1}{\sigma _{N}^{2}N\left( N-1\right) }\sum_{i\neq m}K_{\sigma
_{N},\gamma }\left( V_{im}\left( \beta \right) \right) Y_{im\left(
1,1\right) }\textrm{sgn}\left( W_{im}^{\prime }\gamma \right) +o_{p}(
N^{-1/2}) .
\end{eqnarray*}
\end{lemma}

\begin{lemma}
\label{LE:delta}Suppose Assumptions C5--C7 hold. Then, we have
\begin{eqnarray*}
&&\mathbb{E}\left[ \sigma _{N}^{-2}K_{\sigma _{N},\gamma }\left(
V_{im}\left( \beta \right) \right) Y_{im\left( 1,1\right) }\textrm{sgn}\left(
W_{im}^{\prime }\gamma \right) |Z_{i}\right] \\
&=&f_{V_{im}(\beta)}(0)\mathbb{E}\left[ Y_{im\left( 1,1\right) }\textrm{sgn}\left( W_{im}^{\prime
}\gamma \right) |Z_{i},V_{im}\left( \beta \right) =0
\right] +o_{p}( N^{-1/2}) .
\end{eqnarray*}
\end{lemma}

Note that Lemma \ref{LE:L_N^K_beta} indicates that $\hat{\beta}$ affects $\mathcal{\hat{%
L}}_{N}^{K}\left( \hat{\gamma}\right) $ at the rate of $N^{-1}\sigma
_{N}^{-2}$, and Lemma \ref{LE:L_N^K} shows that $\hat{\gamma}$ affects $%
\mathcal{\hat{L}}_{N}^{K}\left( \hat{\gamma}\right) $ at the rate of $%
N^{-1}$.

\begin{proof}
[Proof of Theorem \ref{TH:testing}] Note that $O_p(N^{-1}\sigma_N^{-2})=o_p(N^{-1/2})$ by Assumption C7 and $\hat{\gamma}-$ $\gamma =O_{p}(
N^{-1/2})$ by Theorem \ref%
{T:crossAsymp}. These results, together with Lemmas \ref{LE:L_N^K_beta} and \ref{LE:L_N^K}, imply that
\begin{equation*}
\mathcal{\hat{L}}_{N}^{K}\left( \hat{\gamma}\right) =\frac{1}{\sigma
_{N}^{2}N\left( N-1\right) }\sum_{i\neq m}K_{\sigma _{N},\gamma }\left(
V_{im}\left( \beta \right) \right) Y_{im\left( 1,1\right) }\text{sgn}\left(
W_{im}^{\prime }\gamma \right) +o_{p} ( N^{-1/2} ) .
\end{equation*}%
We can use the same arguments as those used for proving Lemmas \ref{lemmaA1} and \ref{lemmaA2} to show
\begin{equation*}
\mathcal{\hat{L}}_{N}^{K}\left( \hat{\gamma}\right) \overset{p}{\rightarrow }%
\mathcal{\bar{L}}\left( \gamma \right) \mathcal{=}f_{V_{im}\left( \beta
\right) }\left( 0\right) \mathbb{E}\left[ Y_{im\left( 1,1\right) }\text{sgn}%
\left( W_{im}^{\prime }\gamma \right) |V_{im}\left( \beta \right) =0\right].
\end{equation*}%
Then, we can write
\begin{eqnarray*}
&&\mathcal{\hat{L}}_{N}^{K}\left( \hat{\gamma}\right) -\mathcal{\bar{L}}%
\left( \gamma \right) \\
&=&\frac{1}{\sigma _{N}^{2}N\left( N-1\right) }%
\sum_{i\neq m}K_{\sigma _{N},\gamma }\left( V_{im}\left( \beta \right)
\right) Y_{im\left( 1,1\right) }\text{sgn}\left( W_{im}^{\prime }\gamma
\right) -\mathcal{\bar{L}}\left( \gamma \right) +o_{p}( N^{-1/2})
\\
&=&\frac{2}{N}\sum_{i=1}^{N}\left\{ \mathbb{E}\left[ \sigma
_{N}^{-2}K_{\sigma _{N},\gamma }\left( V_{im}\left( \beta \right) \right)
Y_{im\left( 1,1\right) }\text{sgn}\left( W_{im}^{\prime }\gamma \right)
|Z_{i}\right] -\mathbb{E}\left[ \sigma _{N}^{-2}K_{\sigma _{N},\gamma
}\left( V_{im}\left( \beta \right) \right) Y_{im\left( 1,1\right) }\text{sgn}%
\left( W_{im}^{\prime }\gamma \right) \right] \right\} \\
&&+\mathbb{E}\left[ \sigma _{N}^{-2}K_{\sigma _{N},\gamma }\left(
V_{im}\left( \beta \right) \right) Y_{im\left( 1,1\right) }\text{sgn}\left(
W_{im}^{\prime }\gamma \right) \right] -\mathcal{\bar{L}}\left( \gamma
\right) +o_{p} ( N^{-1/2} ) \\
&=&\frac{2}{N}\sum_{i=1}^{N}f_{V_{im}(\beta)}(0)\left\{ \mathbb{E}\left[ Y_{im\left( 1,1\right) }%
\text{sgn}\left( W_{im}^{\prime }\gamma \right) |Z_{i},V_{im}\left( \beta
\right) =0\right] -\mathbb{E}\left[ Y_{im\left( 1,1\right) }\text{sgn}\left(
W_{im}^{\prime }\gamma \right) |V_{im}\left( \beta \right) =0\right] \right\}
\\
&&+o_{p}( N^{-1/2}),
\end{eqnarray*}
where the second equality uses the decomposition (\ref{eq:A4}) and Lemma \ref{lemmaA6}, and the last equality follows by Lemma \ref{LE:delta}.
Invoking the Lindeberg-L\'{e}vy CLT on the leading term of $\mathcal{\hat{L}}%
_{N}^{K}\left( \hat{\gamma}\right) -\mathcal{\bar{L}(\gamma)}$ in the above yields%
\begin{equation*}
\sqrt{N}( \mathcal{\hat{L}}_{N}^{K}\left( \hat{\gamma}\right) -\mathcal{%
\bar{L}}\left( \gamma \right) ) \overset{d}{\rightarrow }N\left(
0,\Delta_\gamma \right) .
\end{equation*}
\end{proof}

\subsection{Proof of of Theorem \ref{thm:cross_theta_identification}}
\begin{proof}[Proof of Theorem \ref{thm:cross_theta_identification}]
The proof is lengthy, but the underlying idea is straightforward. Using the (conditional) linear independence (equivalent to the regular full rank condition) and the free variation of the first regressor conditions in Assumption CL2, we show that there exists a positive probability of making an incorrect prediction for any $\vartheta\neq\theta$.

\noindent\textbf{Some Definitions.} We normalize the coefficients on
$X_{im1}^{(1)}$, $X_{im2}^{(1)}$, and $W_{i}^{(1)}$ to 1 as indicated in
Assumption CL3. To simplify notation, we express the three utility indexes
with a slight abuse of notation as $X_{im1}^{(1)}+u_{im1}(\vartheta)$,
$X_{im2}^{(1)}+u_{im2}(\vartheta)$, and $W_{im}^{(1)}+u_{imb}(\vartheta)$,
where $u_{im1}(\vartheta)\equiv\tilde{X}_{im1}^{\prime}b+S_{im}^{\prime
}\varrho_{1}$, $u_{im2}(\vartheta)\equiv\tilde{X}_{im2}^{\prime}%
b+S_{im}^{\prime}\varrho_{2}$, and $u_{imb}(\vartheta)\equiv\tilde{W}%
_{im}^{\prime}r+S_{im}^{\prime}\varrho_{b}$, respectively. We use $a\vee b$
and $a\wedge b$ to denote $\max\{a,b\}$ and $\min\{a,b\}$, respectively.
Finally, we define indicator functions for the events on the left-hand sides
of (\ref{eq:lad_ineq_3})--(\ref{eq:lad_ineq_4}), (\ref{eq:lad_ineq_5})--(\ref{eq:lad_ineq_6}), and
(\ref{eq:lad_ineq_7})--(\ref{eq:lad_ineq_8})\footnote{There is no need
to define $I_{im(0,0)}^{+}(\vartheta)$ and $I_{im(0,0)}^{-}(\vartheta)$ separately since
$I_{im(0,0)}^{+}(\vartheta)=I_{im(1,1)}^{-}(\vartheta)$ and $I_{im(0,0)}^{-}(\vartheta)=I_{im(1,1)}^{+}(\vartheta)$.} and
the two uninformative events in Remark \ref{remark:LAD1} in turn:
\begin{align*}
I_{im(0,1)}^{+}(\vartheta) &  =1[X_{im1}^{(1)}+u_{im1}(\vartheta)\leq
0,X_{im2}^{(1)}+u_{im2}(\vartheta)\geq0,W_{im}^{(1)}+u_{imb}(\vartheta
)\leq0],\\
I_{im(0,1)}^{-}(\vartheta) &  =1[X_{im1}^{(1)}+u_{im1}(\vartheta)\geq
0,X_{im2}^{(1)}+u_{im2}(\vartheta)\leq0,W_{im}^{(1)}+u_{imb}(\vartheta
)\geq0],\\
I_{im(1,1)}^{+}(\vartheta) &  =1[X_{im1}^{(1)}+u_{im1}(\vartheta)\geq
0,X_{im2}^{(1)}+u_{im2}(\vartheta)\geq0,W_{im}^{(1)}+u_{imb}(\vartheta
)\geq0],\\
I_{im(1,1)}^{-}(\vartheta) &  =1[X_{im1}^{(1)}+u_{im1}(\vartheta)\leq
0,X_{im2}^{(1)}+u_{im2}(\vartheta)\leq0,W_{im}^{(1)}+u_{imb}(\vartheta
)\leq0],\\
I_{imu}^{+}(\vartheta) &  =1[X_{im1}^{(1)}+u_{im1}(\vartheta)\leq
0,X_{im2}^{(1)}+u_{im2}(\vartheta)\leq0,W_{im}^{(1)}+u_{imb}(\vartheta
)\geq0],\\
I_{imu}^{-}(\vartheta) &  =1[X_{im1}^{(1)}+u_{im1}(\vartheta)\geq
0,X_{im2}^{(1)}+u_{im2}(\vartheta)\geq0,W_{im}^{(1)}+u_{imb}(\vartheta)\leq0].
\end{align*}
Recall that by construction $q_{(1,0)}(Z_{i},Z_{m},\theta)=0\leq
q_{(1,0)}(Z_{i},Z_{m},\vartheta)$ for all $\vartheta\in\Theta$ and
$(Z_{i},Z_{m})$. In what follows, we show for any $\vartheta\in\Theta
\setminus\{\theta\}$, $q_{(1,0)}(Z_{i},Z_{m},\vartheta)>0$ occurs with
positive probability, and hence $Q_{(1,0)}(\theta)<Q_{(1,0)}(\vartheta)$.

\noindent\textbf{Identification of }$(\beta,\rho_{1})$\textbf{.} For any
$\vartheta\in\Theta\setminus\{\theta\}$, we have
\begin{align}
&  P(q_{(1,0)}(Z_{i},Z_{m},\vartheta)>0)\nonumber\\
\geq &  P(\{q_{(1,0)}(Z_{i},Z_{m},\vartheta)>0\}\cap\{I_{im(1,1)}^{+}%
(\theta)+I_{im(1,1)}^{-}(\theta)=1\})\nonumber\\
= &  P(\{q_{(1,0)}(Z_{i},Z_{m},\vartheta)>0\}\cap\{I_{im(1,1)}^{+}%
(\theta)=1\})+P(\{q_{(1,0)}(Z_{i},Z_{m},\vartheta)>0\}\cap\{I_{im(1,1)}%
^{-}(\theta)=1\})\nonumber\\
\geq &  P(\{q_{(1,0)}(Z_{i},Z_{m},\vartheta)>0\}\cap\{I_{im(1,1)}^{+}%
(\theta)=I_{im(1,0)}^{-}(\vartheta)=1\})\nonumber\\
&  +P(\{q_{(1,0)}(Z_{i},Z_{m},\vartheta)>0\}\cap\{I_{im(1,1)}^{-}%
(\theta)=I_{im(1,0)}^{+}(\vartheta)=1\})\nonumber\\
= &  P(q_{(1,0)}(Z_{i},Z_{m},\vartheta)>0|I_{im(1,1)}^{+}(\theta
)=I_{im(1,0)}^{-}(\vartheta)=1)P(I_{im(1,1)}^{+}(\theta)=I_{im(1,0)}%
^{-}(\vartheta)=1)\nonumber\\
&  +P(q_{(1,0)}(Z_{i},Z_{m},\vartheta)>0|I_{im(1,1)}^{-}(\theta)=I_{im(1,0)}%
^{+}(\vartheta)=1)P(I_{im(1,1)}^{-}(\theta)=I_{im(1,0)}^{+}(\vartheta
)=1).\label{eq:clproof_1}%
\end{align}

By definitions, we can write
\begin{align*}
&  P(I_{im(1,1)}^{+}(\theta)=I_{im(1,0)}^{-}(\vartheta)=1)\\
=  &  P(u_{im1}(\vartheta)\leq-X_{im1}^{(1)}\leq u_{im1}(\theta),-X_{im2}%
^{(1)}\leq u_{im2}(\vartheta)\wedge u_{im2}(\theta),-W_{im}^{(1)}\leq
u_{imb}(\vartheta)\wedge u_{imb}(\theta)),
\end{align*}
and
\begin{align*}
&  P(I_{im(1,1)}^{-}(\theta)=I_{im(1,0)}^{+}(\vartheta)=1)\\
=  &  P(u_{im1}(\theta)\leq-X_{im1}^{(1)}\leq u_{im1}(\vartheta),-X_{im2}%
^{(1)}\geq u_{im2}(\vartheta)\vee u_{im2}(\theta),-W_{im}^{(1)}\geq
u_{imb}(\vartheta)\vee u_{imb}(\theta)).
\end{align*}
Assumptions CL2(i) and CL2(iii) imply that events $\{-X_{im2}^{(1)}\leq
u_{im2}(\vartheta)\wedge u_{im2}(\theta),-W_{im}^{(1)}\leq u_{imb}%
(\vartheta)\wedge u_{imb}(\theta)\}$ and $\{-X_{im2}^{(1)}\geq u_{im2}%
(\vartheta)\vee u_{im2}(\theta),-W_{im}^{(1)}\geq u_{imb}(\vartheta)\vee
u_{imb}(\theta)\}$ always have positive probabilities of occurring. Then, when
conditioning on them, Assumption CL2(ii) implies that the event ${u_{im1}%
(\vartheta)\neq u_{im1}(\theta)}$ has a positive probability to occur when
$(b,\varrho_{1})\neq(\beta,\rho_{1})$. Thus, under Assumption CL2(i), we
conclude that either ${u_{im1}(\vartheta)\leq-X_{im1}^{(1)}\leq u_{im1}%
(\theta)}$ or ${u_{im1}(\theta)\leq-X_{im1}^{(1)}\leq u_{im1}(\vartheta)}$
occurs with positive probability. In summary, we can claim that for any
$\vartheta$ with $(b,\varrho_{1})\neq(\beta,\rho_{1})$,
\begin{equation}
P(I_{im(1,1)}^{+}(\theta)=I_{im(1,0)}^{-}(\vartheta)=1)>0\text{ or
}P(I_{im(1,1)}^{-}(\theta)=I_{im(1,0)}^{+}(\vartheta)=1)>0.
\label{eq:clproof_2}%
\end{equation}

Then, if we can establish that
\begin{align}
P(q_{(1,0)}(Z_{i},Z_{m},\vartheta)  &  >0|I_{im(1,1)}^{+}(\theta
)=I_{im(1,0)}^{-}(\vartheta)=1)>0\text{ and}\nonumber\\
P(q_{(1,0)}(Z_{i},Z_{m},\vartheta)  &  >0|I_{im(1,1)}^{-}(\theta
)=I_{im(1,0)}^{+}(\vartheta)=1)>0, \label{eq:clproof_3}%
\end{align}
by substituting (\ref{eq:clproof_2}) into (\ref{eq:clproof_1}), we obtain
$P(q_{(1,0)}(Z_{i},Z_{m},\vartheta)>0)>0$ for any $(b,\varrho_{1})\neq
(\beta,\rho_{1})$. This establishes the identification of $(\beta,\rho_{1})$
in $\theta$. We defer the justification of the inequalities in
(\ref{eq:clproof_3}) to the end of this proof.

\noindent\textbf{Identification of }$(\beta,\rho_{2})$\textbf{.} Similarly, we
can write
\begin{align}
&  P(q_{(1,0)}(Z_{i},Z_{m},\vartheta)>0)\nonumber\\
\geq &  P(\{q_{(1,0)}(Z_{i},Z_{m},\vartheta)>0\}\cap\{I_{imu}^{+}%
(\theta)+I_{imu}^{-}(\theta)=1\})\nonumber\\
\geq &  P(q_{(1,0)}(Z_{i},Z_{m},\vartheta)>0|I_{imu}^{+}(\theta)=I_{im(1,0)}%
^{-}(\vartheta)=1)P(I_{imu}^{+}(\theta)=I_{im(1,0)}^{-}(\vartheta
)=1)\nonumber\\
&  +P(q_{(1,0)}(Z_{i},Z_{m},\vartheta)>0|I_{imu}^{-}(\theta)=I_{im(1,0)}%
^{+}(\vartheta)=1)P(I_{imu}^{-}(\theta)=I_{im(1,0)}^{+}(\vartheta
)=1),\label{eq:clproof_4}%
\end{align}
where
\begin{align*}
&  P(I_{imu}^{+}(\theta)=I_{im(1,0)}^{-}(\vartheta)=1)\\
= &  P(-X_{im1}^{(1)}\geq u_{im1}(\vartheta)\vee u_{im1}(\theta),u_{im2}%
(\theta)\leq-X_{im2}^{(1)}\leq u_{im2}(\vartheta),-W_{im}^{(1)}\leq
u_{imb}(\vartheta)\wedge u_{imb}(\theta)),
\end{align*}
and
\begin{align*}
&  P(I_{imu}^{-}(\theta)=I_{im(1,0)}^{+}(\vartheta)=1)\\
= &  P(-X_{im1}^{(1)}\leq u_{im1}(\vartheta)\wedge u_{im1}(\theta
),u_{im2}(\vartheta)\leq-X_{im2}^{(1)}\leq u_{im2}(\theta),-W_{im}^{(1)}\geq
u_{imb}(\vartheta)\vee u_{imb}(\theta)).
\end{align*}
Then, we can establish the identification of $(\beta,\rho_{2})$ in $\theta$
using similar arguments as above.

\noindent\textbf{Identification of }$(\gamma,\rho_{b})$\textbf{. }The
identification of $(\gamma,\rho_{b})$ can be similarly established based on
the following inequality:
\begin{align}
&  P(q_{(1,0)}(Z_{i},Z_{m},\vartheta)>0)\nonumber\\
\geq &  P(\{q_{(1,0)}(Z_{i},Z_{m},\vartheta)>0\}\cap\{I_{im(0,1)}^{+}%
(\theta)+I_{im(0,1)}^{-}(\theta)=1\})\nonumber\\
\geq &  P(q_{(1,0)}(Z_{i},Z_{m},\vartheta)>0|I_{im(0,1)}^{+}(\theta
)=I_{im(1,0)}^{-}(\vartheta)=1)P(I_{im(0,1)}^{+}(\theta)=I_{im(1,0)}%
^{-}(\vartheta)=1)\nonumber\\
&  +P(q_{(1,0)}(Z_{i},Z_{m},\vartheta)>0|I_{im(0,1)}^{-}(\theta)=I_{im(1,0)}%
^{+}(\vartheta)=1)P(I_{im(0,1)}^{-}(\theta)=I_{im(1,0)}^{+}(\vartheta
)=1),\label{eq:clproof_5}%
\end{align}
where
\begin{align*}
&  P(I_{im(0,1)}^{+}(\theta)=I_{im(1,0)}^{-}(\vartheta)=1)\\
= &  P(-X_{im1}^{(1)}\leq u_{im1}(\vartheta)\wedge u_{im1}(\theta
),-X_{im2}^{(1)}\geq u_{im2}(\vartheta)\vee u_{im2}(\theta),u_{imb}%
(\vartheta)\leq-W_{im}^{(1)}\leq u_{imb}(\theta)),
\end{align*}
and
\begin{align*}
&  P(I_{im(0,1)}^{-}(\theta)=I_{im(1,0)}^{+}(\vartheta)=1)\\
= &  P(-X_{im1}^{(1)}\geq u_{im1}(\vartheta)\vee u_{im1}(\theta),-X_{im2}%
^{(1)}\leq u_{im2}(\vartheta)\wedge u_{im2}(\theta),u_{imb}(\theta)\leq
-W_{im}^{(1)}\leq u_{imb}(\vartheta)).
\end{align*}
Finally, put all these results together to conclude that all elements in
$\theta$ are identified.

\noindent\textbf{On (\ref{eq:clproof_3}).} To complete the
proof, we need to justify the claim in (\ref{eq:clproof_3}) as well as similar
claims for (\ref{eq:clproof_4}) and (\ref{eq:clproof_5}). Since these proofs
follow similar arguments, we only present the proof for $P(q_{(1,0)}%
(Z_{i},Z_{m},\vartheta)>0|I_{im(1,1)}^{+}(\theta)=I_{im(1,0)}^{-}%
(\vartheta)=1)>0$ and omit the others for the sake of brevity. Note that by
definition, when $I_{im(1,0)}^{-}(\vartheta)=1$, $q_{(1,0)}(Z_{i}%
,Z_{m},\vartheta)>0$ if and only if $\Delta p_{(1,0)}(Z_{i},Z_{m})>0$.
Therefore, we show $P(\Delta p_{(1,0)}(Z_{i},Z_{m})>0|I_{im(1,1)}^{+}%
(\theta)=I_{im(1,0)}^{-}(\vartheta)=1)>0$ in the following, and the proof is done.

Under Assumption CL1(i), $Z_{im}$ has distribution symmetric at 0, and hence
we can pick $X_{im2}^{(1)}=W_{im}^{(1)}=0$, $\tilde{X}_{im2}=0$, $\tilde
{W}_{im}=0$, and $S_{im}=0$ so that $\{-X_{im2}^{(1)}\leq u_{im2}%
(\vartheta)\wedge u_{im2}(\theta),-W_{im}^{(1)}\leq u_{imb}(\vartheta)\wedge
u_{imb}(\theta)\}$ holds with equalities for any $\vartheta$. Note that here
we use the \textquotedblleft matching\textquotedblright\ insight for
developing the MRC procedure, that is, we match the utility indexes associated
with the second stand-alone alternative and the bundle effect. Meanwhile,
under Assumptions CL2(i) and CL2(ii), we can still have $\{u_{im1}%
(\vartheta)<-X_{im1}^{(1)}<u_{im1}(\theta)\}$ occur with positive probability
conditional on $E\equiv\{X_{im2}^{(1)}=W_{im}^{(1)}=0,\tilde{X}_{im2}%
=0,\tilde{W}_{im}=0,S_{im}=0\}$. In this scenario, inequality
(\ref{eq:lad_ineq_1}) implies that $\Delta p_{(1,0)}(Z_{i},Z_{m})>0,$ and say
$\Delta p_{(1,0)}(Z_{i},Z_{m})=\delta_{1}>0$ conditional on $E$. Due to the
continuity, there exists a small $\delta_{2}>0$ such that conditional on
\[
E^{\delta_{2}}=\{ \vert X_{im2}^{(1)}\vert +\vert
W_{im}^{(1)}\vert +\Vert \tilde{X}_{im2}\Vert +\Vert
\tilde{W}_{im}\Vert +\Vert S_{im}\Vert \leq\delta
_{2}\}  ,
\]
a small neighborhood of $E$, 
\[
\Delta p_{(1,0)}(Z_{i},Z_{m})\geq\delta_{1}/2>0.
\]
Again, due to the continuity, $P\left(  E^{\delta_{2}}\right)  >0$. Since what
we have discussed so far is only one possible scenario for $\Delta p_{(1,0)}%
(Z_{i},Z_{m})>0$ to occur, we must have $P(q_{(1,0)}(Z_{i},Z_{m}%
,\vartheta)>0|I_{im(1,1)}^{+}(\theta)=I_{im(1,0)}^{-}(\vartheta)=1)>P\left(
E^{\delta_{2}}\right)  >0$. This completes the proof.
\end{proof}

\subsection{Proof of of Theorem \ref{thm:cross_theta_consistency}}

\begin{proof}[Proof of Theorem \ref{thm:cross_theta_consistency}]
Let 
\[
\hat{Q}_{N(1,0)}^{D}(\vartheta):=\frac{2}{N(N-1)}\sum_{i=1}^{N-1}\sum_{m>i}\hat{q}_{im(1,0)}^{D}(\vartheta),
\]
where $\hat{q}_{im(1,0)}^{D}(\vartheta)$ is defined in (\ref{eq:lad_estimator}). We
prove the consistency of $\hat{\theta}$ by verifying the four sufficient
conditions in Theorem 2.1 of \cite{NeweyMcFadden1994}: (S1) $\Theta$ is a compact set,
(S2) $\sup_{\vartheta\in\Theta}\vert\hat{Q}_{N(1,0)}^{D}(\vartheta)-Q_{(1,0)}^{D}(\vartheta)\vert=o_{p}(1)$,
(S3) $Q_{(1,0)}^{D}(\vartheta)$ is continuous in $\vartheta$, and
(S4) $Q_{(1,0)}^{D}(\vartheta)$ is uniquely minimized at $\theta$.

Condition (S1) is trivially satisfied under Assumption CL3. The identification
condition (S4) is an immediate result of Theorem \ref{thm:cross_theta_identification} and the discussion
right after it. The continuity condition (S3) is also satisfied since
$Q_{(1,0)}^{D}(\vartheta)$, by definition, can be expressed as the
sum of terms of the following type: 
\begin{align*}
\int_{\Omega_{\vartheta}}(\vert1-\Delta p_{(1,0)}(z_{i},z_{m})\vert-1)f_{Z}(z_{i})f_{Z}(z_{m})dz_{i}dz_{m},
\end{align*}
where $\Omega_{\vartheta}:=\{x_{im1}'b+s_{im}'\varrho_{1}\geq0,x_{im2}'b+s_{im}'\varrho_{2}\leq0,w_{im}'r+s_{im}'\varrho_{b}\leq0\}$.
Then it is easy to see that Assumptions CL1 and CL2 imply the continuity
of $Q_{(1,0)}^{D}(\vartheta)$.

The remaining task is to show the uniform convergence condition (S2).
To this end, we prove $\sup_{\vartheta\in\Theta}\vert\hat{Q}_{N(1,0)}^{D}(\vartheta)-Q_{N(1,0)}^{D}(\vartheta)\vert=o_{p}(1)$
and $\sup_{\vartheta\in\Theta}\vert Q_{N(1,0)}^{D}(\vartheta)-Q_{(1,0)}^{D}(\vartheta)\vert=o_{p}(1)$,
where $Q_{N(1,0)}^{D}(\vartheta):=\frac{2}{N(N-1)}\sum_{i=1}^{N-1}\sum_{m>i}q_{(1,0)}^{D}(Z_{i},Z_{m},\vartheta)$.

Firstly, we have by definition
\begin{align*}
\sup_{\vartheta\in\Theta}\vert\hat{Q}_{N(1,0)}^{D}(\vartheta)-Q_{N(1,0)}^{D}(\vartheta)\vert & \leq\frac{2}{N(N-1)}\sum_{i=1}^{N-1}\sum_{m>i}\sup_{\vartheta\in\Theta}\vert\hat{q}_{im(1,0)}^{D}(\vartheta)-q_{(1,0)}^{D}(Z_{i},Z_{m},\vartheta)\vert\\
 & \leq\frac{4}{N(N-1)}\sum_{i=1}^{N-1}\sum_{m>i}\vert\Delta\hat{p}_{(1,0)}(Z_{i},Z_{m})-\Delta p_{(1,0)}(Z_{i},Z_{m})\vert\\
 & \leq\sup_{(z_{i},z_{m})\in\mathcal{Z}^{2}}4\vert\Delta\hat{p}_{(1,0)}(z_{i},z_{m})-\Delta p_{(1,0)}(z_{i},z_{m})\vert=o_{p}(1),
\end{align*}
where the last equality follows by Assumption CL5. 

Secondly, since $\vert\Delta p_{(1,0)}(Z_{i},Z_{m})\vert\leq1$, we
can write
\begin{align}
&\sup_{\vartheta\in\Theta}\vert Q_{N(1,0)}^{D}(\vartheta)-Q_{(1,0)}^{D}(\vartheta)\vert \nonumber \\
 \leq &\sup_{\vartheta\in\Theta}\vert\frac{2}{N(N-1)}\sum_{i=1}^{N-1}\sum_{m>i}q_{(1,0)}^{D}(Z_{i},Z_{m},\vartheta)-\mathbb{E}[q_{(1,0)}^{D}(Z_{i},Z_{m},\vartheta)]\vert\nonumber \\
  \leq &\sup_{\bar{\vartheta}\in\bar{\Theta}}\vert\frac{2}{N(N-1)}\sum_{i=1}^{N-1}\sum_{m>i}\psi_{(1,0)}^{D}(Z_{i},Z_{m},\bar{\vartheta})-\mathbb{E}[\psi_{(1,0)}^{D}(Z_{i},Z_{m},\bar{\vartheta})]\vert,\label{eq:lad_consistency_1}
\end{align}
where $\bar{\Theta}:=\{\bar{\vartheta}=(\vartheta,a)|\vartheta\in\Theta,a\in[-1,1]\}$
and 
\begin{align*}
\psi_{(1,0)}^{D}(Z_{i},Z_{m},\bar{\vartheta}):= & [\vert I_{im(1,0)}^{+}(\vartheta)-a\vert+\vert I_{im(1,0)}^{-}(\vartheta)+a\vert]\times[I_{im(1,0)}^{+}(\vartheta)+I_{im(1,0)}^{-}(\vartheta)]\\
 & +[1-(I_{im(1,0)}^{+}(\vartheta)+I_{im(1,0)}^{-}(\vartheta))].
\end{align*}
Let $\mathcal{G}=\{\psi_{(1,0)}^{D}(\cdot,\cdot,\bar{\vartheta})|\bar{\vartheta}=(\vartheta,a),\vartheta\in\Theta,a\in[-1,1]\}$
define the class of bounded functions $\psi_{(1,0)}^{D}(\cdot,\cdot,\bar{\vartheta})$
indexed by $\bar{\vartheta}$. Using similar arguments in the proof
of Lemma \ref{lemmaA2}, we can show that $\mathcal{G}$ is a Euclidean class\footnote{Alternatively, one can also apply results in Section 2.6 of \cite{vdVaartWellner1996} to show
that $\mathcal{G}$ is a VC-class.} of functions. Hence, we have
\begin{equation}
\sup_{\bar{\vartheta}\in\bar{\Theta}}\vert\frac{2}{N(N-1)}\sum_{i=1}^{N-1}\sum_{m>i}\psi_{(1,0)}^{D}(Z_{i},Z_{m},\bar{\vartheta})-\mathbb{E}[\psi_{(1,0)}^{D}(Z_{i},Z_{m},\bar{\vartheta})]\vert=o_{p}(1).\label{eq:lad_consistency_2}
\end{equation}
Plugging (\ref{eq:lad_consistency_2}) into (\ref{eq:lad_consistency_1})
gives $\sup_{\vartheta\in\Theta}\vert Q_{N(1,0)}^{D}(\vartheta)-Q_{(1,0)}^{D}(\vartheta)\vert=o_{p}(1)$.
Then (S2) follows immediately from the triangle inequality, which
completes the proof.
\end{proof}

\section{Proofs for the Panel Data Model}\label{appendixB}
\subsection{Proof of Theorem \ref{T:paneldist}}
In this section, we provide the proof of Theorem
\ref{T:paneldist}, applying the asymptotic theory developed in
\cite{SeoOtsu2018}. Before the main proof, we first present two supporting lemmas, Lemmas \ref{LE:P1} and
\ref{LE:P2}, whose proofs are relegated to
Appendix \ref{appendixE}. The outline of the proof process is as follows. Lemma
\ref{LE:P1} verifies the technical conditions in \cite{SeoOtsu2018}. Lemma
\ref{LE:P2} obtains technical terms for the final asymptotics. Then we apply
the results in \cite{SeoOtsu2018} and get the asymptotics of our estimators in
the proof of Theorem \ref{T:paneldist}. 

\begin{lemma}
\label{LE:P1}Suppose Assumptions P1--P9 hold. Then $\phi_{Ni}\left(  b\right)
$ and $\varphi_{Ni}\left(  r\right)  $ satisfy Assumption M in
\cite{SeoOtsu2018}.
\end{lemma}

\begin{lemma}
\label{LE:P2}Suppose Assumptions P1--P9 hold. Then
\[
\lim_{N\rightarrow\infty}(  Nh_{N}^{k_{1}+k_{2}})  ^{2/3}%
\mathbb{E}[  \phi_{Ni}(  \beta+\rho\left(  Nh_{N}^{k_{1}+k_{2}%
})  ^{-1/3}\right)]  =\frac{1}{2}\rho^{\prime}\mathbb{V}\rho,
\]%
\[
\lim_{N\rightarrow\infty}(  N\sigma_{N}^{2k_{1}})  ^{2/3}%
\mathbb{E}[  \varphi_{Ni}(  \gamma+\delta(  N\sigma_{N}%
^{2k_{1}})  ^{-1/3}) ]  =\frac{1}{2}\delta^{\prime
}\mathbb{W}\delta,
\]%
\[
\lim_{N\rightarrow\infty} (  Nh_{N}^{k_{1}+k_{2}} )  ^{1/3}%
\mathbb{E} [  h_{N}^{k_{1}+k_{2}}\phi_{Ni} (  \beta+\rho_{1} (
Nh_{N}^{k_{1}+k_{2}} )  ^{-1/3} )  \phi_{Ni} (  \beta+\rho_{2} (  Nh_{N}^{k_{1}+k_{2}} )  ^{-1/3} )   ]
=\mathbb{H}_{1}\left(  \rho_{1},\rho_{2}\right)  ,
\]
and
\[
\lim_{N\rightarrow\infty} (  N\sigma_{N}^{2k_{1}} )  ^{1/3}%
\mathbb{E} [  \sigma_{N}^{2k_{1}}\varphi_{Ni} (  \gamma+\delta
_{1} (  N\sigma_{N}^{2k_{1}} )  ^{-1/3} )  \varphi_{Ni} (
\gamma+\delta_{2} (  N\sigma_{N}^{2k_{1}} )  ^{-1/3} )  ]
=\mathbb{H}_{2}\left(  \delta_{1},\delta_{2}\right)  ,
\]
where $\rho,$ $\rho_{1},$ and $\rho_{2}$ are $k_{1}\times1$ vectors, $\delta,$
$\delta_{1},$ and $\delta_{2}$ are $k_{2}\times1$ vectors, $\mathbb{V}$ is a
$k_{1}\times$ $k_{1}$ matrix defined as
\begin{align}
\mathbb{V}= &  -\sum_{t>s}\sum_{d\in\mathcal{D}}\left\{  \int1\left[
x^{\prime}\beta=0\right]  \left(  \frac{\partial\kappa_{dts}^{\left(
1\right)  }\left(  x\right)  }{\partial x}^{\prime}\beta\right)
f_{X_{1ts}|\left\{  X_{2ts}=0,W_{ts}=0\right\}  }\left(  x\right)  xx^{\prime
}\text{\emph{d}}\sigma_{0}^{\left(  1\right)  }\right.  \label{EQ:V}\\
&  \left.  +\int1\left[  x^{\prime}\beta=0\right]  \left(  \frac
{\partial\kappa_{dts}^{\left(  2\right)  }\left(  x\right)  }{\partial
x}^{\prime}\beta\right)  f_{X_{2ts}|\left\{  X_{1ts}=0,W_{ts}=0\right\}
}\left(  x\right)  xx^{\prime}\text{\emph{d}}\sigma_{0}^{\left(  2\right)
}\right\}  ,\nonumber
\end{align}
$\mathbb{W}$ is a $k_{2}\times$ $k_{2}$ matrix defined as
\begin{equation}
\mathbb{W}=-\sum_{t>s}\int1\left[  w^{\prime}\gamma=0\right]  \left(
\frac{\partial\kappa_{\left(  1,1\right)  ts}^{\left(  3\right)  }\left(
w\right)  }{\partial w}^{\prime}\gamma\right)  f_{W_{ts}|\left\{
X_{1ts}=0,X_{2ts}=0\right\}  }\left(  w\right)  ww^{\prime}\text{\emph{d}%
}\sigma_{0}^{\left(  3\right)  },\label{EQ:W}%
\end{equation}
and $\mathbb{H}_{1}$ and $\mathbb{H}_{2}$ are written respectively as
\begin{align}
&  \mathbb{H}_{1}\left(  \rho_{1},\rho_{2}\right)  \label{EQ:H1}\\
= &  \frac{1}{2}\sum_{t>s}\sum_{d\in\mathcal{D}}\left\{  \int\kappa
_{dts}^{\left(  4\right)  }\left(  \bar{x}\right)  \left[  \left\vert \bar
{x}^{\prime}\rho_{1}\right\vert +\left\vert \bar{x}^{\prime}\rho
_{2}\right\vert -\left\vert \bar{x}^{\prime}\left(  \rho_{1}-\rho_{2}\right)
\right\vert \right]  f_{X_{1ts}|\left\{  X_{2ts}=0,W_{ts}=0\right\}  }\left(
0,\bar{x}\right)  \text{\emph{d}}\bar{x}f_{X_{2ts},W_{ts}}\left(  0,0\right)
\right.  \nonumber\\
& \left.  +\int\kappa_{dts}^{\left(  5\right)  }\left(  \bar{x}\right)
\left[  \left\vert \bar{x}^{\prime}\rho_{1}\right\vert +\left\vert \bar
{x}^{\prime}\rho_{2}\right\vert -\left\vert \bar{x}^{\prime}\left(  \rho
_{1}-\rho_{2}\right)  \right\vert \right]  f_{X_{2ts}|\left\{  X_{1ts}%
=0,W_{ts}=0\right\}  }\left(  0,\bar{x}\right)  \text{\emph{d}}\bar
{x}f_{X_{1ts},W_{ts}}\left(  0,0\right)  \right\}  \bar{K}_{2}^{k_{1}+k_{2}%
},\nonumber
\end{align}
and%
\begin{align}
\mathbb{H}_{2}\left(  \delta_{1},\delta_{2}\right)  = &  \frac{1}{2}\sum
_{t>s}\int\kappa_{\left(  1,1\right)  ts}^{\left(  6\right)  }\left(  \bar
{w}\right)  \left[  \left\vert \bar{w}^{\prime}\delta_{1}\right\vert
+\left\vert \bar{w}^{\prime}\delta_{2}\right\vert -\left\vert \bar{w}^{\prime
}\left(  \delta_{1}-\delta_{2}\right)  \right\vert \right]  \label{EQ:H2}\\
&  \cdot f_{W_{ts}|\left\{  X_{1ts}=0,X_{2ts}=0\right\}  }\left(  0,\bar
{w}\right)  \text{\emph{d}}\bar{w}f_{X_{1ts},X_{2ts}}\left(  0,0\right)
\bar{K}_{2}^{2k_{1}}.\nonumber
\end{align}
Technical terms that appear in $\mathbb{V}$, $\mathbb{W}$, $\mathbb{H}_{1}$,
$\mathbb{H}_{2}$ are defined as follows. $\sigma_{0}^{\left(  1\right)  },$
$\sigma_{0}^{\left(  2\right)  },$ and $\sigma_{0}^{\left(  3\right)  }$ are
the surface measures of $\left\{  X_{1ts}\beta=0\right\}  ,$ $\left\{
X_{2ts}\beta=0\right\}  $, and $\left\{  W_{ts}\gamma=0\right\}  ,$
respectively. $\bar{K}_{2}\equiv\int K\left(  u\right)  ^{2}$d$u.$ We
decompose $X_{1ts},X_{2ts}\ $and $W_{ts}$ into $X_{jts}=a\beta+\bar{X}%
_{jts},j=1,2,$ and $W_{ts}=a\gamma+\bar{W}_{ts},$ where $\bar{X}_{jts},j=1,2,$
are orthogonal to $\beta,$ and $\bar{W}_{ts}$ is orthogonal to $\gamma$. The
density for $X_{1ts}$ is written as $f_{X_{1ts}}\left(  a,\bar{X}%
_{1ts}\right)  $ under this decomposition. Densities for $X_{2ts}\ $and
$W_{ts}$ are written similarly. Further,%
\[
\kappa_{dts}^{\left(  1\right)  }\left(  x\right)  \equiv\mathbb{E}%
[Y_{idst}\left(  -1\right)  ^{d_{1}}|X_{i1ts}=x,X_{i2ts}=0,W_{its}=0],
\]%
\[
\kappa_{dts}^{\left(  2\right)  }\left(  x\right)  \equiv\mathbb{E}%
[Y_{idst}\left(  -1\right)  ^{d_{2}}|X_{i2ts}=x,X_{i1ts}=0,W_{its}=0],
\]%
\[
\kappa_{\left(  1,1\right)  ts}^{\left(  3\right)  }\left(  w\right)
\equiv\mathbb{E}\left[  Y_{i(1,1)ts}|W_{its}=w,X_{i1ts}=0,X_{i2ts}=0\right]  ,
\]%
\[
\kappa_{dts}^{\left(  4\right)  }\left(  x\right)  \equiv\mathbb{E}\left[
\left\vert Y_{idst}\right\vert \text{ }|X_{i1ts}=x,X_{i2ts}=0,W_{its}%
=0\right]  ,
\]%
\[
\kappa_{dts}^{\left(  5\right)  }\left(  x\right)  \equiv\mathbb{E}\left[
\left\vert Y_{idst}\right\vert \text{ }|X_{i2ts}=x,X_{i1ts}=0,W_{its}%
=0\right]  ,
\]
and%
\[
\kappa_{\left(  1,1\right)  ts}^{\left(  6\right)  }\left(  w\right)
\equiv\mathbb{E}\left[  \left\vert Y_{i\left(  1,1\right)  ts}\right\vert
\text{ }|W_{its}=w,X_{i1ts}=0,X_{i2ts}=0\right]  .
\]

\end{lemma}

\begin{proof}
[Proof of Theorem \ref{T:paneldist}]Lemma \ref{LE:P1} verifies the technical
conditions for applying the results in \cite{SeoOtsu2018}, and shows that
$\phi_{Ni}\left(  b\right)  $ and $\varphi_{Ni}\left(  r\right)  $ are
\textit{manageable} in the sense of \cite{KimPollard1990}. By Assumption P9,
Lemma 1 in \cite{SeoOtsu2018} and its subsequent analysis, we have%
\[
\hat{\beta}-\beta=O_{p} (   (  Nh_{N}^{k_{1}+k_{2}} )
^{-1/3} )  \text{ and }\hat{\gamma}-\gamma=O_{p} (   (
N\sigma_{N}^{2k_{1}} )  ^{-1/3} )  .
\]
Notice that $\hat{\beta}$ can be equivalently obtained from
\[
\arg\max_{b\in\mathcal{B}} (  Nh_{N}^{k_{1}+k_{2}})  ^{2/3}\cdot
N^{-1}\sum_{i=1}^{N}\phi_{Ni} (  \beta+ (  Nh_{N}^{k_{1}+k_{2}%
} )  ^{-1/3} \{  (  Nh_{N}^{k_{1}+k_{2}} )  ^{1/3} (
b-\beta )   \}   )  .
\]
We get the asymptotics of $\hat{\beta}$ if we can get the asymptotics of
\[
\arg\max_{\rho\in\mathbb{R}^{k_{1}}} (  Nh_{N}^{k_{1}+k_{2}} )
^{2/3}\cdot N^{-1}\sum_{i=1}^{N}\phi_{Ni} (  \beta+ (  Nh_{N}%
^{k_{1}+k_{2}} )  ^{-1/3}\rho)  .
\]
Since $\phi_{Ni}\left(  b\right)  $ is \textit{manageable} in the sense of
\cite{KimPollard1990}, by Theorem 1 in \cite{SeoOtsu2018} and Lemma
\ref{LE:P2}, we have the uniform convergence of the above stochastic process
and
\[
 (  Nh_{N}^{k_{1}+k_{2}} )  ^{2/3}\cdot N^{-1}\sum_{i=1}^{N}\phi
_{Ni} (  \beta+ (  Nh_{N}^{k_{1}+k_{2}})  ^{-1/3}\rho )
\rightsquigarrow\mathcal{Z}_{1}\left(  \rho\right)  ,
\]
where $\mathcal{Z}_{1}\left(  \rho\right)  $ is a Gaussian process with
continuous sample path, expected value $\frac{1}{2}\rho^{\prime}\mathbb{V}%
\rho$ and covariance kernel $\mathbb{H}_{1}\left(  \rho_{1},\rho_{2}\right)
.$ Apply the CMT to obtain
\[
( Nh_{N}^{k_{1}+k_{2}}) ^{1/3}( \hat{\beta}-\beta) \overset{d}{\rightarrow
}\arg\max_{\rho}\mathcal{Z}_{1}\left(  \rho\right) .
\]

Apply similar arguments to $\hat{\gamma}$. By Theorem 1 in \cite{SeoOtsu2018}
and Lemma \ref{LE:P2}, we get%
\[
(N\sigma_{N}^{2k_{1}})^{1/3}\left(  \hat{\gamma}-\gamma\right)
\overset{d}{\rightarrow}\arg\max_{\delta\in\mathbb{R}^{k_{2}}}\mathcal{Z}%
_{2}\left(  \delta\right)  ,
\]
where $\mathcal{Z}_{2}\left(  \delta\right)  $ is a Gaussian process with
continuous sample path, expected value $\frac{1}{2}\delta^{\prime}%
\mathbb{W}\delta$ and covariance kernel $\mathbb{H}_{2}\left(  \delta
_{1},\delta_{2}\right)  .$
\end{proof}

\subsection{Proof of Theorem \ref{Thm:panelinfer}}

\begin{proof}
[Proof of Theorem \ref{Thm:panelinfer}] We show that the
numerical bootstrap works for $\beta $. The discussion for $\gamma $ is
omitted due to the similarity.

The key is to show the results in \cite{HongLi2020} hold by modifying
condition (vi) in Theorem 4.1 (using their notations), for example, for $%
\beta $, that
\begin{equation*}
\Sigma _{1/2}\left( \rho _{1},\rho _{2}\right) =\lim_{N\rightarrow \infty
}(Nh_{N}^{k_{1}+k_{2}})^{1/3}\mathbb{E}[h_{N}^{k_{1}+k_{2}}\phi _{Ni}(\beta
+\rho _{1}(Nh_{N}^{k_{1}+k_{2}})^{-1/3})\phi _{Ni}(\beta +\rho
_{2}(Nh_{N}^{k_{1}+k_{2}})^{-1/3})]
\end{equation*}%
exists for all $\rho _{1},\rho _{2}\in \mathbb{R}^{k_{1}}$. This is indeed
the case as shown by our Lemma \ref{LE:P2} in Appendix \ref{appendixB}. We
now provide the details of the proof.

By some standard calculation as \cite{SeoOtsu2018}, we can show that the
convergence rate of $\hat{\beta}^{\ast }$ is $(\varepsilon
_{N1}^{-1}h_{N}^{k_{1}+k_{2}})^{-1/3}.$ Thus, if we manage to show the limit
of%
\begin{align}
& (\varepsilon _{N1}^{-1}h_{N}^{k_{1}+k_{2}})^{2/3}\cdot
N^{-1}\sum_{i=1}^{N}\phi _{Ni}(\beta +\rho (\varepsilon
_{N1}^{-1}h_{N}^{k_{1}+k_{2}})^{-1/3})+(\varepsilon
_{N1}^{-1}h_{N}^{k_{1}+k_{2}})^{2/3}\cdot (N\varepsilon _{N1})^{1/2}
\label{EQ:betahatstarobj} \\
& \cdot \lbrack N^{-1}\sum_{i=1}^{N}\phi _{Ni}^{\ast }(\beta +\rho
(\varepsilon _{N1}^{-1}h_{N}^{k_{1}+k_{2}})^{-1/3})-N^{-1}\sum_{i=1}^{N}\phi
_{Ni}(\beta +\rho (\varepsilon _{N1}^{-1}h_{N}^{k_{1}+k_{2}})^{-1/3})],
\notag
\end{align}%
then the limiting distribution of $\hat{\beta}^{\ast }$ can be established.

For the first term in equation (\ref{EQ:betahatstarobj}),
\begin{align}
& ( \varepsilon_{N1}^{-1}h_{N}^{k_{1}+k_{2}}) ^{2/3}\cdot
N^{-1}\sum_{i=1}^{N}\phi_{Ni}( \beta+\rho(
\varepsilon_{N1}^{-1}h_{N}^{k_{1}+k_{2}}) ^{-1/3})  \label{EQ:betaojbt1} \\
= &( \varepsilon_{N1}^{-1}h_{N}^{k_{1}+k_{2}}) ^{2/3}\mathbb{E}[ \phi_{Ni}(
\beta+\rho( \varepsilon_{N1}^{-1}h_{N}^{k_{1}+k_{2}}) ^{-1/3}) ]  \notag \\
& +( \varepsilon_{N1}^{-1}h_{N}^{k_{1}+k_{2}}) ^{2/3}\cdot
N^{-1}\sum_{i=1}^{N}[ \phi_{Ni}( \beta+\rho( \varepsilon
_{N1}^{-1}h_{N}^{k_{1}+k_{2}}) ^{-1/3}) -\mathbb{E}[ \phi_{Ni}( \beta+\rho(
\varepsilon_{N1}^{-1}h_{N}^{k_{1}+k_{2}}) ^{-1/3}) ] ]  \notag \\
= &\frac{1}{2}\rho^{\prime}\mathbb{V}\rho+O_{p}( ( \varepsilon
_{N1}^{-1}h_{N}^{k_{1}+k_{2}}) ^{2/3}\cdot N^{-1/2}\cdot h_{N}^{-(
k_{1}+k_{2}) /2}( \varepsilon_{N1}^{-1}h_{N}^{k_{1}+k_{2}}) ^{-1/6}) =\frac{1%
}{2}\rho^{\prime}\mathbb{V}\rho+o_{p}( 1) ,  \notag
\end{align}
where the second equality follows by Lemma \ref{LE:P2}, the assumption on $%
\varepsilon_{N1}$ such that the bias term $h_{N}^{2}$ is a small order term
compared to $( \varepsilon_{N1}^{-1}h_{N}^{k_{1}+k_{2}}) ^{2/3},$ and Markov
inequality, and the last equality holds by $N\varepsilon_{N1}\rightarrow%
\infty.$

The second term in equation (\ref{EQ:betahatstarobj}) is a sample average of
mean 0 series under $P^{\ast}$ multiplied by $(
\varepsilon_{N1}^{-1}h_{N}^{k_{1}+k_{2}}) ^{2/3}( N\varepsilon_{N1}) ^{1/2}.$
The covariance kernel of the second term is
\begin{align}
& ( \varepsilon_{N1}^{-1}h_{N}^{k_{1}+k_{2}}) ^{4/3}\cdot(
N\varepsilon_{N1}) \cdot N^{-1}\cdot\mathbb{E}^{\ast}[ \phi _{Ni}^{\ast}(
\beta+\rho_{1}( \varepsilon_{N1}^{-1}h_{N}^{k_{1}+k_{2}}) ^{-1/3})
\phi_{Ni}^{\ast}( \beta+\rho _{2}( \varepsilon_{N1}^{-1}h_{N}^{k_{1}+k_{2}})
^{-1/3}) ]  \notag \\
= &( \varepsilon_{N1}^{-1}h_{N}^{k_{1}+k_{2}}) ^{1/3}\mathbb{E}^{\ast}[
h_{N}^{k_{1}+k_{2}}\phi_{Ni}^{\ast}( \beta +\rho_{1}(
\varepsilon_{N1}^{-1}h_{N}^{k_{1}+k_{2}}) ^{-1/3}) \phi_{Ni}^{\ast}(
\beta+\rho_{2}( \varepsilon _{N1}^{-1}h_{N}^{k_{1}+k_{2}}) ^{-1/3}) ]  \notag
\\
= &( \varepsilon_{N1}^{-1}h_{N}^{k_{1}+k_{2}}) ^{1/3}\mathbb{E}[
h_{N}^{k_{1}+k_{2}}\phi_{Ni}( \beta+\rho_{1}(
\varepsilon_{N1}^{-1}h_{N}^{k_{1}+k_{2}}) ^{-1/3}) \phi _{Ni}(
\beta+\rho_{2}( \varepsilon_{N1}^{-1}h_{N}^{k_{1}+k_{2}}) ^{-1/3}) ] +o_{p}(
1)  \notag \\
= &\mathbb{H}_{1}( \rho_{1},\rho_{2}) +o_{p}( 1) ,  \label{EQ:betaobjt2}
\end{align}
by the i.i.d. sampling and Lemma \ref{LE:P2}.

Equations (\ref{EQ:betaojbt1}) and (\ref{EQ:betaobjt2}) imply that the term
in equation (\ref{EQ:betahatstarobj}) converges to $\mathcal{Z}_{1}^{\ast
}\left( \rho \right) $ with mean $\frac{1}{2}\rho ^{\prime }\mathbb{V}\rho $
and covariance kernel $\mathbb{H}_{1}\left( \rho _{1},\rho _{2}\right) $.
Thus, $\mathcal{Z}_{1}^{\ast }\left( \rho \right) $ is an independent copy
of $\mathcal{Z}_{1}\left( \rho \right) $ and
\begin{equation*}
(\varepsilon _{N1}^{-1}h_{N}^{k_{1}+k_{2}})^{1/3}(\hat{\beta}^{\ast }-\beta )%
\overset{d}{\rightarrow }\arg \max_{\rho \in \mathbb{R}^{k_{1}}}\mathcal{Z}%
_{1}^{\ast }\left( \rho \right) .
\end{equation*}%
Finally, note, by $N\varepsilon _{N1}\rightarrow \infty $ and
\begin{align*}
(\varepsilon _{N1}^{-1}h_{N}^{k_{1}+k_{2}})^{1/3}(\hat{\beta}^{\ast }-\hat{%
\beta})& =(\varepsilon _{N1}^{-1}h_{N}^{k_{1}+k_{2}})^{1/3}(\hat{\beta}%
^{\ast }-\beta )+(\varepsilon _{N1}^{-1}h_{N}^{k_{1}+k_{2}})^{1/3}(\hat{\beta%
}-\beta ) \\
& =(\varepsilon _{N1}^{-1}h_{N}^{k_{1}+k_{2}})^{1/3}(\hat{\beta}^{\ast
}-\beta )+O_{p}((\varepsilon _{N1}^{-1}h_{N}^{k_{1}+k_{2}})^{1/3}\cdot
(Nh_{N}^{k_{1}+k_{2}})^{-1/3}) \\
& =(\varepsilon _{N1}^{-1}h_{N}^{k_{1}+k_{2}})^{1/3}(\hat{\beta}^{\ast
}-\beta )+o_{p}(1),
\end{align*}%
we have%
\begin{equation*}
(\varepsilon _{N1}^{-1}h_{N}^{k_{1}+k_{2}})^{1/3}(\hat{\beta}^{\ast }-\hat{%
\beta})\overset{d}{\rightarrow }\arg \max_{\rho \in \mathbb{R}^{k_{1}}}%
\mathcal{Z}_{1}^{\ast }\left( \rho \right) .
\end{equation*}

\end{proof}

\subsection{Proof of Theorem \ref{T:paneltest}}
\begin{proof} [Proof of Theorem \ref{T:paneltest}] Recall that under Assumptions
P1--P9, we have shown that our estimator satisfies the technical conditions required
by \cite{SeoOtsu2018} in Lemma \ref{LE:P1}. Then applying Lemma 1 in \cite%
{SeoOtsu2018} to $N^{-1}\mathcal{L}_{N,\gamma }^{P,K}(r),$ we get%
\begin{equation}
N^{-1}\mathcal{L}_{N,\gamma }^{P,K}(r)-\mathcal{\bar{L}}^{P}\left( r\right)
=N^{-1}\mathcal{L}_{N,\gamma }^{P,K}(\gamma )-\mathcal{\bar{L}}^{P}\left(
\gamma \right) +\varepsilon \left\Vert r-\gamma \right\Vert ^{2}+O_{p}(
\sigma _{N}^{2}) +O_{p}( ( N\sigma _{N}^{2k_{1}})
^{-2/3})  \label{EQ:lnpk1}
\end{equation}%
for any small $\varepsilon >0,$ where use some of the results in Lemma \ref{LE:P2}.

Some standard calculations as in the proof of Lemma \ref{LE:P2} yield
\begin{equation}
\mathcal{\bar{L}}^{P}\left( r\right) -\mathcal{\bar{L}}^{P}\left( \gamma
\right) =\frac{1}{2}\left( r-\gamma \right) ^{\prime }\mathbb{W}\left(
r-\gamma \right) +o_{p}\left( \left\Vert r-\gamma \right\Vert ^{2}\right)  
\label{EQ:lnpk2}
\end{equation}%
for any $r$ in a small neighborhood of $\gamma .$

Recall that we assume $( N\sigma _{N}^{2k_{1}})^{2/3}\sigma
_{N}^{2}\rightarrow 0$ in Assumption P9 and that $\mathbb{W}$ is finite. Combining these two
conditions with (\ref{EQ:lnpk1}) and (\ref{EQ:lnpk2}) gives
\begin{equation*}
N^{-1}\mathcal{L}_{N,\gamma }^{P,K}(r)=N^{-1}\mathcal{L}_{N,\gamma
}^{P,K}(\gamma )+O_{p}\left( \left\Vert r-\gamma \right\Vert ^{2}\right)
+O_{p}( \left( N\sigma _{N}^{2k_{1}}\right) ^{-2/3}) ,
\end{equation*}%
which further implies
\begin{equation*}
N^{-1}\mathcal{L}_{N,\gamma }^{P,K}(\hat{\gamma})-N^{-1}\mathcal{L}%
_{N,\gamma }^{P,K}(\gamma )=O_{p}( \left( N\sigma _{N}^{2k_{1}}\right)
^{-2/3}) 
\end{equation*}%
since $\hat{\gamma}-\gamma =O_{p}( \left( N\sigma
_{N}^{2k_{1}}) ^{-1/3}\right)$ as shown in Theorem \ref{T:paneldist}. Therefore, we have
\begin{equation}
\sqrt{N\sigma _{N}^{2k_{1}}}\left( N^{-1}\mathcal{L}_{N,\gamma }^{P,K}(\hat{%
\gamma})-N^{-1}\mathcal{L}_{N,\gamma }^{P,K}(\gamma )\right) =o_{p}\left(
1\right) .  \label{EQ:lnpk3}
\end{equation}

Additionally, by applying the Lindeberg-Feller CLT, we have
\begin{equation*}
\sqrt{N\sigma _{N}^{2k_{1}}}\left( N^{-1}\mathcal{L}_{N,\gamma
}^{P,K}(\gamma )-\mathcal{\bar{L}}^{P}\left( \gamma \right) \right) \overset{%
d}{\rightarrow }N\left( 0,\Delta ^{P}_{\gamma}\right) .  
\end{equation*}%
Combining this with (\ref{EQ:lnpk3}), we arrive at the desired result
\begin{equation*}
\sqrt{N\sigma _{N}^{2k_{1}}}\left( N^{-1}\mathcal{L}_{N,\gamma }^{P,K}(\hat{%
\gamma})-\mathcal{\bar{L}}^{P}\left( \gamma \right) \right) \overset{d}{%
\rightarrow }N\left( 0,\Delta ^{P}_{\gamma}\right) .
\end{equation*}
\end{proof}

\section{Technical Details for Theorem \ref{T:cross_boot}}\label{appendixD}
In this section, we show the validity of the nonparametric
bootstrap for the cross-sectional model (Section \ref{SEC:infer_cross}). A complete proof of the consistency of the bootstrap procedure (Theorem \ref{T:cross_boot}) is lengthy and tedious. Considering that this process is relatively standard in the literature, we only present an outline of the proof. 
To fill in the gaps of the outline,
one first needs to define a product probability space for both the original
random series and the bootstrap series as in \cite{WellnerZhan1996} (an
application can be found in \cite{AbrevayaHuang2005}), define a suitable norm
for functions, and then establish some uniform convergence results following
\cite{Sherman1993, Sherman1994AoS, Sherman1994ET} on this
probability space with this norm. To apply Sherman's results, the criterion functions need to be \textit{manageable} in the sense
of \cite{KimPollard1990}. This is indeed the case for our estimators. For example, equation
(\ref{crossobjrK}) is a summation of
\[
\mathcal{K}_{\sigma_{N}}(  V_{im} (  \hat{\beta} )  )
Y_{im\left(  1,1\right)  }\text{sgn}\left(  W_{im}^{\prime}r\right)
=\mathcal{K}_{\sigma_{N}} (  V_{im} (  \hat{\beta} )  )
Y_{im\left(  1,1\right)  }\left(  2\cdot1\left[  W_{im}^{\prime}r>0\right]
-1\right)  .
\]
The collection of indicator functions $1\left[  W_{im}^{\prime}r>0\right]  $
for $r\in\mathcal{R}$\ is well known to be a Vapnik--Chervonenkis (VC) class.
The kernel function $\mathcal{K}$, when is chosen to satisfy some mild conditions (e.g., uniformly bounded with
compact support and continuously differentiable), 
is also \textit{manageable}. The complication of the bandwidth can be handled in
the same way as in \cite{SeoOtsu2018}.

Now suppose we have handled the technical details mentioned above. In what follows, we present
the proof outline. To ease expositions, we assume $\hat{\beta}$ and $\hat{\gamma}$ are obtained
from maximizing$,$ respectively, $\mathbb{L}_{N,\beta}^{K}\left(  b\right)  $
and $\mathbb{L}_{N,\gamma}^{K} (  r,\hat{\beta} )  $, which are
defined as%
\[
\mathbb{L}_{N,\beta}^{K}\left(  b\right)  =\sum_{i=1}^{N-1}\sum_{m>i}%
\mathcal{K}_{h_{N}}\left(  X_{im}\right)  \chi_{1}\left(  Z_{i},Z_{m},b\right)
\]
and
\[
\mathbb{L}_{N,\gamma}^{K} (  r,\hat{\beta} )  =\sum_{i=1}^{N-1}%
\sum_{m>i}\mathcal{K}_{\sigma_{N}} (  X_{im}^{\prime}\hat{\beta} )
\chi_{2}\left(  Z_{i},Z_{m},r\right),
\]
where $X_{im}$ is a sub-vector of $Z_i-Z_m$. Note that the criterion functions and notations defined here are meant for general U-processes, and thus are different from the ones used in the main text. However, the procedures that we study here share the same structure as those in Section \ref{sec:rcm} (i.e., (\ref{crossobjrK})--(\ref{crossobjrK})). Therefore, if we can demonstrate the consistency of the bootstrap here, we can easily apply the same method to our MRC estimators.

We do normalization such that $\chi_{1}\left(  Z_{m},Z_{i},\beta\right)
=0$ and $\chi_{2}\left(  Z_{i},Z_{m},\gamma\right)  =0$ for all $Z.$%
\footnote{Taking equation (\ref{crossobjrK}) as an example, this can be done by
replacing $\text{sgn}\left(  W^{\prime}_{im}r\right)  $ with $\text{sgn}\left(
W^{\prime}_{im}r\right)  -\text{sgn}\left(  W^{\prime}_{im}\gamma\right)  $ in the
criterion function.} Assume that $\chi_{1}\left(  Z_{i},Z_{m},b\right)
=\chi_{1}\left(  Z_{m},Z_{i},b\right)  ,$ $\beta$ uniquely maximizes
$\mathbb{E}\left[  \chi_{1}\left(  Z_{i},Z_{m},b\right)  |X_{im}=0\right]  ,$
and $\mathcal{K}_{h_{N}}$ is a $p$-th order kernel with $h_{N}$ that satisfies
$Nh_{N}^{p}=o\left(  1\right)  .$ Similarly, we assume that $\chi_{2}\left(
Z_{i},Z_{m},r\right)  =\chi_{2}\left(  Z_{m},Z_{i},r\right)  ,$ $\gamma$
uniquely maximizes $\mathbb{E}\left[  \chi_{2}\left(  Z_{i},Z_{m},r\right)
|X_{im}^{\prime}\beta=0\right]  ,$ and $\mathcal{K}_{\sigma_{N}}$ is a $p$-th
order kernel with $\sigma_{N}$ that satisfies $N\sigma_{N}^{p}=o\left(
1\right)  .$ We similarly denote $\mathbb{L}_{N,\beta}^{K\ast}\left(  b\right)  $
and $\mathbb{L}_{N,\gamma}^{K\ast} (  r,\hat{\beta}^{\ast} )  $
as the corresponding criterion functions using the bootstrap series.

Before establishing the consistency of the bootstrap, it is helpful to do a
recast of how we derive the asymptotics for $(  \hat{\beta},\hat{\gamma
})  .$ $\mathbb{L}_{N,\beta}^{K}\left(  b\right)  $ can be decomposed
into%
\begin{align}
&  2\left[  N\left(  N-1\right)  \right]  ^{-1}\mathbb{L}_{N,\beta}^{K}\left(
b\right) \label{EQ:LNBK_de}\\
  = &\mathbb{E}\left[  \mathcal{K}_{h_{N}}\left(  X_{im}\right)  \chi
_{1}\left(  Z_{i},Z_{m},b\right)  \right] \nonumber\\
&  +\frac{2}{N}\sum_{i=1}^{N}\left\{  \mathbb{E}\left[  \mathcal{K}_{h_{N}%
}\left(  X_{im}\right)  \chi_{1}\left(  Z_{i},Z_{m},b\right)  |Z_{i}\right]
-\mathbb{E}\left[  \mathcal{K}_{h_{N}}\left(  X_{im}\right)  \chi_{1}\left(
Z_{i},Z_{m},b\right)  \right]  \right\} \nonumber\\
&  +\frac{1}{N\left(  N-1\right)  }\sum_{i=1}^{N}\sum_{m=1,m\neq i}%
^{N}\left\{  \mathcal{K}_{h_{N}}\left(  X_{im}\right)  \chi_{1}\left(
Z_{i},Z_{m},b\right)  -\mathbb{E}\left[  \mathcal{K}_{h_{N}}\left(
X_{im}\right)  \chi_{1}\left(  Z_{i},Z_{m},b\right)  |Z_{i}\right]  \right.
\nonumber\\
&  \left.  -\mathbb{E}\left[  \mathcal{K}_{h_{N}}\left(  X_{im}\right)
\chi_{1}\left(  Z_{i},Z_{m},b\right)  |Z_{m}\right]  +\mathbb{E}\left[
\mathcal{K}_{h_{N}}\left(  X_{im}\right)  \chi_{1}\left(  Z_{i},Z_{m}%
,b\right)  \right]  \right\} \nonumber\\
  \equiv &\mathcal{T}_{N1}+\mathcal{T}_{N2}+\mathcal{T}_{N3}.\nonumber
\end{align}
Since $\chi_{1}\left(  Z_{i},Z_{m},\beta\right)  =0$ and $\beta$ uniquely
maximizes $\mathbb{E}\left[  \chi_{1}\left(  Z_{i},Z_{m},b\right)
|X_{im}=0\right]  ,$ we have%
\begin{align*}
\mathcal{T}_{1N}  &  =\mathbb{E}\left[  \chi_{1}\left(  Z_{i},Z_{m},b\right)
|X_{im}=0\right]  +O_{p}\left(  h_{N}^{p}\right) \\
&  =\frac{1}{2}\left(  b-\beta\right)  ^{\prime}\mathbb{V}_{\beta}\left(
b-\beta\right)  +O_{p}\left(  \left\Vert b-\beta\right\Vert ^{3}\right)
+O_{p}\left(  h_{N}^{p}\right)  ,
\end{align*}
where the $O_{p} (  h_{N}^{p} )  $ is the bias term from matching
and $\mathbb{V}_{\beta}$ is a negative definite matrix. For the second
term,
\begin{align}
\mathcal{T}_{N2}  &  =\frac{2}{N}\sum_{i=1}^{N}\left\{  \mathbb{E}\left[
\chi_{1}\left(  Z_{i},Z_{m},b\right)  |Z_{i},X_{im}=0\right]  f_{X_{m}}\left(
X_{i}\right)  -\mathbb{E}\left[  \chi_{1}\left(  Z_{i},Z_{m},b\right)
|X_{im}=0\right]  \right\}  +O_{p}\left(  h_{N}^{p}\right) \nonumber\\
&  =2\left(  b-\beta\right)  ^{\prime}\frac{W_{N1}}{\sqrt{N}}+o_{p}\left(
\left\Vert b-\beta\right\Vert ^{2}\right)  +O_{p}\left(  h_{N}^{p}\right)  ,
\label{EQ:tn2}%
\end{align}
where $W_{N1}$ converges to a normal distribution. For the last term,%
\[
\mathcal{T}_{N3}=o_{p}\left(  N^{-1}\right)  ,
\]
by similar arguments in \cite{Sherman1993}. Put all these results together to obtain
\[
2\left[  N\left(  N-1\right)  \right]  ^{-1}\mathbb{L}_{N,\beta}^{K}\left(
b\right)  =\frac{1}{2}\left(  b-\beta\right)  ^{\prime}\mathbb{V}_{\beta
}\left(  b-\beta\right)  +2\left(  b-\beta\right)  ^{\prime}\frac{W_{N1}%
}{\sqrt{N}}+o_{p}\left(  \left\Vert b-\beta\right\Vert ^{2}\right)
+o_{p}\left(  N^{-1}\right)  .
\]
$\hat{\beta}-\beta$ can be shown to be $O_{p} (  N^{-1/2} )  $. Using
the rate result and invoking the CMT, we have
\begin{equation}
\sqrt{N} (  \hat{\beta}-\beta )  =-2\mathbb{V}_{\beta}^{-1}%
W_{N1}+o_{p}\left(  1\right)  . \label{EQ:beta^}%
\end{equation}

Let $\mathbb{E}^{\ast}$ and $P^{\ast}$\ denote the expectation and the density
for the bootstrap series, respectively. By definition, $P^{\ast}$ puts
$N^{-1}$ probability on each of $\left\{  Z_{i}\right\}  _{i=1}^{N}.$ We can
do a similar decomposition of $\mathbb{L}_{N,\beta}^{K\ast}\left(  b\right)  $
as for $\mathbb{L}_{N,\beta}^{K}\left(  b\right)  :$%
\begin{align}
&  2\left[  N\left(  N-1\right)  \right]  ^{-1}\mathbb{L}_{N,\beta}^{K\ast
}\left(  b\right) \nonumber\\
  = &\left[  N\left(  N-1\right)  \right]  ^{-1}\sum_{i=1}^{N}\sum_{m=1,m\neq
i}^{N}\mathcal{K}_{h_{N}}\left(  X_{im}^{\ast}\right)  \chi_{1}\left(
Z_{i}^{\ast},Z_{m}^{\ast},b\right) \nonumber\\
  = &\mathbb{E}^{\ast}\left[  \mathcal{K}_{h_{N}}\left(  X_{im}^{\ast}\right)
\chi_{1}\left(  Z_{i}^{\ast},Z_{m}^{\ast},b\right)  \right] \nonumber\\
&  +\frac{2}{N}\sum_{i=1}^{N}\left\{  \mathbb{E}^{\ast}\left[  \mathcal{K}%
_{h_{N}}\left(  X_{im}^{\ast}\right)  \chi_{1}\left(  Z_{i}^{\ast},Z_{m}%
^{\ast},b\right)  |Z_{i}^{\ast}\right]  -\mathbb{E}^{\ast}\left[
\mathcal{K}_{h_{N}}\left(  X_{im}^{\ast}\right)  \chi_{1}\left(  Z_{i}^{\ast
},Z_{m}^{\ast},b\right)  \right]  \right\} \nonumber\\
&  +\frac{1}{N\left(  N-1\right)  }\sum_{i=1}^{N}\sum_{m=1,m\neq i}%
^{N}\left\{  \mathcal{K}_{h_{N}}\left(  X_{im}^{\ast}\right)  \chi_{1}\left(
Z_{i}^{\ast},Z_{m}^{\ast},b\right)  -\mathbb{E}^{\ast}\left[  \mathcal{K}%
_{h_{N}}\left(  X_{im}^{\ast}\right)  \chi_{1}\left(  Z_{i}^{\ast},Z_{m}%
^{\ast},b\right)  |Z_{i}^{\ast}\right]  \right. \nonumber\\
&  \left.  -\mathbb{E}^{\ast}\left[  \mathcal{K}_{h_{N}}\left(  X_{im}^{\ast
}\right)  \chi_{1}\left(  Z_{i}^{\ast},Z_{m}^{\ast},b\right)  |Z_{m}^{\ast
}\right]  +\mathbb{E}^{\ast}\left[  \mathcal{K}_{h_{N}}\left(  X_{im}^{\ast
}\right)  \chi_{1}\left(  Z_{i}^{\ast},Z_{m}^{\ast},b\right)  \right]
\right\} \nonumber\\
  \equiv &\mathcal{T}_{N1}^{\ast}+\mathcal{T}_{N2}^{\ast}+\mathcal{T}%
_{N3}^{\ast}. \label{EQ:LNKstartb}%
\end{align}

For $\mathcal{T}_{N1}^{\ast},$ by the definition of of $\mathbb{E}^{\ast},$
\begin{align}
\mathcal{T}_{N1}^{\ast}    = &\mathbb{E}^{\ast}\left[  \mathcal{K}_{h_{N}%
}\left(  X_{im}^{\ast}\right)  \chi_{1}\left(  Z_{i}^{\ast},Z_{m}^{\ast
},b\right)  \right]  \nonumber \\
=&\frac{1}{N^{2}}\sum_{i=1}^{N}\sum_{m=1}^{N}\mathcal{K}_{h_{N}}\left(
X_{im}\right)  \chi_{1}\left(  Z_{i},Z_{m},b\right) \nonumber\\
  = &\frac{1}{2}\left(  b-\beta\right)  ^{\prime}\mathbb{V}_{\beta}\left(
b-\beta\right)  +2\left(  b-\beta\right)  ^{\prime}\frac{W_{N1}}{\sqrt{N}%
}+O_{p}\left(  \left\Vert b-\beta\right\Vert ^{3}\right)  +o_{p}\left(
N^{-1}\right)  ,\label{EQ:LNKstartb1}
\end{align}
where the last line holds because the second line is the same as $2\left[  N\left(
N-1\right)  \right]  ^{-1}\mathbb{L}_{N,\beta}^{K}\left(  b\right)  $ plus a
small order term.

Both $\mathcal{T}_{N2}$ and $\mathcal{T}_{N2}^{\ast}$ are sample averages of
mean zero series. Moreover,\textbf{ }for $\left\Vert b-\beta\right\Vert
=O (  N^{-1/2} )  ,$ $N\mathcal{T}_{N2}=2\sqrt{N}\left(
b-\beta\right)  ^{\prime}W_{N1}+o_{p}\left(  1\right)  $ and $W_{N1}$
converges to a normal distribution$.$ Then, we can apply the results in
\cite{gine1990bootstrapping} that $NT_{N2}^{\ast}$ approximates the
distribution of $NT_{N2}$ for $\left\Vert b-\beta\right\Vert =O (
N^{-1/2} )  .$ That is, for $\left\Vert b-\beta\right\Vert =O (
N^{-1/2} )  ,$
\begin{equation}
\mathcal{T}_{N2}^{\ast}=2\left(  b-\beta\right)  ^{\prime}\frac{W_{N1}^{\ast}%
}{\sqrt{N}}+o_{p}\left(  N^{-1}\right)  ,\label{EQ:LNKstartb2}%
\end{equation}
where $W_{N1}^{\ast}$ is an independent copy of $W_{N1}.$

For $\mathcal{T}_{N3}^{\ast},$ terms across $i$ and $m$ are uncorrelated with
each other under $P^{\ast}$. Therefore, $\left[  \mathbb{E}^{\ast}\left(
\mathcal{T}_{N3}^{\ast2}\right)  \right]  ^{1/2}$ is of the same order as
\begin{align*}
&  N^{-1}\left\{  \mathbb{E}^{\ast}\left[  \left[  \mathcal{K}_{h_{N}}\left(
X_{im}^{\ast}\right)  \chi_{1}\left(  Z_{i}^{\ast},Z_{m}^{\ast},b\right)
\right]  ^{2}\right]  \right\}  ^{1/2}  =N^{-1}\left\{  \frac{1}{N^{2}}\sum_{i=1}^{N}\sum_{m=1}^{N}\mathcal{K}%
_{h_{N}}\left(  X_{im}\right)  ^{2}\chi_{1}\left(  Z_{i},Z_{m},b\right)
^{2}\right\}^{1/2}.
\end{align*}
Note that $\chi_{1}\left(  Z_{i},Z_{m},\beta\right)  =0$, so the above term is
$o_{p}\left(  N^{-1}\right)  $ for $\left\Vert b-\beta\right\Vert =O\left(
N^{-1/2}\right)  .$ That means $\left[  \mathbb{E}^{\ast}\left(
\mathcal{T}_{N3}^{\ast2}\right)  \right]  ^{1/2}=o_{p}\left(  N^{-1}\right)
.$ By Markov inequality,
\begin{equation}
\mathcal{T}_{N3}^{\ast}=o_{p}\left(  N^{-1}\right)  .\label{EQ:LNKstartb3}%
\end{equation}

To summarize, equations (\ref{EQ:LNKstartb})--(\ref{EQ:LNKstartb3}) imply that for $\left\Vert
b-\beta\right\Vert =O\left(  N^{-1/2}\right)  ,$%
\[
2\left[  N\left(  N-1\right)  \right]  ^{-1}\mathbb{L}_{N,\beta}^{K\ast
}\left(  b\right)  =\frac{1}{2}\left(  b-\beta\right)  ^{\prime}%
\mathbb{V}_{\beta}\left(  b-\beta\right)  +2\left(  b-\beta\right)  ^{\prime
}\frac{W_{N1}}{\sqrt{N}}+2\left(  b-\beta\right)  ^{\prime}\frac{W_{N1}^{\ast
}}{\sqrt{N}}+o_{p}\left(  N^{-1}\right)  .
\]
We can similarly show that $\hat{\beta}^{\ast}-\beta=O_{p}\left(
N^{-1/2}\right)  .$ By the CMT again,
\[
\sqrt{N} (  \hat{\beta}^{\ast}-\beta)  =-2\mathbb{V}_{\beta}%
^{-1}W_{N1}-2\mathbb{V}_{\beta}^{-1}W_{N1}^{\ast}+o_{p}\left(  1\right)  .
\]
Substitute equation (\ref{EQ:beta^}) into the equation above to get%
\begin{equation}
\sqrt{N} (  \hat{\beta}^{\ast}-\hat{\beta} )  =-2\mathbb{V}_{\beta
}^{-1}W_{N1}^{\ast}+o_{p}\left(  1\right)  .\label{EQ:betahatstar}%
\end{equation}
Since $W_{N1}^{\ast}$ is an independent copy of $W_{N1},$ $\sqrt{N} (
\hat{\beta}^{\ast}-\hat{\beta} )  $ converges to the same distribution as
the one $\sqrt{N} (  \hat{\beta}-\beta )  $ converges to.

We turn to the inference for $\hat{\gamma}.$ Again, we decompose
$\mathbb{L}_{N,\gamma}^{K} (  r,\hat{\beta} )  $ and analyze each
term in the decompositions one by one. We start with%
\begin{align}
&  2\left[  N\left(  N-1\right)  \right]  ^{-1}\mathbb{L}_{N,\gamma}%
^{K}\left(  r,\hat{\beta}\right) \label{EQ:LNKr}\\
  = & \left[  N\left(  N-1\right)  \right]  ^{-1}\sum_{i=1}^{N}\sum_{m=1,m\neq
i}^{N}\mathcal{K}_{\sigma_{N}}\left(  X_{im}^{\prime}\hat{\beta}\right)
\chi_{2}\left(  Z_{i},Z_{m},r\right) \nonumber\\
  = &\left[  N\left(  N-1\right)  \right]  ^{-1}\sum_{i=1}^{N}\sum_{m=1,m\neq
i}^{N}\mathcal{K}_{\sigma_{N}}\left(  X_{im}^{\prime}\beta\right)  \chi
_{2}\left(  Z_{i},Z_{m},r\right) \nonumber\\
&  +\left[  N\left(  N-1\right)  \right]  ^{-1}\sum_{i=1}^{N}\sum_{m=1,m\neq
i}^{N}\chi_{2}\left(  Z_{i},Z_{m},r\right)  \nabla\mathcal{K}_{\sigma_{N}%
}\left(  X_{im}^{\prime}\beta\right)  X_{im}^{\prime}\frac{\left(  \hat{\beta
}-\beta\right)  }{\sigma_{N}}\nonumber\\
&  +\left[  N\left(  N-1\right)  \right]  ^{-1}\sum_{i=1}^{N}\sum_{m=1,m\neq
i}^{N}\chi_{2}\left(  Z_{i},Z_{m},r\right)  \nabla^{2}\mathcal{K}_{\sigma_{N}%
}\left(  X_{im}^{\prime}\tilde{\beta}\right)  \left[  X_{im}^{\prime}%
\frac{\left(  \hat{\beta}-\beta\right)  }{\sigma_{N}}\right]  ^{2}\nonumber\\
  \equiv & \mathcal{T}_{N4}+\mathcal{T}_{N5}+\mathcal{T}_{N6},\nonumber
\end{align}
where $\tilde{\beta}$ is a vector that lies between $\beta$ and $\hat
{\beta}$ and makes the equality hold.

Note that $\mathcal{T}_{N4}$ shares the same structure as $2\left[  N\left(
N-1\right)  \right]  ^{-1}\mathbb{L}_{N,\beta}^{K}\left(  b\right)  ,$ and so
\begin{equation}
\mathcal{T}_{N4}=\frac{1}{2}\left(  r-\gamma\right)  ^{\prime}\mathbb{V}%
_{\gamma}\left(  r-\gamma\right)  +2\left(  r-\gamma\right)  ^{\prime}%
\frac{W_{N2}}{\sqrt{N}}+O_{p}\left(  \left\Vert r-\gamma\right\Vert
^{3}\right)  +O_{p}\left(  \sigma_{N}^{p}\right)  +o_{p}\left(  N^{-1}\right)
, \label{EQ:LNKr1}%
\end{equation}
where $\mathbb{V}_{\gamma}$ is a negative definite matrix, and $W_{N2}$
converges to a normal distribution.

For term $\mathcal{T}_{N5},$
\[
\mathcal{T}_{N5}=\left\{  \left[  N\left(  N-1\right)  \right]  ^{-1}%
\sum_{i=1}^{N}\sum_{m=1,m\neq i}^{N}\chi_{2}\left(  Z_{i},Z_{m},r\right)
\sigma_{N}^{-1}\nabla\mathcal{K}_{\sigma_{N}}\left(  X_{im}^{\prime}%
\beta\right)  X_{im}^{\prime}\right\}  \left(  \hat{\beta}-\beta\right)  .
\]
Note that
\begin{align}
&  \left[  N\left(  N-1\right)  \right]  ^{-1}\sum_{i=1}^{N}\sum_{m=1,m\neq
i}^{N}\chi_{2}\left(  Z_{i},Z_{m},r\right)  \sigma_{N}^{-1}\nabla
\mathcal{K}_{\sigma_{N}}\left(  X_{im}^{\prime}\beta\right)  X_{im}^{\prime
}\label{EQ:tn5_o}\\
  = &\mathbb{E}\left[  \chi_{2}\left(  Z_{i},Z_{m},r\right)  \sigma_{N}%
^{-1}\nabla\mathcal{K}_{\sigma_{N}}\left(  X_{im}^{\prime}\beta\right)
X_{im}^{\prime}\right] \nonumber\\
&  +\frac{2}{N}\sum_{i=1}^{N}\left\{  \mathbb{E}\left[  \chi_{2}\left(
Z_{i},Z_{m},r\right)  \sigma_{N}^{-1}\nabla\mathcal{K}_{\sigma_{N}}\left(
X_{im}^{\prime}\beta\right)  X_{im}^{\prime}|Z_{i}\right]  -\mathbb{E}\left[
\chi_{2}\left(  Z_{i},Z_{m},r\right)  \sigma_{N}^{-1}\nabla\mathcal{K}%
_{\sigma_{N}}\left(  X_{im}^{\prime}\beta\right)  X_{im}^{\prime}\right]
\right\} \nonumber\\
&  +\frac{1}{N\left(  N-1\right)  }\sum_{i=1}^{N}\sum_{m=1,m\neq i}%
^{N}\left\{  \chi_{2}\left(  Z_{i},Z_{m},r\right)  \sigma_{N}^{-1}%
\nabla\mathcal{K}_{\sigma_{N}}\left(  X_{im}^{\prime}\beta\right)
X_{im}^{\prime}-\mathbb{E}\left[  \chi_{2}\left(  Z_{i},Z_{m},r\right)
\sigma_{N}^{-1}\nabla\mathcal{K}_{\sigma_{N}}\left(  X_{im}^{\prime}%
\beta\right)  X_{im}^{\prime}|Z_{i}\right]  \right. \nonumber\\
&  \left.  -\mathbb{E}\left[  \chi_{2}\left(  Z_{i},Z_{m},r\right)  \sigma
_{N}^{-1}\nabla\mathcal{K}_{\sigma_{N}}\left(  X_{im}^{\prime}\beta\right)
X_{im}^{\prime}|Z_{m}\right]  +\mathbb{E}\left[  \chi_{2}\left(  Z_{i}%
,Z_{m},r\right)  \sigma_{N}^{-1}\nabla\mathcal{K}_{\sigma_{N}}\left(
X_{im}^{\prime}\beta\right)  X_{im}^{\prime}\right]  \right\}  \nonumber
\end{align}
with the lead term $\mathbb{E} [  \chi_{2}\left(  Z_{i},Z_{m},r\right)
\sigma_{N}^{-1}\nabla\mathcal{K}_{\sigma_{N}}\left(  X_{im}^{\prime}%
\beta\right)  X_{im}^{\prime} ]  $. As $\int\nabla\mathcal{K}%
_{\sigma_{N}}\left(  u\right)  \text{d}u=0,$ the lead term cancels one more
$\sigma_{N}^{-1}$ when calculating the expectation. Since $\chi_{2}\left(
Z_{i},Z_{m},\gamma\right)  =0,$ we can write
\begin{equation}
\mathbb{E}\left[  \chi_{2}\left(  Z_{i},Z_{m},r\right)  \sigma_{N}^{-1}%
\nabla\mathcal{K}_{\sigma_{N}}\left(  X_{im}^{\prime}\beta\right)
X_{im}\right]  =\left(  r-\gamma\right)  ^{\prime}\Gamma+O (  \left\Vert
r-\gamma\right\Vert ^{2} )  +O\left(  \sigma_{N}^{p}\right)  ,
\label{EQ:tn5lead}%
\end{equation}
for some matrix $\Gamma.$ Collect all these result to obtain
\begin{align}
\mathcal{T}_{N5}  &  =\left(  r-\gamma\right)  ^{\prime}\Gamma\left(
\hat{\beta}-\beta\right)  +O_{p}\left(  \left\Vert r-\gamma\right\Vert
^{2}\left\Vert \hat{\beta}-\beta\right\Vert \right)  +O\left(  \sigma_{N}%
^{p}\right) \label{EQ:LNKr2}\\
&  =-2\left(  r-\gamma\right)  ^{\prime}\Gamma\mathbb{V}_{\beta}^{-1}%
\frac{W_{N1}}{\sqrt{N}}+O_{p}\left(  \left\Vert r-\gamma\right\Vert
^{2}\left\Vert \hat{\beta}-\beta\right\Vert \right)  +o_{p}\left(
\frac{\left\Vert r-\gamma\right\Vert }{\sqrt{N}}\right)  +O\left(  \sigma
_{N}^{p}\right), \nonumber
\end{align}
where the second line follows by substituting in equation (\ref{EQ:beta^}).

For term $\mathcal{T}_{N6},$ it is not hard to see that%
\begin{equation}
\mathcal{T}_{N6}=O_{p}\left(  \left.  \left\Vert r-\gamma\right\Vert
\left\Vert \hat{\beta}-\beta\right\Vert ^{2}\right/  \sigma_{N}^{2}\right)
=o_{p}\left(  N^{-1}\right)  , \label{EQ:LNKr3}%
\end{equation}
for $\Vert r-\gamma\Vert=O (  N^{-1/2} )  $, and $\sigma_{N}$ in Assumption C7.

To summarize, equations (\ref{EQ:LNKr}), (\ref{EQ:LNKr1}), (\ref{EQ:LNKr2}),
and (\ref{EQ:LNKr3}) imply that for $\left\Vert r-\gamma\right\Vert =O\left(
N^{-1/2}\right)  ,$%
\begin{align*}
2\left[  N\left(  N-1\right)  \right]  ^{-1}\mathbb{L}_{N,\gamma}^{K}\left(
r,\hat{\beta}\right)     = &\frac{1}{2}\left(  r-\gamma\right)  ^{\prime
}\mathbb{V}_{\gamma}\left(  r-\gamma\right) \\
&  +2\left(  r-\gamma\right)  ^{\prime}\frac{W_{N2}}{\sqrt{N}}-2\left(
r-\gamma\right)  ^{\prime}\Gamma\mathbb{V}_{\beta}^{-1}\frac{W_{N1}}{\sqrt{N}%
}+o_{p}\left(  N^{-1}\right)  .
\end{align*}
Once $\hat{\gamma}-\gamma=O_{p}\left(  N^{-1/2}\right)$ is established,  applying the CMT gives
\begin{equation}
\sqrt{N}\left(  \hat{\gamma}-\gamma\right)  =-2\mathbb{V}_{\gamma}^{-1}%
W_{N2}+2\mathbb{V}_{\gamma}^{-1}\Gamma\mathbb{V}_{\beta}^{-1}W_{N1}%
+o_{p}\left(  1\right)  . \label{EQ:gammahat}%
\end{equation}

For $\mathbb{L}_{N,\gamma}^{K\ast}\left(  r,b\right)  ,$ we decompose it
similarly:
\begin{align}
&  2\left[  N\left(  N-1\right)  \right]  ^{-1}\mathbb{L}_{N,\gamma}^{K\ast
}\left(  r,\hat{\beta}^{\ast}\right) \label{EQ:LNKrstar}\\
&  =\left[  N\left(  N-1\right)  \right]  ^{-1}\sum_{i=1}^{N}\sum_{m=1,m\neq
i}^{N}\mathcal{K}_{\sigma_{N}}\left(  X_{im}^{\ast\prime}\hat{\beta}^{\ast
}\right)  \chi_{2}\left(  Z_{i}^{\ast},Z_{m}^{\ast},r\right) \nonumber\\
&  =\left[  N\left(  N-1\right)  \right]  ^{-1}\sum_{i=1}^{N}\sum_{m=1,m\neq
i}^{N}\mathcal{K}_{\sigma_{N}}\left(  X_{im}^{\ast\prime}\hat{\beta}\right)
\chi_{2}\left(  Z_{i}^{\ast},Z_{m}^{\ast},r\right) \nonumber\\
&  +\left[  N\left(  N-1\right)  \right]  ^{-1}\sum_{i=1}^{N}\sum_{m=1,m\neq
i}^{N}\chi_{2}\left(  Z_{i}^{\ast},Z_{m}^{\ast},r\right)  \nabla
\mathcal{K}_{\sigma_{N}}\left(  X_{im}^{\ast\prime}\hat{\beta}\right)  \left(
\frac{X_{im}^{\ast\prime}\left(  \hat{\beta}^{\ast}-\hat{\beta}\right)
}{\sigma_{N}}\right) \nonumber\\
&  +\left[  N\left(  N-1\right)  \right]  ^{-1}\sum_{i=1}^{N}\sum_{m=1,m\neq
i}^{N}\chi_{2}\left(  Z_{i}^{\ast},Z_{m}^{\ast},r\right)  \nabla
^{2}\mathcal{K}_{\sigma_{N}}\left(  X_{im}^{\ast\prime}\tilde{\beta}^{\ast
}\right)  \left(  \frac{X_{im}^{\ast\prime}\left(  \hat{\beta}^{\ast}%
-\hat{\beta}\right)  }{\sigma_{N}}\right)  ^{2}\nonumber\\
&  \equiv\mathcal{T}_{N4}^{\ast}+\mathcal{T}_{N5}^{\ast}+\mathcal{T}%
_{N6}^{\ast},\nonumber
\end{align}
where $\tilde{\beta}^{\ast}$ is some vector that lies between $b$ and
$\hat{\beta}$ and makes the equality hold.

For $\mathcal{T}_{N4}^{\ast},$ we have%
\begin{align}
\mathcal{T}_{N4}^{\ast}    =&\left[  N\left(  N-1\right)  \right]  ^{-1}%
\sum_{i=1}^{N}\sum_{m=1,m\neq i}^{N}\mathcal{K}_{\sigma_{N}}\left(
X_{im}^{\ast\prime}\hat{\beta}\right)  \chi_{2}\left(  Z_{i}^{\ast}%
,Z_{m}^{\ast},r\right) \nonumber\\
  = &\mathbb{E}^{\ast}\left[  \mathcal{K}_{\sigma_{N}}\left(  X_{im}%
^{\ast\prime}\hat{\beta}\right)  \chi_{2}\left(  Z_{i}^{\ast},Z_{m}^{\ast
},r\right)  \right] \nonumber\\
&  +\frac{2}{N}\sum_{i=1}^{N}\left\{  \mathbb{E}^{\ast}\left[  \mathcal{K}%
_{\sigma_{N}}\left(  X_{im}^{\ast\prime}\hat{\beta}\right)  \chi_{2}\left(
Z_{i}^{\ast},Z_{m}^{\ast},r\right)  |Z_{i}^{\ast}\right]  -\mathbb{E}^{\ast
}\left[  \mathcal{K}_{\sigma_{N}}\left(  X_{im}^{\ast\prime}\hat{\beta
}\right)  \chi_{2}\left(  Z_{i}^{\ast},Z_{m}^{\ast},r\right)  \right]
\right\} \nonumber\\
&  +\frac{1}{N\left(  N-1\right)  }\sum_{i=1}^{N}\sum_{m=1,m\neq i}%
^{N}\left\{  \mathcal{K}_{\sigma_{N}}\left(  X_{im}^{\ast\prime}\hat{\beta
}\right)  \chi_{2}\left(  Z_{i}^{\ast},Z_{m}^{\ast},r\right)  -\mathbb{E}%
^{\ast}\left[  \mathcal{K}_{\sigma_{N}}\left(  X_{im}^{\ast\prime}\hat{\beta
}\right)  \chi_{2}\left(  Z_{i}^{\ast},Z_{m}^{\ast},r\right)  |Z_{i}^{\ast
}\right]  \right. \nonumber\\
&  \left.  -\mathbb{E}^{\ast}\left[  \mathcal{K}_{\sigma_{N}}\left(
X_{im}^{\ast\prime}\hat{\beta}\right)  \chi_{2}\left(  Z_{i}^{\ast}%
,Z_{m}^{\ast},r\right)  |Z_{m}^{\ast}\right]  +\mathbb{E}^{\ast}\left[
\mathcal{K}_{\sigma_{N}}\left(  X_{im}^{\ast\prime}\hat{\beta}\right)
\chi_{2}\left(  Z_{i}^{\ast},Z_{m}^{\ast},r\right)  \right]  \right\}
.\nonumber\\
  \equiv &\mathcal{T}_{N4,1}^{\ast}+\mathcal{T}_{N4,2}^{\ast}+\mathcal{T}%
_{N4,3}^{\ast}.\label{EQ:LNKrstar1}
\end{align}
We analyze the three terms in (\ref{EQ:LNKrstar1}) one by one.

For $\mathcal{T}_{N4,1}^{\ast},$ we have%
\begin{align}
\mathcal{T}_{N4,1}^{\ast}  &  =\mathbb{E}^{\ast}\left[  \mathcal{K}%
_{\sigma_{N}}\left(  X_{im}^{\ast\prime}\hat{\beta}\right)  \chi_{2}\left(
Z_{i}^{\ast},Z_{m}^{\ast},r\right)  \right] \label{EQ:LNKrstar1_1}\\
&  =\frac{1}{N^{2}}\sum_{i=1}^{N}\sum_{m=1}^{N}\mathcal{K}_{\sigma_{N}}\left(
X_{im}^{\prime}\hat{\beta}\right)  \chi_{2}\left(  Z_{i},Z_{m},r\right)
\nonumber\\
&  =\frac{1}{2}\left(  r-\gamma\right)  ^{\prime}\mathbb{V}_{\gamma}\left(
r-\gamma\right)  +2\left(  r-\gamma\right)  ^{\prime}\frac{W_{N2}}{\sqrt{N}%
}-2\left(  r-\gamma\right)  ^{\prime}\Gamma\mathbb{V}_{\beta}^{-1}\frac
{W_{N1}}{\sqrt{N}}+o_{p}\left(  N^{-1}\right)  ,\nonumber
\end{align}
where the last line holds because the term in the second line is $2\left[
N\left(  N-1\right)  \right]  ^{-1}\mathbb{L}_{N,\gamma}^{K}(
r,\hat{\beta})  $ plus a small order term.

For $\mathcal{T}_{N4,2}^{\ast},$ we first note that
\begin{align}
&  \sum_{i=1}^{N}\left\{  \mathbb{E}\left[  \mathcal{K}_{\sigma_{N}}\left(
X_{im}^{\prime}\beta\right)  \chi_{2}\left(  Z_{i},Z_{m},r\right)
|Z_{i}\right]  -\mathbb{E}\left[  \mathcal{K}_{\sigma_{N}}\left(
X_{im}^{\prime}\beta\right)  \chi_{2}\left(  Z_{i},Z_{m},r\right)  \right]
\right\} \label{EQ:t421}\\
  = &\sqrt{N}\left(  r-\gamma\right)  ^{\prime}W_{N2}+o_{p}\left(  1\right)
,\nonumber
\end{align}
for $\left\Vert r-\gamma\right\Vert =O\left(
N^{-1/2}\right)   ,$ which holds for
the same reason as for equation (\ref{EQ:tn2})$.$ Since $\hat{\beta}%
-\beta=O_{p}\left(  N^{-1/2}\right)  ,$ by the equicontinuity we claimed,%
\begin{align}
\sum_{i=1}^{N}  &  \left\{  \mathbb{E}\left[  \mathcal{K}_{\sigma_{N}}\left(
X_{im}^{\prime}\beta\right)  \chi_{2}\left(  Z_{i},Z_{m},r\right)
|Z_{i}\right]  -\mathbb{E}\left[  \mathcal{K}_{\sigma_{N}}\left(
X_{im}^{\prime}\beta\right)  \chi_{2}\left(  Z_{i},Z_{m},r\right)  \right]
\right\} \label{EQ:t422}\\
&  -\sum_{i=1}^{N}\left\{  \mathbb{E}_{N}\left[  \mathcal{K}_{\sigma_{N}%
}\left(  X_{im}^{\prime}\hat{\beta}\right)  \chi_{2}\left(  Z_{i}%
,Z_{m},r\right)  |Z_{i}\right]  -\mathbb{E}_{N}\left[  \mathcal{K}_{\sigma
_{N}}\left(  X_{im}^{\prime}\hat{\beta}\right)  \chi_{2}\left(  Z_{i}%
,Z_{m},r\right)  \right]  \right\}  =o_{p}\left(  1\right)  .\nonumber
\end{align}
Then equations (\ref{EQ:t421}) and (\ref{EQ:t422}) imply that%
\begin{align*}
&  2\sum_{i=1}^{N}\left\{  \mathbb{E}_{N}\left[  \mathcal{K}_{\sigma_{N}%
}\left(  X_{im}^{\prime}\hat{\beta}\right)  \chi_{2}\left(  Z_{i}%
,Z_{m},r\right)  |Z_{i}\right]  -\mathbb{E}_{N}\left[  \mathcal{K}_{\sigma
_{N}}\left(  X_{im}^{\prime}\hat{\beta}\right)  \chi_{2}\left(  Z_{i}%
,Z_{m},r\right)  \right]  \right\} \\
&  =2\sqrt{N}\left(  r-\gamma\right)  ^{\prime}W_{N2}+o_{p}\left(  1\right)  ,
\end{align*}
which approximates a normal distribution for $\left\Vert r-\gamma\right\Vert =O\left(
N^{-1/2}\right)   .$ Thus, we are able to apply the results of
\cite{gine1990bootstrapping} and get that
\[
N\mathcal{T}_{N4,2}^{\ast}=2\sum_{i=1}^{N}\left\{  \mathbb{E}^{\ast}\left[
\mathcal{K}_{\sigma_{N}}\left(  X_{im}^{\ast\prime}\hat{\beta}\right)
\chi_{2}\left(  Z_{i}^{\ast},Z_{m}^{\ast},r\right)  |Z_{i}^{\ast}\right]
-\mathbb{E}^{\ast}\left[  \mathcal{K}_{\sigma_{N}}\left(  X_{im}^{\ast\prime
}\hat{\beta}\right)  \chi_{2}\left(  Z_{i}^{\ast},Z_{m}^{\ast},r\right)
\right]  \right\}
\]
approximates the distribution of the above term. That is
\begin{equation}
\mathcal{T}_{N4,2}^{\ast}=2\left(  r-\gamma\right)  ^{\prime}\frac
{W_{N2}^{\ast}}{\sqrt{N}}+o_{p}\left(  N^{-1}\right)  , \label{EQ:LNKrstar1_2}%
\end{equation}
where $W_{N2}^{\ast}$ is an independent copy of $W_{N2}.$

The same arguments as we used for $\mathcal{T}_{N3}^{\ast}$ lead to
\begin{equation}
\mathcal{T}_{N4,3}^{\ast}=o_{p}\left(  N^{-1}\right)  . \label{EQ:LNKrstar1_3}%
\end{equation}

Then, equations (\ref{EQ:LNKrstar1}), (\ref{EQ:LNKrstar1_1}), (\ref{EQ:LNKrstar1_2}%
), and (\ref{EQ:LNKrstar1_3}) together imply that for $\Vert r-\gamma\Vert=O (  N^{-1/2})  $,
\begin{align}
\mathcal{T}_{N4}^{\ast}    = &\frac{1}{2}\left(  r-\gamma\right)  ^{\prime
}\mathbb{V}_{\gamma}\left(  r-\gamma\right)  +2\left(  r-\gamma\right)
^{\prime}\frac{W_{N2}}{\sqrt{N}}-2\left(  r-\gamma\right)  ^{\prime}%
\Gamma\mathbb{V}_{\beta}^{-1}\frac{W_{N1}}{\sqrt{N}}\label{EQ:tn4}\\
&  +2\left(  r-\gamma\right)  ^{\prime}\frac{W_{N2}^{\ast}}{\sqrt{N}}%
+o_{p}\left(  N^{-1}\right)  .\nonumber
\end{align}

For $\mathcal{T}_{N5}^{\ast},$ we have
\begin{align*}
\mathcal{T}_{N5}^{\ast}  &  =\left[  N\left(  N-1\right)  \right]  ^{-1}%
\sum_{i=1}^{N}\sum_{m=1,m\neq i}^{N}\chi_{2}\left(  Z_{i}^{\ast},Z_{m}^{\ast
},r\right)  \nabla\mathcal{K}_{\sigma_{N}}\left(  X_{im}^{\ast\prime}%
\hat{\beta}\right)  \left(  \frac{X_{im}^{\ast\prime}\left(  \hat{\beta}%
^{\ast}-\hat{\beta}\right)  }{\sigma_{N}}\right) \\
&  =\left\{  \left[  N\left(  N-1\right)  \right]  ^{-1}\sum_{i=1}^{N}%
\sum_{m=1,m\neq i}^{N}\chi_{2}\left(  Z_{i}^{\ast},Z_{m}^{\ast},r\right)
\sigma_{N}^{-1}\nabla\mathcal{K}_{\sigma_{N}}\left(  X_{im}^{\ast\prime}%
\hat{\beta}\right)  X_{im}^{\ast\prime}\right\}  \left(  \hat{\beta}^{\ast
}-\hat{\beta}\right)  .
\end{align*}
Following the same analysis for $\mathcal{T}_{N5},$ the lead term for $\mathcal{T}%
_{N5}^{\ast}$ is
\[
\mathbb{E}^{\ast}\left[  \chi_{2}\left(  Z_{i}^{\ast},Z_{m}^{\ast},r\right)
\sigma_{N}^{-1}\nabla\mathcal{K}_{\sigma_{N}}\left(  X_{im}^{\ast\prime}%
\hat{\beta}\right)  X_{im}^{\ast\prime}\right]  \left(  \hat{\beta}^{\ast
}-\hat{\beta}\right)  .
\]
Note that%
\begin{align*}
&  \mathbb{E}^{\ast}\left[  \chi_{2}\left(  Z_{i}^{\ast},Z_{m}^{\ast
},r\right)  \sigma_{N}^{-1}\nabla\mathcal{K}_{\sigma_{N}}\left(  X_{im}%
^{\ast\prime}\hat{\beta}\right)  X_{im}^{\ast\prime}\right]  \left(
\hat{\beta}^{\ast}-\hat{\beta}\right) \\
&  =\left[  \frac{1}{N^{2}}\sum_{i=1}^{N}\sum_{m=1}^{N}\chi_{2}\left(
Z_{i},Z_{m},r\right)  \sigma_{N}^{-1}\nabla\mathcal{K}_{\sigma_{N}}\left(
X_{im}^{\prime}\beta\right)  X_{im}\right]  \left(  \hat{\beta}^{\ast}%
-\hat{\beta}\right) \\
&  =\left(  r-\gamma\right)  ^{\prime}\Gamma\left(  \hat{\beta}^{\ast}%
-\hat{\beta}\right)  +O_{p}\left(  \left\Vert r-\gamma\right\Vert
^{2}\left\Vert \hat{\beta}^{\ast}-\hat{\beta}\right\Vert \right)
+o_{p}\left(  N^{-1}\right) \\
&  =-2\left(  r-\gamma\right)  ^{\prime}\Gamma\mathbb{V}_{\beta}^{-1}%
\frac{W_{N1}^{\ast}}{\sqrt{N}}+o_{p}\left(  \frac{\left\Vert r-\gamma
\right\Vert }{\sqrt{N}}\right)  +O_{p}\left(  \left\Vert r-\gamma\right\Vert
^{2}\left\Vert \hat{\beta}^{\ast}-\hat{\beta}\right\Vert \right)
+o_{p}\left(  N^{-1}\right)  ,
\end{align*}
where we use equations (\ref{EQ:tn5_o}) and (\ref{EQ:tn5lead}) to get the
third line, and we substitute equation (\ref{EQ:betahatstar}) in the last
line. Then we have%
\begin{equation}
\mathcal{T}_{N5}^{\ast}=-2\left(  r-\gamma\right)  ^{\prime}\Gamma
\mathbb{V}_{\beta}^{-1}\frac{W_{N1}^{\ast}}{\sqrt{N}}+o_{p}\left(
\frac{\left\Vert r-\gamma\right\Vert }{\sqrt{N}}\right)  +O_{p}\left(
\left\Vert r-\gamma\right\Vert ^{2}\left\Vert \hat{\beta}^{\ast}-\hat{\beta
}\right\Vert \right)  +o_{p}\left(  N^{-1}\right)  . \label{EQ:tn5}%
\end{equation}

For $\mathcal{T}_{N6}^{\ast},$ it is not hard to see that%
\begin{equation}
\mathcal{T}_{N6}^{\ast}=O_{p}\left(  \frac{\left\Vert r-\gamma\right\Vert
\left\Vert \hat{\beta}^{\ast}-\hat{\beta}\right\Vert ^{2}}{\sigma_{N}^{2}%
}\right)  . \label{EQ:tn6}%
\end{equation}

Then, equations (\ref{EQ:LNKrstar}), (\ref{EQ:tn4}), (\ref{EQ:tn5}), and
(\ref{EQ:tn6}) together imply that for $\Vert r-\gamma\Vert=O\left(  N^{-1/2}\right)  $,
\begin{align*}
&  2\left[  N\left(  N-1\right)  \right]  ^{-1}\mathbb{L}_{N,\gamma}^{K\ast
}\left(  r,\hat{\beta}^{\ast}\right) \\
  =&\frac{1}{2}\left(  r-\gamma\right)  ^{\prime}\mathbb{V}_{\gamma}\left(
r-\gamma\right)  +2\left(  r-\gamma\right)  ^{\prime}\frac{W_{N2}}{\sqrt{N}%
}-2\left(  r-\gamma\right)  ^{\prime}\Gamma\mathbb{V}_{\beta}^{-1}\frac
{W_{N1}}{\sqrt{N}}+2\left(  r-\gamma\right)  ^{\prime}\frac{W_{N2}^{\ast}%
}{\sqrt{N}}\\
&  -2\left(  r-\gamma\right)  ^{\prime}\Gamma\mathbb{V}_{\beta}^{-1}%
\frac{W_{N1}^{\ast}}{\sqrt{N}}+o_{p}\left(  N^{-1}\right),
\end{align*}
from which we can show that $\hat{\gamma}^{\ast}-\gamma=O_{p}\left(  N^{-1/2}\right)  .$
With this result and by the CMT, we get%
\[
\sqrt{N}\left(  \hat{\gamma}^{\ast}-\gamma\right)  =-2\mathbb{V}_{\gamma}%
^{-1}W_{N2}+2\mathbb{V}_{\gamma}^{-1}\Gamma\mathbb{V}_{\beta}^{-1}%
W_{N1}-2\mathbb{V}_{\gamma}^{-1}W_{N2}^{\ast}+2\mathbb{V}_{\gamma}^{-1}%
\Gamma\mathbb{V}_{\beta}^{-1}W_{N1}^{\ast}+o_{p}\left(  1\right),
\]
which implies%
\[
\sqrt{N}\left(  \hat{\gamma}^{\ast}-\hat{\gamma}\right)  =-2\mathbb{V}%
_{\gamma}^{-1}W_{N2}^{\ast}+2\mathbb{V}_{\gamma}^{-1}\Gamma\mathbb{V}_{\beta
}^{-1}W_{N1}^{\ast}+o_{p}\left(  N^{-1}\right)  
\]
by substituting equation (\ref{EQ:gammahat}) in.

Since $W_{N1}^{\ast}$ and $W_{N2}^{\ast}$ are independent copies of $W_{N1}$
and $W_{N2},$ respectively, $\sqrt{N}\left(  \hat{\gamma}^{\ast}-\hat{\gamma
}\right)  $ approximates the distribution of $\sqrt{N}\left(  \hat{\gamma
}-\gamma\right)  .$

\section{Additional Discussions\label{appendixAdd}}
We discuss the convergence rates of our cross-sectional MRC and panel data MS estimators in Section \ref{appendixAdd_1}. We explain why we choose not to construct the panel data MS estimator for $\gamma$ via index matching in Section \ref{appendixAdd_2}.

\subsection{On the Convergence Rates of MRC and MS Estimators\label{appendixAdd_1}}
The rate of convergence of the MRC estimators for the
cross-sectional model is $N^{-1/2}$. The intuition is that each observation
can be matched with $Nh_{N}^{k_{1}+k_{2}}$ and $N\sigma_{N}^{2k_{1}}$
observations (in probability) for estimating $\beta$ and $\gamma$,
respectively. Thus, all observations are useful and the curse of
dimensionality is not a concern for the cross-sectional estimators. In
contrast, an observation is useful for the panel data MS estimators if all its
covariates are very close to each other across two time periods. For a fixed
$T,$ the probability of an observation being useful is proportional to
$h_{N}^{k_{1}+k_{2}}$ or $\sigma_{N}^{2k_{1}}.$ Therefore, the number of
\textquotedblleft useful\textquotedblright\ observations is proportional to
$Nh_{N}^{k_{1}+k_{2}}$ and $N\sigma_{N}^{2k_{1}},$ and the estimation suffers
from the curse of dimensionality. We may remedy this problem if $T\rightarrow
\infty$ such that $Th_{N}^{k_{1}+k_{2}}\rightarrow\infty$ and $T\sigma
_{N}^{2k_{1}}\rightarrow\infty$ and each observation can be ``matched'' at
certain two time periods with large probability. Since this paper focuses on
models for short panel data, we leave the discussion on this issue to future
research. Further, the criterion functions have a \textquotedblleft
sharp\textquotedblright\ edge as in \cite{KimPollard1990}. As a result, the
convergence rates for $\hat{\beta}$ and $\hat{\gamma}$ are expected to be
$(Nh_{N}^{k_{1}+k_{2}})^{-1/3}$ and $(N\sigma_{N}^{2k_{1}})^{-1/3}$,
respectively. The cross-sectional estimators, however, do not suffer this
\textquotedblleft sharp\textquotedblright\ edge effect because the criterion
functions are U-statistics and the \textquotedblleft sharp\textquotedblright%
\ edge effect vanishes after we decompose the criterion functions (see, e.g.,
the second line of equation (\ref{eq:A4})) when deriving the asymptotics.

\subsection{On the Construction of Panel Data MS Estimator \label{appendixAdd_2}}
Unlike the cross-section case, our panel data MS procedure does not
estimate $\gamma$ through matching $X_{jt}^{\prime}\hat{\beta}$ and $X_{js}%
^{\prime}\hat{\beta}$, $j=1,2$, because of the following observation. As
demonstrated in Theorem \ref{T:paneldist}, the convergence rate of
$\hat{\beta}$ is $(Nh_{N}^{k_{1}+k_{2}})^{-1/3}$. If we knew the true value of
$\beta$, $\hat{\gamma}$ could be obtained by matching $X_{jt}^{\prime}\beta$
and $X_{js}^{\prime}\beta$. Using the same arguments, we can show that this
infeasible estimator has a rate of convergence $(N\sigma_{N}^{2})^{-1/3}$. As
a result, it is only possibly beneficial to match $X_{jt}\hat{\beta}$ and
$X_{js}\hat{\beta}$ for estimating $\gamma$ when $\sigma_{N}^{2}
=o(h_{N}^{k_{1}+k_{2}})$; otherwise the asymptotics of $\hat{\beta}$ would
dominate the limiting distribution of $\hat{\gamma}$. However, this involves
careful selection of tuning parameters $h_{N}$ and $\sigma_{N}$. In more
general cases, this selection relies not only on the dimension of the
regressor space but also on the cardinality of the choice set, and thus should
be made on a case-by-case basis. Because of the lack of a universally
applicable treatment, we choose to simply match the covariates to avoid this
complexity and the ensuing discussion.

\subsection{Identification of the Case with Three Alternatives \label{appendixAdd_3}}
This section discusses how to create identification inequalities for models with three stand-alone alternatives and four potential bundles. More complex cases can be handled similarly. Given these inequalities, one can prove the identification and construct estimation procedures for model coefficients using the same methods employed in Sections \ref{SEC2} and \ref{SEC3}. The technical conditions required for identification are similar to those of the two-alternative model studied in Sections \ref{SEC2} and \ref{SEC3}. We omit them for conciseness and only focus on the identification strategies.

Throughout this section, we consider choice set $\mathcal{J}=\{0,1,2,3,(1,2),(1,3),(2,3),\left(
1,2,3\right) \}$ or equivalently $\mathcal{D}=\{d|d=(d_{1},d_{2},d_{3})\in\{0,1\}^{3}\}$, which is the set of all possible $d=(d_{1},d_{2},d_{3})\in\{0,1\}^{3}$.

\subsubsection{Cross-Sectional MRC Method}\label{SEC:idenJ3}
In the cross-sectional model, assume that an agent chooses $d$ that maximizes the latent utility
\begin{align}
U_{d}=& \sum_{j=1}^{3}F_{j}(X_{j}^{\prime }\beta ,\epsilon _{j})\cdot
d_{j}+ \eta _{110}\cdot F_{110}\left(W_{1}^{\prime }\gamma _{1}\right)
\cdot d_{1}\cdot d_{2}+ \eta _{101}\cdot F_{101}\left(W_{2}^{\prime }\gamma
_{2}\right) \cdot d_{1}\cdot d_{3} \nonumber \\
& +\eta _{011}\cdot F_{011}\left( W_{3}^{\prime }\gamma _{3}\right)\cdot
d_{2}\cdot d_{3} + \eta _{111}\cdot F_{111}\left(W_{4}^{\prime
}\gamma _{4}\right) \cdot d_{1}\cdot d_{2}\cdot d_{3},\label{eq:cross_3_1}
\end{align}
where $X_{1},X_{2},X_{3}\in \mathbb{R}^{k_{1}},\ W_{1}\in \mathbb{R}%
^{k_{2}},W_{2}\in \mathbb{R}^{k_{3}},W_{3}\in \mathbb{R}^{k_{4}},W_{4}\in
\mathbb{R}^{k_{5}}$, $\eta \equiv (\eta _{110},\eta _{101},\eta _{011},\eta
_{111})\in \mathbb{R}_{+}^{4}$ are the bundle-specific unobserved
heterogeneity, and all $F_j(\cdot)$'s and $F_d(\cdot)$'s are strictly increasing in all arguments. Note that by properly re-organizing the covariate vector, expression (\ref{eq:cross_3_1}) can accommodate both choice-specific
and common regressors. See \cite{CameronTrivedi2005} p.498 for a more detailed discussion. Thus, for notational convenience, here we do not distinguish common regressors form those choice-specific regressors as we did in Section \ref{SEC:LAD}. Let $Z\equiv (X_{1} ,X_{2} ,X_{3} ,W_{1} ,W_{2} ,W_{3} ,W_{4} ) $ to collect all observed covariates. We assume $(\epsilon _{1},\epsilon _{2},\epsilon _{3},\eta )\perp Z$.

The observed dependent variable $Y_{d}$ takes the form
\begin{equation*}
Y_{d}=1[U_{d}>U_{d^{\prime}},\forall d^{\prime}\in\mathcal{D}\setminus d].
\end{equation*}
Then, similar to the two-alternatives\footnote{This refers to models with $J=2$ or equivalently $\mathcal{J}=\{0,1,2,(1,2)\}$ that we discussed in Sections \ref{SEC2} and \ref{SEC3}.} case, by $(\epsilon_{1},\epsilon_{2}
\epsilon_{3},\eta)\perp Z,$%
\begin{equation*}
P(Y_{(1,d_{2},d_{3})}=1|Z,X_{2}=x_{2},X_{3}=x_{3},W_{1}=w_{1},W_{2}=w_{2},W_{3}=w_{3},W_{4}=w_{4})
\end{equation*}
is increasing in $X_{1}^{\prime}\beta$ for any constant vectors $\left(
x_{2},x_{3},w_{1},w_{2},w_{3},w_{4}\right)$ and any $\left( d_{2},d_{3}\right)
$. Let $\left( Y_{i},Z_{i}\right) $ and $\left( Y_{m},Z_{m}\right) $
be two independent copies of $\left( Y,Z\right) $ with
\begin{equation*}
Y\equiv(Y_{(0,0,0)},Y_{(1,0,0)},Y_{(0,1,0)},Y_{(0,0,1)},Y_{(1,1,0)},Y_{(1,0,1)},Y_{(0,1,1)},Y_{(1,1,1)}).
\end{equation*}
Then we have the following (conditional) moment inequalities for all $d_{2}$
and $d_{3}$:
\begin{align}
& \left\{ X_{i1}^{\prime}\beta\geq X_{m1}^{\prime}\beta\right\}
\label{EQ:beta} \\
 \Leftrightarrow & \left\{ P(Y_{i\left( 1,d_{2},d_{3}\right)
}=1|Z_{i},X_{i2}=x_{2},X_{i3}=x_{3},W_{i1}=w_{1},W_{i2}=w_{2},W_{i3}=w_{3},W_{i4}=w_{4})\right.
\notag \\
& \geq\left. P(Y_{m\left( 1,d_{2},d_{3}\right)
}=1|Z_{m},X_{m2}=x_{2},X_{m3}=x_{3},W_{m1}=w_{1},W_{m2}=w_{2},W_{m3}=w_{3},W_{m4}=w_{4})\right\} .
\notag
\end{align}
Similar moment inequalities can be obtained for $X_{ij}^{\prime}\beta\geq
X_{mj}^{\prime}\beta$ by fixing $\left\{ X_{1},X_{2},X_{3}\right\}
\backslash\left\{ X_{j}\right\} $ and $W_{1},W_{2},W_{3},W_{4}$.
Collectively, these moment inequalities establish the identification of $%
\beta,$ with some standard technical conditions as for the two-alternatives
case.

Once $\beta $ is identified, we can proceed to identify $\gamma $'s by fixing $X_{j}^{\prime }\beta ,$ $j=1,2,3,$
, or simply by fixing $X_{1},X_{2},X_{3}$. We
follow the former strategy. For $\gamma _{1},$ we fix $(X_{1}^{\prime
}\beta ,X_{2}^{\prime }\beta ,X_{3}^{\prime }\beta )$ at some constant
vector $(v_{1},v_{2},v_{3})$ and fix $(W_{i2},W_{i3},W_{i4})$ at some constant vector $(w_2,w_3,w_4)$. Then
by $(\epsilon _{1},\epsilon _{2},\epsilon _{3},\eta )\perp Z,$%
\begin{equation*}
P(Y_{(1,1,d_{3})}=1|Z,X_{1}^{\prime }\beta =v_{1},X_{2}^{\prime }\beta
=v_{2},X_{3}^{\prime }\beta =v_{3},W_{2}=w_{2},W_{3}=w_{3},W_{4}=w_{4})
\end{equation*}%
is increasing in $W_{1}^{\prime }\gamma_1 $ for any $d_{3}.$ As a result, for
any $d_{3}$, $(v_{1},v_{2},v_{3})$, and $(w_{2},w_{3},w_{4})$, we have
\begin{align}
& \left\{ W_{i1}^{\prime }\gamma _{1}\geq W_{m1}^{\prime }\gamma
_{1}\right\}   \label{EQ:gamma1} \\
 \Leftrightarrow &\left\{ P(Y_{i(1,1,d_{3})}=1|Z_{i},X_{i1}^{\prime }\beta
=v_{1},X_{i2}^{\prime }\beta =v_{2},X_{i3}^{\prime }\beta
=v_{3},W_{i2}=w_{2},W_{i3}=w_{3},W_{i4}=w_{4})\right.   \notag \\
& \geq \left. P(Y_{m(1,1,d_{3})}=1|Z_{m},X_{m1}^{\prime }\beta
=v_{1},X_{m2}^{\prime }\beta =v_{2},X_{m3}^{\prime }\beta
=v_{3},W_{m2}=w_{2},W_{m3}=w_{3},W_{m4}=w_{4})\right\} .  \notag
\end{align}%
We can similarly identify $\gamma _{2}$ and $\gamma _{3}.$

Once $\beta$, $%
\gamma_{1},\gamma_{2},$ and $\gamma_{3}$ are identified, we can identify $\gamma_{4}$ by fixing either $(
X_{1},X_{2},X_{3},W_{1},W_{2},W_{3})$ or $\left(
X_{1}^{\prime}\beta,X_{2}^{\prime}\beta
,X_{3}^{\prime}\beta,W_{1}^{\prime}\gamma_{1},W_{2}^{\prime}%
\gamma_{2},W_{3}^{\prime}\gamma_{3}\right)$. We consider
the latter and fix $(X_{1}^{\prime}\beta,X_{2}^{\prime}\beta,X_{3}^{%
\prime}\beta
,W_{1}^{\prime}\gamma_{1},W_{2}^{\prime}\gamma_{2},W_{3}^{\prime}\gamma
_{3})$ at some constant vector $(v_{1},v_{2},v_{3},v_{4},v_{5},v_{6})$. Then $%
(\epsilon_{1},\epsilon_{2},\epsilon_{3},\eta)\perp Z$ implies that%
\begin{equation*}
P\left( Y_{(1,1,1)}=1|Z,X_{1}^{\prime}\beta=v_{1},X_{2}^{\prime}\beta
=v_{2},X_{3}^{\prime}\beta=v_{3},W_{1}^{\prime}\gamma_{1}=v_{4},W_{2}^{%
\prime }\gamma_{2}=v_{5},W_{3}^{\prime}\gamma_{3}=v_{6}\right)
\end{equation*}
is increasing in $W_{4}^{\prime}\gamma_{4}$ for any fixed $%
(v_{1},v_{2},v_{3},v_{4},v_{5},v_{6})$. Thus, the following
inequality holds for any $(v_{1},v_{2},v_{3},v_{4},v_{5},v_{6})^{\prime},$%
\begin{align}
& \left\{ W_{i4}^{\prime}\gamma_{4}\geq W_{m4}^{\prime}\gamma_{4}\right\}
\label{EQ:gamma4} \\
 \Leftrightarrow &\left\{ P(Y_{i(1,1,1)}=1|Z_{i},X_{i1}^{\prime}\beta
=v_{1},X_{i2}^{\prime}\beta=v_{2},X_{i3}^{\prime}\beta=v_{3},W_{i1}^{\prime
}\gamma_{1}=v_{4},W_{i2}^{\prime}\gamma_{2}=v_{5},W_{i3}^{\prime}\gamma
_{3}=v_{6})\right.  \notag \\
& \geq\left.
P(Y_{m(1,1,1)}=1|Z_{m},X_{m1}^{\prime}\beta=v_{1},X_{m2}^{\prime}%
\beta=v_{2},X_{m3}^{\prime}\beta=v_{3},W_{m1}^{\prime}%
\gamma_{1}=v_{4},W_{m2}^{\prime}\gamma_{2}=v_{5},W_{m3}^{\prime}\gamma
_{3}=v_{6})\right\} .  \notag
\end{align}

The inequalities established in equations (\ref{EQ:beta})--(\ref{EQ:gamma4}) can be used to construct localized MRC estimators as in Section \ref{sec:rcm}. However, this matching-based identification strategy cannot identify coefficients on common regressors, i.e., covariates that appear in more than one index.

\subsubsection{LAD Method}
The LAD estimation for $\theta:=(\beta,\gamma_{1},\gamma_{2},\gamma_{3},\gamma_{4})$
is analogous to the simple case presented in the main text. Here, we only outline the procedure for the cross-sectional model for illustration purposes. The panel data LAD estimation can be constructed similarly and is omitted for brevity.\footnote{For the panel data setting, the identification relies on the conditional
homogeneity condition $\xi_{s}\overset{d}{=}\xi_{t}|(\alpha,Z_{s},Z_{t})$.}
\begin{enumerate}
\item By $(\epsilon_{1},\epsilon_{2},\epsilon_{3},\eta)\perp Z$ and the
monotonicity of $F_{j}(\cdot)$'s and $F_{d}(\cdot)$'s functions,
we can write 
\begin{align}
\{X_{im1}'\beta\geq0,X_{im2}'\beta\leq0,X_{im3}'\beta\leq0, & W_{im1}'\gamma_{1}\leq0,W_{im2}'\gamma_{2}\leq0,W_{im3}'\gamma_{3},\leq0,W_{im4}'\gamma_{4}\leq0\}\label{eq:lad_J3_1}\\
 & \Rightarrow\nonumber \\
\Delta p_{(1,0,0)}(Z_{i},Z_{m}):=P(Y_{i(1,0,0)} & =1|Z_{i})-P(Y_{m(1,0,0)}=1|Z_{i})\geq0\label{eq:lad_J3_2}
\end{align}
and
\begin{align}
\{X_{im1}'\beta\leq0,X_{im2}'\beta\geq0,X_{im3}'\beta\geq0, & W_{im1}'\gamma_{1}\geq0,W_{im2}'\gamma_{2}\geq0,W_{im3}'\gamma_{3},\geq0,W_{im4}'\gamma_{4}\geq0\}\label{eq:lad_J3_3}\\
 & \Rightarrow\nonumber \\
\Delta p_{(1,0,0)}(Z_{i},Z_{m}) & \leq0,\label{eq:lad_J3_4}
\end{align}
i.e., the event (\ref{eq:lad_J3_1}) and (\ref{eq:lad_J3_3}) are
respectively sufficient (but not necessary) conditions for inequality
(\ref{eq:lad_J3_2}) and (\ref{eq:lad_J3_4}) to hold.
\item Define
\[
Q_{(1,0,0)}(\vartheta):=\mathbb{E}[q_{(1,0,0)}(Z_{i},Z_{m},\vartheta)],
\]
where
\begin{align*}
 & q_{(1,0,0)}(Z_{i},Z_{m},\vartheta)\\
= & [\vert I_{im(1,0,0)}^{+}(\vartheta)-1[\Delta p_{(1,0,0)}(Z_{i},Z_{m})\geq0]\vert+\vert I_{im(1,0,0)}^{-}(\vartheta)-1[\Delta p_{(1,0,0)}(Z_{i},Z_{m})\leq0]\vert]\\
 & \times[I_{im(1,0,0)}^{+}(\vartheta)+I_{im(1,0,0)}^{-}(\vartheta)]
\end{align*}
with
\begin{align*}
 & I_{im(1,0,0)}^{+}(\vartheta)\\
= & 1[X_{im1}'\beta\geq0,X_{im2}'\beta\leq0,X_{im3}'\beta\leq0,W_{im1}'\gamma_{1}\leq0,W_{im2}'\gamma_{2}\leq0,W_{im3}'\gamma_{3},\leq0,W_{im4}'\gamma_{4}\leq0]\\
 & I_{im(1,0,0)}^{-}(\vartheta)\\
= & 1[X_{im1}'\beta\leq0,X_{im2}'\beta\geq0,X_{im3}'\beta\geq0,W_{im1}'\gamma_{1}\geq0,W_{im2}'\gamma_{2}\geq0,W_{im3}'\gamma_{3},\geq0,W_{im4}'\gamma_{4}\geq0].
\end{align*}
The intuition underlying the construction of this population criterion
function is the same as the discussion in Section \ref{SEC:LAD} (below (\ref{eq:lad_obj})--(\ref{eq:lad_ineq_2_2})). Then the identification of $\theta$ can be established using arguments similar to those for Theorem \ref{thm:cross_theta_identification}.
\item To compute the estimator $\hat{\theta}$ for $\theta$, one can consider
the ``debiased'' sample criterion analogous to (\ref{eq:lad_estimator}), i.e.,
\begin{equation}
\hat{\theta}=\arg\min_{\vartheta\in\Theta}\sum_{i=1}^{N-1}\sum_{m>i}\hat
{q}^{D}_{im(1,0,0)}(\vartheta), 
\end{equation}
where
\begin{align*}
\hat{q}_{im(1,0,0)}^{D}(\vartheta)=  &  [\vert I_{im(1,0,0)}^{+}(\vartheta)-\Delta
\hat{p}_{(1,0,0)}(Z_{i},Z_{m})\vert+\vert I_{im(1,0,0)}^{-}(\vartheta)+\Delta\hat
{p}_{(1,0,0)}(Z_{i},Z_{m})\vert]\\
&  \times[I_{im(1,0,0)}^{+}(\vartheta)+I_{im(1,0,0)}^{-}(\vartheta)]+[1-(I_{im(1,0,0)}
^{+}(\vartheta)+I_{im(1,0,0)}^{-}(\vartheta))].
\end{align*}
\item Repeat Steps 1--3 to obtain sample criteria associated with $\Delta p_{(0,1,0)}(Z_{i},Z_{m})$,
$\Delta p_{(0,0,1)}(Z_{i},Z_{m})$, $\Delta p_{(1,1,0)}(Z_{i},Z_{m})$,
$\Delta p_{(1,0,1)}(Z_{i},Z_{m})$, $\Delta p_{(0,1,1)}(Z_{i},Z_{m})$,
$\Delta p_{(0,0,0)}(Z_{i},Z_{m})$, and $\Delta p_{(1,1,1)}(Z_{i},Z_{m})$.
Collectively, use some convex combination of these criteria to obtain
$\hat{\theta}$ with better finite sample performance.
\end{enumerate}

\subsubsection{Panel Data MS Method}\label{SEC:idenJ3P}
We specify the latent utility of the panel data bundle choice model with $J=3$ as follows:
\begin{align*}
U_{dt}=& \sum_{j=1}^{3}F_{j}(X_{jt}^{\prime }\beta ,\epsilon _{jt},\alpha
_{j})\cdot d_{j}+\eta _{110t}\cdot F_{110}\left(  W_{1t}^{\prime }\gamma
_{1}+\alpha _{b1}\right) \cdot d_{1}\cdot d_{2}+\eta
_{101t}\cdot F_{101}\left(W_{2t}^{\prime }\gamma _{2}+\alpha _{b2}\right) \cdot
d_{1}\cdot d_{3} \\
& +\eta _{011t}\cdot F_{011}\left(W_{3t}^{\prime }\gamma _{3}+\alpha
_{b3}\right) \cdot d_{2}\cdot d_{3}+\eta _{111t}\cdot F_{111}\left( 
W_{4t}^{\prime }\gamma _{4}+\alpha _{b4}\right) \cdot d_{1}\cdot
d_{2}\cdot d_{3},
\end{align*}%
where $X_{1t},X_{2t},X_{3t}\in \mathbb{R}^{k_{1}},\ W_{1t}\in \mathbb{R}%
^{k_{2}},W_{2t}\in \mathbb{R}^{k_{3}},W_{3t}\in \mathbb{R}^{k_{4}},$ $W_{4t}\in
\mathbb{R}^{k_{5}}$, and $\eta _{t}\equiv (\eta _{110t},\eta _{101t},\eta
_{011t},\eta _{111t})\in \mathbb{R}_{+}^{4}$. All $F_j(\cdot)$'s and $F_d(\cdot)$'s are strictly increasing in all arguments. An agent chooses $d$ that maximizes the latent utility, i.e.,
\begin{equation*}
Y_{dt}=1[U_{dt}>U_{d^{\prime }t},\forall d^{\prime }\in \mathcal{D}\setminus
d].
\end{equation*}%
Denote $Z_{t} =(X_{1t},X_{2t},X_{3t},W_{1t},W_{2t},W_{3t},W_{4t})$, $\xi
_{t}=(\epsilon _{1t},\epsilon _{2t},\epsilon _{3t},\eta _{110t},\eta
_{101t},\eta _{011t},\eta _{111t})$, and $\alpha  =\left( \alpha _{1},\alpha _{2},\alpha _{3},\alpha
_{b1},\alpha _{b2},\alpha _{b3},\alpha _{b4}\right)$. The key identification assumption is that $\xi _{s}\overset{d}{=}\xi _{t}|(\alpha
,Z_{s},Z_{t})$ holds for any two time periods $s$ and $t$.

The identification of the panel data model is similar to the cross-sectional model discussed in Section \ref{SEC:idenJ3}. The primary difference is that in a panel data setting, we need to match and compare the same agent in time periods $s$ and $t$, rather than two independent agents $i$ and $m$. As a result, we can construct the same set of identification inequalities as Section \ref{SEC:idenJ3}, but replace the $(i,m)$ therein with $(s,t)$. For example, analogous to (\ref{EQ:beta}) and (\ref{EQ:gamma1}), we can write for any $%
d_{2},d_{3},x_{1},x_{2},x_{3},w_{1},w_{2},w_{3},w_{4}$, and $c$
\begin{align}
& \left\{ X_{1t}^{\prime }\beta \geq X_{1s}^{\prime }\beta \right\}
\label{EQ:beta1P} \\
 \Leftrightarrow &\{P(Y_{\left( 1,d_{2},d_{3}\right)t
}=1|Z_t,X_{2t}=x_{2},X_{3t}=x_{3},W_{1t}=w_{1},W_{2}=w_{2t},W_{3t}=w_{3},W_{4t}=w_{4},\alpha =c)
\notag \\
& \geq P(Y_{\left( 1,d_{2},d_{3}\right)s }=1|Z_s,X_{2s}=x_{2},X_{3s}=x_{3},W_{1s}=w_{1},W_{2s}=w_{2},W_{3s}=w_{3},W_{4s}=w_{4},\alpha =c)\}
\notag
\end{align}
and
\begin{align}
& \left\{W_{1t}^{\prime }\gamma _{1}\geq W_{1s}^{\prime }\gamma _{1}\right\}
\label{EQ:gamma1P} \\
\Leftrightarrow & P(Y_{\left( 1,1,d_{3}\right)
t}=1|Z_t,X_{1t}=x_{1},X_{2t}=x_{2},X_{3t}=x_{3},W_{2t}=w_{2},W_{3t}=w_{3},W_{4t}=w_{4},\alpha =c)
\notag \\
& \geq P(Y_{(1,1,d_{3})s}=1|Z_s,X_{1s}=x_{1},X_{2s}=x_{2},X_{3s}=x_{3},W_{2s}=w_{2},W_{3s}=w_{3},W_{4s}=w_{4},\alpha =c).  \notag
\end{align}
Then, under similar regular conditions, we can identify and estimate $(\beta,\gamma_1)$ based on inequalities (\ref{EQ:beta1P}) and (\ref{EQ:gamma1P}). Identification inequalities for $(\gamma_2,\gamma_3,\gamma_4)$ can be similarly constructed and hence are omitted for brevity.  

The panel data MS estimator differs from the cross-sectional MRC estimator in that its convergence rate decreases as the number of matching variables increases. For example, using similar analysis, we can show that for MS estimators based on inequalities (\ref{EQ:beta1P}) and (\ref{EQ:gamma1P}), we have
\begin{equation*}
\hat{\beta}-\beta =O_{p}\left(
(Nh_{N}^{2k_{1}+k_{2}+k_{3}+k_{4}+k_{5}})^{-1/3}\right) \text{ and }\hat{%
\gamma}_{1}-\gamma _{1}=O_{p}\left(
(N\sigma_{N}^{3k_{1}+k_{3}+k_{4}+k_{5}})^{-1/3}\right),
\end{equation*}%
where we assume we use the same type of estimator as in Section \ref{SEC3.1}
and the same bandwidth ($h_{N}$ or $\sigma_{N}$) for all covariates. The convergence rates
of $\hat{\gamma}_{2},\hat{\gamma}_{3}$, and $\hat{\gamma}_{4}$ can be
similarly obtained. Generally speaking, we can show that the MS estimator has a convergence rate of $O_p((Nh_N^{\bar{k}})^{-1/3})$ ($O_p((N\sigma_N^{\bar{k}})^{-1/3})$), where $\bar{k}$ represents the number of matching variables.

\section{Proofs of Technical Lemmas} \label{appendixE} 
We prove Lemmas \ref{lemmaA1}--\ref{LE:delta} and Lemmas \ref{LE:P1}--\ref{LE:P2} in order.
For Lemmas \ref{lemmaA1}--\ref{lemmaA7}, we use the following results, which can be derived under Assumptions C2 and C5 by applying standard arguments (e.g., Example 6.4 of \cite{KimPollard1990}).  
We use $\Vert \cdot \Vert _{F}$ to
denote the Frobenius norm; that is, for any matrix $\mathbf{A}$, $\Vert
\mathbf{A}\Vert _{F}=\sqrt{\text{trace}\left( \mathbf{A}\mathbf{A}^{\prime
}\right) }$.

Let $\mathcal{N}_{\beta }$ denote a neighborhood of $\beta$. Then, for $\varrho _{i}(b)$ defined in (\ref{EQ:rhoi}),
\begin{enumerate}
\item[-] All mixed second partial derivatives of $\varrho _{i}(b)$ exist on $%
\mathcal{N}_{\beta }$.
\item[-] There is an integrable function $M_{1}(Z_{i})$ such that for all $%
b\in \mathcal{N}_{\beta }$, $\Vert \nabla ^{2}\varrho _{i}(b)-\nabla
^{2}\varrho _{i}(\beta )\Vert _{F}\leq M_{1}(Z_{i})\Vert b-\beta \Vert $.
\item[-] $\mathbb{E}[\Vert \nabla \varrho _{i}(\beta )\Vert ^{2}]<\infty $, $%
\mathbb{E}[\Vert \nabla ^{2}\varrho _{i}(\beta )\Vert _{F}]<\infty $, and $%
\mathbb{E}[\nabla ^{2}\varrho _{i}(\beta )]$ is negative definite.
\end{enumerate}

Let $\mathcal{N}_{\gamma }$ denote a neighborhood of $%
\gamma$. Then, for $\tau _{i}(r)$ and $\mu (v_{1},v_{2},r)$ defined respectively in (\ref{EQ:Tau}) and (\ref{EQ:Mu}), 
\begin{enumerate}
\item[-] All mixed second partial derivatives of $\tau _{i}(r)$ exist on $%
\mathcal{N}_{\gamma }$.
\item[-] There is an integrable function $M_{2}(Z_{i})$ such that for all $%
r\in \mathcal{N}_{\gamma }$ $\Vert \nabla ^{2}\tau _{i}(r)-\nabla ^{2}\tau
_{i}(\gamma )\Vert _{F}\leq M_{2}(Z_{i})\Vert r-\gamma \Vert $.
\item[-] $\mathbb{E}[\Vert \nabla \tau _{i}(\gamma )\Vert ^{2}]<\infty $, $%
\mathbb{E}[\Vert \nabla ^{2}\tau _{i}(\gamma )\Vert _{F}]<\infty $, and $%
\mathbb{E}[\nabla ^{2}\tau _{i}(\gamma )]$ is negative definite.
\item[-] All mixed second partial derivatives of $\mu (v_{1},v_{2},r)$ exist
on $\mathcal{N}_{\gamma }$.
\item[-] There exists an continuous, integrable function $M_{3}(v_{1},v_{2})$
such that for all $r\in \mathcal{N}_{\gamma }$, $\Vert \nabla _{11}^{2}\mu
(v_{1},v_{2},r)-\nabla _{11}^{2}\mu (v_{1},v_{2},\gamma )\Vert _{F}\leq
M_{3}(v_{1},v_{2})\Vert r-\gamma \Vert $.
\item[-] $\mathbb{E}[\Vert \nabla _{13}^{2}\mu (V_{i}\left( \beta \right)
,V_{i}\left( \beta \right) ,\gamma )\Vert _{F}^{2}]<\infty $ and $\sup_{r\in
\mathcal{N}_{\gamma }}\mathbb{E}[\Vert \nabla _{33}^{2}\nabla _{1}\mu
(V_{i}\left( \beta \right) ,V_{i}\left( \beta \right) ,\gamma )\Vert
_{F}]<\infty $.
\end{enumerate}

Further, using the definition of $\tau _{i}(r)$, we
can write $\tau _{i}(r)$ as
\begin{equation*}
\tau _{i}(r)=f_{V_{im}}(0)\mathbb{E}%
[h_{im}(r)|V_{i}=v_{i},W_{i}=w_{i},V_{im}=0]=\int
B(v_{i},v_{i},w_{i},w_{m})S_{im}(r)f_{V,W}(v_{i},w_{m})dw_{m}.
\end{equation*}

\begin{proof}
[Proof of Lemma \ref{lemmaA1}]First, note that
\begin{equation}
\frac{1}{2}\sigma_{N}^{2}\sup_{r\in\mathcal{R}}|\hat{\mathcal{L}}_{N}%
^{K}(r)-\mathcal{L}_{N}^{K}(r)|\leq\frac{1}{N(N-1)}\sum_{i\neq m}%
|K_{\sigma_{N},\gamma}(\hat{V}_{im})-K_{\sigma_{N},\gamma}(V_{im})|.
\label{eq:A1}%
\end{equation}
Applying a second-order Taylor expansion to the right-hand side of
(\ref{eq:A1}) yields the lead term of the form
\begin{align*}
&  \frac{1}{\sigma_{N}N(N-1)}\sum_{i\neq m}|\nabla_{1}K_{\sigma_{N},\gamma
}(V_{im})X_{im1}^{\prime}(\hat{\beta}-\beta)+\nabla_{2}K_{\sigma_{N},\gamma
}(V_{im})X_{im2}^{\prime}(\hat{\beta}-\beta)|\\
\leq &  \frac{1}{\sigma_{N}N(N-1)}\sum_{i\neq m}\sum_{j=1}^{2}|\nabla
_{j}K_{\gamma}(V_{im}/\sigma_{N})X_{imj}^{\prime}(\hat{\beta}-\beta)|,
\end{align*}
where $\nabla_{j}K_{\gamma}(\cdot)$ ($\nabla_{j}K_{\sigma_{N},\gamma}(\cdot)$)
denotes the partial derivative of $K_{\gamma}(\cdot)$ ($K_{\sigma_{N},\gamma
}(\cdot)$) w.r.t. its $j$-th argument corresponding to alternative $j$. It
suffices to focus on the term with $j=1$. Then by Assumptions C1--C7,
and the Cauchy-Schwarz inequality, we conclude that
\begin{align*}
&  \frac{1}{N(N-1)}\sum_{i\neq m}\left\vert \nabla_{1}K_{\gamma}(V_{im}%
/\sigma_{N})\frac{X_{im1}^{\prime}(\hat{\beta}-\beta)}{\sigma_{N}}\right\vert
\leq\frac{1}{N(N-1)}\sum_{i\neq m}|\nabla_{1}K_{\gamma}(V_{im}/\sigma
_{N})|\Vert X_{im1}\Vert\frac{\Vert\hat{\beta}-\beta\Vert}{\sigma_{N}}\\
=  &  [O_{p}(\sigma_{N}^{2})+o_{p}(1)]O_{p}(\Vert\hat{\beta}-\beta\Vert
/\sigma_{N})=o_{p}(1),
\end{align*}
where the penultimate equality follows from the SLLN of U-statistics (see
e.g., \cite{Serfling2009}) and $\mathbb{E}[|\nabla_{1}K_{\gamma}(V_{im}%
/\sigma_{N})|\Vert X_{im1}\Vert]=O_{p}(\sigma_{N}^{2})$, and the last equality
is due to the $\sqrt{N}$-consistency of $\hat{\beta}$. Then the desired
convergence result follows.
\end{proof}

\begin{proof}
[Proof of Lemma \ref{lemmaA2}]Note that Lemma \ref{lemmaA2} follows from
$\sup_{r\in\mathcal{R}}|\mathcal{L}_{N}^{K}(r)-\mathbb{E}[\mathcal{L}_{N}%
^{K}(r)]|=o_{p}(1)$ and $\sup_{r\in\mathcal{R}}|\mathbb{E}[\mathcal{L}_{N}%
^{K}(r)]-\mathcal{L}(r)|=o(1)$.

Define $\mathcal{F}_{N}=\{K_{\sigma_{N},\gamma}(v_{im})h_{im}(r)|r\in
\mathcal{R}\}$ with $\sigma_{N}>0$ and $\sigma_{N}\rightarrow0$.
$\mathcal{F}_{N}$ is a sub-class of the fixed class $\mathcal{F}%
\equiv\{K_{\gamma}(v_{im}/\sigma)h_{im}(r)|r\in\mathcal{R},\sigma
>0\}=\mathcal{F}_{r}\mathcal{F}_{\sigma}$, with $\mathcal{F}_{\sigma}%
\equiv\{K_{\gamma}(v_{im}/\sigma)|\sigma>0\}$, which is Euclidean for the
constant envelope $\sup_{v\in\mathbb{R}^{2}}|K_{\gamma}(v)|$ by Lemma 22 in
\cite{NolanPollard1987} and $\mathcal{F}_{r}\equiv\{h_{im}(r)|r\in
\mathcal{R}\}$. Noticing that $h_{im}(r)$ is uniformly bounded by 2, Example
2.11 and Lemma 2.15 in \cite{PakesPollard1989} then implies that
$\mathcal{F}_{r}$ is Euclidean for the constant envelope 2. Putting all these
results together, we conclude using Lemma 2.14 in \cite{PakesPollard1989} that
$\mathcal{F}$ is Euclidean for the constant envelope $2\sup_{v\in
\mathbb{R}^{2}}|K_{\gamma}(v)|$. Applying Corollary 4 in \cite{Sherman1994AoS}%
, we obtain that $\sup_{r\in\mathcal{R}}|\mathcal{L}_{N}^{K}(r)-\mathbb{E}%
[\mathcal{L}_{N}^{K}(r)]|=O_{p}(N^{-1}/\sigma_{N}^{2})=o_{p}(1)$ by Assumption C7.

It remains to show that $\sup_{r\in\mathcal{R}}|\mathbb{E}[\mathcal{L}_{N}%
^{K}(r)]-\mathcal{L}(r)|=o(1)$. Let $\eta(\cdot)\equiv f_{V_{im}}%
(\cdot)\mathbb{E}[h_{im}(r)|V_{im}=\cdot]$. Then $\mathcal{L}(r)=\eta(0)$ by
definition. We can write by Assumptions C6 and C7 that
\begin{align*}
&  \sup_{r\in\mathcal{R}}|\mathbb{E}[\mathcal{L}_{N}^{K}(r)]-\mathcal{L}%
(r)|=\sup_{r\in\mathcal{R}}|\sigma_{N}^{-2}\mathbb{E}[K_{\gamma}(V_{im}%
/\sigma_{N})h_{im}(r)]-\eta(0)|\\
=  &  \sup_{r\in\mathcal{R}}|\int\sigma_{N}^{-2}K_{\gamma}(v/\sigma_{N}%
)\eta(v)\text{d}v-\eta(0)|=\sup_{r\in\mathcal{R}}|\int\sigma_{N}^{-2}%
K_{\gamma}(v/\sigma_{N})[\eta(0)+v^{\prime}\nabla\eta(\bar{v})]\text{d}%
v-\eta(0)|\\
=  &  \sup_{r\in\mathcal{R}}|\int K_{\gamma}(u)[\eta(0)+\sigma_{N}u^{\prime
}\nabla\eta(u_{N})]\text{d}u-\eta(0)|=\sup_{r\in\mathcal{R}}|\sigma_{N}\int
K_{\gamma}(u)u^{\prime}\nabla\eta(u_{N})\text{d}u|\\
\leq &  \sigma_{N}\cdot\sup_{r\in\mathcal{R}}\int|K_{\gamma}(u)|\Vert
u\Vert\Vert\nabla\eta(u_{N})\Vert\text{d}u=O(\sigma_{N})=o(1),
\end{align*}
where the third equality applies a mean-value expansion and the fourth
equality uses a change of variables. Then, the desired result follows.
\end{proof}

\begin{proof}
[Proof of Lemma \ref{lemmaA3}]The proof proceeds by verifying the four
sufficient conditions for Theorem 2.1 in \cite{NeweyMcFadden1994}: (S1)
$\mathcal{R}$ is a compact set, (S2) $\sup_{r\in\mathcal{R}}|\hat{\mathcal{L}%
}_{N}^{K}(r)-\mathcal{L}(r)|=o_{p}(1)$, (S3) $\mathcal{L}(r)$ is continuous in
$r$, and (S4) $\mathcal{L}(r)$ is uniquely maximized at $\gamma$.

Condition (S1) is satisfied by Assumption C4. Lemmas \ref{lemmaA1} and
\ref{lemmaA2} together imply that condition (S2) holds. Note that
\begin{align*}
&  \mathbb{E}[h_{im}(r)|V_{im}=0]= \mathbb{E}\left\{  \mathbb{E}%
[Y_{im(1,1)}(\text{sgn}(W_{im}^{\prime}r)-\text{sgn}(W_{im}^{\prime}%
\gamma))|Z_{i},Z_{m}]|V_{im}=0\right\} \\
=  &  \mathbb{E}\left[  (P(Y_{i(1,1)}=1|Z_{i})-P(Y_{m(1,1)}=1|Z_{m}%
))(\text{sgn}(W_{im}^{\prime}r)-\text{sgn}(W_{im}^{\prime}\gamma
))|V_{im}=0\right]  .
\end{align*}
Then the identification condition (S4) can be verified by similar arguments
used in the proof of Theorem 2.1 given that $f_{V_{im}%
}(0)>0$ by Assumption C5.

It remains to verify the continuity of $\mathcal{L}(r)$ in $r$. $\mathbb{E}[h_{im}(r)|V_{im}=0]$ can be expressed as the
sum of terms like
\begin{align*}
&  P(Y_{im(1,1)}=d,W_{im}^{(1)}>-\tilde{W}_{im}^{\prime}\tilde{r}%
|V_{im}=0)\\
=  &  \int\int_{-\tilde{W}_{im}^{\prime}\tilde{r}/r^{(1)}}^{\infty
}P(Y_{im(1,1)}=d|W_{im}^{(1)}=w,\tilde{W}_{im},V_{im}=0)f_{W_{im}^{(1)}%
|\tilde{W}_{im},V_{im}=0}(w)\text{d}w\text{d}F_{\tilde{W}_{im}|V_{im}=0}.
\end{align*}
for some $d\in\{-1,0,1\}$. Then $\mathcal{L}(r)$ is continuous in $r$ if
$f_{W_{im}^{(1)}|\tilde{W}_{im},V_{im}=0}(\cdot)$ does not have any mass
points, which is guaranteed by Assumption C3. This completes the proof.
\end{proof}

\begin{proof}
[Proof of Lemma \ref{lemmaA4}]First, we express $\mathbb{E}[\mathcal{L}%
_{N}^{K}(r)]$ as the integral
\begin{align}
&  \int\sigma_{N}^{-2}K_{\sigma_{N},\gamma}(v_{im})B(v_{i},v_{m},w_{i}%
,w_{m})S_{im}(r)\text{d}F_{V,W}(v_{i},w_{i})\text{d}F_{V,W}(v_{m}%
,w_{m})\label{eq:A7}\\
=  &  \int K_{\gamma}(u_{im})B(v_{m}+u_{im}\sigma_{N},v_{m},w_{i},w_{m}%
)S_{im}(r)f_{V,W}(v_{m}+u_{im}\sigma_{N},w_{i})\text{d}w_{i}\text{d}%
u_{im}\text{d}F_{V,W}(v_{m},w_{m}),\nonumber
\end{align}
where we apply the change of variables $u_{im}=v_{im}/\sigma_{N}$ to obtain
the equality.

Take the $\kappa_{\gamma}^{\text{th}}$-order Taylor expansion inside the
integral around $v_{m}$ to obtain the lead term
\begin{align}
&  \int K_{\gamma}(u_{im})B(v_{m},v_{m},w_{i},w_{m})S_{im}(r)f_{V,W}%
(v_{m},w_{i})\text{d}w_{i}\text{d}u_{im}\text{d}F_{V,W}(v_{m},w_{m}%
)\nonumber\\
=  &  \int B(v_{m},v_{m},w_{i},w_{m})S_{im}(r)f_{V,W}(v_{m},w_{i}%
)\text{d}w_{i}\text{d}F_{V,W}(v_{m},w_{m})=\mathbb{E}[\tau_{m}(r)],
\label{eq:A8}%
\end{align}
where the first equality follows by $\int K_{\gamma}(u)du=1$.\footnote{More precisely, we apply $\kappa_{1\gamma}^{\text{th}}$-order and $\kappa_{2\gamma}^{\text{th}}$-order Taylor expansions to $B(\cdot,v_{m},w_{i},w_{m}%
)$ and $f_{V,W}(\cdot,w_{i})$, respectively.} All remaining
terms are zero except the last one which is of order $O(\sigma_{N}^{\kappa_\gamma
}\delta_{N})=o(\Vert r-\gamma\Vert^{2})$ by Assumptions C6 and C7.

Note that a second-order Taylor expansion around $\gamma$ gives
\begin{align*}
&  \tau_{m}(r)-\tau_{m}(\gamma)= (r-\gamma)^{\prime}\nabla\tau_{m}%
(\gamma)+\frac{1}{2}(r-\gamma)^{\prime}\nabla^{2}\tau_{m}(\bar{r})(r-\gamma)\\
=  &  (r-\gamma)^{\prime}\nabla\tau_{m}(\gamma)+\frac{1}{2}(r-\gamma)^{\prime
}\nabla^{2}\tau_{m}(\gamma)(r-\gamma)+\frac{1}{2}(r-\gamma)^{\prime}\left[
\nabla^{2}\tau_{m}(\bar{r})-\nabla^{2}\tau_{m}(\gamma)\right]  (r-\gamma),
\end{align*}
and hence by Assumption C7
\begin{align}
\mathbb{E}[\tau_{m}(r)]=  &  (r-\gamma)^{\prime}\mathbb{E}[\nabla\tau
_{m}(\gamma)]+\frac{1}{2}(r-\gamma)^{\prime}\mathbb{E}[\nabla^{2}\tau
_{m}(\gamma)](r-\gamma)+o(\Vert r-\gamma\Vert^{2})\nonumber\\
=  &  \frac{1}{2}(r-\gamma)^{\prime}\mathbb{V}_{\gamma}(r-\gamma)+o(\Vert
r-\gamma\Vert^{2}), \label{eq:A9}%
\end{align}
where the second equality uses the fact that $\mathbb{E}[\nabla\tau_{m}%
(\gamma)]=0$ since $\mathbb{E}[\tau_{m}(r)]$ is maximized at $\gamma$. Then,
applying (\ref{eq:A7}), (\ref{eq:A8}), and (\ref{eq:A9}) proves the lemma.
\end{proof}

\begin{proof}
[Proof of Lemma \ref{lemmaA5}]We can establish a representation for
$\mathbb{E}[\mathcal{L}_{N}^{K}(r)|Z_{m}]$ using the same arguments as in
proving Lemma \ref{lemmaA4}, but no longer integrating over $Z_{m}$.
Specifically, a change of variables $u_{im}=v_{im}/\sigma_{N}$ gives
\begin{align*}
&  \mathbb{E}[\mathcal{L}_{N}^{K}(r)|Z_{m}=z_{m}]=\int\sigma_{N}^{-2}%
K_{\sigma_{N},\gamma}(v_{im})B(v_{i},v_{m},w_{i},w_{m})S_{im}(r)\text{d}%
F_{V,W}(v_{i},w_{i})\\
=  &  \int K_{\gamma}(u_{im})B(v_{m}+u_{im}\sigma_{N},v_{m},w_{i},w_{m}%
)S_{im}(r)f_{V,W}(v_{m}+u_{im}\sigma_{N},w_{i})\text{d}w_{i}\text{d}u_{im}.
\end{align*}
The lead term of the $\kappa_{\gamma}^{\text{th}}$-order Taylor expansion
inside the integral around $v_{m}$ is
\begin{align*}
&  \int K_{\gamma}(u_{im})B(v_{m},v_{m},w_{i},w_{m})S_{im}(r)f_{V,W}%
(v_{m},w_{i})\text{d}w_{i}\text{d}u_{im}\\
=  &  \int B(v_{m},v_{m},w_{i},w_{m})S_{im}(r)f_{V,W}(v_{m},w_{i}%
)\text{d}w_{i}=\tau_{m}(r)
\end{align*}
by $\int K_{\gamma}(u)du=1$, and the sample average of the bias term is of
order $o(\Vert r-\gamma\Vert^{2})$.

Then, applying Lemma \ref{lemmaA4} and Assumption C7, we can write
\begin{align*}
&  \frac{2}{N}\sum_{m}\mathbb{E}[\mathcal{L}_{N}^{K}(r)|Z_{m}]-2\mathbb{E}%
[\mathcal{L}_{N}^{K}(r)]\\
=  &  \frac{2}{N}\sum_{m}\tau_{m}(r)-(r-\gamma)^{\prime}\mathbb{V}_{\gamma
}(r-\gamma)+o_{p}(\Vert r-\gamma\Vert^{2})\\
=  &  \frac{2}{N}\sum_{m}(r-\gamma)^{\prime}\nabla\tau_{m}(\gamma)+\frac{1}%
{N}\sum_{m}(r-\gamma)^{\prime}\left(  \nabla^{2}\tau_{m}(\gamma)-\mathbb{V}%
_{\gamma}\right)  (r-\gamma)+o_{p}(\Vert r-\gamma\Vert^{2}).
\end{align*}
Then the desired result follows by noticing that $N^{-1}\sum_{m}\nabla^{2}%
\tau_{m}(\gamma)-\mathbb{V}_{\gamma}=o_{p}(1)$ by the SLLN.
\end{proof}

\begin{proof}
[Proof of Lemma \ref{lemmaA6}]By construction, we can write
\[
\rho_{N}(r)=\frac{1}{N(N-1)}\sum_{i\neq m}\rho_{im}(r),
\]
where $\rho_{im}(r)\equiv(K_{\sigma_{N},\gamma}(V_{im})h_{im}(r)-2N^{-1}%
\sum_{i}\mathbb{E}[K_{\sigma_{N},\gamma}(V_{im})h_{im}(r)|Z_{i}]+\mathbb{E}%
[K_{\sigma_{N},\gamma}(V_{im})h_{im}(r)])/\sigma_{N}^{2}$.

Note that $\rho_{im}(\gamma)=0$ and $|\rho_{im}(r)|$ is bounded by a multiple
of $M/\sigma_{N}^{2}$ where $M$ is a positive constant. Define $\mathcal{F}%
_{N}^{\ast}=\{\rho_{im}^{\ast}(r)|r\in\mathcal{R}\}$ where $\rho_{im}^{\ast
}(r)\equiv\sigma_{N}^{2}\rho_{im}(r)/M$. Then the Euclidean properties of the
($P$-degenerate) class of functions $\mathcal{F}_{N}^{\ast}$ are deduced using
similar arguments for proving Lemma \ref{lemmaA2} in combination with
Corollary 17 and Corollary 21 in \cite{NolanPollard1987}. As $\int
K_{\sigma_{N},\gamma}(v)^{p}$d$v=O(\sigma_{N}^{2})$ for $p=1,2$, we have
$\sup_{r\in\mathcal{R}_{N}}\mathbb{E}[\rho_{im}^{\ast}(r)^{2}]=O(\sigma
_{N}^{2})$. Then applying Theorem 3 of \cite{Sherman1994ET} gives
\[
\frac{1}{N(N-1)}\sum_{i\neq m}\rho_{im}^{\ast}(r)=O_{p}(N^{-1}\sigma
_{N}^{\alpha}),
\]
where $0<\alpha<1$ and hence $q_{N}(r)=O_{p}(N^{-1}\sigma_{N}^{\alpha
-2})=O_{p}(N^{-1}\sigma_{N}^{-2})$.
\end{proof}

\begin{proof}
[Proof of Lemma \ref{lemmaA7}]The first step is to plug (\ref{eq:A10}) into
(\ref{eq:A5}) so that $\Delta\mathcal{L}_{N}^{K}(r)$ can be expressed as
\begin{equation}
\frac{1}{\sigma_{N}^{3}N(N-1)(N-2)}\sum_{i\neq m\neq l}[\nabla_{1}%
K_{\sigma_{N},\gamma}(V_{im})X_{im1}^{\prime}\psi_{l,\beta}+\nabla
_{2}K_{\sigma_{N},\gamma}(V_{im})X_{im2}^{\prime}\psi_{l,\beta}]h_{im}(r)
\label{eq:A11}%
\end{equation}
plus a remainder term of higher order since $\sqrt{N}\sigma_{N}\rightarrow
\infty$. (\ref{eq:A11}) is a third order U-process. We use the U-statistic
decomposition found in \cite{Sherman1994ET} (see also \cite{Serfling2009}) to
derive a linear representation. Note that (\ref{eq:A11}) has unconditional
mean zero, so as is its mean conditional on each of its first two arguments $(Z_i,Z_m)$.
Furthermore, by Theorem 3 of \cite{Sherman1994ET} and similar arguments for
proving Lemma \ref{lemmaA6}, the remainder term of the decomposition
(projection) is of order $O_{p}(N^{-1}\sigma_{N}^{-2})$. Hence it suffices to
derive a linear representation for its mean conditional on its third
argument:
\[
\frac{1}{\sigma_{N}^{3}N}\sum_{l}\int B(v_{i},v_{m},w_{i},w_{m})S_{im}%
(r)\nabla K_{\sigma_{N},\gamma}(v_{im})^{\prime}\left(
\begin{array}
[c]{c}%
x_{im1}^{\prime}\psi_{l,\beta}\\
x_{im2}^{\prime}\psi_{l,\beta}%
\end{array}
\right)  \text{d}F_{X,W}(x_{i},w_{i})\text{d}F_{X,W}(x_{m},w_{m}).
\]
where $F_{X,W}(\cdot,\cdot)$ denotes the joint distribution function of
$(X_{i},W_{i})$. While the above integral is expressed w.r.t. $(x_{i},w_{i})$
and $(x_{m},w_{m})$, it will prove convenient to express the integral in terms
of $v_{i}$ and $v_{m}$. We do so as follows:
\begin{equation}
\frac{1}{\sigma_{N}^{3}N}\sum_{l}\left(  \int\nabla K_{\sigma_{N},\gamma
}(v_{im})^{\prime}G(v_{i},v_{m},r)f_{V}(v_{i})f_{V}(v_{m})\text{d}%
v_{i}\text{d}v_{m}\right)  \psi_{l,\beta}. \label{eq:A12}%
\end{equation}

Apply a change of variables in (\ref{eq:A12}) with $u_{im}=v_{im}/\sigma_{N}$
to obtain
\begin{align*}
&  \frac{1}{\sigma_{N}^{3}N}\sum_{l}\left(  \int\nabla K_{\sigma_{N},\gamma
}(v_{im})^{\prime}G(v_{i},v_{m},r)f_{V}(v_{i})f_{V}(v_{m})\text{d}%
v_{i}\text{d}v_{m}\right)  \psi_{l,\beta}\\
=  &  \frac{1}{\sigma_{N}N}\sum_{l}\left(  \int\nabla K_{\gamma}%
(u_{im})^{\prime}G(v_{m}+u_{im}\sigma_{N},v_{m},r)f_{V}(v_{m}+u_{im}\sigma
_{N})f_{V}(v_{m})\text{d}u_{im}\text{d}v_{m}\right)  \psi_{l,\beta}.
\end{align*}
As $\int\nabla K_{\gamma}(u)^{\prime}$d$u=0$ and $\int u_{j}\nabla
_{j}K_{\gamma}(u)$d$u=-1$ for $j=1,2$, a second-order expansion inside the
integral around $u_{im}\sigma_{N}=0$ yields the lead term\footnote{The
second-order term is of order $O_{p}(\delta_{N}\sigma_{N}/\sqrt{N})$, which
will be $o_{p}(\Vert r-\gamma\Vert^{2})$.}
\[
-\frac{1}{N}\sum_{l}\left(  \int\nabla_{1}\mu(v_{m},v_{m},r)f_{V}%
(v_{m})\text{d}v_{m}\right)  \psi_{l,\beta}.
\]
We next apply a second-order Taylor expansion of $\nabla_{1}\mu(v_{m}%
,v_{m},r)$ around $\gamma$ to yield
\[
(r-\gamma)^{\prime}\frac{1}{N}\sum_{l}\left(  -\int\nabla_{13}^{2}\mu
(v_{m},v_{m},\gamma)f_{V}(v_{m})\text{d}v_{m}\right)  \psi_{l,\beta}%
+o_{p}(\Vert r-\gamma\Vert^{2}),
\]
which concludes the proof of this lemma.
\end{proof}

\begin{proof}
[Proof of Lemma \ref{LE:L_N^K_beta}] Recall that $%
V_{im}\left( \beta \right) =\left( X_{im1}^{\prime }\beta ,X_{im2}^{\prime
}\beta \right) .$ The following calculation is useful:\textbf{\ }%
\begin{eqnarray}
&&K_{\sigma _{N},\gamma } ( V_{im} ( \hat{\beta} )  )
-K_{\sigma _{N},\gamma }\left( V_{im}\left( \beta \right) \right)
 \notag\\
&=&K\left( \frac{X_{im1}^{\prime }\hat{\beta}}{\sigma _{N}}\right) K\left(
\frac{X_{im2}^{\prime }\hat{\beta}}{\sigma _{N}}\right) -K\left( \frac{%
X_{im1}^{\prime }\beta }{\sigma _{N}}\right) K\left( \frac{X_{im2}^{\prime
}\beta }{\sigma _{N}}\right)  \notag \\
&=&K^{\prime }\left( \frac{X_{im1}^{\prime }\beta }{\sigma _{N}}\right)
K\left( \frac{X_{im2}^{\prime }\beta }{\sigma _{N}}\right) \frac{%
X_{im1}^{\prime } ( \hat{\beta}-\beta  ) }{\sigma _{N}}\notag\\
&&+K\left(
\frac{X_{im1}^{\prime }\beta }{\sigma _{N}}\right) K^{\prime }\left( \frac{%
X_{im2}^{\prime }\beta }{\sigma _{N}}\right) \frac{X_{im2}^{\prime } (
\hat{\beta}-\beta  ) }{\sigma _{N}}+O_{p} ( N^{-1}\sigma
_{N}^{-2} ) ,  \label{EQ:k_diff}
\end{eqnarray}%
where the last line applies the Taylor expansion. Suppose we have a random
vector $q_{im}$ such that $\mathbb{E}[q_{im}|V_{im}\left( \beta
\right) =v] $ exists and is $\kappa _{\gamma }$-th differentiable in $v,$
with derivatives up to $\kappa _{\gamma }$-th being uniformly bounded. We
make the following claim%
\begin{equation}
\frac{1}{N\left( N-1\right) }\sum_{i\neq m}K^{\prime }\left(
\frac{X_{im1}^{\prime }\beta }{\sigma _{N}}\right) K\left( \frac{%
X_{im2}^{\prime }\beta }{\sigma _{N}}\right) q_{im}=O_{p}(
N^{-1/2}\sigma _{N}^{2}),  \label{EQ:k_diff_claim}
\end{equation}%
whose proof is deferred to the end. Using (\ref{EQ:k_diff_claim}), we can write
\begin{eqnarray}
&&\frac{1}{\sigma _{N}^{2}N\left( N-1\right) }\sum_{i\neq m}K^{\prime
}\left( \frac{X_{im1}^{\prime }\beta }{\sigma _{N}}\right) K\left( \frac{%
X_{im2}^{\prime }\beta }{\sigma _{N}}\right) \frac{X_{im1}^{\prime }(
\hat{\beta}-\beta ) }{\sigma _{N}}Y_{im\left( 1,1\right) }\text{sgn}%
\left( W_{im}^{\prime }\gamma \right)  \notag \\
&=&\left[ \frac{1}{\sigma _{N}^{2}N\left( N-1\right) }\sum_{i\neq
m}K^{\prime }\left( \frac{X_{im1}^{\prime }\beta }{\sigma _{N}}\right)
K\left( \frac{X_{im2}^{\prime }\beta }{\sigma _{N}}\right) X_{im1}^{\prime
}Y_{im\left( 1,1\right) }\text{sgn}\left( W_{im}^{\prime }\gamma \right) %
\right] \frac{( \hat{\beta}-\beta) }{\sigma _{N}}  \notag \\
&=&O_{p}(N^{-1/2}( \hat{\beta}-\beta)\sigma _{N}^{-1}%
) =O_{p}( N^{-1}\sigma _{N}^{-1}) ,  \label{EQ:k_diff_1}
\end{eqnarray}%
where we let $q_{im}=X_{im1}^{\prime }Y_{im\left( 1,1\right) }$sgn$\left(
W_{im}^{\prime }\gamma \right)$ and use the result $\hat{\beta}-\beta
=O_{p} ( N^{-1/2} ) .$ Clearly, this also holds if we replace $%
X_{im1}^{\prime } ( \hat{\beta}-\beta  ) $ with $X_{im2}^{\prime
} ( \hat{\beta}-\beta  ) .$

Combine (\ref{EQ:k_diff}) and (\ref{EQ:k_diff_1}) to obtain%
\begin{eqnarray*}
&&\frac{1}{\sigma _{N}^{2}N\left( N-1\right) }\sum_{i\neq m}\left[ K_{\sigma
_{N},\gamma } ( V_{im} ( \hat{\beta} )  ) -K_{\sigma
_{N},\gamma }\left( V_{im}\left( \beta \right) \right) \right] Y_{im\left(
1,1\right) }\text{sgn}\left( W_{im}^{\prime }\gamma \right) \\
&=&2\cdot O_{p}\left( N^{-1}\sigma _{N}^{-1}\right) +O_{p}\left(
N^{-1}\sigma _{N}^{-2}\right) =O_{p}\left( N^{-1}\sigma _{N}^{-2}\right) ,
\end{eqnarray*}%
as desired.

The remaining task is to show the claim in (\ref{EQ:k_diff_claim}). Some
standard results on U-statistics imply that the variance of (\ref%
{EQ:k_diff_claim}) is $O_{p}( N^{-1})$. So we only need show the
expectation of the U-statistic is $O( N^{-1/2})$. Under Assumption C5, we can
calculate the expectation of the summand of (\ref{EQ:k_diff_claim}) as follows:%
\begin{eqnarray*}
&&\mathbb{E}\left[ K^{\prime }\left( \frac{X_{im1}^{\prime }\beta }{\sigma
_{N}}\right) K\left( \frac{X_{im2}^{\prime }\beta }{\sigma _{N}}\right)
q_{im}\right] =\mathbb{E}\left[ K^{\prime }\left( \frac{X_{im1}^{\prime }\beta }{\sigma
_{N}}\right) K\left( \frac{X_{im2}^{\prime }\beta }{\sigma _{N}}\right)
\mathbb{E}\left( q_{im}|V_{im}(\beta)
\right) \right] \\
&=&\int \int K^{\prime }\left( \frac{v_{1}}{\sigma _{N}}\right) K\left(
\frac{v_{2}}{\sigma _{N}}\right) \mathbb{E}[q_{im}|V_{im}(\beta)=(v_1,v_2)] f_{V_{im}(\beta)}( v_{1},v_{2}) dv_{1}dv_{2}\\
&=&\sigma_N^2\int \int K^{\prime }\left( u_{1}\right) K\left( u_{2}\right) \mathbb{E}%
[q_{im}|V_{im}(\beta)=(u_{1}\sigma _{N},u_{2}\sigma _{N})]f_{V_{im}(\beta)}( u_{1}\sigma _{N},u_{2}\sigma _{N}) du_{1}du_{2} \\
&=&\sigma_N^2\mathbb{E}%
[q_{im}|V_{im}(\beta)=(0,0)]f_{V_{im}(\beta)}(0, 0) \int \int K^{\prime }\left( u_{1}\right) K\left( u_{2}\right)
du_{1}du_{2}+O_{p} ( \sigma_N^4 ) \\
&=&0+O_{p} ( N^{-1/2} ) =O_{p} ( N^{-1/2} ),
\end{eqnarray*}%
where the last line holds by Assumptions C6 and C7. This shows the claim.
\end{proof}

\begin{proof}
[Proof of Lemma \ref{LE:L_N^K}] Applying similar arguments for obtaining (\ref{eq:A13}), we can write  
\begin{eqnarray*}
\mathcal{\hat{L}}_{N}^{K}\left( r\right) &=&\frac{1}{2}\left( r-\gamma
\right) ^{\prime }\mathbb{V}_{\gamma }\left( r-\gamma \right) +O_{p} (
N^{-1/2}\left\Vert r-\gamma \right\Vert  ) +o_{p} ( \left\Vert
r-\gamma \right\Vert ^{2} ) \\
&&+\frac{1}{\sigma _{N}^{2}N\left( N-1\right) }\sum_{i\neq m}K_{\sigma
_{N},\gamma } ( V_{im} ( \hat{\beta} )  ) Y_{im\left(
1,1\right) }\text{sgn}\left( W_{im}^{\prime }\gamma \right) +O_{p}\left(
N^{-1}\sigma _{N}^{-2}\right) .
\end{eqnarray*}%
Then, using the result in Lemma \ref{LE:L_N^K_beta}, we have
\begin{eqnarray*}
\mathcal{\hat{L}}_{N}^{K}\left( r\right) &=&\frac{1}{2}\left( r-\gamma
\right) ^{\prime }\mathbb{V}_{\gamma }\left( r-\gamma \right) +O_{p} (
N^{-1/2}\left\Vert r-\gamma \right\Vert  ) +o_{p} ( \left\Vert
r-\gamma \right\Vert ^{2} ) \\
&&+\frac{1}{\sigma _{N}^{2}N\left( N-1\right) }\sum_{i\neq m}K_{\sigma
_{N},\gamma }\left( V_{im}\left( \beta \right) \right) Y_{im\left(
1,1\right) }\text{sgn}\left( W_{im}^{\prime }\gamma \right) +o_{p} (
N^{-1/2} ) +O_{p} ( N^{-1}\sigma _{N}^{-2} ) .
\end{eqnarray*}%
Since $\ N^{1/2}\sigma _{N}^{2}\rightarrow \infty $ by Assumption C7,
$O_{p}\left( N^{-1}\sigma _{N}^{-2}\right) =o_{p}\left(
N^{-1/2}\right)$, which completes the proof.
\end{proof}

\begin{proof}
[Proof of Lemma \ref{LE:delta}] Using the definition that $%
V_{im}\left( \beta \right) =\left( X_{im1}^{\prime }\beta ,X_{im2}^{\prime
}\beta \right) $, we can write
\begin{eqnarray*}
&&\mathbb{E}\left[ \sigma _{N}^{-2}K_{\sigma _{N},\gamma }\left(
V_{im}\left( \beta \right) \right) Y_{im\left( 1,1\right) }\text{sgn}\left(
W_{im}^{\prime }\gamma \right) |Z_{i}\right] \\
&=&\mathbb{E}\left[ \left. \sigma _{N}^{-2}K\left( \frac{X_{im1}^{\prime
}\beta }{\sigma _{N}}\right) K\left( \frac{X_{im2}^{\prime }\beta }{\sigma
_{N}}\right) Y_{im\left( 1,1\right) }\text{sgn}\left( W_{im}^{\prime }\gamma
\right) \right\vert Z_{i}\right] \\
&=&\mathbb{E}\left[ \left. \sigma _{N}^{-2}K\left( \frac{X_{im1}^{\prime
}\beta }{\sigma _{N}}\right) K\left( \frac{X_{im2}^{\prime }\beta }{\sigma
_{N}}\right) \mathbb{E}\left[ Y_{im\left( 1,1\right) }\text{sgn}\left(
W_{im}^{\prime }\gamma \right) |Z_{i},X_{m1},X_{m2}\right] \right\vert Z_{i}%
\right] \\
&=&f_{V_{im}(\beta)}(0)\mathbb{E}\left[ Y_{im\left( 1,1\right) }\text{sgn}\left( W_{im}^{\prime
}\gamma \right) |Z_{i},V_{im}\left( \beta \right) =0%
\right] +o_{p} ( N^{-1/2} ),
\end{eqnarray*}%
where the last line uses the same argument as in proving Lemma \ref{lemmaA5}.
\end{proof}

\begin{proof}
[Proof of Lemma \ref{LE:P1}]We verify the conditions in Assumption M,
specifically, M.i, M.ii, and M.iii, in \cite{SeoOtsu2018} one by one. Recall
that
\begin{align*}
\phi_{Ni}\left(  b\right)   &  \equiv\sum_{t>s}\sum_{d\in\mathcal{D}%
}\{\mathcal{K}_{h_{N}}(X_{i2ts},W_{its})Y_{idst}\left(  -1\right)  ^{d_{1}%
}\left(  1[X_{i1ts}^{\prime}b>0]-1[X_{i1ts}^{\prime}\beta>0]\right)  \\
&  +\mathcal{K}_{h_{N}}(X_{i1ts},W_{its})Y_{idst}\left(  -1\right)  ^{d_{2}%
}\left(  1[X_{i2st}^{\prime}b>0]-1[X_{i2st}^{\prime}\beta>0]\right)  \},
\end{align*}
The \textquotedblleft$h_{n}$\textquotedblright\ in \cite{SeoOtsu2018} needs to
be set as $h_{N}^{k_{1}+k_{2}}$ for the term $\phi_{Ni}\left(  b\right)  $.
Similarly, the \textquotedblleft$h_{n}$\textquotedblright\ in
\cite{SeoOtsu2018} needs to be set as $\sigma_{N}^{2k_{1}}$ for the term
$\varphi_{Ni}\left(  r\right)  $. We conduct the analysis for $\phi
_{Ni}\left(  b\right)  $ first$,$ and the results for $\varphi_{Ni}\left(
r\right)  $ follow similarly. For $\phi_{Ni}\left(  b\right)  ,$ we analyze
the term
\[
\mathcal{K}_{h_{N}}(X_{i2ts},W_{its})Y_{idst}\left(  -1\right)  ^{d_{1}%
}\left(  1[X_{i1ts}^{\prime}b>0]-1[X_{i1ts}^{\prime}\beta>0]\right)  ,
\]
and the remaining terms can be analyzed similarly, thanks to the similar
structure of those terms.

\noindent\textbf{On Assumption M.i in} \cite{SeoOtsu2018}. Note that%
\begin{align}
&  \mathbb{E}[\mathcal{K}_{h_{N}}(X_{i2ts},W_{its})Y_{idst}\left(  -1\right)
^{d_{1}}\left(  1\left[  X_{i1ts}^{\prime}b>0\right]  -1\left[  X_{i1ts}%
^{\prime}\beta>0\right]  \right)  ]\label{EQ:TECH0}\\
= &  \int_{\mathbb{R}^{k_{2}}}\int_{\mathbb{R}^{k_{1}}}\mathbb{E}%
[Y_{idst}\left(  -1\right)  ^{d_{1}}\left(  1\left[  X_{i1ts}^{\prime
}b>0\right]  -1\left[  X_{i1ts}^{\prime}\beta>0\right]  \right)
|X_{i2ts}=x_{2},W_{its}=w]\nonumber\\
&  \cdot\mathcal{K}_{h_{N}}(x_{2},w)f_{X_{2ts},W_{ts}}\left(  x_{2},w\right)
\text{d}x_{2}\text{d}w\nonumber\\
= &  \mathbb{E}[Y_{idst}\left(  -1\right)  ^{d_{1}}\left(  1\left[
X_{i1ts}^{\prime}b>0\right]  -1\left[  X_{i1ts}^{\prime}\beta>0\right]
\right)  |X_{i2ts}=0,W_{its}=0]f_{X_{2ts},W_{ts}}\left(  0,0\right)
+O(h_{N}^{2}),\nonumber
\end{align}
where the second equality holds by change of variables and Taylor expansion,
and the small order term holds by the properties of the kernel function and
the boundedness conditions imposed in Assumptions P5--P7. Further, by
Assumption P8, the small order term satisfies $O(h_{N}^{2})=o((Nh_{N}%
^{k_{1}+k_{2}})^{-2/3}).$ As a result, we only need to focus on the lead term
\begin{align}
&  \mathbb{E}[Y_{idst}\left(  -1\right)  ^{d_{1}}\left(  1\left[
X_{i1ts}^{\prime}b>0\right]  -1\left[  X_{i1ts}^{\prime}\beta>0\right]
\right)  |X_{i2ts}=0,W_{its}=0]\nonumber\\
= &  \mathbb{E}[\mathbb{E}[Y_{idst}\left(  -1\right)  ^{d_{1}}|X_{i1ts}%
,X_{i2ts}=0,W_{its}=0]\left(  1\left[  X_{i1ts}^{\prime}b>0\right]  -1\left[
X_{i1ts}^{\prime}\beta>0\right]  \right)  |X_{i2ts}=0,W_{its}=0]\nonumber\\
\equiv &  \mathbb{E}[\kappa_{dts}^{\left(  1\right)  }\left(  X_{i1ts}\right)
\left(  1\left[  X_{i1ts}^{\prime}b>0\right]  -1\left[  X_{i1ts}^{\prime}%
\beta>0\right]  \right)  |X_{i2ts}=0,W_{its}=0]\nonumber\\
= &  \int_{\mathbb{R}^{k_{1}}}\kappa_{dts}^{\left(  1\right)  }\left(
x\right)  \left(  1\left[  x^{\prime}b>0\right]  -1\left[  x^{\prime}%
\beta>0\right]  \right)  f_{X_{1ts}|\left\{  X_{2ts}=0,W_{ts}=0\right\}
}\left(  x\right)  \text{d}x\nonumber\\
\equiv &  \int_{\mathbb{R}^{k_{1}}}\kappa_{dts}^{\left(  1\right)  }\left(
x\right)  \left(  1\left[  x^{\prime}b>0\right]  -1\left[  x^{\prime}%
\beta>0\right]  \right)  f_{X_{1ts}|\left\{  X_{2ts}=0,W_{ts}=0\right\}
}\left(  x\right)  \text{d}x,\label{EQ:TECH1}%
\end{align}
where to ease notations we denote
\[
\kappa_{dts}^{\left(  1\right)  }\left(  x\right)  =\mathbb{E}[Y_{idst}\left(
-1\right)  ^{d_{1}}|X_{i1ts}=x,X_{i2ts}=0,W_{its}=0]
\]
in the third line.

We now get the first and second derivatives of the above term w.r.t. $b$
around $\beta.$ Since the calculation is the classical differential geometry
and very similar to that in Sections 5 and 6.4 of \cite{KimPollard1990} and
Section B.1 of \cite{SeoOtsu2018}, we only present key steps and omit some
standard similar details. First, define the mapping%
\[
T_{b}= ( I-\left\Vert b\right\Vert _{2}^{-2}bb^{\prime} ) ( I-\beta
\beta^{\prime} ) +\left\Vert b\right\Vert _{2}^{-2}b\beta^{\prime}.
\]
Then $T_{b}$ maps $\left\{  X_{i1ts}^{\prime}b>0\right\}  $ onto $\left\{
X_{i1ts}^{\prime}\beta>0\right\}  $ and $\left\{  X_{i1ts}^{\prime
}b=0\right\}  $ onto $\left\{  X_{i1ts}^{\prime}\beta=0\right\}  .$ With
equation (\ref{EQ:TECH1}), equations (5.2) and (5.3) in \cite{KimPollard1990}
imply%
\begin{align*}
&  \frac{\partial}{\partial b}\mathbb{E}\left[  \left.  Y_{idst}\left(
-1\right)  ^{d_{1}}\left(  1\left[  X_{i1ts}^{\prime}b>0\right]  -1\left[
X_{i1ts}^{\prime}\beta>0\right]  \right)  \right\vert X_{i2ts}=0,W_{its}%
=0\right] \\
&  =\frac{\partial}{\partial b}\int_{\mathbb{R}^{k_{1}}}\kappa_{dts}^{\left(
1\right)  }\left(  x\right)  \left(  1\left[  x^{\prime}b>0\right]  -1\left[
x^{\prime}\beta>0\right]  \right)  f_{X_{1ts}|\left\{  X_{2ts}=0,W_{ts}%
=0\right\}  }\left(  x\right)  \text{d}x,\\
&  =\left\Vert b\right\Vert _{2}^{-2}b^{\prime}\beta\left(  I-\left\Vert
b\right\Vert _{2}^{-2}bb^{\prime}\right)  \int1\left[  x^{\prime}%
\beta=0\right]  \kappa_{dts}^{\left(  1\right)  }\left(  T_{b}x\right)
xf_{X_{1ts}|\left\{  X_{2ts}=0,W_{ts}=0\right\}  }\left(  T_{b}x\right)
\text{d}\sigma_{0}^{\left(  1\right)  },
\end{align*}
where $\sigma_{0}^{\left(  1\right)  }$ is the surface measure of $\left\{
X_{i1ts}^{\prime}\beta=0\right\}  .$ Note that $T_{\beta}x=x,$ then%
\[
1\left[  x^{\prime}\beta=0\right]  \kappa_{dts}^{\left(  1\right)  }\left(
T_{\beta}x\right)  =1\left[  x^{\prime}\beta=0\right]  \kappa_{dts}^{\left(
1\right)  }\left(  x\right)  =0
\]
because%
\begin{align*}
\left.  \kappa_{dts}^{\left(  1\right)  }\left(  X_{i1t}\right)  \right\vert
_{x^{\prime}\beta=0}  &  =\mathbb{E}\left[  \left.  Y_{idst}\left(  -1\right)
^{d_{1}}\right\vert X_{i1ts}^{\prime}\beta=0,X_{i2ts}=0,W_{its}=0\right] \\
&  =\mathbb{E}\left[  \left.  \left(  Y_{idt}-Y_{ids}\right)  \left(
-1\right)  ^{d_{1}}\right\vert X_{i1t}^{\prime}\beta=X_{i1s}^{\prime}%
\beta,X_{i2t}=X_{i2s},W_{it}=W_{is}\right] \\
&  =0,
\end{align*}
where the last line holds by Assumption P2 that $\xi_{s}\overset{d}{=}\xi_{t}|
( \alpha,Z^{T} ) .$ This implies that
\[
\left.  \frac{\partial}{\partial b}\mathbb{E}\left[  \left.  Y_{idst}\left(
-1\right)  ^{d_{1}}\left(  1\left[  X_{i1ts}^{\prime}b>0\right]  -1\left[
X_{i1ts}^{\prime}\beta>0\right]  \right)  \right\vert X_{i2ts}=0,W_{its}%
=0\right]  \right\vert _{b=\beta}=0.
\]
For the same reason, the nonzero component of the second order derivative at
$b=\beta$ only comes from the derivative of $\kappa_{dts}^{\left(  1\right)
}\left(  X_{i1ts}\right)  .$ Notice that $\left.  \frac{\partial\kappa
_{dts}^{\left(  1\right)  }\left(  T_{b}x\right)  }{\partial b}\right\vert
_{b=\beta}=-\left(  \frac{\partial\kappa_{dts}^{\left(  1\right)  }\left(
x\right)  }{\partial x}^{\prime}\beta\right)  x,$ we have%
\begin{align*}
&  \left.  \frac{\partial^{2}}{\partial b\partial b^{\prime}}\mathbb{E}\left[
\left.  Y_{idst}\left(  -1\right)  ^{d_{1}}\left(  1\left[  X_{i1ts}^{\prime
}b>0\right]  -1\left[  X_{i1ts}^{\prime}\beta>0\right]  \right)  \right\vert
X_{i2ts}=0,W_{its}=0\right]  \right\vert _{b=\beta}\\
=  & -\int1\left[  x^{\prime}\beta=0\right]  \left(  \frac{\partial
\kappa_{dts}^{\left(  1\right)  }\left(  x\right)  }{\partial x}^{\prime}%
\beta\right)  f_{X_{1ts}|\left\{  X_{2ts}=0,W_{ts}=0\right\}  }\left(
x\right)  xx^{\prime}\text{d}\sigma_{0}^{\left(  1\right)  } \equiv
\mathbb{V}_{dts}^{\left(  1\right)  }.
\end{align*}
With the first and second derivatives obtained, we can move on to apply the
Taylor expansion to the expectation term in equation (\ref{EQ:TECH1}) around
$b=\beta$\ as%
\begin{align*}
&  \mathbb{E}\left[  \left.  Y_{idst}\left(  -1\right)  ^{d_{1}}\left(
1\left[  X_{i1ts}^{\prime}b>0\right]  -1\left[  X_{i1ts}^{\prime}%
\beta>0\right]  \right)  \right\vert X_{i2ts}=0,W_{its}=0\right] \\
=  & \frac{1}{2}\left(  b-\beta\right)  ^{\prime}\mathbb{V}_{dts}^{\left(
1\right)  }\left(  b-\beta\right)  +o\left(  \left\Vert b-\beta\right\Vert
^{2}\right)  .
\end{align*}
Substitute this back to equation (\ref{EQ:TECH0}), we have%
\begin{align*}
&  \mathbb{E}\left[  \mathcal{K}_{h_{N}}(X_{i2ts},W_{its})Y_{idst}\left(
-1\right)  ^{d_{1}}\left(  1\left[  X_{i1ts}^{\prime}b>0\right]  -1\left[
X_{i1ts}^{\prime}\beta>0\right]  \right)  \right] \\
=  & \frac{1}{2}\left(  b-\beta\right)  ^{\prime}\mathbb{V}_{dts}^{\left(
1\right)  }\left(  b-\beta\right)  +o\left(  \left\Vert b-\beta\right\Vert
^{2}\right)  +o\left(  \left(  Nh_{N}^{k_{1}+k_{2}}\right)  ^{-2/3}\right)  .
\end{align*}

We can similarly define
\[
\mathbb{V}_{dts}^{\left(  2\right)  }\equiv-\int1\left[  x^{\prime}%
\beta=0\right]  \left(  \frac{\partial\kappa_{dts}^{\left(  2\right)  }\left(
x\right)  }{\partial x}^{\prime}\beta\right)  f_{X_{2ts}|\left\{
X_{1ts}=0,W_{ts}=0\right\}  }\left(  x\right)  xx^{\prime}\text{d}\sigma
_{0}^{\left(  2\right)  },
\]
where
\[
\kappa_{dts}^{\left(  2\right)  }\left(  x\right)  \equiv\mathbb{E}\left[
\left.  Y_{idst}\left(  -1\right)  ^{d_{2}}\right\vert X_{i2ts}=x,X_{i1ts}%
=0,W_{its}=0\right]  ,
\]
and $\sigma_{0}^{\left(  2\right)  }$ is the surface measure of $\left\{
X_{i2ts}^{\prime}\beta=0\right\}  $. A similar derivation yields
\begin{align*}
&  \mathbb{E}\left[  \mathcal{K}_{h_{N}}(X_{i1ts},W_{its})Y_{idst}\left(
-1\right)  ^{d_{2}}\left(  1\left[  X_{i2st}^{\prime}b>0\right]  -1\left[
X_{i2st}^{\prime}\beta>0\right]  \right)  \right] \\
=  & \frac{1}{2}\left(  b-\beta\right)  ^{\prime}\mathbb{V}_{dts}^{\left(
2\right)  }\left(  b-\beta\right)  +o\left(  \left\Vert b-\beta\right\Vert
^{2}\right)  +o\left(  \left(  Nh_{N}^{k_{1}+k_{2}}\right)  ^{-2/3}\right)  .
\end{align*}
Put everything so far together, we have%
\begin{equation}
\mathbb{E}\left[  \phi_{Ni}\left(  b\right)  \right]  =\frac{1}{2}\left(
b-\beta\right)  ^{\prime}\left[  \sum_{t>s}\sum_{d\in\mathcal{D}}\left(
\mathbb{V}_{dts}^{\left(  1\right)  }+\mathbb{V}_{dts}^{\left(  2\right)
}\right)  \right]  \left(  b-\beta\right)  +o ( \left\Vert b-\beta\right\Vert
^{2} ) +o ( ( Nh_{N}^{k_{1}+k_{2}} ) ^{-2/3}) . \label{EQ:phiNi}%
\end{equation}
Assumption M.i in \cite{SeoOtsu2018} is then verified by
\begin{align*}
\mathbb{V} \equiv & \sum_{t>s}\sum_{d\in\mathcal{D}}\left(  \mathbb{V}%
_{dts}^{\left(  1\right)  }+\mathbb{V}_{dts}^{\left(  2\right)  }\right)  .\\
=  & -\sum_{t>s}\sum_{d\in\mathcal{D}}\left\{  \int1\left[  x^{\prime}%
\beta=0\right]  \left(  \frac{\partial\kappa_{dts}^{\left(  1\right)  }\left(
x\right)  }{\partial x}^{\prime}\beta\right)  f_{X_{1ts}|\left\{
X_{2ts}=0,W_{ts}=0\right\}  }\left(  x\right)  xx^{\prime}\text{d}\sigma
_{0}^{\left(  1\right)  }\right. \\
&  \left.  +\int1\left[  x^{\prime}\beta=0\right]  \left(  \frac
{\partial\kappa_{dts}^{\left(  2\right)  }\left(  x\right)  }{\partial
x}^{\prime}\beta\right)  f_{X_{2ts}|\left\{  X_{1ts}=0,W_{ts}=0\right\}
}\left(  x\right)  xx^{\prime}\text{d}\sigma_{0}^{\left(  2\right)  }\right\}
.
\end{align*}

We can obtain similar results for $\varphi_{Ni}\left(  r\right)  $ using
exactly the same line of analysis, i.e.,%
\begin{equation}
\mathbb{E}\left[  \varphi_{Ni}\left(  r\right)  \right]  =\frac{1}{2}\left(
r-\gamma\right)  ^{\prime}\mathbb{W}\left(  r-\gamma\right)  +o ( \left\Vert
r-\gamma\right\Vert ^{2} ) +o( ( N\sigma_{N}^{2k_{1}} ) ^{-2/3}) ,
\label{EQ:phiNir}%
\end{equation}
where%
\[
\mathbb{W\equiv}-\sum_{t>s}\int1\left[  w^{\prime}\gamma=0\right]  \left(
\frac{\partial\kappa_{\left(  1,1\right)  ts}^{\left(  3\right)  }\left(
w\right)  }{\partial w}^{\prime}\gamma\right)  f_{W_{ts}|\left\{
X_{1ts}=0,X_{2ts}=0\right\}  }\left(  w\right)  ww^{\prime}\text{d}\sigma_{0}
\]
with
\[
\kappa_{\left(  1,1\right)  ts}^{\left(  3\right)  }\left(  w\right)
\equiv\mathbb{E}\left[  Y_{i(1,1)ts}|W_{its}=w,X_{i1ts}=0,X_{i2ts}=0\right]
\]
and $\sigma_{0}^{\left(  3\right)  }$ being the surface measure of $\left\{
W_{its}^{\prime}\gamma=0\right\}  $.

\noindent\textbf{On Assumption M.ii in }\cite{SeoOtsu2018}\textbf{. }Again, we
first verify this condition for%
\[
\mathcal{K}_{h_{N}}(X_{i2ts},W_{its})Y_{idst}\left(  -1\right)  ^{d_{1}%
}\left(  1\left[  X_{i1ts}^{\prime}b>0\right]  -1\left[  X_{i1ts}^{\prime
}\beta>0\right]  \right)  ,
\]
and other components in $\phi_{Ni}\left(  b\right)  $ follow similarly. We
evaluate norm of the difference of the above term at $b=b_{1}$ and $b=b_{2}$
multiplied by $h_{N}^{\left(  k_{1}+k_{2}\right)  /2}$:%
\begin{align}
&  h_{N}^{\left(  k_{1}+k_{2}\right)  /2}\left\Vert \mathcal{K}_{h_{N}%
}(X_{i2ts},W_{its})Y_{idst}\left(  -1\right)  ^{d_{1}}\left(  1\left[
X_{i1ts}^{\prime}b_{1}>0\right]  -1\left[  X_{i1ts}^{\prime}b_{2}>0\right]
\right)  \right\Vert \nonumber\\
=  & \left\{  \mathbb{E}\left[  \mathbb{E}\left[  \left.  h_{N}^{k_{1}+k_{2}%
}\mathcal{K}_{h_{N}}^{2}(X_{i2ts},W_{its})Y_{idst}^{2}\right\vert
X_{i1ts}\right]  \left(  1\left[  X_{i1ts}^{\prime}b_{1}>0\right]  -1\left[
X_{i1ts}^{\prime}b_{2}>0\right]  \right)  ^{2}\right]  \right\}
^{1/2}\nonumber\\
=  & \left\{  \mathbb{E}\left[  \mathbb{E}\left[  \left.  h_{N}^{k_{1}+k_{2}%
}\mathcal{K}_{h_{N}}^{2}(X_{i2ts},W_{its})Y_{idst}^{2}\right\vert
X_{i1ts}\right]  \left\vert 1\left[  X_{i1ts}^{\prime}b_{1}>0\right]
-1\left[  X_{i1ts}^{\prime}b_{2}>0\right]  \right\vert \right]  \right\}
^{1/2}\\
\geq &  C_{1}\mathbb{E}\left[  \left\vert 1\left[  X_{i1ts}^{\prime}%
b_{1}>0\right]  -1\left[  X_{i1ts}^{\prime}b_{2}>0\right]  \right\vert
\right]  \geq C_{2}\left\Vert b-\beta\right\Vert _{2},\nonumber
\end{align}
where the second equality holds by the fact that the difference of two
indicator functions can only take values $-1,0,$ or 1. $C_{1}$ and $C_{2}$ are
two positive constants. Applying the same analysis to all other terms in
$\phi_{Ni}\left(  b\right)  ,$ we have%
\[
h_{N}^{\left(  k_{1}+k_{2}\right)  /2}\left\Vert \phi_{Ni}\left(
b_{1}\right)  -\phi_{Ni}\left(  b_{2}\right)  \right\Vert \geq C_{3}\left\Vert
b_{1}-b_{2}\right\Vert
\]
for some positive constant $C_{3}.$

Applying similar analysis to $\varphi_{Ni}\left(  r\right)  $ leads to%
\[
\sigma_{N}^{k_{1}}\left\Vert \varphi_{Ni}\left(  r_{1}\right)  -\varphi
_{Ni}\left(  r_{2}\right)  \right\Vert \geq C_{4}\left\Vert r_{1}%
-r_{2}\right\Vert
\]
for some positive constant $C_{4}.$

\noindent\textbf{On Assumption M.iii in} \cite{SeoOtsu2018}\textbf{. }We still
begin with the analysis on%
\[
\mathcal{K}_{h_{N}}(X_{i2ts},W_{its})Y_{idst}\left(  -1\right)  ^{d_{1}%
}\left(  1\left[  X_{i1ts}^{\prime}b>0\right]  -1\left[  X_{i1ts}^{\prime
}\beta>0\right]  \right)  ,
\]
and other components in $\phi_{Ni}\left(  b\right)  $ follow similarly. We
evaluate the square of difference of the above term at $b=b_{1}$ and
$b=b_{2},$ such that $\left\Vert b_{1}-b_{2}\right\Vert <\varepsilon$,
multiplied by $h_{N}^{k_{1}+k_{2}}$:%
\begin{align*}
&  h_{N}^{k_{1}+k_{2}}\mathbb{E}\left[  \sup_{b_{1},b_{2}\in\mathcal{B}%
,\left\Vert b_{1}-b_{2}\right\Vert <\varepsilon}\left[  \mathcal{K}_{h_{N}%
}(X_{i2ts},W_{its})Y_{idst}\left(  -1\right)  ^{d_{1}}\left(  1\left[
X_{i1ts}^{\prime}b_{1}>0\right]  -1\left[  X_{i1ts}^{\prime}b_{2}>0\right]
\right)  \right]  ^{2}\right] \\
\leq &  \mathbb{E}\left[  \mathbb{E}\left[  \left.  h_{N}^{k_{1}+k_{2}%
}\mathcal{K}_{h_{N}}^{2}(X_{i2ts},W_{its})Y_{idst}^{2}\right\vert
X_{i1ts}\right]  \sup_{b_{1},b_{2}\in\mathcal{B},\left\Vert b_{1}%
-b_{2}\right\Vert <\varepsilon}\left\vert 1\left[  X_{i1ts}^{\prime}%
b_{1}>0\right]  -1\left[  X_{i1ts}^{\prime}b_{2}>0\right]  \right\vert \right]
\\
\leq &  C_{5}\mathbb{E}\left[  \sup_{b_{1},b_{2}\in\mathcal{B},\left\Vert
b_{1}-b_{2}\right\Vert <\varepsilon}\left\vert 1\left[  X_{i1ts}^{\prime}%
b_{1}>0\right]  -1\left[  X_{i1ts}^{\prime}b_{2}>0\right]  \right\vert
\right]  \leq C_{6}\varepsilon,
\end{align*}
where the second line follows by the fact that the maximum of absolute value
of the difference of two indicator functions is 1, and the last line holds by
Assumptions P4 and P5. Apply the same analysis to all other terms in
$\phi_{Ni}\left(  b\right) $ to obtain
\[
h_{N}^{k_{1}+k_{2}}\mathbb{E}\left[  \sup_{b_{1},b_{2}\in\mathcal{B}%
,\left\Vert b_{1}-b_{2}\right\Vert <\varepsilon}\left[  \phi_{Ni}\left(
b_{1}\right)  -\phi_{Ni}\left(  b_{2}\right)  \right]  ^{2}\right]  \leq
C_{7}\varepsilon
\]
for some positive constant $C_{7}.$

Similar analysis on $\varphi_{Ni}\left(  r\right)  $ leads to%
\[
h_{N}^{2k_{1}}\mathbb{E}\left[  \sup_{r_{1},r_{2}\in\mathcal{R},\left\Vert
r_{1}-r_{2}\right\Vert <\varepsilon}\left[  \varphi_{Ni}\left(  r_{1}\right)
-\varphi_{Ni}\left(  r_{2}\right)  \right]  ^{2}\right]  \leq C_{8}\varepsilon
\]
for some positive constant $C_{8}.$
\end{proof}

\begin{proof}
[Proof of Lemma \ref{LE:P2}]The first two equalities are direct results from
Lemma B.1 by Taylor expansions. Taking $b=\beta+\rho( Nh_{N}%
^{k_{1}+k_{2}} ) ^{-1/3}$ in equation (\ref{EQ:phiNi}) yields%
\[
\left(  Nh_{N}^{k_{1}+k_{2}}\right)  ^{2/3}\mathbb{E}\left[  \phi_{Ni}\left(
\beta+\rho\left(  Nh_{N}^{k_{1}+k_{2}}\right)  ^{-1/3}\right)  \right]
=\frac{1}{2}\rho^{\prime}\mathbb{V}\rho+o\left(  1\right)  \rightarrow\frac
{1}{2}\rho^{\prime}\mathbb{V}\rho.
\]
Similarly, setting $r=\gamma+\delta( N\sigma_{N}^{2k_{1}} ) ^{-1/3}$ in
equation (\ref{EQ:phiNir}) gives%
\[
\left(  N\sigma_{N}^{2k_{1}}\right)  ^{2/3}\mathbb{E}\left[  \varphi
_{Ni}\left(  \gamma+\delta\left(  N\sigma_{N}^{2k_{1}}\right)  ^{-1/3}\right)
\right]  =\frac{1}{2}\delta^{\prime}\mathbb{W\delta}+o\left(  1\right)
\rightarrow\frac{1}{2}\delta^{\prime}\mathbb{W\delta}.
\]

$\mathbb{H}_{1}$ and $\mathbb{H}_{2}$ can be obtained in the same way as in
\cite{KimPollard1990}. We omit some similar yet tedious details and refer
interested readers to Section 6.4 in \cite{KimPollard1990}. To get
$\mathbb{H}_{1},$ we let $\upsilon_{N}\equiv( Nh_{N}^{k_{1}+k_{2}} ) ^{1/3}$
and define%
\begin{align*}
\mathbb{L}\left(  \rho_{1}-\rho_{2}\right)   &  \equiv\lim_{N\rightarrow
\infty}\upsilon_{N}\mathbb{E}\left\{  h_{N}^{k_{1}+k_{2}}\left[  \phi
_{Ni}\left(  \beta+\rho_{1}\upsilon_{N}^{-1}\right)  -\phi_{Ni}\left(
\beta+\rho_{2}\upsilon_{N}^{-1}\right)  \right]  ^{2}\right\}  ,\\
\mathbb{L}\left(  \rho_{1}\right)   &  \equiv\lim_{N\rightarrow\infty}%
\upsilon_{N}\mathbb{E}\left\{  h_{N}^{k_{1}+k_{2}}\left[  \phi_{Ni}\left(
\beta+\rho_{1}\upsilon_{N}^{-1}\right)  -\phi_{Ni}\left(  \beta\right)
\right]  ^{2}\right\}  ,\text{ and}\\
\mathbb{L}\left(  \rho_{2}\right)   &  \equiv\lim_{N\rightarrow\infty}%
\upsilon_{N}\mathbb{E}\left\{  h_{N}^{k_{1}+k_{2}}\left[  \phi_{Ni}\left(
\beta+\rho_{2}\upsilon_{N}^{-1}\right)  -\phi_{Ni}\left(  \beta\right)
\right]  ^{2}\right\}  .
\end{align*}
Since $\phi_{Ni}\left(  \beta\right)  =0,$%
\begin{equation}
\mathbb{H}_{1}\left(  \rho_{1},\rho_{2}\right)  =\frac{1}{2}\left[
\mathbb{L}\left(  \rho_{1}\right)  +\mathbb{L}\left(  \rho_{2}\right)
-\mathbb{L}\left(  \rho_{1}-\rho_{2}\right)  \right]  . \label{EQ:H}%
\end{equation}
We calculate $\mathbb{L}$ as follows. Notice that
\begin{align*}
&  \mathbb{E}\left[  h_{N}^{k_{1}+k_{2}}\mathcal{K}_{h_{N}}(X_{ijts}%
,W_{its})^{2}\right] \\
=  & \int\frac{1}{h_{N}^{k_{1}+k_{2}}}\left[  \Pi_{\iota=1}^{k_{1}}K\left(
\frac{X_{ijts,\iota}}{h_{N}}\right)  \Pi_{\iota=1}^{k_{2}}K\left(
\frac{W_{its,\iota}}{h_{N}}\right)  \right]  ^{2}\text{d}F_{X_{ijts},W_{its}%
}=O(1),\\
&  \mathbb{E}\left[  h_{N}^{k_{1}+k_{2}}\mathcal{K}_{h_{N}}(X_{ijts}%
,W_{its})\mathcal{K}_{h_{N}}(X_{ij^{\prime}t^{\prime}s^{\prime}}%
,W_{it^{\prime}s^{\prime}})\right] \\
=  & \int\int\frac{1}{h_{N}^{k_{1}+k_{2}}}\Pi_{\iota=1}^{k_{1}}K\left(
\frac{X_{ijts,\iota}}{h_{N}}\right)  \Pi_{\iota=1}^{k_{2}}K\left(
\frac{W_{its,\iota}}{h_{N}}\right) \\
&  \cdot\Pi_{\iota=1}^{k_{1}}K\left(  \frac{X_{ij^{\prime}t^{\prime}s^{\prime
},\iota}}{h_{N}}\right)  \Pi_{\iota=1}^{k_{2}}K\left(  \frac{W_{it^{\prime
}s^{\prime},\iota}}{h_{N}}\right)  \text{d}F_{X_{ijts},W_{its}}\text{d}%
F_{X_{ij^{\prime}t^{\prime}s^{\prime}},W_{it^{\prime}s^{\prime}}}\\
=  & o\left(  1\right)  ,\text{ for any }\left(  j,t,s\right)  \neq\left(
j^{\prime},t^{\prime},s^{\prime}\right)  ,
\end{align*}
and%
\begin{align*}
&  \phi_{Ni}\left(  \beta+\rho_{1}\upsilon_{N}^{-1}\right)  -\phi_{Ni}\left(
\beta+\rho_{2}\upsilon_{N}^{-1}\right) \\
=  & \sum_{t>s}\sum_{d\in\mathcal{D}}\left\{  \mathcal{K}_{h_{N}}%
(X_{i2ts},W_{its})Y_{idst}\left(  -1\right)  ^{d_{1}}\left(  1\left[
X_{i1ts}^{\prime}\left(  \beta+\rho_{1}\upsilon_{N}^{-1}\right)  >0\right]
-1\left[  X_{i1ts}^{\prime}\left(  \beta+\rho_{2}\upsilon_{N}^{-1}\right)
>0\right]  \right)  \right. \\
&  \left.  +\mathcal{K}_{h_{N}}(X_{i1ts},W_{its})Y_{idst}\left(  -1\right)
^{d_{2}}\left(  1\left[  X_{i2st}^{\prime}\left(  \beta+\rho_{1}\upsilon
_{N}^{-1}\right)  >0\right]  -1\left[  X_{i2st}^{\prime}\left(  \beta+\rho
_{2}\upsilon_{N}^{-1}\right)  >0\right]  \right)  \right\}  .
\end{align*}
Thus, the interaction terms after expanding the square in $\mathbb{L}$'s are
asymptotically negligible, and we only need to focus on the squares of each
term above.

We now derive the square of the first term multiplied by $h_{N}^{k_{1}+k_{2}}$
and the rest follow similarly.
\begin{align*}
&  \mathbb{E}\left\{  h_{N}^{k_{1}+k_{2}}\left[  \mathcal{K}_{h_{N}}%
(X_{i2ts},W_{its})Y_{idst}\left(  -1\right)  ^{d_{1}}\left(  1\left[
X_{i1ts}^{\prime}\left(  \beta+\rho_{1}\upsilon_{N}^{-1}\right)  >0\right]
-1\left[  X_{i1ts}^{\prime}\left(  \beta+\rho_{2}\upsilon_{N}^{-1}\right)
>0\right]  \right)  \right]  ^{2}\right\} \\
= & \mathbb{E}\left\{  \mathbb{E}\left[  \left.  \left\vert Y_{idst}%
\right\vert \left\vert 1\left[  X_{i1ts}^{\prime}\left(  \beta+\rho
_{1}\upsilon_{N}^{-1}\right)  >0\right]  -1\left[  X_{i1ts}^{\prime}\left(
\beta+\rho_{2}\upsilon_{N}^{-1}\right)  >0\right]  \right\vert \text{
}\right\vert X_{i2ts},W_{its}\right]  h_{N}^{k_{1}+k_{2}}\mathcal{K}_{h_{N}%
}^{2}(X_{i2ts},W_{its})\right\} \\
= & \mathbb{E}\left[  \left.  \left\vert Y_{idst}\right\vert \left\vert
1\left[  X_{i1ts}^{\prime}\left(  \beta+\rho_{1}\upsilon_{N}^{-1}\right)
>0\right]  -1\left[  X_{i1ts}^{\prime}\left(  \beta+\rho_{2}\upsilon_{N}%
^{-1}\right)  >0\right]  \right\vert \text{ }\right\vert X_{i2ts}%
=0,W_{its}=0\right] \\
&  \cdot f_{X_{2ts},W_{ts}}\left(  0,0\right)  \left[  \int K^{2}\left(
u\right)  \text{d}u\right]  ^{k_{1}+k_{2}}+O\left(  h_{N}\right) \\
= & \mathbb{E}\left\{  \mathbb{E}\left[  \left\vert Y_{idst}\right\vert \text{
}|X_{i1ts},X_{i2ts}=0,W_{its}=0\right]  \right.  \text{ }\\
&  \cdot\left\vert 1\left[  X_{i1ts}^{\prime}\left(  \beta+\rho_{1}%
\upsilon_{N}^{-1}\right)  >0\right]  -1\left[  X_{i1ts}^{\prime}\left(
\beta+\rho_{2}\upsilon_{N}^{-1}\right)  >0\right]  \right\vert \text{
}|\left.  X_{i2ts}=0,W_{its}=0\right\} \\
&  \cdot f_{X_{2ts},W_{ts}}\left(  0,0\right)  \left[  \int K^{2}\left(
u\right)  \text{d}u\right]  ^{k_{1}+k_{2}}+O\left(  h_{N}\right) \\
= & \int\kappa_{dts}^{\left(  4\right)  }\left(  x\right)  \left\vert 1\left[
x^{\prime}\left(  \beta+\rho_{1}\upsilon_{N}^{-1}\right)  >0\right]  -1\left[
x^{\prime}\left(  \beta+\rho_{2}\upsilon_{N}^{-1}\right)  >0\right]
\right\vert f_{X_{1ts}|\left\{  X_{2ts}=0,W_{its}=0\right\}  }\left(
x\right)  \text{d}x\\
&  \cdot f_{X_{2ts},W_{ts}}\left(  0,0\right)  \bar{K}_{2}^{k_{1}+k_{2}%
}+O\left(  h_{N}\right) ,
\end{align*}
where $\kappa_{dts}^{\left(  4\right)  }\left(  x\right)  \equiv
\mathbb{E}\left[  \left\vert Y_{idst}\right\vert \text{ }|X_{i1ts}%
=x,X_{i2ts}=0,W_{its}=0\right] $ and $\bar{K}_{2}\equiv\int K^{2}\left(
u\right)  \text{d}u$. By Assumption P9, $\upsilon_{N}h_{N}= ( Nh_{N}%
^{k_{1}+k_{2}} ) ^{1/3}h_{N}\rightarrow0,$ so the bias term is asymptotic
negligible. To calculate the above integral, we follow \cite{KimPollard1990}
and decompose $X_{1ts}=a\beta+\bar{X}_{1ts},$ with $\bar{X}_{1ts}$ orthogonal
to $\beta.$ We use $f_{X_{1ts}}\left(  a,\bar{x}\right)  $ to denote the
density of $X_{1ts}$ at $X_{1ts}=a\beta+\bar{x}$. Using the results in
\cite{KimPollard1990},
\begin{align*}
&  \lim_{N\rightarrow\infty}\upsilon_{N}\cdot\int\kappa_{dts}^{\left(
4\right)  }\left(  x\right)  \left\vert 1\left[  x^{\prime}\left(  \beta
+\rho_{1}\upsilon_{N}^{-1}\right)  >0\right]  -1\left[  x^{\prime}\left(
\beta+\rho_{2}\upsilon_{N}^{-1}\right)  >0\right]  \right\vert f_{X_{1ts}%
|\left\{  X_{2ts}=0,W_{ts}=0\right\}  }\left(  x\right)  \text{d}x\\
=  & \int\kappa_{dts}^{\left(  4\right)  }\left(  \bar{x}\right)  \left\vert
\bar{x}^{\prime}\left(  \rho_{1}-\rho_{2}\right)  \right\vert f_{X_{1ts}%
|\left\{  X_{2ts}=0,W_{ts}=0\right\}  }\left(  0,\bar{x}\right)  \text{d}%
\bar{x}.
\end{align*}
To summarize, we have%
\begin{align*}
&  \lim_{N\rightarrow\infty}\upsilon_{N}\mathbb{E}\left\{  h_{N}^{k_{1}+k_{2}%
}\left[  \mathcal{K}_{h_{N}}(X_{i2ts},W_{its})Y_{idst}\left(  -1\right)
^{d_{1}}\left(  1\left[  X_{i1ts}^{\prime}\left(  \beta+\rho_{1}\upsilon
_{N}^{-1}\right)  >0\right]  -1\left[  X_{i1ts}^{\prime}\left(  \beta+\rho
_{2}\upsilon_{N}^{-1}\right)  >0\right]  \right)  \right]  ^{2}\right\} \\
=  & \int\kappa_{dts}^{\left(  4\right)  }\left(  \bar{x}\right)  \left\vert
\bar{x}^{\prime}\left(  \rho_{1}-\rho_{2}\right)  \right\vert f_{X_{1ts}%
|\left\{  X_{2ts}=0,W_{ts}=0\right\}  }\left(  0,\bar{x}\right)  \text{d}%
\bar{x}f_{X_{2ts},W_{ts}}\left(  0,0\right)  \bar{K}_{2}^{k_{1}+k_{2}}.
\end{align*}
Similarly, we have for the second term
\begin{align*}
&  \lim_{N\rightarrow\infty}\upsilon_{N}\mathbb{E}\left\{  h_{N}^{k_{1}+k_{2}%
}\left[  \mathcal{K}_{h_{N}}(X_{i1ts},W_{its})Y_{idst}\left(  -1\right)
^{d_{2}}\left(  1\left[  X_{i2st}^{\prime}\left(  \beta+\rho_{1}\upsilon
_{N}^{-1}\right)  >0\right]  -1\left[  X_{i2st}^{\prime}\left(  \beta+\rho
_{2}\upsilon_{N}^{-1}\right)  >0\right]  \right)  \right]  ^{2}\right\} \\
=  & \int\kappa_{dts}^{\left(  5\right)  }\left(  \bar{x}\right)  \left\vert
\bar{x}^{\prime}\left(  \rho_{1}-\rho_{2}\right)  \right\vert f_{X_{2ts}%
|\left\{  X_{1ts}=0,W_{ts}=0\right\}  }\left(  0,\bar{x}\right)  \text{d}%
\bar{x}f_{X_{1ts},W_{ts}}\left(  0,0\right)  \bar{K}_{2}^{k_{1}+k_{2}},
\end{align*}
where%
\[
\kappa_{dts}^{\left(  5\right)  }\left(  x\right)  \equiv\mathbb{E}\left[
\left\vert Y_{idst}\right\vert \text{ }|X_{i2ts}=x,X_{i1ts}=0,W_{its}%
=0\right]  .
\]
Therefore,
\begin{align*}
\mathbb{L}\left(  \rho_{1}-\rho_{2}\right)   &  =\lim_{N\rightarrow\infty
}\upsilon_{N}\mathbb{E}\left\{  h_{N}^{k_{1}+k_{2}}\left[  \phi_{Ni}\left(
\beta+\rho_{1}\upsilon_{N}^{-1}\right)  -\phi_{Ni}\left(  \beta+\rho
_{2}\upsilon_{N}^{-1}\right)  \right]  ^{2}\right\} \\
&  =\sum_{t>s}\sum_{d\in\mathcal{D}}\left\{  \int\kappa_{dts}^{\left(
4\right)  }\left(  \bar{x}\right)  \left\vert \bar{x}^{\prime}\left(  \rho
_{1}-\rho_{2}\right)  \right\vert f_{X_{1ts}|\left\{  X_{2ts}=0,W_{ts}%
=0\right\}  }\left(  0,\bar{x}\right)  \text{d}\bar{x}f_{X_{2ts},W_{ts}%
}\left(  0,0\right)  \right. \\
&  \left.  +\int\kappa_{dts}^{\left(  5\right)  }\left(  \bar{x}\right)
\left\vert \bar{x}^{\prime}\left(  \rho_{1}-\rho_{2}\right)  \right\vert
f_{X_{2ts}|\left\{  X_{1ts}=0,W_{ts}=0\right\}  }\left(  0,\bar{x}\right)
\text{d}\bar{x}f_{X_{1ts},W_{ts}}\left(  0,0\right)  \right\}  \bar{K}%
_{2}^{k_{1}+k_{2}}.
\end{align*}
Finally, by equation (\ref{EQ:H}),
\begin{align*}
&  \mathbb{H}_{1}\left(  \rho_{1},\rho_{2}\right) \\
=  & \frac{1}{2}\sum_{t>s}\sum_{d\in\mathcal{D}}\left\{  \int\kappa
_{dts}^{\left(  4\right)  }\left(  \bar{x}\right)  \left[  \left\vert \bar
{x}^{\prime}\rho_{1}\right\vert +\left\vert \bar{x}^{\prime}\rho
_{2}\right\vert -\left\vert \bar{x}^{\prime}\left(  \rho_{1}-\rho_{2}\right)
\right\vert \right]  f_{X_{1ts}|\left\{  X_{2ts}=0,W_{ts}=0\right\}  }\left(
0,\bar{x}\right)  \text{d}\bar{x}f_{X_{2ts},W_{ts}}\left(  0,0\right)  \right.
\\
&  +\left.  \int\kappa_{dts}^{\left(  5\right)  }\left(  \bar{x}\right)
\left[  \left\vert \bar{x}^{\prime}\rho_{1}\right\vert +\left\vert \bar
{x}^{\prime}\rho_{2}\right\vert -\left\vert \bar{x}^{\prime}\left(  \rho
_{1}-\rho_{2}\right)  \right\vert \right]  f_{X_{2ts}|\left\{  X_{1ts}%
=0,W_{ts}=0\right\}  }\left(  0,\bar{x}\right)  \text{d}\bar{x}f_{X_{1ts}%
,W_{ts}}\left(  0,0\right)  \right\}  \bar{K}_{2}^{k_{1}+k_{2}}.
\end{align*}

The same arguments can be applied to obtain $\mathbb{H}_{2}$ as
\begin{align*}
\mathbb{H}_{2}\left(  \delta_{1},\delta_{2}\right)  =  & \frac{1}{2}\sum
_{t>s}\int\kappa_{\left(  1,1\right)  ts}^{\left(  6\right)  }\left(  \bar
{w}\right)  \left[  \left\vert \bar{w}^{\prime}\delta_{1}\right\vert
+\left\vert \bar{w}^{\prime}\delta_{2}\right\vert -\left\vert \bar{w}^{\prime
}\left(  \delta_{1}-\delta_{2}\right)  \right\vert \right]  f_{W_{ts}|\left\{
X_{1ts}=0,X_{2ts}=0\right\}  }\left(  0,\bar{w}\right)  \text{d}\bar{w}\\
&  \cdot f_{X_{1ts},X_{2ts}}\left(  0,0\right)  \bar{K}_{2}^{2k_{1}},
\end{align*}
where $W_{ts}=a\gamma+\bar{W}_{ts}$ with $\bar{W}_{ts}$ orthogonal to $\gamma
$, $f_{W_{ts}}\left(  a,\bar{w}\right)  $ denotes the density of $W_{ts}$ at
$W_{ts}=a\gamma+\bar{W}_{ts},$ and
\[
\kappa_{\left(  1,1\right)  ts}^{\left(  6\right)  }\left(  w\right)
\equiv\mathbb{E}\left[  \left\vert Y_{i\left(  1,1\right)  ts}\right\vert
\text{ }|W_{its}=w,X_{i1ts}=0,X_{i2ts}=0\right]  \text{.}%
\]
\end{proof}

\end{appendix}

\bibliography{biblist}

\begin{thebibliography}{77}
\newcommand{\enquote}[1]{``#1''}
\expandafter\ifx\csname natexlab\endcsname\relax\def\natexlab#1{#1}\fi

\bibitem[\protect\citeauthoryear{Abrevaya, Hausman, and Khan}{Abrevaya et~al.}{2010}]{AbrevayaEtal2010}
\textsc{Abrevaya, J., J.~A. Hausman, and S.~Khan} (2010): \enquote{Testing for causal effects in a generalized regression model with endogenous regressors,} \emph{Econometrica}, 78, 2043--2061.

\bibitem[\protect\citeauthoryear{Abrevaya and Huang}{Abrevaya and Huang}{2005}]{AbrevayaHuang2005}
\textsc{Abrevaya, J. and J.~Huang} (2005): \enquote{On the bootstrap of the maximum score estimator,} \emph{Econometrica}, 73, 1175--1204.

\bibitem[\protect\citeauthoryear{Ahn, Ichimura, Powell, and Ruud}{Ahn et~al.}{2018}]{AhnEtal2018}
\textsc{Ahn, H., H.~Ichimura, J.~L. Powell, and P.~A. Ruud} (2018): \enquote{Simple estimators for invertible index models,} \emph{Journal of Business \& Economic Statistics}, 36, 1--10.

\bibitem[\protect\citeauthoryear{Allen and Rehbeck}{Allen and Rehbeck}{2019}]{AllenRehbeck2019}
\textsc{Allen, R. and J.~Rehbeck} (2019): \enquote{Identification with additively separable heterogeneity,} \emph{Econometrica}, 87, 1021--1054.

\bibitem[\protect\citeauthoryear{Allen and Rehbeck}{Allen and Rehbeck}{2022}]{allen2022latent}
---\hspace{-.1pt}---\hspace{-.1pt}--- (2022): \enquote{Latent complementarity in bundles models,} \emph{Journal of Econometrics}, 228, 322--341.

\bibitem[\protect\citeauthoryear{Altonji and Matzkin}{Altonji and Matzkin}{2005}]{AltonjiMatzkin2005}
\textsc{Altonji, J.~G. and R.~L. Matzkin} (2005): \enquote{Cross section and panel data estimators for nonseparable models with endogenous regressors,} \emph{Econometrica}, 73, 1053--1102.

\bibitem[\protect\citeauthoryear{Arcones and Gine}{Arcones and Gine}{1992}]{ArconesGine}
\textsc{Arcones, M.~A. and E.~Gine} (1992): \enquote{On the bootstrap of U and V statistics,} \emph{The Annals of Statistics}, 655--674.

\bibitem[\protect\citeauthoryear{Augereau, Greenstein, and Rysman}{Augereau et~al.}{2006}]{augereau2006coordination}
\textsc{Augereau, A., S.~Greenstein, and M.~Rysman} (2006): \enquote{Coordination versus differentiation in a standards war: 56K modems,} \emph{The RAND Journal of Economics}, 37, 887--909.

\bibitem[\protect\citeauthoryear{Berry, Khwaja, Kumar, Musalem, Wilbur, Allenby, Anand, Chintagunta, Hanemann, and Jeziorski}{Berry et~al.}{2014}]{berry2014structural}
\textsc{Berry, S., A.~Khwaja, V.~Kumar, A.~Musalem, K.~C. Wilbur, G.~Allenby, B.~Anand, P.~Chintagunta, W.~M. Hanemann, and P.~Jeziorski} (2014): \enquote{Structural models of complementary choices,} \emph{Marketing Letters}, 25, 245--256.

\bibitem[\protect\citeauthoryear{Berry and Haile}{Berry and Haile}{2009}]{berry2009nonparametric}
\textsc{Berry, S.~T. and P.~A. Haile} (2009): \enquote{Nonparametric Identification of Multinomial Choice Demand Models with Heterogeneous Consumers,} Working Paper 15276, National Bureau of Economic Research.

\bibitem[\protect\citeauthoryear{Cameron and Trivedi}{Cameron and Trivedi}{2005}]{CameronTrivedi2005}
\textsc{Cameron, A.~C. and P.~K. Trivedi} (2005): \emph{Microeconometrics: methods and applications}, Cambridge University Press.

\bibitem[\protect\citeauthoryear{Candelaria}{Candelaria}{2020}]{candelaria2020semiparametric}
\textsc{Candelaria, L.~E.} (2020): \enquote{A semiparametric network formation model with unobserved linear heterogeneity,} \emph{arXiv preprint arXiv:2007.05403}.

\bibitem[\protect\citeauthoryear{Cavanagh and Sherman}{Cavanagh and Sherman}{1998}]{CavanaghSherman1998}
\textsc{Cavanagh, C. and R.~P. Sherman} (1998): \enquote{Rank estimators for monotonic index models,} \emph{Journal of Econometrics}, 84, 351--382.

\bibitem[\protect\citeauthoryear{Chen, Linton, and {van Keilegom}}{Chen et~al.}{2003}]{ChenEtal2003}
\textsc{Chen, X., O.~Linton, and I.~{van Keilegom}} (2003): \enquote{Estimation of semiparametric models when the criterion function is not smooth,} \emph{Econometrica}, 71, 1591--1608.

\bibitem[\protect\citeauthoryear{Chernozhukov, Fern{\'a}ndez-Val, and Newey}{Chernozhukov et~al.}{2019}]{chernozhukov2019nonseparable}
\textsc{Chernozhukov, V., I.~Fern{\'a}ndez-Val, and W.~K. Newey} (2019): \enquote{Nonseparable multinomial choice models in cross-section and panel data,} \emph{Journal of Econometrics}, 211, 104--116.

\bibitem[\protect\citeauthoryear{Chung and Rao}{Chung and Rao}{2003}]{chung2003general}
\textsc{Chung, J. and V.~R. Rao} (2003): \enquote{A general choice model for bundles with multiple-category products: Application to market segmentation and optimal pricing for bundles,} \emph{Journal of Marketing Research}, 40, 115--130.

\bibitem[\protect\citeauthoryear{Devroye}{Devroye}{1978}]{devroye1978uniform}
\textsc{Devroye, L.} (1978): \enquote{The uniform convergence of nearest neighbor regression function estimators and their application in optimization,} \emph{IEEE Transactions on Information Theory}, 24, 142--151.

\bibitem[\protect\citeauthoryear{Dub{\'e}}{Dub{\'e}}{2004}]{dube2004multiple}
\textsc{Dub{\'e}, J.-P.} (2004): \enquote{Multiple discreteness and product differentiation: Demand for carbonated soft drinks,} \emph{Marketing Science}, 23, 66--81.

\bibitem[\protect\citeauthoryear{Dunker, Hoderlein, and Kaido}{Dunker et~al.}{2022}]{dunker2022nonparametric}
\textsc{Dunker, F., S.~Hoderlein, and H.~Kaido} (2022): \enquote{Nonparametric identification of random coefficients in endogenous and heterogeneous aggregate demand models,} \emph{arXiv preprint arXiv:2201.06140}.

\bibitem[\protect\citeauthoryear{Ershov, Lalibert{\'e}, Marcoux, and Orr}{Ershov et~al.}{2021}]{ershov2021estimating}
\textsc{Ershov, D., J.-W. Lalibert{\'e}, M.~Marcoux, and S.~Orr} (2021): \enquote{Estimating complementarity with large choice sets: An application to mergers,} \emph{Available at SSRN: https://ssrn.com/abstract=3802097}.

\bibitem[\protect\citeauthoryear{Fan}{Fan}{2013}]{Fan2013}
\textsc{Fan, Y.} (2013): \enquote{Ownership consolidation and product characteristics: A study of the US daily newspaper market,} \emph{American Economic Review}, 103, 1598--1628.

\bibitem[\protect\citeauthoryear{Foubert and Gijsbrechts}{Foubert and Gijsbrechts}{2007}]{foubert2007shopper}
\textsc{Foubert, B. and E.~Gijsbrechts} (2007): \enquote{Shopper response to bundle promotions for packaged goods,} \emph{Journal of Marketing Research}, 44, 647--662.

\bibitem[\protect\citeauthoryear{Fox}{Fox}{2007}]{Fox2007}
\textsc{Fox, J.~T.} (2007): \enquote{Semiparametric estimation of multinomial discrete-choice models using a subset of choices,} \emph{The RAND Journal of Economics}, 38, 1002--1019.

\bibitem[\protect\citeauthoryear{Fox and Lazzati}{Fox and Lazzati}{2017}]{FoxLazzati2017}
\textsc{Fox, J.~T. and N.~Lazzati} (2017): \enquote{A note on identification of discrete choice models for bundles and binary games,} \emph{Quantitative Economics}, 8, 1021--1036.

\bibitem[\protect\citeauthoryear{Fritsch and Guenther}{Fritsch and Guenther}{2019}]{fritsch2019package}
\textsc{Fritsch, S. and F.~Guenther} (2019): \enquote{Package ‘neuralnet’,} \emph{Training of Neural Networks}, 2, 30.

\bibitem[\protect\citeauthoryear{Gao and Li}{Gao and Li}{2020}]{gao2020robust}
\textsc{Gao, W.~Y. and M.~Li} (2020): \enquote{Robust semiparametric estimation in panel multinomial choice models,} \emph{arXiv preprint arXiv:2009.00085}.

\bibitem[\protect\citeauthoryear{Gao, Li, and Xu}{Gao et~al.}{2023}]{gao2023logical}
\textsc{Gao, W.~Y., M.~Li, and S.~Xu} (2023): \enquote{Logical differencing in dyadic network formation models with nontransferable utilities,} \emph{Journal of Econometrics}, 235, 302--324.

\bibitem[\protect\citeauthoryear{Gentzkow}{Gentzkow}{2007}]{Gentzkow2007}
\textsc{Gentzkow, M.} (2007): \enquote{Valuing new goods in a model with complementarity: Online newspapers,} \emph{American Economic Review}, 97, 713--744.

\bibitem[\protect\citeauthoryear{Gentzkow, Shapiro, and Sinkinson}{Gentzkow et~al.}{2014}]{gentzkow2014competition}
\textsc{Gentzkow, M., J.~M. Shapiro, and M.~Sinkinson} (2014): \enquote{Competition and ideological diversity: Historical evidence from us newspapers,} \emph{American Economic Review}, 104, 3073--3114.

\bibitem[\protect\citeauthoryear{Gin{\'e} and Zinn}{Gin{\'e} and Zinn}{1990}]{gine1990bootstrapping}
\textsc{Gin{\'e}, E. and J.~Zinn} (1990): \enquote{Bootstrapping general empirical measures,} \emph{The Annals of Probability}, 851--869.

\bibitem[\protect\citeauthoryear{Graham}{Graham}{2017}]{graham2017econometric}
\textsc{Graham, B.~S.} (2017): \enquote{An econometric model of network formation with degree heterogeneity,} \emph{Econometrica}, 85, 1033--1063.

\bibitem[\protect\citeauthoryear{Hall, Racine, and Li}{Hall et~al.}{2004}]{hall2004cross}
\textsc{Hall, P., J.~Racine, and Q.~Li} (2004): \enquote{Cross-validation and the estimation of conditional probability densities,} \emph{Journal of the American Statistical Association}, 99, 1015--1026.

\bibitem[\protect\citeauthoryear{Han}{Han}{1987}]{Han1987}
\textsc{Han, A.~K.} (1987): \enquote{Nonparametric analysis of a generalized regression model; the maximum rank correlation estimator,} \emph{Journal of Econometrics}, 35, 303--316.

\bibitem[\protect\citeauthoryear{Hansen}{Hansen}{2008}]{hansen2008uniform}
\textsc{Hansen, B.~E.} (2008): \enquote{Uniform convergence rates for kernel estimation with dependent data,} \emph{Econometric Theory}, 24, 726--748.

\bibitem[\protect\citeauthoryear{Hayfield and Racine}{Hayfield and Racine}{2008}]{hayfield2008nonparametric}
\textsc{Hayfield, T. and J.~S. Racine} (2008): \enquote{Nonparametric econometrics: The np package,} \emph{Journal of Statistical Software}, 27, 1--32.

\bibitem[\protect\citeauthoryear{Hendel}{Hendel}{1999}]{hendel1999estimating}
\textsc{Hendel, I.} (1999): \enquote{Estimating multiple-discrete choice models: An application to computerization returns,} \emph{The Review of Economic Studies}, 66, 423--446.

\bibitem[\protect\citeauthoryear{Hong and Li}{Hong and Li}{2020}]{HongLi2020}
\textsc{Hong, H. and J.~Li} (2020): \enquote{The numerical bootstrap,} \emph{The Annals of Statistics}, 48, 397--412.

\bibitem[\protect\citeauthoryear{Hong, Mahajan, and Nekipelov}{Hong et~al.}{2015}]{HongEtal2015}
\textsc{Hong, H., A.~Mahajan, and D.~Nekipelov} (2015): \enquote{Extremum estimation and numerical derivatives,} \emph{Journal of Econometrics}, 188, 250--263.

\bibitem[\protect\citeauthoryear{Horowitz}{Horowitz}{1992}]{Horowitz1992}
\textsc{Horowitz, J.~L.} (1992): \enquote{A smoothed maximum score estimator for the binary response model,} \emph{Econometrica}, 60, 505--531.

\bibitem[\protect\citeauthoryear{Horowitz}{Horowitz}{2001}]{Horowitz2001}
---\hspace{-.1pt}---\hspace{-.1pt}--- (2001): \enquote{The bootstrap,} in \emph{Handbook of Econometrics}, Elsevier, vol.~5, 3159--3228.

\bibitem[\protect\citeauthoryear{Iaria and Wang}{Iaria and Wang}{2021}]{iaria2021empirical}
\textsc{Iaria, A. and A.~Wang} (2021): \enquote{An empirical model of quantity discounts with large choice sets,} \emph{CEPR Discussion Paper No. DP16666, Available at SSRN: https://ssrn.com/abstract=3960270}.

\bibitem[\protect\citeauthoryear{Jin, Ying, and Wei}{Jin et~al.}{2001}]{JinEtal2001}
\textsc{Jin, Z., Z.~Ying, and L.~J. Wei} (2001): \enquote{A simple resampling method by perturbing the minimand,} \emph{Biometrika}, 88, 381--390.

\bibitem[\protect\citeauthoryear{Khan, Ouyang, and Tamer}{Khan et~al.}{2021}]{KOT2021}
\textsc{Khan, S., F.~Ouyang, and E.~Tamer} (2021): \enquote{Inference on semiparametric multinomial response models,} \emph{Quantitative Economics}, 12, 743--777.

\bibitem[\protect\citeauthoryear{Kim and Pollard}{Kim and Pollard}{1990}]{KimPollard1990}
\textsc{Kim, J. and D.~Pollard} (1990): \enquote{Cube root asymptotics,} \emph{The Annals of Statistics}, 18, 191--219.

\bibitem[\protect\citeauthoryear{Kim, Misra, and Shapiro}{Kim et~al.}{2020}]{KimEtal2020}
\textsc{Kim, Y., S.~Misra, and B.~Shapiro} (2020): \enquote{Valuing brand collaboration: Evidence from a natural experiment,} \emph{Available at SSRN: https://ssrn.com/abstract=3335833}.

\bibitem[\protect\citeauthoryear{Kohler and Langer}{Kohler and Langer}{2021}]{kohler2021rate}
\textsc{Kohler, M. and S.~Langer} (2021): \enquote{On the rate of convergence of fully connected deep neural network regression estimates,} \emph{The Annals of Statistics}, 49, 2231--2249.

\bibitem[\protect\citeauthoryear{Lee}{Lee}{1995}]{Lee1995}
\textsc{Lee, L.-f.} (1995): \enquote{Semiparametric maximum likelihood estimation of polychotomous and sequential choice models,} \emph{Journal of Econometrics}, 65, 381--428.

\bibitem[\protect\citeauthoryear{Lee and Pun}{Lee and Pun}{2006}]{LeePun2006}
\textsc{Lee, S. M.~S. and M.~C. Pun} (2006): \enquote{On m out of n bootstrapping for nonstandard m-estimation with nuisance parameters,} \emph{Journal of American Statistical Association}, 101, 1185--1197.

\bibitem[\protect\citeauthoryear{Lewbel}{Lewbel}{2000}]{Lewbel2000}
\textsc{Lewbel, A.} (2000): \enquote{Semiparametric qualitative response model estimation with unknown heteroscedasticity or instrumental variables,} \emph{Journal of Econometrics}, 97, 145--177.

\bibitem[\protect\citeauthoryear{Lewbel and Nesheim}{Lewbel and Nesheim}{2019}]{lewbel2019sparse}
\textsc{Lewbel, A. and L.~Nesheim} (2019): \enquote{Sparse demand systems: corners and complements,} Cemmap Working Paper CWP45/19, Institute for Fiscal Studies.

\bibitem[\protect\citeauthoryear{Li and Racine}{Li and Racine}{2007}]{LiRacine2007}
\textsc{Li, Q. and J.~S. Racine} (2007): \emph{Nonparametric econometrics: theory and practice}, Princeton University Press.

\bibitem[\protect\citeauthoryear{Liu, Chintagunta, and Zhu}{Liu et~al.}{2010}]{liu2010complementarities}
\textsc{Liu, H., P.~K. Chintagunta, and T.~Zhu} (2010): \enquote{Complementarities and the demand for home broadband internet services,} \emph{Marketing Science}, 29, 701--720.

\bibitem[\protect\citeauthoryear{Manski}{Manski}{1987}]{Manski1987}
\textsc{Manski, C.~F.} (1987): \enquote{Semiparametric analysis of random effects linear models from binary panel data,} \emph{Econometrica (1986-1998)}, 55, 357--362.

\bibitem[\protect\citeauthoryear{Manski and Sherman}{Manski and Sherman}{1980}]{manski1980empirical}
\textsc{Manski, C.~F. and L.~Sherman} (1980): \enquote{An empirical analysis of household choice among motor vehicles,} \emph{Transportation Research Part A: General}, 14, 349--366.

\bibitem[\protect\citeauthoryear{Mullen, Ardia, Gil, Windover, and Cline}{Mullen et~al.}{2011}]{mullen2011deoptim}
\textsc{Mullen, K., D.~Ardia, D.~L. Gil, D.~Windover, and J.~Cline} (2011): \enquote{DEoptim: An R package for global optimization by differential evolution,} \emph{Journal of Statistical Software}, 40, 1--26.

\bibitem[\protect\citeauthoryear{Nevo, Rubinfeld, and McCabe}{Nevo et~al.}{2005}]{nevo2005academic}
\textsc{Nevo, A., D.~L. Rubinfeld, and M.~McCabe} (2005): \enquote{Academic journal pricing and the demand of libraries,} \emph{American Economic Review}, 95, 447--452.

\bibitem[\protect\citeauthoryear{Newey and McFadden}{Newey and McFadden}{1994}]{NeweyMcFadden1994}
\textsc{Newey, W.~K. and D.~McFadden} (1994): \enquote{Large sample estimation and hypothesis testing,} \emph{Handbook of Econometrics}, 4, 2111--2245.

\bibitem[\protect\citeauthoryear{Nolan and Pollard}{Nolan and Pollard}{1987}]{NolanPollard1987}
\textsc{Nolan, D. and D.~Pollard} (1987): \enquote{U-processes: rates of convergence,} \emph{The Annals of Statistics}, 15, 780--799.

\bibitem[\protect\citeauthoryear{Ouyang, Yang, and Zhang}{Ouyang et~al.}{2020}]{OYZ_infinity}
\textsc{Ouyang, F., T.~T. Yang, and H.~Zhang} (2020): \enquote{Semiparametric identification and estimation of discrete choice models for bundles,} \emph{Economics Letters}, 109321.

\bibitem[\protect\citeauthoryear{Pakes and Pollard}{Pakes and Pollard}{1989}]{PakesPollard1989}
\textsc{Pakes, A. and D.~Pollard} (1989): \enquote{Simulation and the asymptotics of optimization estimators,} \emph{Econometrica}, 57, 1027--1057.

\bibitem[\protect\citeauthoryear{Pakes and Porter}{Pakes and Porter}{2021}]{PakesPorter2016}
\textsc{Pakes, A. and J.~Porter} (2021): \enquote{Moment Inequalities for Multinomial Choice with Fixed Effects,} \emph{Forthcoming at Quantitative Economics}.

\bibitem[\protect\citeauthoryear{Seo and Otsu}{Seo and Otsu}{2018}]{SeoOtsu2018}
\textsc{Seo, M.~H. and T.~Otsu} (2018): \enquote{Local M-estimation with discontinuous criterion for dependent and limited observations,} \emph{The Annals of Statistics}, 46, 344--369.

\bibitem[\protect\citeauthoryear{Serfling}{Serfling}{2009}]{Serfling2009}
\textsc{Serfling, R.~J.} (2009): \emph{Approximation theorems of mathematical statistics}, vol. 162, John Wiley \& Sons.

\bibitem[\protect\citeauthoryear{Sher and Kim}{Sher and Kim}{2014}]{SherKim2014}
\textsc{Sher, I. and K.~i. Kim} (2014): \enquote{Identifying combinatorial valuations from aggregate demand,} \emph{Journal of Economic Theory}, 153, 428--458.

\bibitem[\protect\citeauthoryear{Sherman}{Sherman}{1993}]{Sherman1993}
\textsc{Sherman, R.~P.} (1993): \enquote{The limiting distribution of the maximum rank correlation estimator,} \emph{Econometrica}, 61, 123--137.

\bibitem[\protect\citeauthoryear{Sherman}{Sherman}{1994{\natexlab{a}}}]{Sherman1994AoS}
---\hspace{-.1pt}---\hspace{-.1pt}--- (1994{\natexlab{a}}): \enquote{Maximal inequalities for degenerate U-processes with applications to optimization estimators,} \emph{The Annals of Statistics}, 439--459.

\bibitem[\protect\citeauthoryear{Sherman}{Sherman}{1994{\natexlab{b}}}]{Sherman1994ET}
---\hspace{-.1pt}---\hspace{-.1pt}--- (1994{\natexlab{b}}): \enquote{U-processes in the analysis of a generalized semiparametric regression estimator,} \emph{Econometric theory}, 10, 372--395.

\bibitem[\protect\citeauthoryear{Shi, Shum, and Song}{Shi et~al.}{2018}]{ShiEtal2018}
\textsc{Shi, X., M.~Shum, and W.~Song} (2018): \enquote{Estimating semiparametric panel multinomial choice models using cyclic monotonicity,} \emph{Econometrica}, 86, 737--761.

\bibitem[\protect\citeauthoryear{Song and Chintagunta}{Song and Chintagunta}{2006}]{song2006measuring}
\textsc{Song, I. and P.~K. Chintagunta} (2006): \enquote{Measuring cross-category price effects with aggregate store data,} \emph{Management Science}, 52, 1594--1609.

\bibitem[\protect\citeauthoryear{Storn and Price}{Storn and Price}{1997}]{StornPrice1997}
\textsc{Storn, R.~M. and K.~V. Price} (1997): \enquote{Differential evolution--a simple and efficient heuristic for global optimization over continuous spaces,} \emph{Journal of Global Optimization}, 11, 341--359.

\bibitem[\protect\citeauthoryear{Subbotin}{Subbotin}{2007}]{Subbotin2007}
\textsc{Subbotin, V.} (2007): \enquote{Asymptotic and bootstrap properties of rank regressions,} \emph{Available at SSRN: https://ssrn.com/abstract=1028548}.

\bibitem[\protect\citeauthoryear{Toth}{Toth}{2017}]{toth2017estimation}
\textsc{Toth, P.} (2017): \enquote{Semiparametric estimation in network formation models with homophily and degree heterogeneity,} \emph{Available at SSRN: https://ssrn.com/abstract=2988698}.

\bibitem[\protect\citeauthoryear{Train, McFadden, and Ben-Akiva}{Train et~al.}{1987}]{train1987demand}
\textsc{Train, K.~E., D.~L. McFadden, and M.~Ben-Akiva} (1987): \enquote{The demand for local telephone service: A fully discrete model of residential calling patterns and service choices,} \emph{The RAND Journal of Economics}, 109--123.

\bibitem[\protect\citeauthoryear{{v}an~der Vaart and Wellner}{{v}an~der Vaart and Wellner}{1996}]{vdVaartWellner1996}
\textsc{{v}an~der Vaart, A. and J.~Wellner} (1996): \emph{Weak convergence and empirical processes}, Springer, New York.

\bibitem[\protect\citeauthoryear{Wang}{Wang}{2023}]{wang2023testing}
\textsc{Wang, R.} (2023): \enquote{Testing and Identifying Substitution and Complementarity Patterns,} \emph{arXiv preprint arXiv:2304.00678}.

\bibitem[\protect\citeauthoryear{Wellner and Zhan}{Wellner and Zhan}{1996}]{WellnerZhan1996}
\textsc{Wellner, J.~A. and Y.~Zhan} (1996): \enquote{Bootstrapping Z-estimators,} \emph{University of Washington Department of Statistics Technical Report}, 308.

\bibitem[\protect\citeauthoryear{Yan and Yoo}{Yan and Yoo}{2019}]{YanYoo2019}
\textsc{Yan, J. and H.~I. Yoo} (2019): \enquote{Semiparametric estimation of the random utility model with rank-ordered choice data,} \emph{Journal of Econometrics}, 211, 414--438.

\end{thebibliography}
\end{document}